\documentclass[12pt]{article}
\usepackage{amsmath}
\usepackage{graphicx,psfrag,epsf}
\usepackage{enumerate}
\usepackage{natbib}
\usepackage{url} 
\usepackage{amsfonts}
\usepackage{amsmath}
\usepackage{amsthm}
\usepackage{float}
\usepackage{xcolor}
\usepackage{wrapfig}
\usepackage{algorithm, algpseudocode}
\usepackage[english]{babel}
\newtheorem{theorem}{Theorem}
\newtheorem{corollary}{Corollary}
\newtheorem{assumption}{Assumption}
\newtheorem{lemma}[theorem]{Lemma}
\newtheorem{definition}{Definition}
\newtheorem{proposition}[theorem]{Proposition}

\newcommand{\naive}{na\"{\i}ve}
\newcommand{\Naive}{Na\"{\i}ve}

\newcommand{\cS}{{\mathcal S}}

\newcommand{\dif}{\mathrm{d}}

\newcommand{\fwd}{\mathtt{f}}
\newcommand{\bck}{\mathtt{b}}

 \newcommand{\defeq}{\stackrel{\textup{\tiny def}}{=}}

\def\E{E}
\newcommand{\mW}{{|W^\downarrow}|}
\newcommand{\assign}{\doteq}

\def\naive{na\"{\i}ve}

\newcommand{\PP}{\text{PoissProc}}

\newcommand{\Mx}[1]{{\max_s |A_s(#1)|}}

\newcommand{\SM}{B}
\newcommand\numberthis{\addtocounter{equation}{1}\tag{\theequation}}

\newcommand{\lb}{\ell}
\newcommand{\ub}{u}
\newcommand{\algname}{Rao-Teh}

\def\ptheta{\nu}
\newcommand{\prior}{P}

\newcommand{\blind}{0}

\begin{document}

\def\spacingset#1{\renewcommand{\baselinestretch}%
{#1}\small\normalsize} \spacingset{1}


\if0\blind
{
  \title{\bf Efficient Parameter Sampling for Markov Jump Processes}
  \author{Boqian Zhang and Vinayak Rao, \\
          Department of Statistics, Purdue University, USA \\
          email: \texttt{zhan1977@purdue.edu, varao@purdue.edu}
  }
  \maketitle
} \fi

\if1\blind
{
  \bigskip
  \bigskip
  \bigskip
  \begin{center}
    {\LARGE\bf Efficient parameter sampling for Markov jump processes}
\end{center}
  \medskip
} \fi
\bigskip
\begin{abstract}
Markov jump processes  are continuous-time stochastic processes widely used in a variety of applied disciplines. 
Inference typically proceeds via Markov chain Monte Carlo, the state-of-the-art being a uniformization-based auxiliary variable Gibbs sampler. 
This was designed for situations where the process parameters are known, and Bayesian inference over unknown parameters is typically carried out by incorporating it into a larger Gibbs sampler.
This strategy of sampling parameters given path, and path given parameters can result in poor Markov chain mixing. 
In this work, we propose a simple and efficient algorithm to address this problem. 
Our scheme brings Metropolis-Hastings approaches for discrete-time hidden Markov models to the continuous-time setting, resulting in a complete and clean recipe for parameter and path inference in Markov jump processes. 
In our experiments, we demonstrate superior performance over Gibbs sampling, a more \naive\ Metropolis-Hastings algorithm, as well as another popular approach, particle Markov chain Monte Carlo.
We also show our sampler inherits geometric mixing from an `ideal' sampler that is computationally much more expensive.
Supplementary material for the article is available online.
\end{abstract}

\noindent%
\vfill

\newpage
\spacingset{1.5} 

\section{Introduction}
\label{sec:intro}
Markov jump processes (MJPs) are continuous-time stochastic processes widely used in fields like computational chemistry~\citep{gillespie97}, molecular genetics~\citep{FearnSher2006}, mathematical finance~\citep{Elliott06}, queuing theory~\citep{Breuer2003}, artificial intelligence~\citep{XuShe10} and social-network analysis~\citep{pan2016markov}. 
MJPs have been used as realistic, mechanistic and interpretable models of a wide variety of phenomena, among others, the references above have used them to model temporal evolution of the state of a chemical reaction or queuing network, segmentation of a strand of DNA, and user activity on social media.
Their continuous-time dynamics however raise computational challenges when, given noisy measurements, one wants to make inferences 
over the latent MJP trajectory as well as any process parameters. 
In contrast to {discrete-time} hidden Markov models, one cannot 
{\em a priori} bound the number of trajectory state transitions, and the transition times themselves are continuous-valued. 
The state-of-the-art approach is an auxiliary variable Gibbs sampler from~\cite{RaoTeh13}, we will refer to this as the {\algname} algorithm. 
This Markov chain Monte Carlo (MCMC) algorithm was designed to simulate paths when the MJP parameters are known. 
Parameter inference is typically carried out by incorporating it into a Gibbs sampler that also conditionally simulates parameters given the currently sampled trajectory. 

In many situations, the MJP trajectory and parameters exhibit strong coupling, so that alternately sampling path given parameters, and parameters given path can result in poor mixing.  
To address this, we propose an efficient Metropolis-Hastings (MH) sampler (algorithm~\ref{alg:MH_improved}). 
In our experiments, we demonstrate superior performance over Gibbs sampling, a more \naive\ MH sampler (algorithm~\ref{alg:MH_naive}), as well as particle Markov chain Monte Carlo~\citep{Andrieu10}. 
We also prove that under relatively mild conditions, our sampler inherits geometric ergodicity from an `ideal' sampler that is computationally much more expensive.

\section{Markov jump processes (MJPs)} 
\label{sec:mjp}
A Markov jump process~\citep{Cinlar1975} is a right-continuous piecewise-constant stochastic process taking values in a state space $\cS$. 
We assume a finite number of states $N$, with $\cS = \{1,\ldots,N\}$. 
Then, the MJP is parameterized by two quantities, an $N$-component probability vector $\pi_0$ and a rate-matrix $A$. 
The former gives the distribution over states at the initial time (we assume this is $0$), while the latter is an $N \times N$-matrix governing the dynamics of the system.  
An off-diagonal element $A_{ij}$ gives the rate of transitioning from state $i$ to $j$. 
The rows of $A$ sum to $0$, so that $A_{ii}=-\sum_{j \neq i} A_{ij}  $. 
We write $A_i$ for the negative of the $i$th diagonal element $A_{ii}$, so that $A_i = -A_{ii}$ gives the total rate at which the system leaves state $i$ for any other state.
To simulate an MJP over an interval $[0,t_{end})$, one follows Gillespie's algorithm~\citep{gillespie97}: 
first sample an initial state $s_0$ from $\pi_0$, and defining $t_0 = t_{curr} = 0$ and $k = 0$, repeat the following while $t_{curr} < t_{end}$:
\begin{itemize}
  \item Simulate a wait-time $\Delta t_k$ from an exponential distribution with rate $A_{s_k}$.  
    Set $t_{k+1} = t_{curr} = t_{k} + \Delta t_k$. 
    The MJP remains in state $s_k$ until time $t_{k+1}$.
  \item Jump to a new state $s_{k+1} \neq s_k$ with probability equal to $A_{s_ks_{k+1}}/A_{s_k}$. Set $k=k+1$.
\end{itemize}
The times $T=(t_1, \dotsc, t_{k - 1})$ and states $S=(s_1, \dotsc, s_{k - 1 })$, along with the initial state $s_0$, define the MJP path. 
We use both $(s_0,S,T)$ and  $\{S(t), t \in [0,t_{end})\}$ (and sometimes just $S(\cdot)$) to refer to the MJP path.
See the top-left panel in figure~\ref{fig:MH_improved} for a sample path.

\vspace{-.15in}
\subsection{Structured rate matrices}
While the rate matrix $A$ can have $N(N-1)$ independent elements, in typical applications, especially with large state-spaces, it is determined by a much smaller set of parameters. 
We will write these as $\theta$, with $A$ a deterministic function of these parameters: $A \equiv A(\theta)$. 
The parameters $\theta$ are often more interpretable than the elements of $A$, and correspond directly to physical, biological or environmental parameters of interest. 
For example:
\begin{description}
  \item[Immigration-death processes] 
    Here, $\theta = (\alpha,\beta)$, with $\alpha$ the arrival-rate and $\beta$ the death-rate. 
    The state represents the size of a population or queue. 
    New individuals enter with rate $\alpha$, so off-diagonal elements $A_{i,i+1}$ equal $\alpha$.
    Each individual dies at a rate $\beta$, so that $A_{i,i-1}=i\beta$ for each $i$.
    All other transitions have rate $0$. 
  \item[Birth-death processes] 
    This variant of the earlier MJP moves from state $i$ to $i+1$ with rate $i\alpha$, with growth-rate proportional to population size. 
    The death-rate is $\beta$, so that $A_{i,i-1}=i\beta$ for each $i$.
    Other off-diagonal elements are $0$, and again $\theta=(\alpha,\beta)$.
  \item[Codon substitution models] 
    These characterize transitions between codons at a DNA locus over evolutionary time. 
    There are $61$ codons, and in the simplest case, all transitions have the same rate~\citep{jukescantor69}: $A_{ij} = \alpha\ \forall i \neq j$. 
    Thus the $61\times 61$ matrix $A$ is determined by a single $\alpha$. 
    Other models~\citep{goldman1994codon} group transitions as `synonymous' and `nonsynonymous', based on whether old and new codons encode the same amino acid. 
    Synonymous and nonsynonymous transitions have their own rates, so $A$ is now determined by 2 parameters $\alpha$ and $\beta$. 
\end{description}

 \section{Bayesian modeling and inference for MJPs}
\label{sec:bayes_model}
We first set up our Bayesian model of the data generation process. 
We model a latent piecewise-constant path $S(\cdot)$ over $[0,t_{end})$ as an $N$-state MJP with rate matrix $A(\theta)$ and prior $\pi_0$ over $s_0 = S(0)$, the state at time $0$. 
We place a prior $\prior(\theta)$ over the unknown $\theta$. 
For simplicity, we assume $\pi_0$ is known (or we set it to a uniform distribution over the $N$ states). 
We have noisy measurements $X$ of the latent process, with likelihood $P(X|\{S(t),\ t \in [0,t_{end})\})$.
Again, for clarity we ignore any unknown parameters in the likelihood, else we can include them in $\theta$.
We assume the observation process has the following structure: for fixed $X$, for any partition $\tilde{W} = \{\tilde{w}_1 = 0, \dotsc, \tilde{w}_{|\tilde{W}|}=t_{end}\}$ of the interval $[0,t_{end})$ (where $|\cdot|$ denotes cardinality), there exist known functions $\ell_i$ such that the likelihood factors as:
\begin{align}
  \label{eq:lik_factor}
  P(X|\{S(t),\ t \in [0,t_{end})\}) = \prod_{i=1}^{|\tilde{W}|-1} \ell_i(\{S(t),\ t \in [\tilde{w}_{i},\tilde{w}_{i+1})\})
\end{align}
A common example is a finite set of independent observations $X = \{x_1,\dotsc,x_{|X|}\}$ at times $T^X = \{t^X_1,\dotsc, t^X_{|X|}\}$, each observation depending on the state of the MJP at that time:
\vspace{-.1in}
\begin{align}
  \label{eq:lik_iid}
  P(X|\{S(t),\ t \in [0,t_{end})\}) = \prod_{i=1}^{|X|} P(x_i|S(t^X_i)).
\end{align}
Other examples are an inhomogeneous Poisson process~\citep{FearnSher2006}, renewal process~\citep{rao2011gaussian} or even another MJP~\citep{Nodelman+al:UAI02,RaoTeh13}, modulated by $(s_0, S, T)$.
The first example, called a Markov modulated Poisson process (MMPP)~\citep{scottmmpp03}, associates a positive rate $\lambda_s$ with each state $s$, with $\ell_i(\{S(t),\ t \in [w_{i},w_{i+1})\})$ equal to the likelihood of the Poisson events within $[w_{i},w_{i+1})$ under an inhomogeneous Poisson process with piecewise-constant rate $\{\lambda_{S(t)},\ t \in [w_{i},w_{i+1})\}$.

With $A(\cdot)$ and $\pi_0$ assumed known, the overall Bayesian model is then 
\vspace{-.1in}
\begin{align}
  \label{eq:bayes_model}
  \theta \sim P(\theta), \quad (s_0, S, T) \sim \text{MJP}(\pi_0, A(\theta)), \quad X \sim P(X|s_0,S,T).
\end{align}
Given $X$, one is interested in the posterior distribution over the latent quantities, $(\theta,s_0, S, T)$. 

\vspace{-.1in}
\subsection{Trajectory inference given the MJP parameters $\theta$}
This was addressed in~\citet{RaoTeh13}  and extended to a broader class of jump processes in~\cite{RaoTeh12}~\citep[also see][]{FearnSher2006, Hobolth09, Elhaygibbssampling}).
\citet{RaoTeh13,RaoTeh12} both involve MJP path representations with auxiliary {\em candidate} jump times that are later {\em thinned}.  
We focus on the former, a simpler and more popular algorithm, based on the idea of {\em uniformization}~\citep{Jen1953}. 

Uniformization involves a parameter $\Omega(\theta) \ge \max_i A_i(\theta)$; \cite{RaoTeh13} suggest $\Omega(\theta) = 2 \max_i A_i(\theta)$. 
Define $B(\theta) = \left(I +\frac{1}{\Omega(\theta)}A(\theta)\right)$; this is a stochastic matrix with nonnegative elements, and rows adding up to $1$.
Unlike the sequential wait-and-jump Gillespie algorithm, uniformization first simulates a random grid of candidate transition-times $W$ over $[0,t_{end})$, and then assigns these state values:
\begin{itemize}
  \item Simulate $W$ from a Poisson process with rate $\Omega(\theta) \ge \max_i A_i(\theta)$: 
    $W \sim \text{PoissProc}(\Omega(\theta))$.
  \item Assign states $(v_0,V)$ to the times $0 \cup W$, with $v_0 \sim \pi_0$, and $P(v_{i+1}=s|v_i) = B_{v_is}(\theta)$.
\end{itemize}
Setting $\Omega(\theta) > \max_i A_i(\theta)$ results in more candidate-times than actual MJP transitions, at the same time, unlike $A(\theta)$, the matrix $B(\theta)$ can thin these through self-transitions. 
Write $U$ for the elements $W$ with self-transitions, and $T$ for the rest.
Define $s_0=v_0$, and $S=\{v_i \in V \text{ s.t.\ } v_i \neq v_{i-1}\}$ as the elements in $V$ corresponding to $T$, then $(s_0,S,T)$ sampled this way for any $\Omega(\theta) \ge \max_i A_i(\theta)$
has the same distribution as under Gillespie's algorithm~\citep{Jen1953,RaoTeh13}. The third panel in figure~\ref{fig:MH_improved} shows these sets.

Introducing the thinned variables allowed~\cite{RaoTeh13} to develop an efficient MCMC sampler (algorithm~\ref{alg:Unif_gibbs}). 
At a high-level, each MCMC iteration simulates a new grid $W$ conditioned on the path $(s_0,S,T)$, and then a new path conditioned on $W$. 
\cite{RaoTeh13} show that the resulting Markov chain targets the desired posterior distribution over trajectories, and is ergodic for any $\Omega(\theta)$ strictly greater than all the $A_i(\theta)$'s. 
\begin{algorithm}[H]
  \caption{The~\cite{RaoTeh13} MCMC sampler for MJP trajectories}
   \label{alg:Unif_gibbs}
  \begin{tabular}{l l}
   \textbf{Input:  } & \text{Prior $\pi_0$, observations $X$}, 
                       \text{the previous path $(s_0, S, T)$}.\\ 
                     & \text{Parameter $\Omega(\theta) > \max_i A_i(\theta)$}, where
                     $A(\theta)$ is the MJP rate-matrix.\\
   \textbf{Output:  }& \text{New MJP trajectory $(s'_0, S', T')$}.\\
   \hline
   \end{tabular}
   \begin{algorithmic}[1]
\State \textbf{ Simulate the thinned candidate times $U$ given the MJP path $(s_0, S,T)$ } 
from a piecewise-constant rate-$(\Omega(\theta)-A_{S(\cdot)}(\theta))$ Poisson process: 

$ U \sim \text{PoissProc}(\Omega(\theta) - A_{S(t)}(\theta)), \quad t \in [0,t_{end}).$ 
\State \textbf{
  Discard the states $(s_0,S)$, and write 
  $W = T \cup U$}.

  \State \textbf{ Simulate states $(v_0,V)$ on $0 \cup W$ from a discrete-time HMM} 
  with initial distribution over states $\pi_0$ and transition matrix $B(\theta) = \left(I+\frac{1}{\Omega(\theta)}A(\theta)\right)$.
Following equation~\eqref{eq:lik_factor}, between two consecutive times $(\tilde{w}_i,\tilde{w}_{i+1})$ in $\tilde{W} \defeq 0 \cup W \cup t_{end}$, state $s$ has 
likelihood $\ell_i(s) \equiv \ell_i(\{S(t) = s,\ t \in [\tilde{w}_i,\tilde{w}_{i+1})\})$. The simulation involves two steps: 
\begin{description}
  \item[Forward pass:] 
    Set $\fwd_0(\cdot) = \pi_0$.
    Sequentially update $\fwd_i(\cdot)$ at time ${w}_i \in {W}$ given $\fwd_{i-1}$: 
        \vspace{-.1in}
        $$\textbf{for } i=1\rightarrow |{W}|\textbf{ do:} \quad \fwd_i(s') = \sum_{s \in \cS} \fwd_{i-1}(s)\cdot \ell_{i}(s) \cdot B_{ss'}(\theta), \quad \forall s' \in \cS.\qquad\qquad\quad $$
  \item[Backward pass:]
    Simulate ${v}_{|{W}|} \sim \bck_{|{W}|}(\cdot)$, where $\bck_{|{W}|}(s) \propto \fwd_{|{W}|}(s)\cdot\ell_{|{W}|+1}(s) \quad \forall s \in \cS.$ 
    $$ \textbf{for } i=(|{W}|-1)\rightarrow 0\textbf{ do:} \quad {v}_i \sim \bck_i(\cdot),\ \  \text{where } 
    \bck_i(s) \propto \fwd_i(s)\cdot B_{sv_{i+1}}(\theta) \cdot \ell_{i+1}(s)  \quad \forall s \in \cS.$$
\end{description}
        \vspace{-.1in}
\State \textbf{Discard self-transitions}: Set $s'_0 = v_0$. Let $T'$ be the set of times in ${W}$ when $V$ changes state. Define $S'$ as the corresponding set of state values. Return $(s'_0, S', T')$.
\end{algorithmic}
\end{algorithm}

\subsection{Joint inference over MJP path $(s_0, S, T)$ and parameters $\theta$}
For fixed parameters $\theta$, the efficiency of the Rao-Teh algorithm has been established, both empirically~\citep{RaoTeh13} and theoretically~\citep{miasojedow2017}.
In practice, the parameters are typically unknown, and often, these are of primary interest. 
One then has to characterize the complete posterior $P(\theta, s_0, S, T|X)$ of the Bayesian model of equation~\eqref{eq:bayes_model}. 
This is typically carried out by incorporating the previous algorithm into a Gibbs sampler that targets the joint $P(\theta, s_0, S, T|X)$ by conditionally simulating $(s_0, S, T)$ given $\theta$ and then $\theta$ given $(s_0, S, T)$. 
Algorithm~\ref{alg:MJP_gibbs}~\citep[see also][]{RaoTeh13} outlines this:
\begin{algorithm}[H]
  \caption{Gibbs sampling for path and parameter inference for MJPs}
   \label{alg:MJP_gibbs}
  \begin{tabular}{l l}
   \textbf{Input:  } 
                      & \text{The current MJP path $(s_0, S, T)$, the current MJP parameters $\theta$}.\\ 
   \textbf{Output:  }& \text{New MJP trajectory $(s'_0, S', T')$ and 
                            parameters $\theta'$}.\\
   \hline
   \end{tabular}
   \begin{algorithmic}[1]
  \State  Simulate a new path from the conditional 
  $P(s'_0, S', T'|X,s_0,S,T,\theta)$ by 
  algorithm~\ref{alg:Unif_gibbs}.
  \State Simulate a new parameter $\theta'$ from the conditional 
  $P(\theta'|X,s'_0,S',T')$ (see equation~\eqref{eq:param_cond}).
   \end{algorithmic}
\end{algorithm} 
\vspace{-.1in}
The distribution $P(\theta'|X,s'_0,S',T')$ depends on 
the amount of time $\tau_i$ spent in each state $i$, and the number of transitions $c_{ij}$ between each pair of states $i,j$: 
\begin{align}
  \label{eq:param_cond}
  P(\theta'|X,s'_0,S',T') \propto P(\theta') \prod_{i \in \cS} \exp(-A_i(\theta')\tau_i) 
  \prod_{j \in \cS} \left(\frac{A_{ij}(\theta')}{A_i(\theta')}\right)^{c_{ij}}.
\end{align}
In some circumstances, this can be directly sampled from, otherwise, one has to use a Markov kernel like Metropolis-Hastings to update $\theta$ to $\theta'$. 
In any event, this introduces no new technical challenges.
  \begin{figure}
  \centering
  \begin{minipage}[hp]{0.3\linewidth}
  \centering
    \vspace{-0 in}
    \includegraphics [width=0.98\textwidth, angle=0]{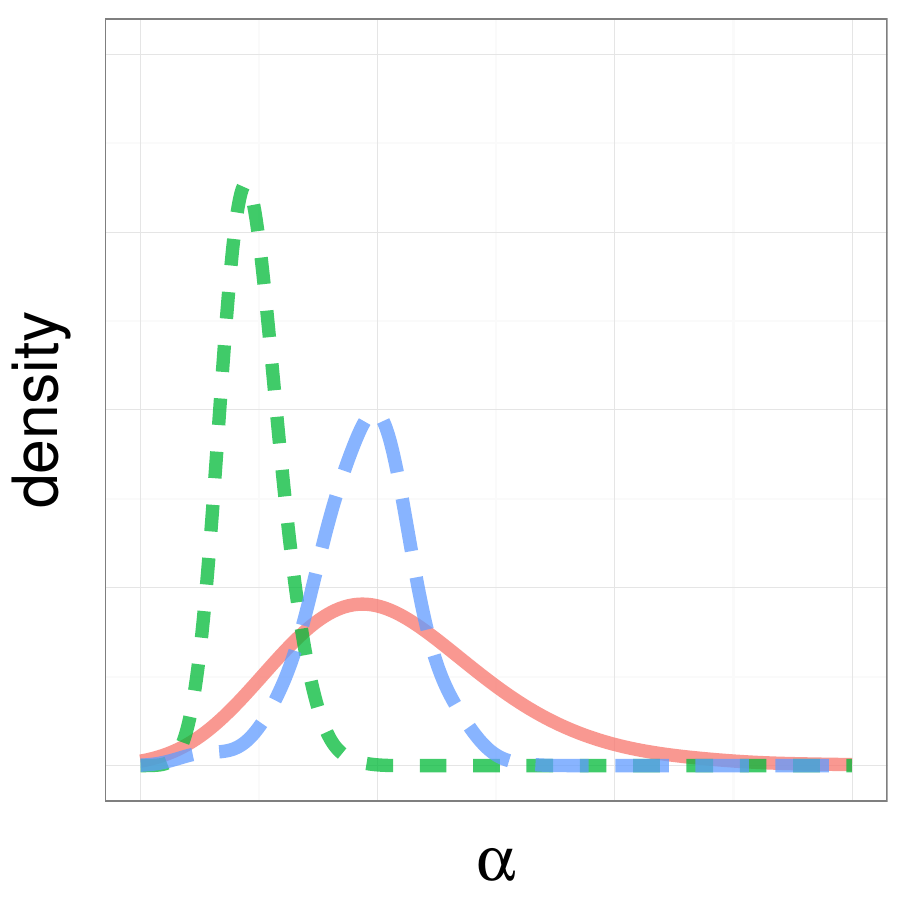}
   \vspace{0.06 in}
  \end{minipage}
  \begin{minipage}[!hp]{0.3\linewidth}
  \centering
    \includegraphics [width=0.98\textwidth, angle=0]{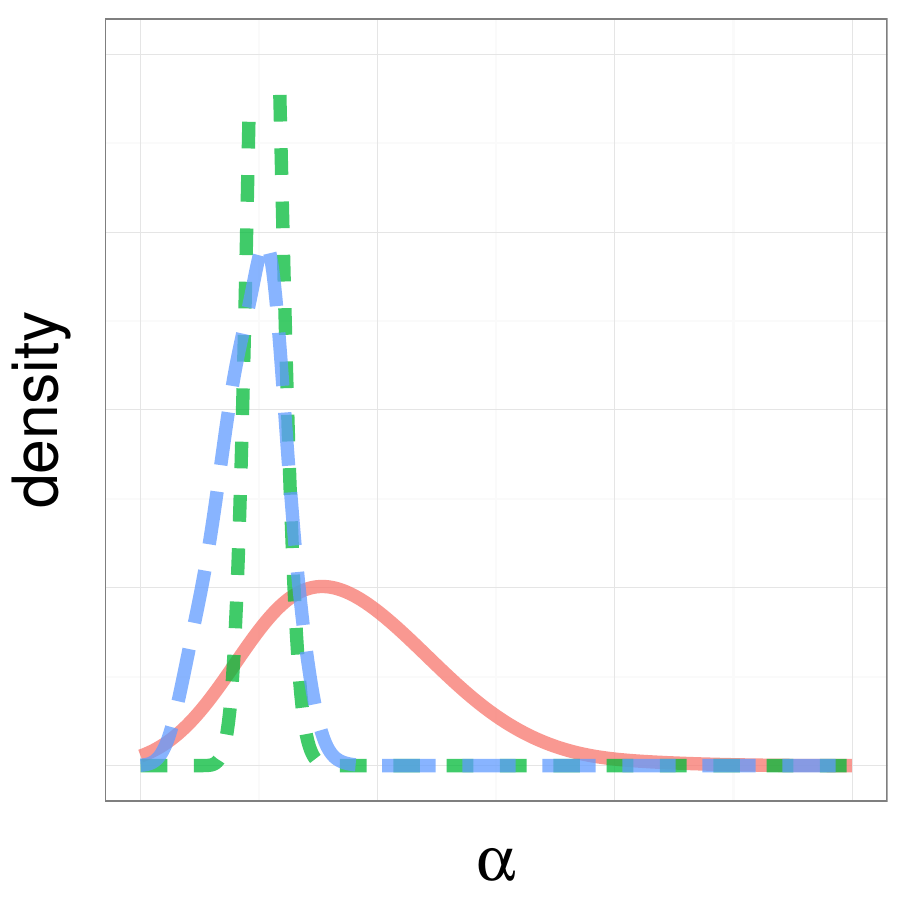}
    \vspace{-0 in}
  \end{minipage}
  \vspace{-.3in}
  \caption{Prior density over an MJP parameter (solid curve), along with two conditionals: given observations only (long-dashes), and given observations and MJP path (short-dashes). 
    As $t_{end}$ increases from $10$ (left) to $100$ (right), the conditionals become more concentrated, implying stronger path-parameter coupling. 
  The plots are from section~\ref{sec:immig} with 3 states.}
     \label{fig:hist}
  \end{figure}
  However, the resulting Gibbs sampler can mix very poorly because of coupling between path and parameters.
  We illustrate this in figure~\ref{fig:hist}~\citep[inspired by][]{papaspiliopoulos2007general}, which shows the posterior distribution of an MJP parameter (long-dashes) is less concentrated than the distribution conditioned on both observations as well as path (short-dashes). 
  The coupling is strengthened as the trajectory grows longer (right panel), and the Gibbs sampler can mix very poorly with long observation periods, even if the observations themselves are only mildly informative about the parameters. Before we describe our actual algorithm, we outline a \naive\ attempt around this coupling.

 \vspace{-.2in}
\section{\Naive\ parameter inference via Metropolis-Hastings}
For discrete-time HMMs, path-parameter coupling can be circumvented by marginalizing out the Markov trajectory, and directly sampling from the marginal posterior $P(\theta|X)$.
In its simplest form, this involves a Metropolis-Hastings (MH) scheme that proposes a new parameter $\vartheta$ from a proposal distribution $q(\vartheta|\theta)$, accepting or rejecting according to the usual MH probability.
The marginal probabilities over $X$ given parameters are computed using the forward-filtering backward-sampling (FFBS) algorithm~\citep{fruhwirth1994data,carter1996markov, RaoTeh13}.
The Rao-Teh algorithm, which recasts posterior simulation for continuous-time models as discrete-time simulation on a random grid, then provides a simple way to incorporate such an MH-scheme into continuous-time settings: directly update $\theta$, conditioning on the random grid $W$, but marginalizing out $(v_0, V)$.

Specifically, given $\theta$ and the Poisson grid $W$, rather than simulating new path values (the backward pass in algorithm~\ref{alg:Unif_gibbs}), and then conditionally updating $\theta$ (the second step in algorithm~\ref{alg:MJP_gibbs}), we {\em first} propose a parameter $\vartheta$ from $q(\vartheta|\theta)$. This is accepted with probability 
$$ \texttt{acc} = \min\left(1, 
\frac{P(X|W,\vartheta) P(W|\vartheta)P(\vartheta)q(\theta|\vartheta)} {P(X|W,\theta) P(W|\theta)P(\theta)q(\vartheta|\theta)}\right),$$ 
thereby targeting the distribution $P(\theta|W,X)$.
In the equation above, $P(X|W,\theta)$ is the probability of the observations $X$ given $W$ with $(v_0,V)$ marginalized out. 
Uniformization says this is the marginal probability of $X$ under a discrete-time HMM on $W$, with transition matrix $B(\theta)$. This can be computed using the forward pass of FFBS algorithm (steps 4 and 6 of algorithm~\ref{alg:MH_naive} below). 
The term $P(W|\theta)$ is the probability of $W$ under a rate-$\Omega(\theta)$ Poisson process. 
These, and the corresponding terms for $\vartheta$ allow the acceptance probability to be computed.
Only {\em after} accepting or rejecting $\vartheta$ do we simulate new states $(v'_0,V')$, using the new parameter $\theta'$ in a backward pass over $W$. 
The new trajectory and parameter are used to simulate a new grid $W'$, and the process is repeated.
Algorithm~\ref{alg:MH_naive} includes all details of this algorithm (see also figure~\ref{fig:naive_mh} in the supplementary material).

\begin{algorithm}[H]
   \caption{\Naive\  MH for parameter inference for MJPs }
   \label{alg:MH_naive}
  \begin{tabular}{l l}
   \textbf{Input:  } & \text{Observations $X$}, 
                       \text{the MJP path $(s_0, S, T)$, the  MJP parameters $\theta$ }and $\pi_0$.\\ 
   \textbf{Output:  }& \text{A new MJP trajectory $(s'_0, S', T')$, 
                            new MJP parameter $\theta'$}.\\
   \hline
   \end{tabular}
   \begin{algorithmic}[1]
     \State Set $\Omega(\theta) > \max_s{A_s(\theta)}$ for
     some function $\Omega(\cdot)$, e.g.\ $\Omega(\theta) = 
      2\max_s A_s(\theta)$.
      \State \textbf{Simulate the thinned times $U$} from a rate-$(\Omega(\theta)-A_{S(\cdot)}(\theta))$ Poisson process: 

      $\qquad \qquad \qquad \qquad U \sim \text{PoissProc}(\Omega(\theta) - A_{S(t)}(\theta)), \quad t\in[0,t_{end})$.

      \State 
      Set $W = T \cup U$ and discard $(s_0,S)$. Define $\tilde{W} = 0 \cup W \cup t_{end}$.
    \State 
    \textbf{Forward pass:}
    Set $B(\theta) = I + \frac{1}{\Omega(\theta)}A(\theta)$ and
    $\fwd^\theta_0(\cdot) = \pi_0$. Recall $\ell_i(\cdot)$ from equation~\eqref{eq:lik_factor}.
\vspace{-.25in}
$$\textbf{for } i=1\rightarrow |{W}|\textbf{ do:} \quad \fwd^{\theta}_i(s') = \sum_{s \in \cS} \fwd^\theta_{i-1}(s)\cdot \ell_{i}(s) \cdot B_{ss'}(\theta), \quad \forall s' \in \cS.\qquad\qquad\quad $$
    \State \textbf{Propose $\vartheta \sim q(\cdot| \theta)$.}
    For elements of ${W}$, calculate $\fwd^\vartheta_i(\cdot)$ similar to above.
      \State \textbf{Accept/Reject:} 
      For $\theta$ (and similarly for $\vartheta$), set  
      $P(W|\theta) = \Omega(\theta)^{|W|}\exp(-\Omega(\theta)t_{end})$, 
      $P(X|W,\theta) = \sum_{s \in \cS} \fwd_{|{W}|}^\theta(s)\cdot\ell_{|{W}|+1}(s)$.
      With probability $\texttt{acc}$, set $\theta' = \vartheta$, else $\theta'=\theta$; 
          \begin{align}
            \label{eq:ncp_acc}
            \texttt{acc} &=  1 \wedge \frac{P(\vartheta|W, X)}{P(\theta|W, X)} \frac{q(\theta|\vartheta)}{q(\vartheta|\theta)}
          =  1 \wedge \frac{P(X| W,\vartheta) P(W | \vartheta)P(\vartheta)}
            {P(X|W, \theta)P(W | \theta)P(\theta)} \frac{q(\theta|\vartheta)}{q(\vartheta|\theta)}.
          \end{align}
    \State 
    \textbf{Backward pass:}
    Simulate $v_{|W|} \sim \bck^{\theta'}_{|W|}(\cdot)$, where $\bck^{\theta'}_{|W|}(s) \propto \fwd^{\theta'}_{|W|}(s)\cdot\ell_{|W|+1}(s) \quad \forall s \in \cS.$ 
\vspace{-.25in}
    $$ \textbf{for } i=(|W|-1)\rightarrow 0\textbf{ do:} \quad v_i \sim \bck^{\theta'}_i(\cdot),\ \ \text{where } 
    \bck^{\theta'}_i(s) \propto \fwd^{\theta'}_i(s)\cdot B_{sv_{i+1}}(\theta') \cdot \ell_{i+1}(s)  \ \forall s \in \cS.$$
    \State Set $s'_0=v_0$. Let $T'$ be the set of times in $W$ when $V$ changes state. Define $S'$ as the corresponding set of state values. Return $(s'_0, S', T', \theta')$.
\end{algorithmic}
\end{algorithm}
\vspace{-.1in}
The resulting MCMC algorithm updates $\theta$ with the MJP trajectory 
integrated out, and by instantiating less `missing' information, can be expected to mix better. 
This can be quantified by the so-called Bayesian fraction of missing information~\citep{liu1994fraction, papaspiliopoulos2007general}. 
%
We note that even with the state values $(v_0,V)$ marginalized out, $\theta$ is updated {\em conditioned on $W$}. 
The distribution of $W$ depends on $\theta$: $W$ follows a rate-$\Omega(\theta)$ Poisson process. This dependence manifests in the $P(W|\theta)$ and $P(W|\vartheta)$ terms in equation~\eqref{eq:ncp_acc}. 
The fact that the MH-acceptance involves the probability of the observations  $X$ is inevitable, however the $P(W|\theta)$ term is an artifact of the computational algorithm of Rao-Teh. 
In our experiments, we show that this term significantly hurts acceptance probabilities and mixing. 
For a given $\theta$, $|W|$ is Poisson distributed with mean and variance $\Omega(\theta)$. 
If the proposed $\vartheta$ is such that $\Omega(\vartheta)$ is half $\Omega(\theta)$, then the ratio $P(W|\vartheta)/P(W|\theta)$ will be small, and $\vartheta$ is unlikely to be accepted. 
The next section describes our main algorithm that gets around this.

 \section{An improved Metropolis-Hasting algorithm}
The algorithm we propose symmetrizes the probability of $W$ under the old and new parameters, so that $P(W|\theta)$ disappears from the acceptance ratio. 
Now, the probability of accepting a proposal $\vartheta$ will depend only on the prior probabilities of $\theta$ and $\vartheta$, as well as how well they both explain the data given $W$.
This is in contrast to the previous algorithm, where one must also factor in how well each parameter explains the current value of the grid $W$.
This results in a MCMC sampler that mixes significantly more rapidly. 
Since we need not account for the probabilities $P(W|\theta)$, we also have a simpler MCMC scheme.

As before, the MCMC iteration begins with $(s_0, S, T, \theta)$. 
Instead of simulating the thinned events $U$ like earlier algorithms, we {\em first} generate a new parameter $\vartheta$ from some distribution $q(\vartheta|\theta)$. 
Treat this as an auxiliary variable, so that the augmented space now is $(s_0,S, T, \theta,\vartheta)$. 
Define a function $\Omega(\theta,\vartheta) > \max_s A_s(\theta)$ that is symmetric in its arguments (the number of arguments will distinguish $\Omega(\cdot,\cdot)$ from $\Omega(\cdot)$ of the earlier sections).
Two examples are $\Omega(\theta,\vartheta) = \kappa \max_s A_s(\theta) + \kappa \max_s A_s(\vartheta)$, for $\kappa \ge 1$, and 
$\Omega(\theta,\vartheta) = \kappa \max\left(\max_s A_s(\theta), \max_s A_s(\vartheta)\right)$, for $\kappa > 1$.

We will treat the path $(s_0,S,T)$ as simulated by  uniformization, but now with the dominating Poisson rate equal to $\Omega(\theta,\vartheta)$  instead of $\Omega(\theta)$. 
The transition matrix $B(\theta,\vartheta)$ of the embedded Markov chain is $B(\theta,\vartheta) = I + \frac{1}{\Omega(\theta,\vartheta)}A(\theta)$, so that the resulting trajectory $(s_0,S,T)$ will still be a realization from a MJP with rate-matrix $A(\theta)$.

Following the Rao-Teh algorithm, the conditional distribution of the thinned events $U$ given $(s_0,S,T,\theta,\vartheta)$ is a piecewise-constant Poisson with rate $\Omega(\theta, \vartheta) - A_{S(t)}(\theta), t \in [0,t_{end})$. 
This reconstructs the set $W = U \cup T$,  and as we saw~\citep[see also][]{RaoTeh13}, $P(W|\theta,\vartheta)$ is a homogeneous Poisson process with rate $\Omega(\theta, \vartheta)$. 
Having imputed $W$, discard the state values, so that the MCMC state space is $(W, \theta, \vartheta)$.
Now, propose swapping $\theta$ with $\vartheta$. 
From the symmetry of $\Omega(\cdot,\cdot)$, the Poisson grid $W$ has the same probability both
before and after this proposal, so unlike the previous scheme, the ratio 
equals $1$.  
This simplifies computation, and as suggested in the previous section, can significantly improve mixing.
An acceptance probability of
$ 
  \min\left(1, \frac{P(X|W,\vartheta,\theta) P(\vartheta) q(\theta|\vartheta)}
   {P(X|W,\theta,\vartheta) P(\theta)q(\vartheta|\theta)}\right)
   $ 
   targets the conditional $P(W,\theta,\vartheta|X) \propto P(\theta)q(\vartheta|\theta)P(W,X|\theta,\vartheta)$.
   The terms $P(X|\vartheta)$ and  $P(X|\theta)$ can be calculated from the forward pass of FFBS, and after
   accepting or rejecting the proposal, a new trajectory is sampled by
   completing the backward pass. Finally, the thinned events and auxiliary parameter are
   discarded. Algorithm~\ref{alg:MH_improved} and 
   figure~\ref{fig:MH_improved} outline the details of these steps. 
\begin{algorithm}[H]
   \caption{Symmetrized MH for parameter inference for MJPs }
   \label{alg:MH_improved}
  \begin{tabular}{l l}
   \textbf{Input:  } & \text{The observations $X$,}
                      \text{the MJP path $(s_0, S, T)$, MJP parameters $\theta$} and $\pi_0$.\\ 
   \textbf{Output:  }& \text{A new MJP trajectory $(s'_0, S', T')$, 
                            new MJP parameters $\theta'$}.\\
   \hline
   \end{tabular}
   \begin{algorithmic}[1]
     \State \textbf{Sample $\vartheta \sim q(\cdot| \theta)$}, and 
      set 
	$\Omega \assign \Omega(\theta,\vartheta)$ 
    for some symmetric $\Omega(\theta,\vartheta) > \max_s A_s(\theta)$.
    \State \textbf{Simulate the thinned times $U$} from a rate-$(\Omega-A_{S(\cdot)}(\theta))$ Poisson process:

    \vspace{-.05in}
    $\qquad \qquad \qquad \qquad U \sim \text{PoissProc}(\Omega - A_{S(t)}(\theta)), \quad t \in [0,t_{end})$.
    \State Set $W = T \cup U$ and discard $(s_0,S)$. Define $\tilde{W} = 0 \cup W \cup t_{end}$.
    \State \textbf{Forward pass:} Set $B(\theta,\vartheta) = I + \frac{A(\theta)}{\Omega(\theta, \vartheta)}$ and $\fwd^{\theta, \vartheta}_0(\cdot) = \pi_0$. Recall $\ell_i(\cdot)$ from equation~\eqref{eq:lik_factor}.
    \vspace{-.25in}
        $$\textbf{for } i=1\rightarrow |{W}|\textbf{ do:} \quad \fwd^{\theta,\vartheta}_i(s') = \sum_{s \in \cS} \fwd^{\theta,\vartheta}_{i-1}(s)\cdot \ell_{i}(s) \cdot B_{ss'}(\theta,\vartheta), \quad \forall s' \in \cS.\qquad\qquad\quad
    \vspace{-.1in}
        $$
        Similarly, for $B(\vartheta,\theta) = I + \frac{A(\vartheta)}{\Omega(\vartheta, \theta)}$, calculate $\fwd^{\vartheta,\theta}_i(\cdot)$ for all elements of ${W}$.
    \State \textbf{Swap} 
    $(\theta, \vartheta)$ with probability
     $
     1 \wedge \frac{P(X| W,\vartheta,\theta)P(\vartheta)q(\theta|\vartheta)}
        {P(X| W,\theta, \vartheta)P(\theta) q(\vartheta|\theta)}.
        $
    Write the new parameters as $(\theta',\vartheta')$.
        Here $P(X|W,\theta,\vartheta) = \sum_{s \in \cS} \fwd^{\theta,\vartheta}_{|W|}(s)\ell_{|W|+1}(s)$, $P(X|W,\vartheta,\theta) = \sum_{s \in \cS} \fwd^{\vartheta,\theta}_{|W|}(s)\ell_{|W|+1}(s)$.
    \State \textbf{Backward pass:}
    Simulate $v_{|W|} \sim \bck^{\theta',\vartheta'}_{|W|}(\cdot)$, where $\bck^{\theta',\vartheta'}_{|W|}\!(s) \propto \fwd^{\theta',\vartheta'}_{|W|}\!(s)\cdot\ell_{|W|+1}(s) \ \  \forall s \in \cS.$ 
    \vspace{-.25in}
    $$\hspace{-.15in} \textbf{for } i=(|W|-1)\rightarrow 0\textbf{ do:} \ v_i \sim \bck^{\theta',\vartheta'}_i\!(\cdot), \text{where } 
    \bck^{\theta',\vartheta'}_i\!(s) \propto \fwd^{\theta',\vartheta'}_i\!(s)\cdot B_{sv_{i+1}}(\theta',\vartheta') \cdot \ell_{i+1}(s)  \ \forall s \in \cS.$$
    
    \vspace{-.1in}
    \State Set $s'_0=v_0$. Let $T'$ be the set of times in $W$ when $V$ changes state. Define $S'$ as the corresponding set of state values. Return $(s'_0, S', T', \theta')$.
\end{algorithmic}
\end{algorithm}

\setlength{\unitlength}{0.8cm}
  \begin{figure}[h!]
  \centering
  \begin{minipage}[!hp]{0.32\linewidth}
  \centering
    \includegraphics [width=0.96\textwidth, angle=0]{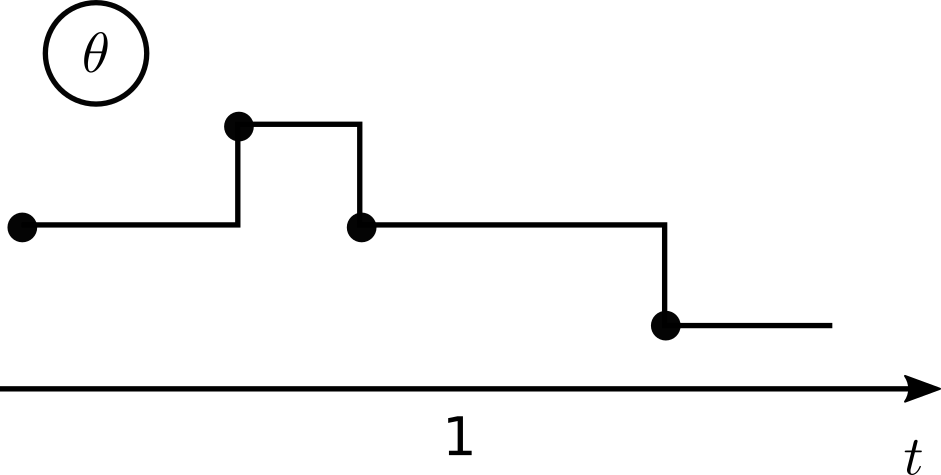}
      \end{minipage}
  \begin{minipage}[hp]{0.32\linewidth}
  \centering
    \includegraphics [width=0.96\textwidth, angle=0]{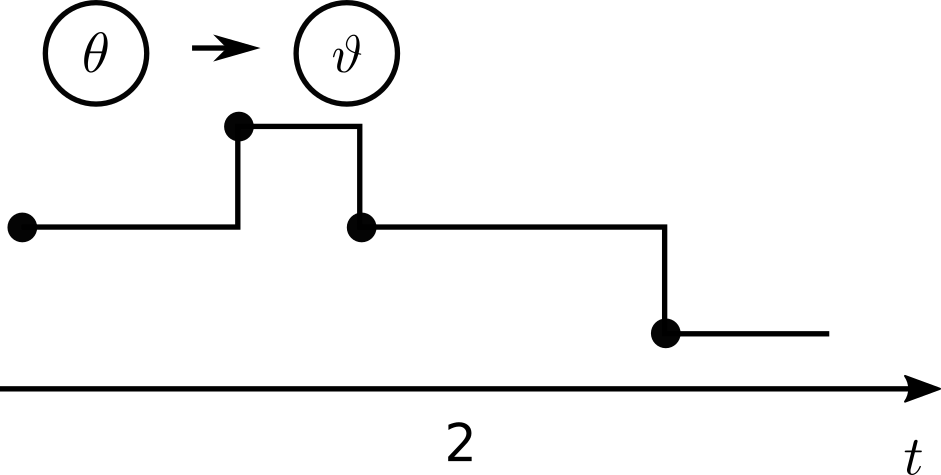}
    \vspace{-0 in}
  \end{minipage}
  \begin{minipage}[hp]{0.32\linewidth}
  \centering
    \includegraphics [width=0.96\textwidth, angle=0]{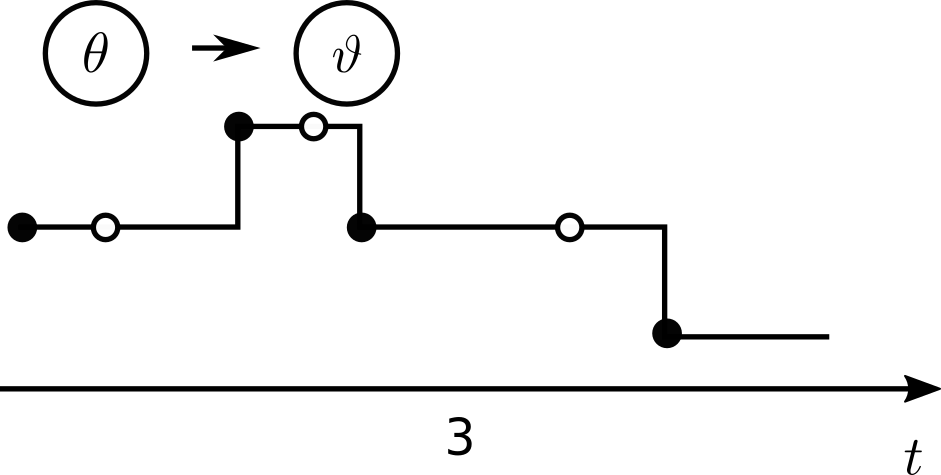}
    \vspace{-0 in}
  \end{minipage}
  \begin{minipage}[hp]{0.32\linewidth}
  \centering
    \includegraphics [width=0.96\textwidth, angle=0]{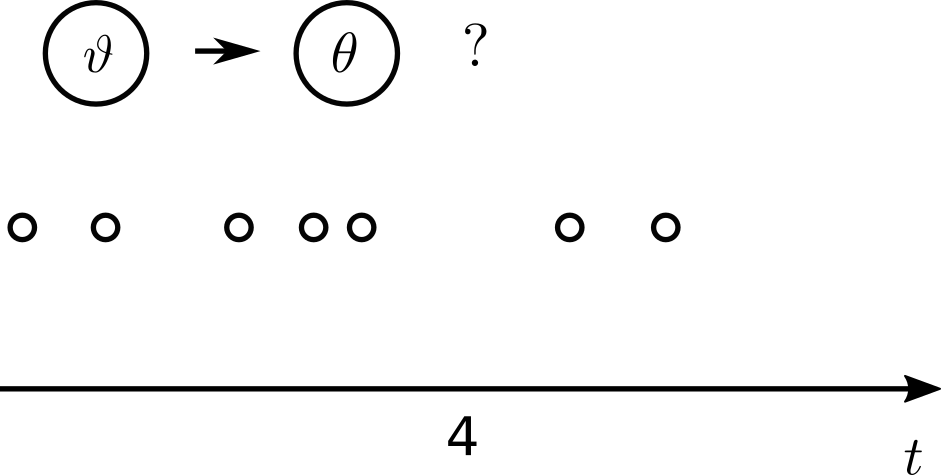}
    \vspace{-0 in}
  \end{minipage}
  \begin{minipage}[hp]{0.32\linewidth}
  \centering
    \includegraphics [width=0.96\textwidth, angle=0]{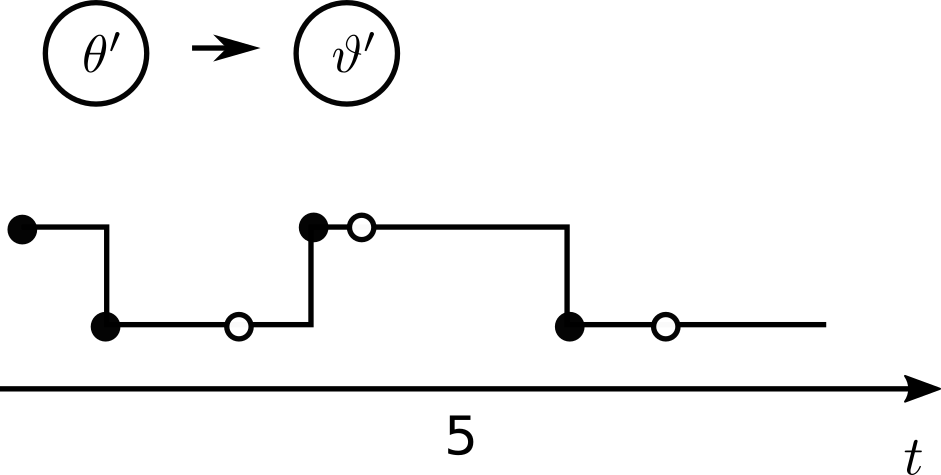}
    \vspace{-0 in}
  \end{minipage}
  \begin{minipage}[hp]{0.32\linewidth}
  \centering
    \includegraphics [width=0.96\textwidth, angle=0]{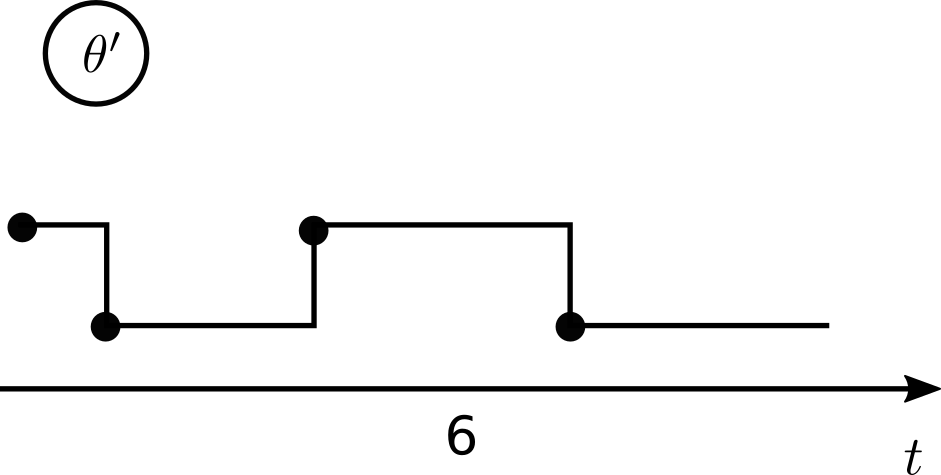}
  \end{minipage}
    \caption{Symmetrized MH algorithm: Steps 1-3: Starting with a trajectory and parameter $\theta$, simulate an auxiliary parameter $\vartheta$, and then the thinned events
      $U$ from a rate $\Omega(\theta,\vartheta) - A_{S(\cdot)}$ Poisson
      process. Step 4: Discard state values, and propose swapping $\theta$ and $\vartheta$. Step 5:
      Run a forward pass to accept or reject this proposal, calling the new parameters $(\theta',\vartheta')$. 
    Use these to simulate a new trajectory. Step 6: Discard $\vartheta'$ and the thinned events.} 
   \label{fig:MH_improved}
  \end{figure}

\begin{proposition}
  The sampler described in Algorithm~\ref{alg:MH_improved} has the posterior distribution $P(\theta,s_0,S,T|X)$ as its stationary distribution.
\end{proposition}
\begin{proof}
  Consider a realization $(\theta,s_0,S,T)$ from the posterior distribution $P(\theta, s_0, S, T|X)$. An iteration of the algorithm first simulates $\vartheta$ from $q(\vartheta|\theta)$. By construction, the marginal distribution over all but the last variable in the set $(\theta, s_0, S, T, \vartheta)$ is still the posterior.

  The algorithm next simulates $U$ from a Poisson process with rate $\Omega(\theta,\vartheta) - A_{S(\cdot)}(\theta)$. Write $W = T \cup U$.
  The random grid $W$ consists of the actual and thinned candidate transition times, and is distributed according to a rate-$\Omega(\theta, \vartheta)$ Poisson process (Proposition 2 in~\cite{RaoTeh13}). 
  Thus, the triplet $(W,\theta,\vartheta)$ has probability proportional to $P(\theta)q(\vartheta|\theta)\text{PoissProc}(W|\Omega(\theta,\vartheta)) P(X|W,\theta,\vartheta)$.
  Next, the algorithm proposes swapping $\theta$ and $\vartheta$ with $W$ fixed (a deterministic proposal), and accepts with MH-acceptance probability 
  $$\texttt{acc} = 
  1 \wedge \frac{P(\vartheta)q(\theta|\vartheta)
P(X|W,\vartheta,\theta)}{P(\theta)q(\vartheta|\theta)
P(X|W,\theta,\vartheta)} =
1 \wedge \frac{P(\vartheta)q(\theta|\vartheta)\text{PoissProc}(W|\Omega(\vartheta,\theta)) P(X|W,\vartheta,\theta)}{P(\theta)q(\vartheta|\theta)\text{PoissProc}(W|\Omega(\theta,\vartheta))
P(X|W,\theta,\vartheta)},$$
where we exploit the symmetry of $\Omega(\cdot,\cdot)$.
Write the new parameters as $(\theta', \vartheta')$. 

This MH step has stationary distribution over $(W,\theta',\vartheta')$ proportional to
$P(\theta')q(\vartheta'|\theta)$ $\text{PoissProc}(W|\Omega(\theta',\vartheta'))
P(X|W,\theta',\vartheta')$, so that the triplet $(W,\theta', \vartheta')$ has the same distribution as $(W,\theta, \vartheta)$.
The algorithm uses $B(\theta',\vartheta')$ to make a backward pass through $W$, simulating state values on $W$ from the conditional of a Markov chain with transition matrix $B(\theta',\vartheta')$ given observations $X$. 
Dropping the self-transition times results in $(\theta', s'_0, S', T', \vartheta')$. 
From uniformization (see also Lemma 1 in~\cite{RaoTeh13}), the trajectory $(s'_0, S', T')$ is distributed according to the conditional of a rate-$A(\theta')$ MJP given observations $X$.
Finally, dropping $\vartheta'$ 
results in $(\theta',s'_0,S',T')$ from the posterior given $X$, proving stationarity.
\end{proof}

\section{Related work}\label{sec:comments}

Our paper modifies the algorithm from~\citet{RaoTeh13} to include parameter inference.
That algorithm requires a uniformization rate $\Omega(\theta) > \max_s A_s(\theta)$, and empirical results from that paper suggest $\Omega(\theta) = 2\max_s A_s(\theta)$.
The uniformization rate $\Omega(\theta,\vartheta)$ in our algorithm includes a proposed parameter $\vartheta$, must be symmetric in both arguments and must be greater than both $\max_s A_s(\theta)$ and $\max_s A_s(\vartheta)$. 
A natural and simple setting is $\Omega(\theta,\vartheta) = \max_s A_s(\theta) + \max_s A_s(\vartheta)$. 
When $\theta$ is known, our algorithm has $\vartheta$ equal to $\theta$ (i.e.\ the `proposed' $\vartheta$ equals $\theta$), and our uniformization rate reduces to $2\max_s A_s(\theta)$. 
This provides a principled motivation for the particular choice of $\Omega$ in~\citet{RaoTeh13}.

Of course, we can consider other choices, such as $\Omega(\theta,\vartheta) = \kappa(\max A_i(\theta) + \max A_i(\vartheta))$ for $\kappa > 1$.  
These result in more thinned events, and so more computation, with faster MCMC mixing. 
We study the effect of $\kappa$ in our experiments, but find the smallest setting of $\kappa=1$ performs best.
It is also possible to have non-additive settings for $\Omega(\theta,\vartheta)$, for example, $\Omega(\theta,\vartheta) = \kappa \max( \max_i A_i(\theta), \max A_i(\vartheta))$ for some $\kappa > 1$. We investigate this too.

A key idea in our paper, as well as~\cite{RaoTeh13}, is to impute the random grid of candidate transition times $W$ every MCMC iteration. 
Conditioned on $W$, the MJP trajectory follows an HMM with transition matrix $B$. 
By running the FFBS algorithm over $W$, we can marginalize out the states associated with $W$, and calculate the marginal $P(X|W,\theta)$. 
Another approach to parameter inference that integrates out state values follows~\citet{FearnSher2006}. 
 This algorithm makes a sequential forward pass through all {\em observations} $X$ (rather than $W$). 
 Unlike with $W$ fixed, one cannot a priori bound the number of transitions between two successive observations, so that~\citet{FearnSher2006} have to use matrix exponentials of $A$ (rather than just $B$) to calculate transition probabilities.
 The resulting algorithm is cubic, rather than quadratic in the number of states, and the number of expensive matrix exponentiations needed scales with the number of observations, rather than the number of transitions.
 Further, matrix exponentiation results in a dense matrix, so that~\cite{FearnSher2006} cannot exploit sparsity in the transition matrix.
 In our framework, $B=I+\frac{1}{\Omega}A$ inherits sparsity present in $A$. Thus if $A$ is tri-diagonal, our algorithm is {\em linear} in the number of states.

 A second approach to marginalizing out state information is particle MCMC~\citep{Andrieu10}. 
 This algorithm, described in section~\ref{sec:pmcmc} in the supplementary material, uses particle filtering to get an unbiased estimate of $P(X|\theta)$. 
 Plugging this estimate into the MH acceptance probability results in an MCMC sampler that targets the correct posterior, however the resulting scheme does not exploit the Markovian structure of the MJP the way FFBS can. 
 In particular, observations that are informative of the MJP state can result in marginal probability estimates that have large variance, resulting in slow mixing. 
 By contrast, given $W$, FFBS can compute the marginal probability $P(X|W,\theta)$ {\em exactly}. 

The basic idea of marginalizing out information to accelerate MCMC convergence is formalized by the idea of the Bayesian fraction of missing information~\citep{liu1994fraction}. 
In this context, papers such as~\citet{papaspiliopoulos2007general,yu2011center} have studied MCMC algorithms for hierarchical latent variable models. 
The Gibbs sampler of algorithm~\ref{alg:MJP_gibbs} can be viewed as operating on a centered parametrization~\citep{papaspiliopoulos2007general} or sufficient augmentation~\citep{yu2011center} of a hierarchical model involving the parameter $\theta$, the latent variables $(v_0, V, W)$ and the observations $X$. 
These papers then suggest noncentered parametrizations or ancillary augmentations, which in our context correspond to simulating $\theta$, $W$, and an {\em independent} set of $(|W|+1)$ i.i.d.\ uniform random variables $Q$. 
Through a sequence of inverse-cdf transforms, the state values $(v_0,V)$ are then written as a deterministic function of $Q$ and $\theta$: $(v_0,V) = f_\theta(Q)$, after which the observations $X$ are produced. 
Now, proposing a new parameter $\vartheta$ automatically proposes a new set of state variables $(v_0',V') = f_\vartheta(Q)$, so that problem of path-parameter coupling is avoided.
A similar idea could also be used to avoid couplng between $\theta$ and the Poisson process $W$.
However now, updating $Q$ given $\theta$ and $(v_0, V, W)$ raises significant challenges to mixing.
By contrast, our approach marginalizes out the variables $(v_0,V)$ (or $Q$), and will mix significantly faster.
Nevertheless, results from the literature on NCPs can suggest further improvements to our approach, and give guidance about conditions under which 
approaches like ours outperform centered parametrisations like algorithm~\ref{alg:MJP_gibbs}, or when a mixture of centered and non-centered
updates could be useful~\citep{yu2011center}.

%
 %
Our approach of first simulating $\vartheta$, and then simulating $W$ from a Poisson process whose rate is symmetric in $\theta$ and $\vartheta$ is related to~\citet{Neal04Drag}. In that work, to simulate from an `energy' model $P(x,y) \propto \exp(-E(x,y))$, the author proposes a new parameter $x^*$, and then updates $y$ via intermediate transitions to be symmetric in $x$ and $x^*$, before proposing to swap $x$ and $x^*$. Our approach exploits the specific structure of the Poisson and Markov jump processes to do this directly, avoiding the need for any tempered transitions. 

Our algorithm is also related to work on MCMC for doubly-intractable distributions.  Algorithms like~\cite{Moller2006,murray2006,Andrieu09} all attempt to evaluate an intractable likelihood under a proposed parameter $\vartheta$ by introducing auxiliary variables, typically sampled independently under the proposed parameters. 
For MJPs, this would involve proposing $\vartheta$, generating a new grid $W^*$, and then using $P(X|W,\theta)$ and $P(X|W^*,\vartheta)$ in the MH acceptance step. 
This is more involved (with two sets of grids), and introduces additional variance that reduces acceptance rates. 
While~\cite{murray2006} suggest annealing schemes to try to address this issue, we exploit the uniformization structure to provide a cleaner solution: generate a single set of auxiliary variables that depends symmetrically on both the new and old parameters. 

 \vspace{-.2in}
\section{Experiments}\label{sec:expts}
\vspace{-.1in}
In the following, we evaluate Python implementations of a number of algorithms, focusing on our contribution, the symmetrized MH algorithm (algorithm~\ref{alg:MH_improved}), and as well as the \naive\ MH algorithm (algorithm~\ref{alg:MH_naive}).
We evaluate different variants of these algorithms, corresponding to different uniformizing Poisson rates. 
For \naive\ MH, we set $\Omega(\theta) = \kappa \max_s A_s(\theta) $ with $\kappa$  equal to $1.5, 2$ and $3$ (here $\kappa$ must be greater than $1$), 
while for symmetrized MH, where the uniformizing rate depends on both the current and proposed parameters, we consider
 $\Omega(\theta, \vartheta) = \kappa (\max A(\theta) + \max A(\vartheta))$ 
 ($\kappa = 1$ and $1.5$), and 
$\Omega(\theta, \vartheta) = 1.5 \max(\max A(\theta), \max A(\vartheta))$.
We evaluate two other baselines: Gibbs sampling (algorithm~\ref{alg:MJP_gibbs}), 
and particle MCMC~\citep[][see also section~\ref{sec:pmcmc} in the appendix]{Andrieu10}. 
Gibbs sampling involves a uniformization step to update the MJP trajectory (step 1 in algorithm~\ref{alg:MJP_gibbs}), for which we use $\Omega(\theta) = \kappa \max_s A_s(\theta)$ for $\kappa=1.5,2,3$. 
Unless specified, our results were obtained from $100$ independent MCMC runs, each of $10000$ iterations.
We found particle MCMC to be more computationally intensive, and limited each run to $3000$ iterations, the number of particles being $5, 10$ and $20$.


For each run of each MCMC algorithm, we calculated the effective sample size (ESS) of the posterior samples of the MJP parameters using the R package \texttt{rcoda}~\citep{Rcoda2006}. 
This estimates the number of independent samples returned by the MCMC algorithm, and dividing this by the runtime of a simulation gives the ESS per unit time (ESS/sec). 
We used this to compare different samplers and different parameter settings.

\subsection{A simple synthetic MJP}
\begin{wrapfigure}{r}{.4\textwidth}
    \vspace{-.5in}
  \begin{minipage}[!hp]{.05\linewidth}
    \hspace{.1in}
  \end{minipage}
  \begin{minipage}[!hp]{.9\linewidth}
  \centering
    \includegraphics [width=\textwidth, angle=0]{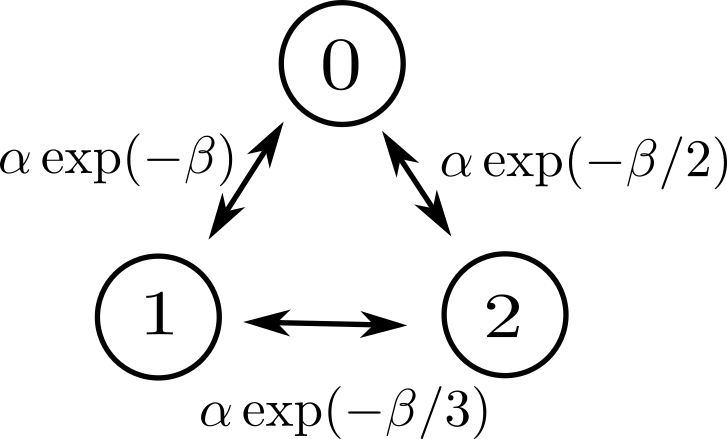}
    \caption{A 3-state MJP with exponentially decaying rates}
    \label{fig:exp_model}
      \end{minipage}
    \vspace{-.1in}
  \end{wrapfigure}
  Consider an MJP with a uniform distribution over states at time $0$, and with transitions between states $i$ and $j$ having rate $\alpha \exp(-\beta/(i+j))$, for two parameters $(\alpha,\beta) \defeq \theta$. 
We consider three settings: $3$ states (figure~\ref{fig:exp_model}), $5$ states, and $10$ states.
We place Gamma$(\alpha_0,\alpha_1)$, and Gamma$(\beta_0, \beta_1)$ priors on the parameters $\alpha$ and $\beta$, with $(\alpha_0,\alpha_1,\beta_0,\beta_1)$ having values $(3,2,5,2)$ respectively. 
For each run, we draw random parameters from the prior to construct a transition matrix $A$, and simulate an MJP trajectory.
We simulate observations uniformly at integer values on the time interval $[0, 20]$. 
Each observation is Gaussian distributed with mean equal to the state at that time, and variance equal to $1$.  
For the MH proposal, we used a lognormal distribution centered at the current parameter value, with variance $\sigma^2$ whose effect we study.  

\noindent \textbf{Results:}
Figure~\ref{fig:POST_EXP} shows the MCMC estimates of the posterior distribution over $\alpha, P(\alpha|X)$ from the Gibbs sampler as well as our symmetrized MH sampler. 
Visually these agree, and we quantify this by running a Kolmogorov-Smirnov two-sample test using $1000$ samples from each algorithm: this returns a p-value of $0.5085$, clearly failing to reject the null hypothesis that both samples come from the same distribution. 
The same is true for $\beta$, though we do not include it here.
The figure also shows the average acceptance probabilities for the two MH samplers: we see that for the same proposal distribution, symmetrization significantly improves acceptance probability. This shows the benefit of eliminating the $P(W|\theta)$ terms from the acceptance probability (we will investigate this further). Figure~\ref{fig:TRACE_EXP} shows traceplots and autocorrelation plots for $\alpha$ from the symmetrized MH and Gibbs samplers. Clearly, our sampler mixes much more efficiently than Gibbs, with \naive\ MH (included in the supplementary material) worse than both.

  \begin{figure}[H]
  \begin{minipage}[!hp]{0.49\linewidth}
    \includegraphics [width=0.49\textwidth, angle=0]{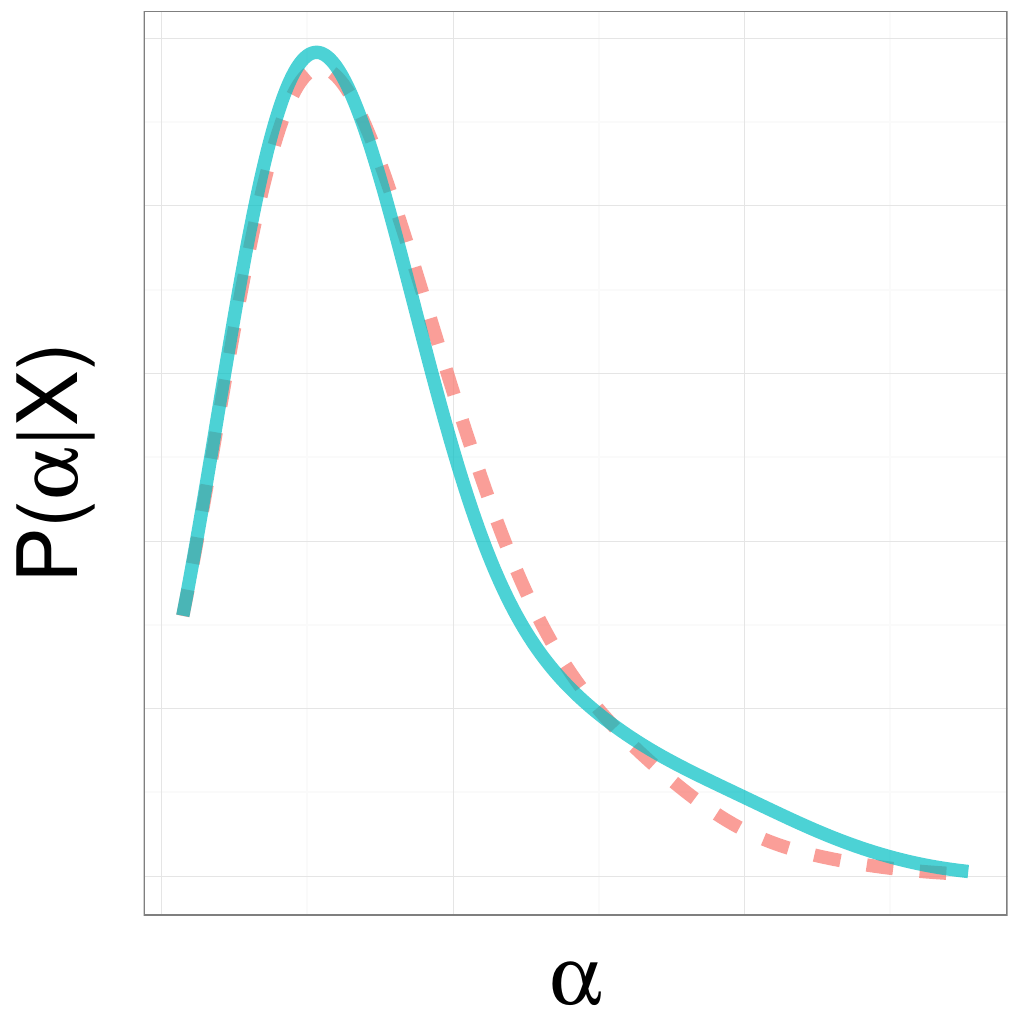}
    \includegraphics [width=0.49\textwidth, angle=0]{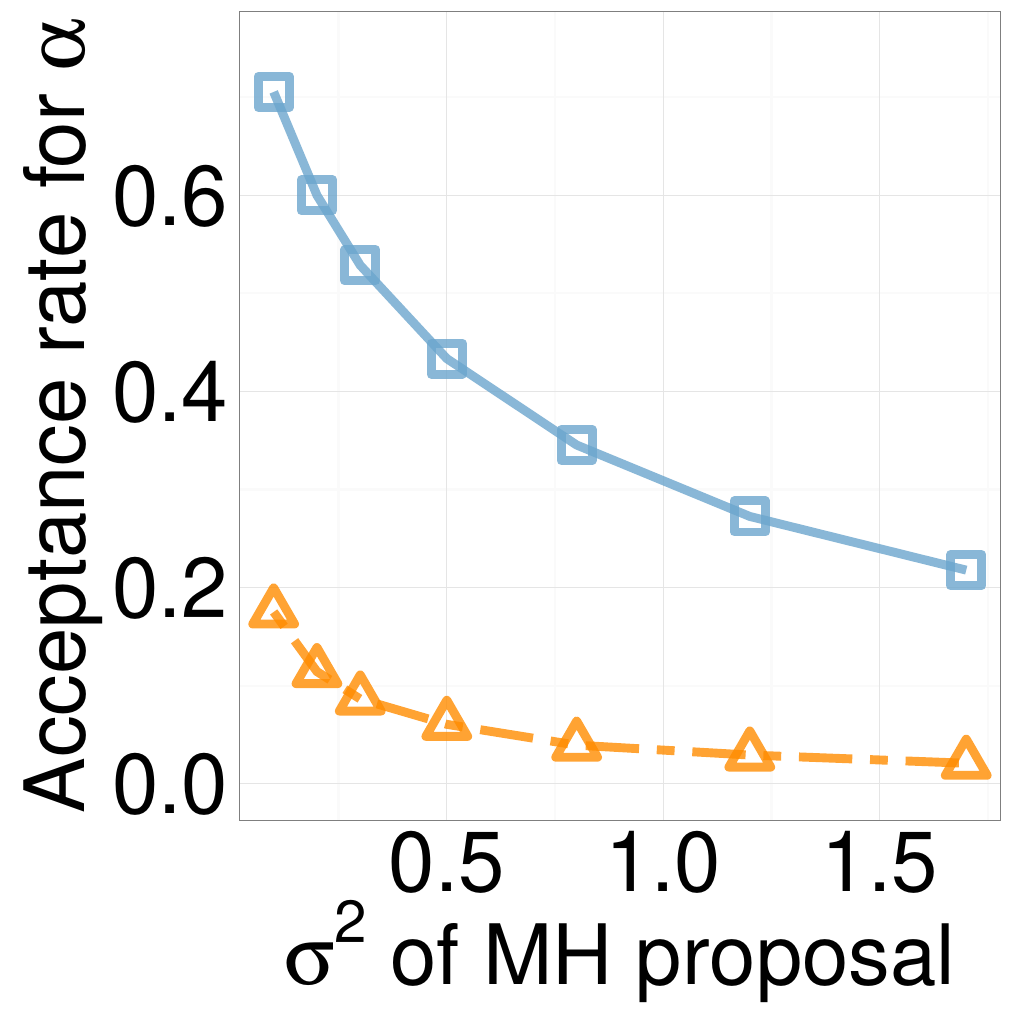}
  \end{minipage}
  \begin{minipage}[!hp]{0.49\linewidth}
    \vspace{-.2in}
    \caption{(Left) posterior $P(\alpha|X)$ from Gibbs (dashed line) and symmetrized MH (solid line) for the synthetic model. (Right) acceptance probabilities of $\alpha$ for symmetrized (squares) and \naive\ (triangles) MH.} 
    \label{fig:POST_EXP}
  \end{minipage}
  \end{figure}
  \begin{figure}[H]
  \centering
  \begin{minipage}[!hp]{0.97\linewidth}
    \includegraphics [width=0.24\textwidth, angle=0]{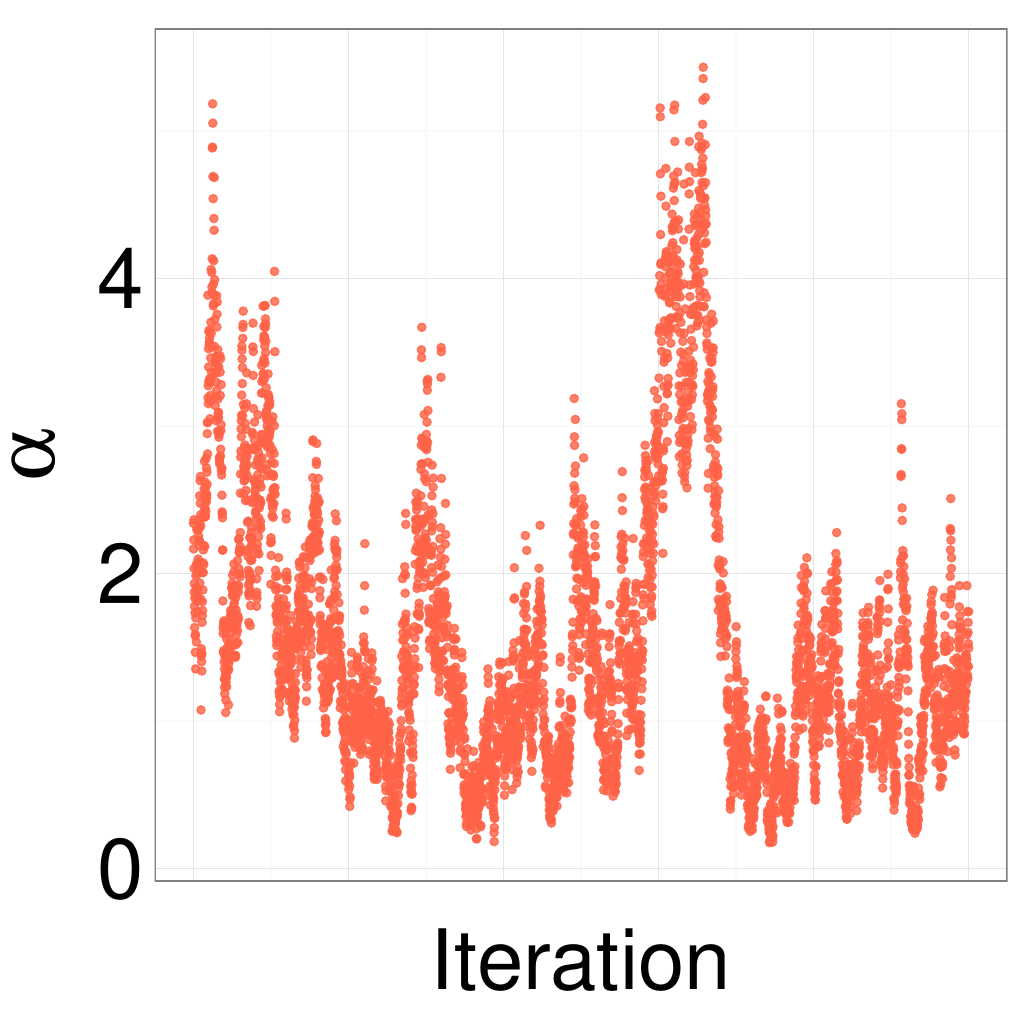}
    \includegraphics [width=0.24\textwidth, angle=0]{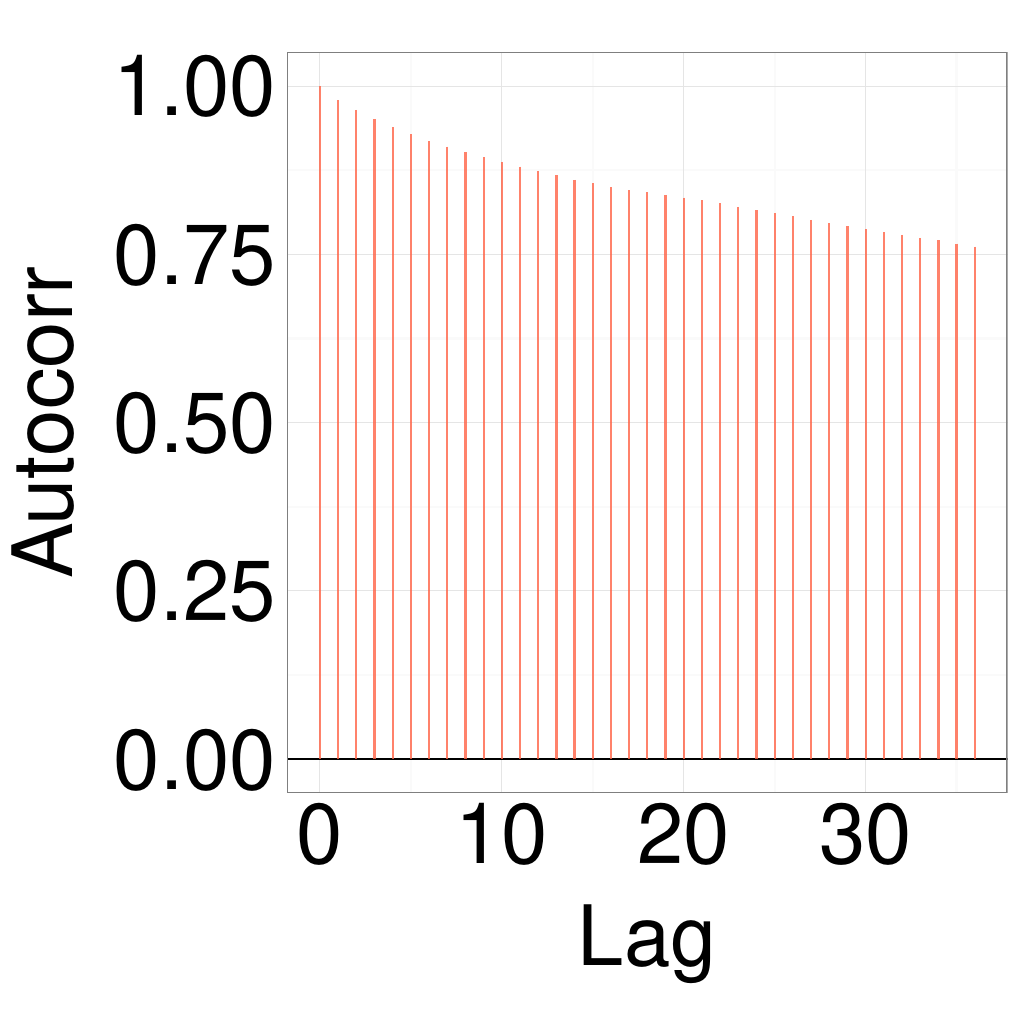}
    \includegraphics [width=0.24\textwidth, angle=0]{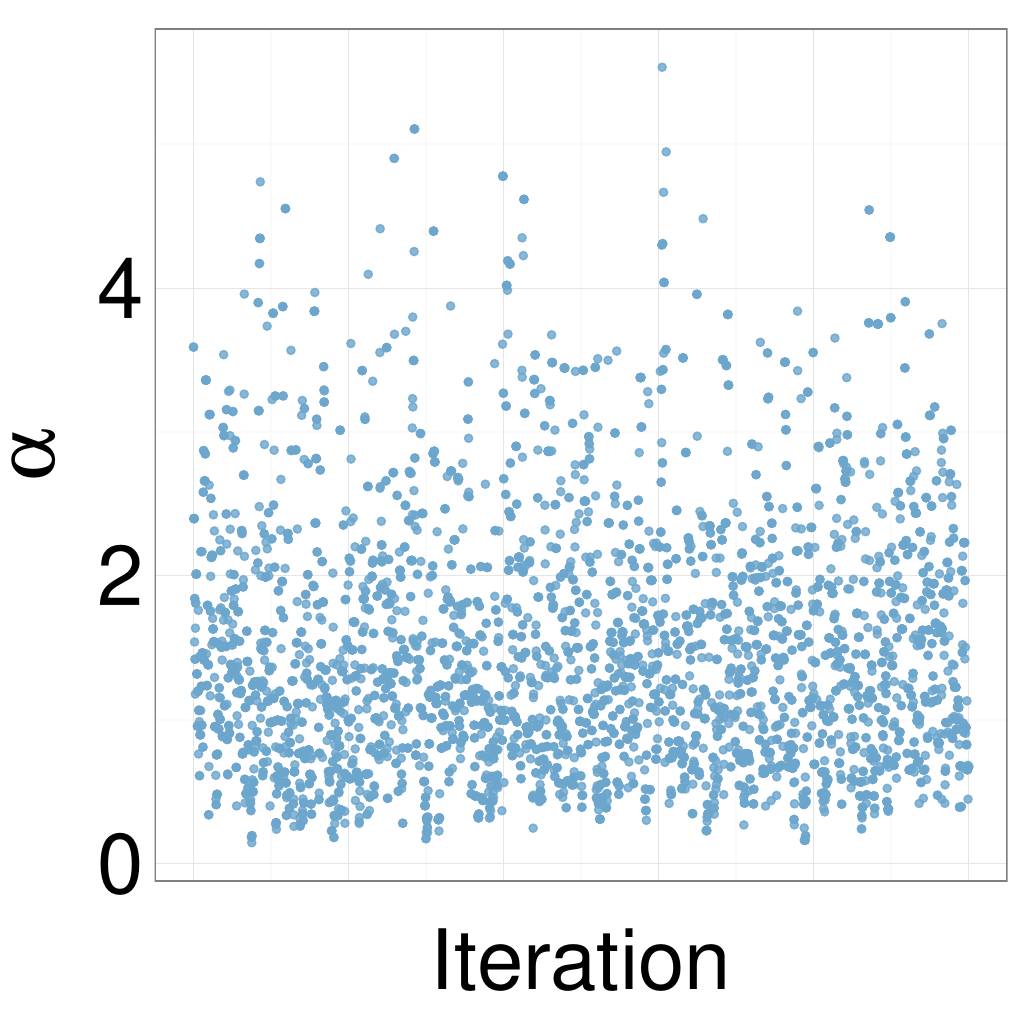}
    \includegraphics [width=0.24\textwidth, angle=0]{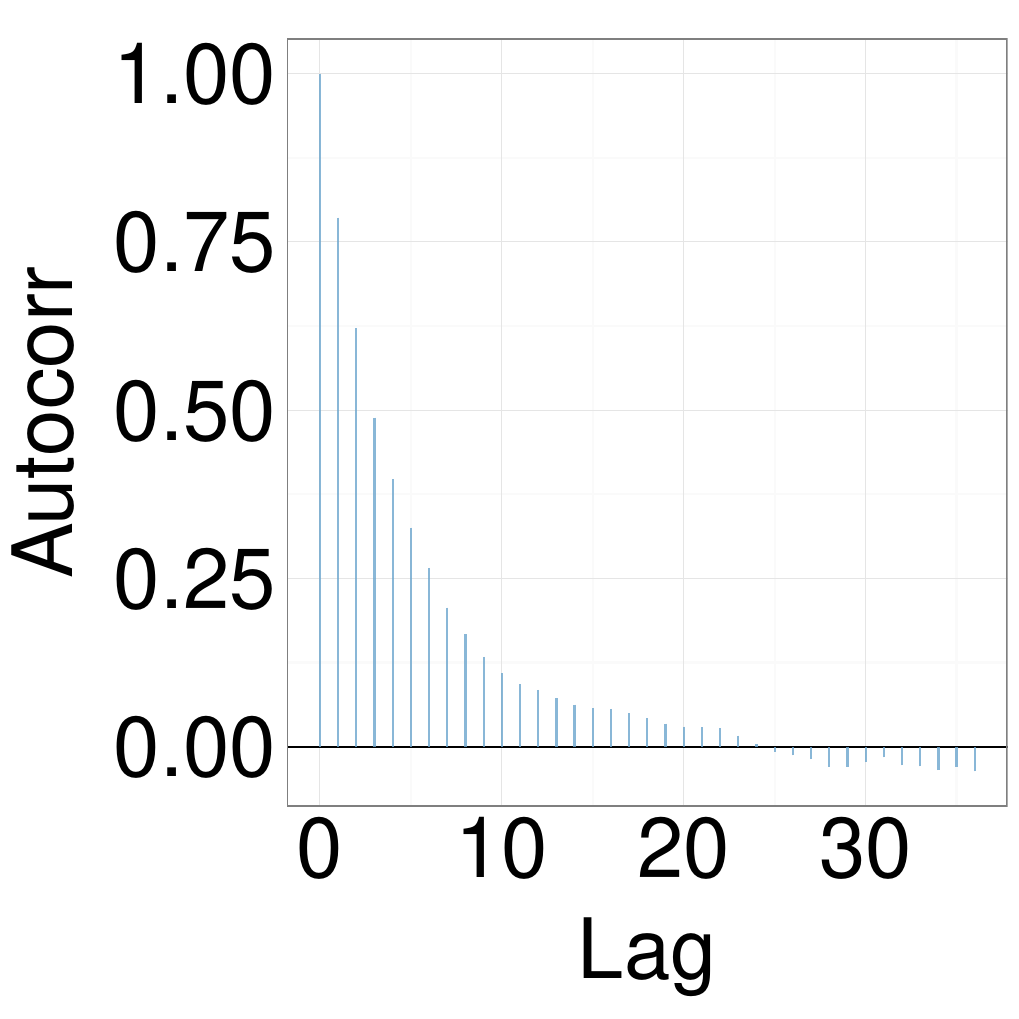}
  \end{minipage}
  \caption{Trace and autocorrelation plots for Gibbs (left two panels) and symmetrized MH (right two panels). All plots are for the synthetic model with $10$ states.}
     \label{fig:TRACE_EXP}
  \end{figure}
  To quantify performance, figure~\ref{fig:ESS_EXP_D10} plots the ESS/sec in the top row, and ESS per 1000 samples in the bottom row for $\alpha$  and $\beta$. 
  The left two columns consider $\alpha$ and $\beta$ for MJPs with $3$ states, and the right two, with $10$ states. 
  We include results for $5$ states in the supplementary material, the conclusions are the same. 
  For each plot, we vary the scale-parameter $\sigma^2$ of the log-normal proposal $q(\vartheta|\theta)$, and look at its effects on ESS/s and ESS. 
  Note that the conditional over parameters given trajectory is not conjugate, so that the Gibbs sampler is really a Metropolis-within-Gibbs (MWG) sampler with an associated lognormal proposal distribution parameterized by $\sigma^2$.

  We see that our symmetrized MH algorithm, shown with blue squares, is significantly more efficient than the baselines over a wide range of $\sigma^2$ values, including the natural choice of $1$.
  Among the baselines, Gibbs (red circles) does better than \naive\ MH (yellow triangles), confirming that the dependency of the Poisson grid on the MJP parameters (as indicated in figure~\ref{fig:POST_EXP}) does indeed slow down mixing. 
 This, coupled with the fact that MWG tends to have higher MH acceptance than \naive\ MH results in Gibbs having superior performance. 
  Our symmetrized MH avoids this problem at no additional computational cost.
  Particle MCMC (black diamonds) has the worst performance. 

  \begin{figure}[H]
  \centering
  \begin{minipage}[!hp]{0.24\linewidth}

  \centering
    \includegraphics [width=0.99\textwidth, angle=0]{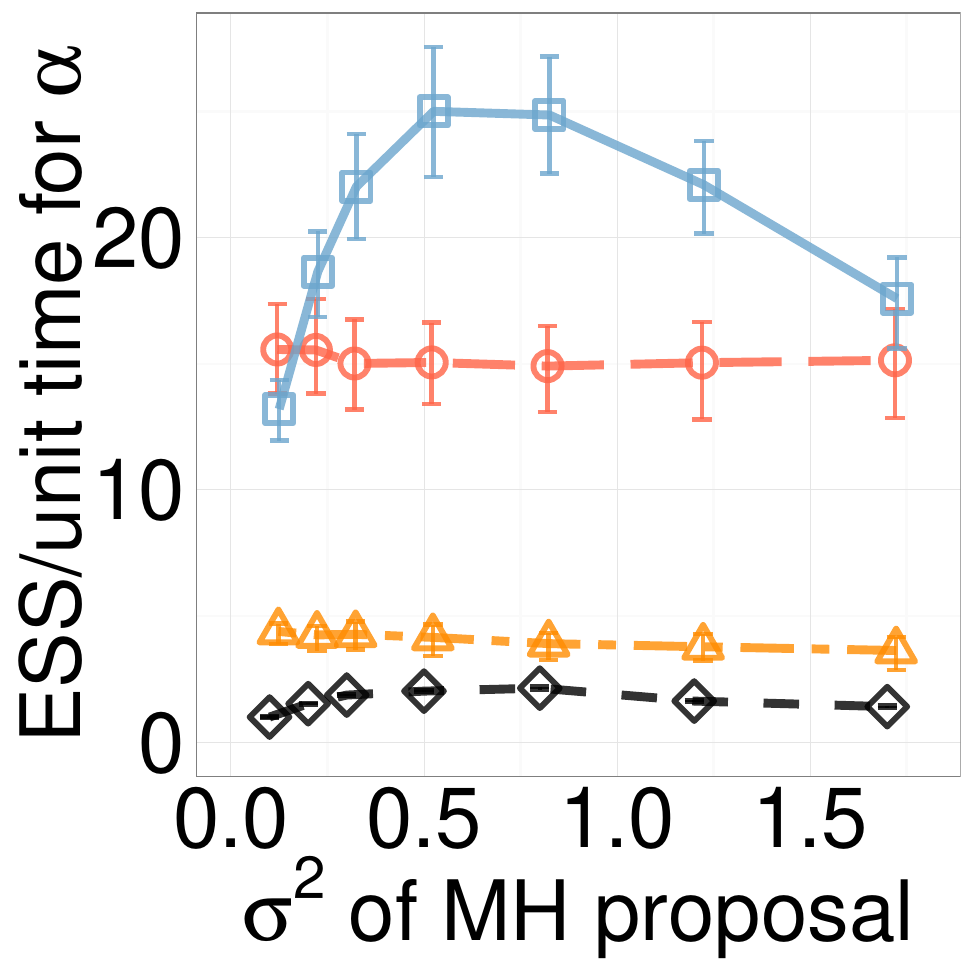}
\end{minipage}
  \begin{minipage}[hp]{0.24\linewidth}
  \centering
    \includegraphics [width=0.99\textwidth, angle=0]{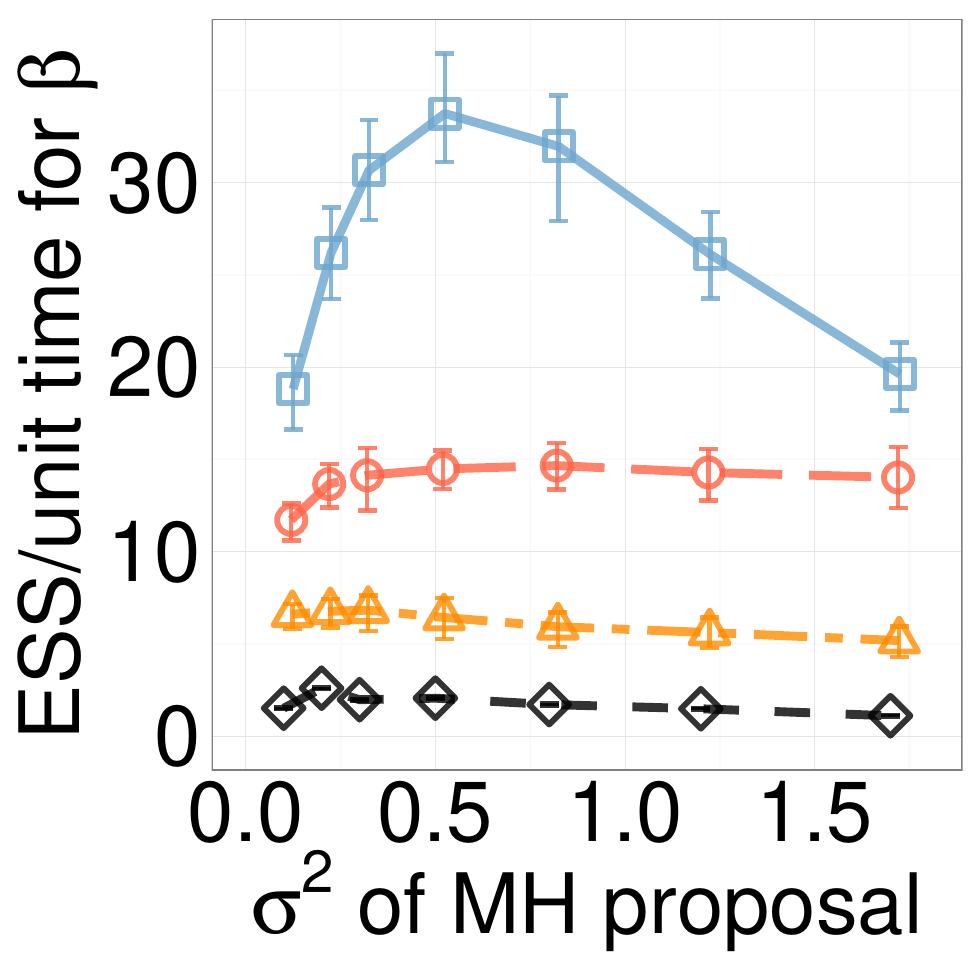}
	\end{minipage}
  \begin{minipage}[hp]{0.24\linewidth}
  \centering
    \includegraphics [width=0.99\textwidth, angle=0]{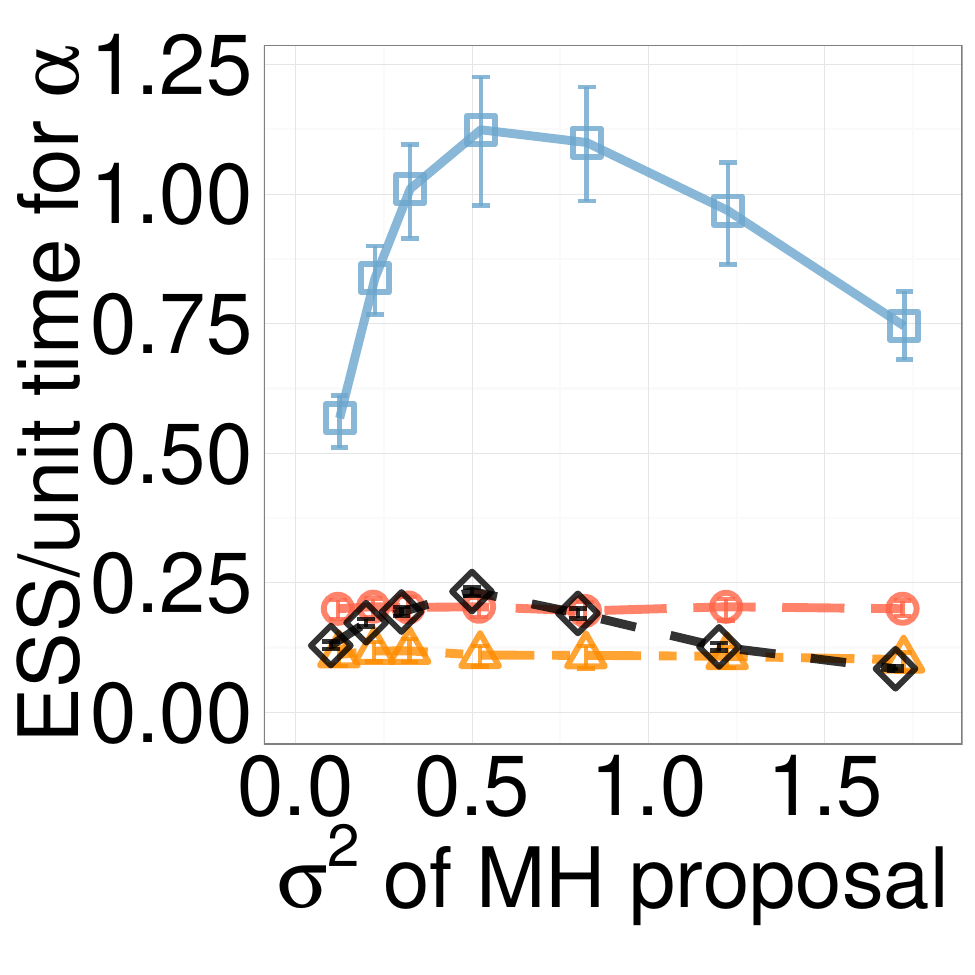}
	\end{minipage}
  \begin{minipage}[hp]{0.24\linewidth}
  \centering
    \includegraphics [width=0.99\textwidth, angle=0]{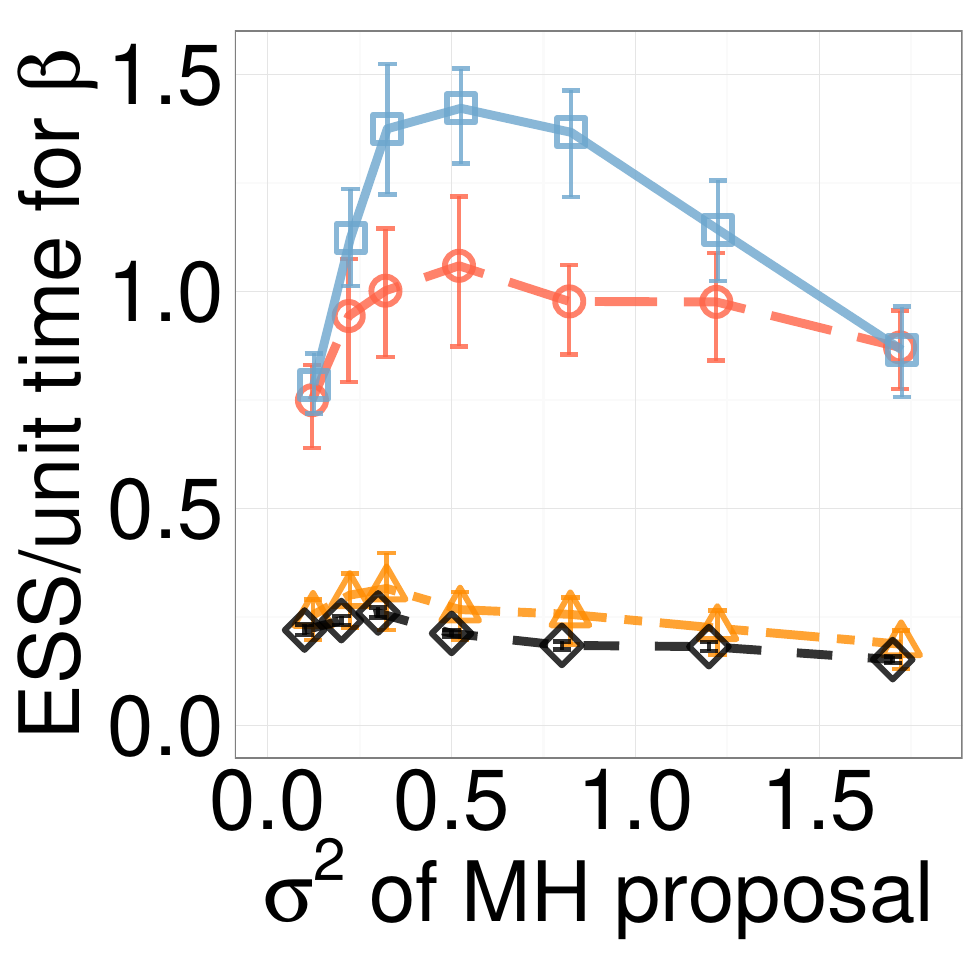}
	\end{minipage}
  \centering
  \begin{minipage}[!hp]{0.24\linewidth}
  \centering
    \includegraphics [width=0.99\textwidth, angle=0]{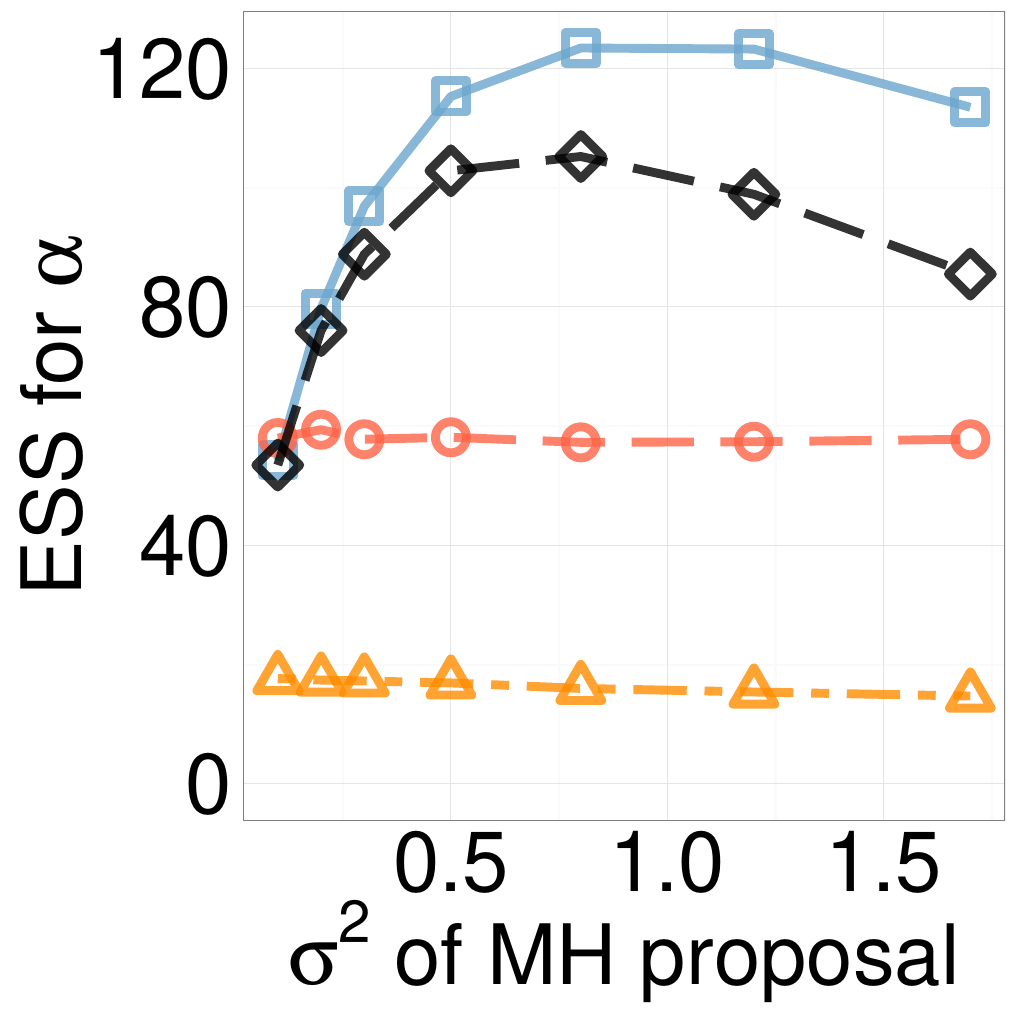}
\end{minipage}
  \begin{minipage}[hp]{0.24\linewidth}
  \centering
    \includegraphics [width=0.99\textwidth, angle=0]{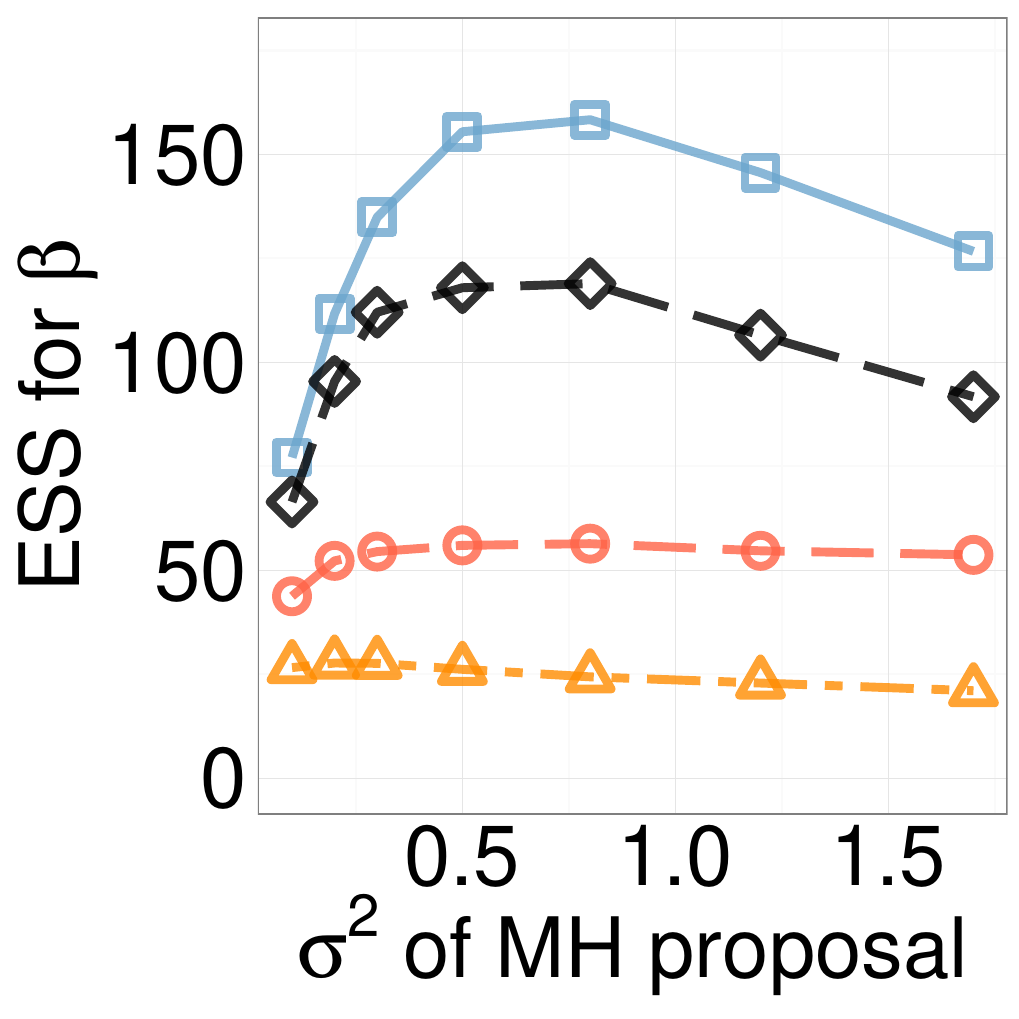}
	\end{minipage}
  \begin{minipage}[hp]{0.24\linewidth}
  \centering
    \includegraphics [width=0.99\textwidth, angle=0]{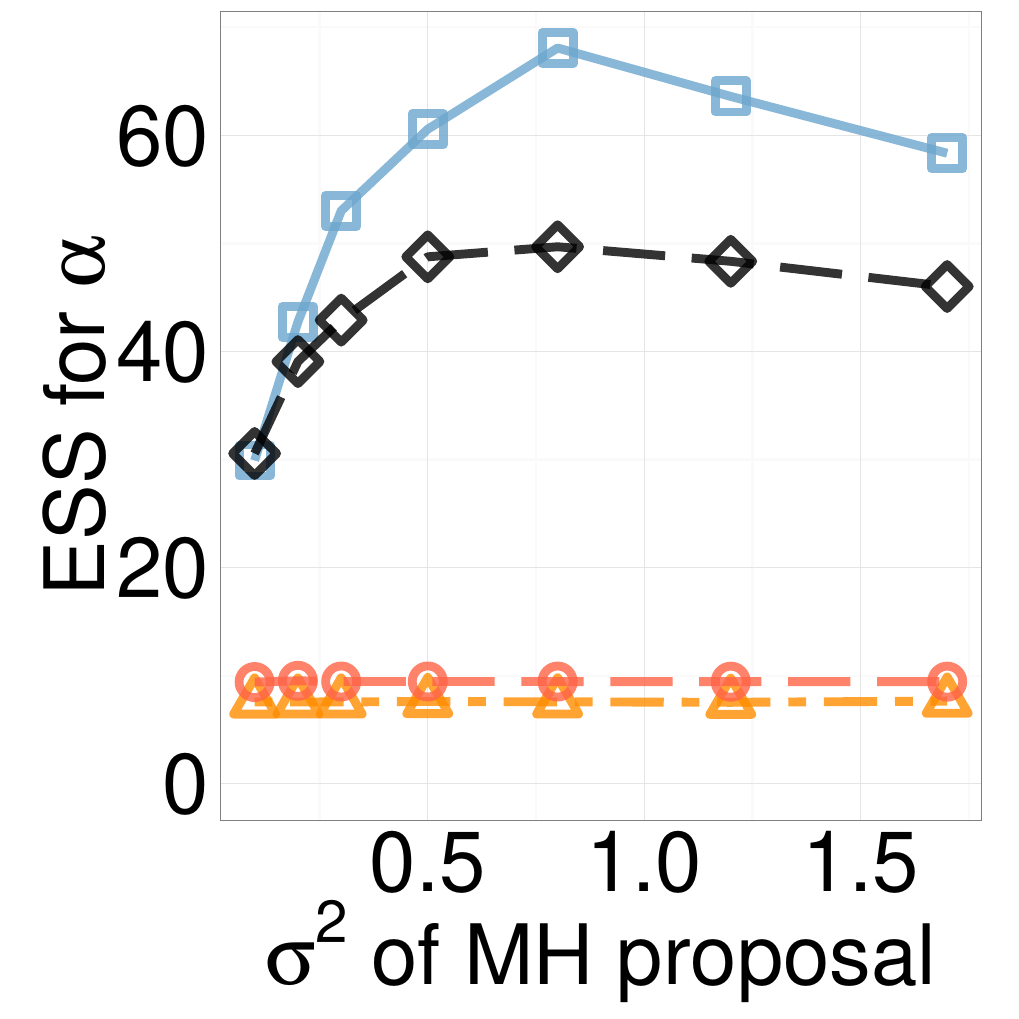}
	\end{minipage}
  \begin{minipage}[hp]{0.24\linewidth}
  \centering
    \includegraphics [width=0.99\textwidth, angle=0]{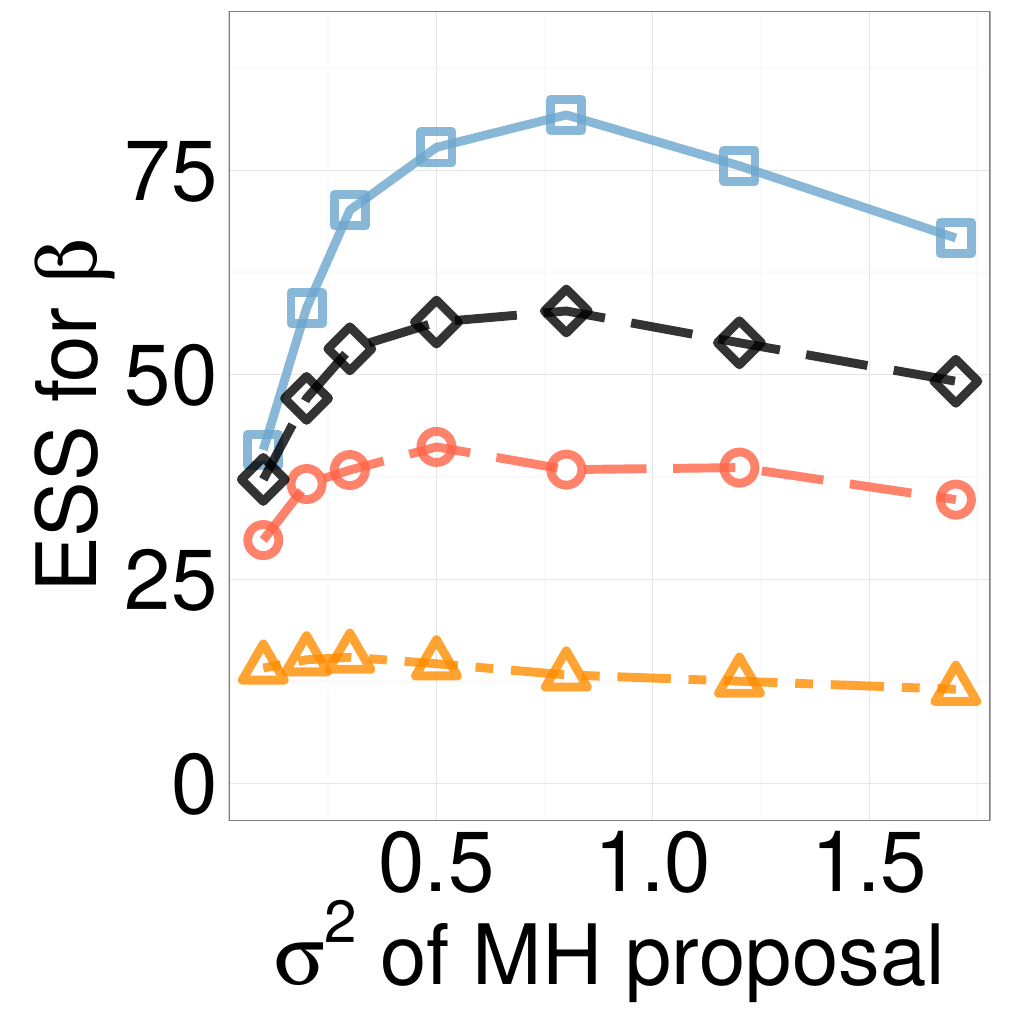}
	\end{minipage}
  \caption{ESS/sec (top row) and ESS per 1000 samples (bottom row) of different algorithms on the synthetic  model. The left two panels are $\alpha$ and $\beta$ for 3 states, the right two, for 10 states. Blue squares, yellow triangles, red circles and black diamonds are the symmetrized MH, \naive\ MH, Gibbs and particle MCMC algorithm.} 
     \label{fig:ESS_EXP_D10}
  \end{figure}

Among the three setting of our algorithm, the simple additive setting
 ({squares}) does best, slightly better than the {max-of-max} setting (circles). 
The {additive setting with a multiplicative factor of $1.5$} ({triangles}) does worse than both the {additive choice with $\kappa=1$ and the max-of-max choice} but still better than the other algorithms. 
The results in figure~\ref{fig:ESS_EXP_D10} for 10 states shows that ESS is slightly lower, and thus mixing is slightly poorer for all samplers. This, coupled with greater computational cost per iteration results in a drop in ESS/s across all algorithms, compared with 3 states. Our symmetrized MH algorithm still outperforms the other samplers, and we observe the same pattern of relative performance for different settings of our sampler (figure~\ref{fig:mhESS_EXP}), with a uniformization rate of  $\Omega(\theta,\vartheta) = \max_s A_s(\theta) + \max_s A(\vartheta)$ giving the best performance. 

  \begin{figure}[H]
  \centering
  \begin{minipage}[!hp]{0.24\linewidth}
  \centering
    \includegraphics [width=0.99\textwidth, angle=0]{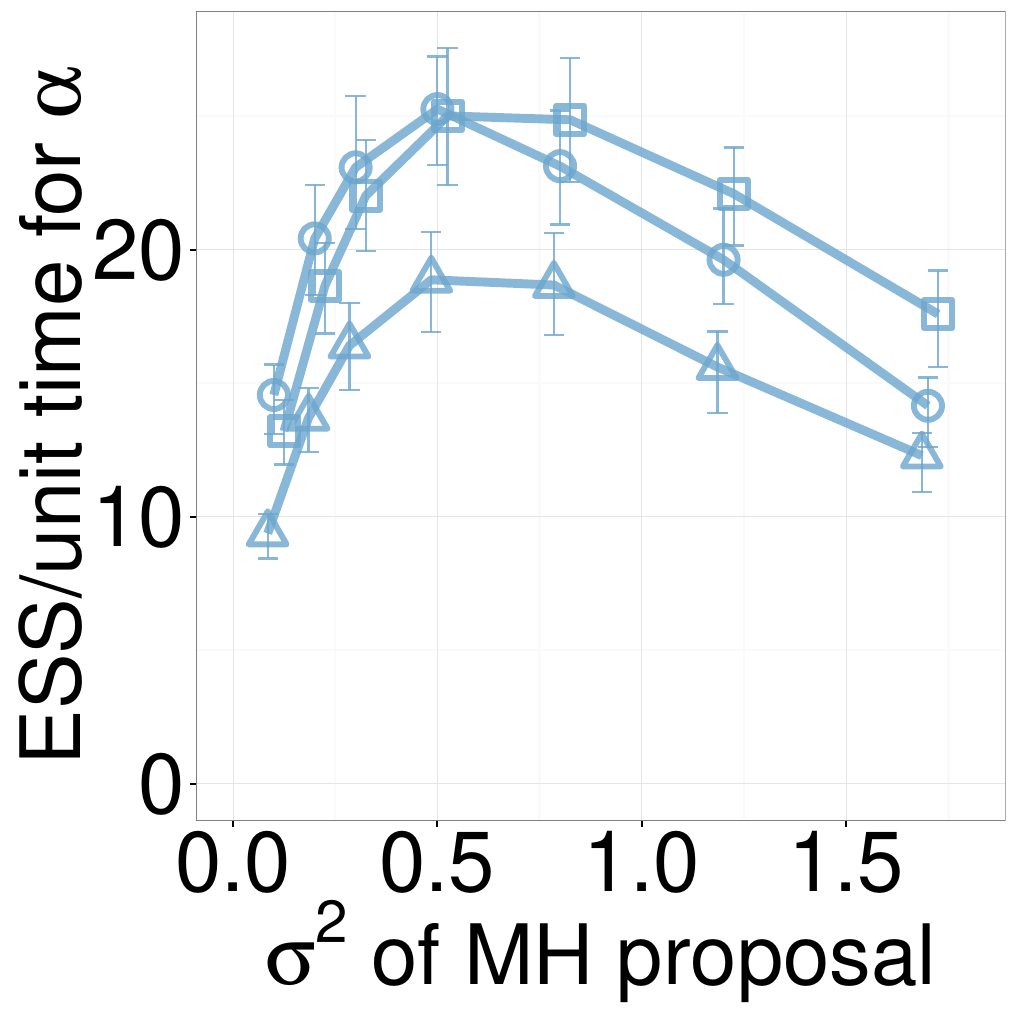}
\end{minipage}
  \begin{minipage}[hp]{0.24\linewidth}
  \centering
    \includegraphics [width=0.99\textwidth, angle=0]{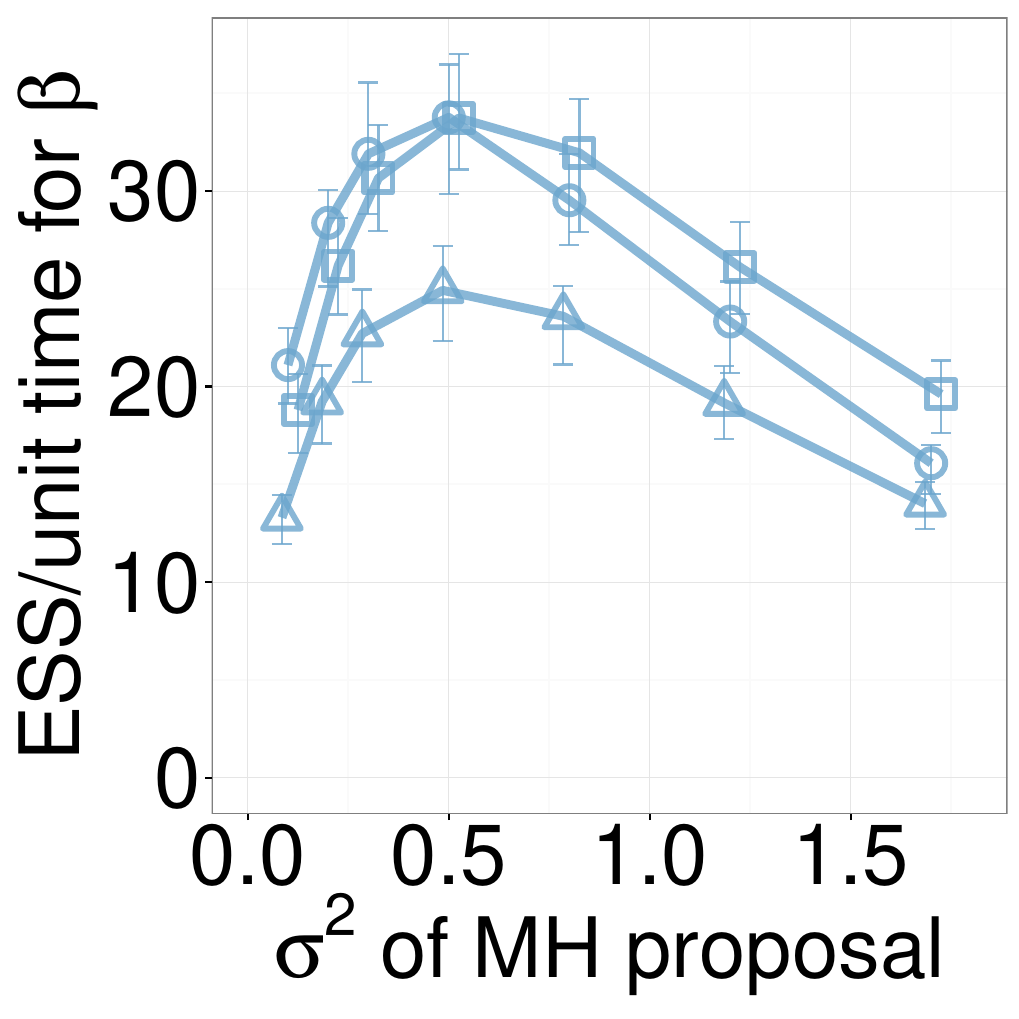}
	\end{minipage}
  \begin{minipage}[hp]{0.24\linewidth}
  \centering
    \includegraphics [width=0.99\textwidth, angle=0]{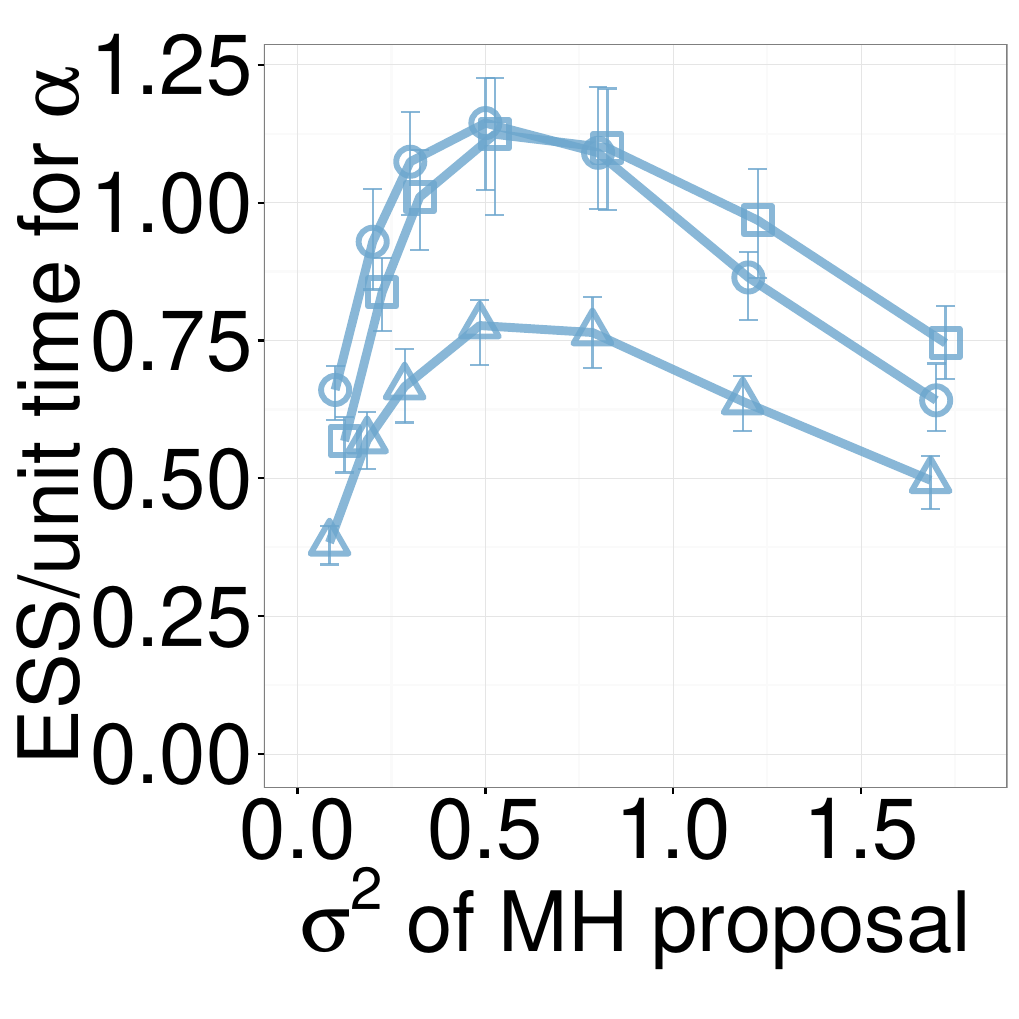}
	\end{minipage}
  \begin{minipage}[hp]{0.24\linewidth}
  \centering
    \includegraphics [width=0.99\textwidth, angle=0]{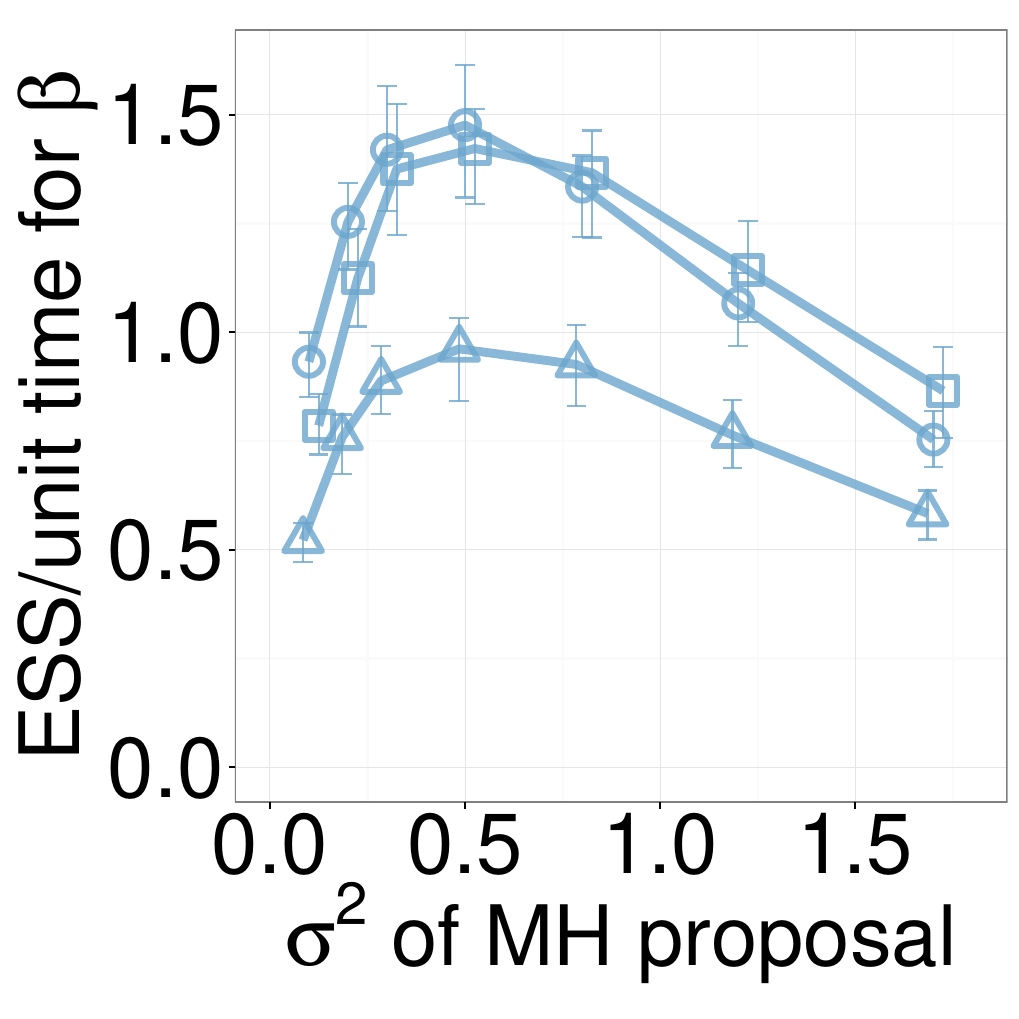}
	\end{minipage}
    \caption{ESS/sec of symmetrized MH for different choices of $\Omega(\theta,\vartheta)$ for the synthetic model. The left two panels are $\alpha$ and $\beta$ for 3 states, and the right two for 10 states. Squares, circles and trianges correspond to $\Omega(\theta,\vartheta)$ set to $(\max_s A_s(\theta) + \max_s A_s(\vartheta))$, $\max(\max_s A_s(\theta), \max_s A_s(\vartheta))$ and  $1.5(\max_s A_s(\theta) + \max_s A_s(\vartheta))$.}
     \label{fig:mhESS_EXP}
  \end{figure}



  \begin{figure}[H]    
  \centering
  \begin{minipage}[hp]{0.24\linewidth}
  \centering
    \includegraphics [width=0.99\textwidth, angle=0]{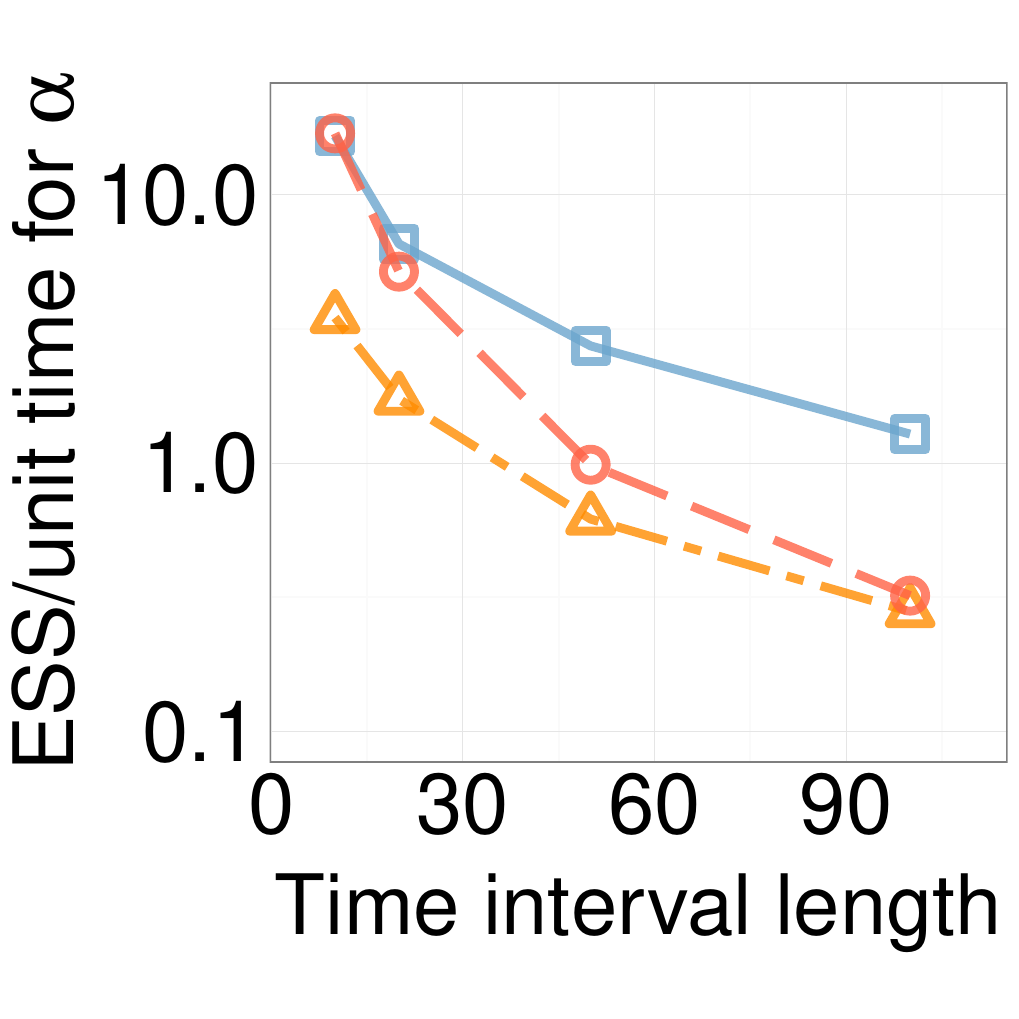}
    \end{minipage}
  \begin{minipage}[hp]{0.24\linewidth}
  \centering
    \includegraphics [width=0.99\textwidth, angle=0]{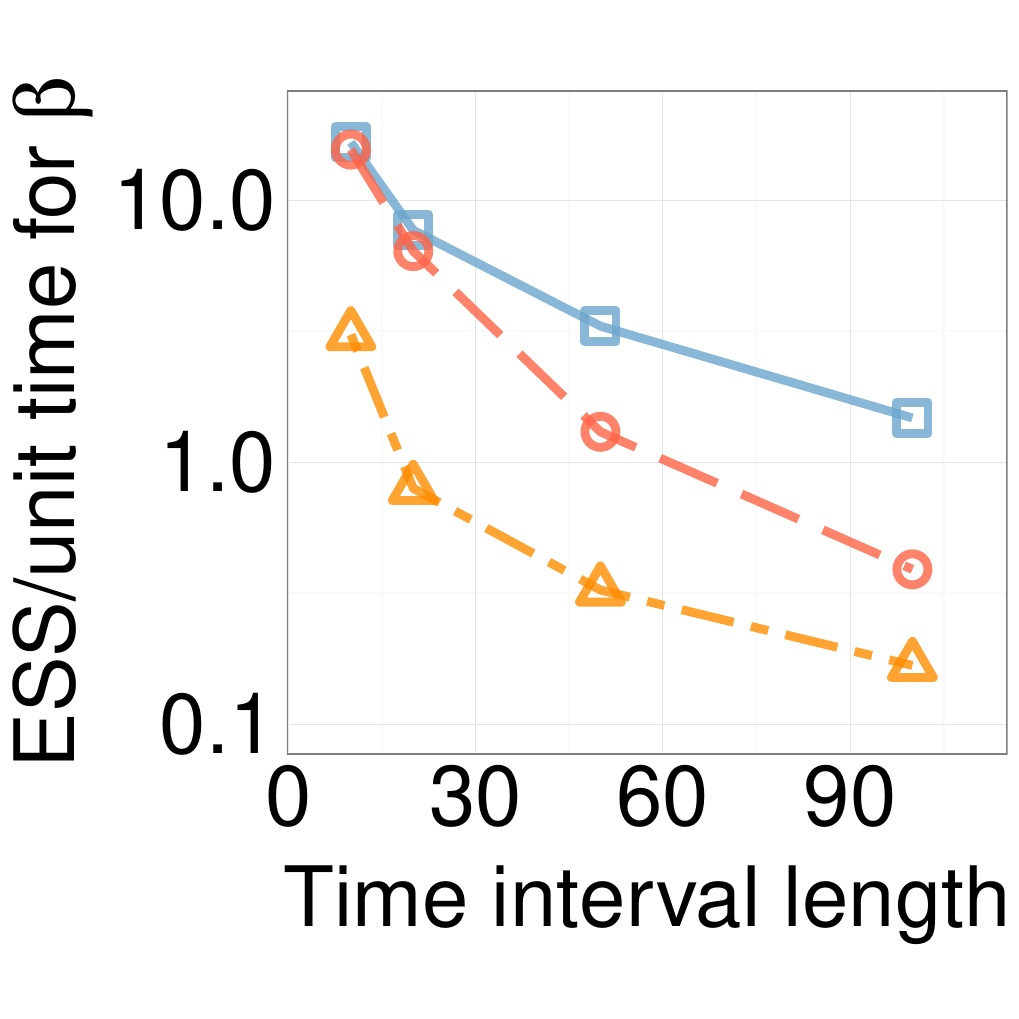}
  \end{minipage}
  \begin{minipage}[hp]{0.24\linewidth}
  \centering
    \includegraphics [width=0.99\textwidth, angle=0]{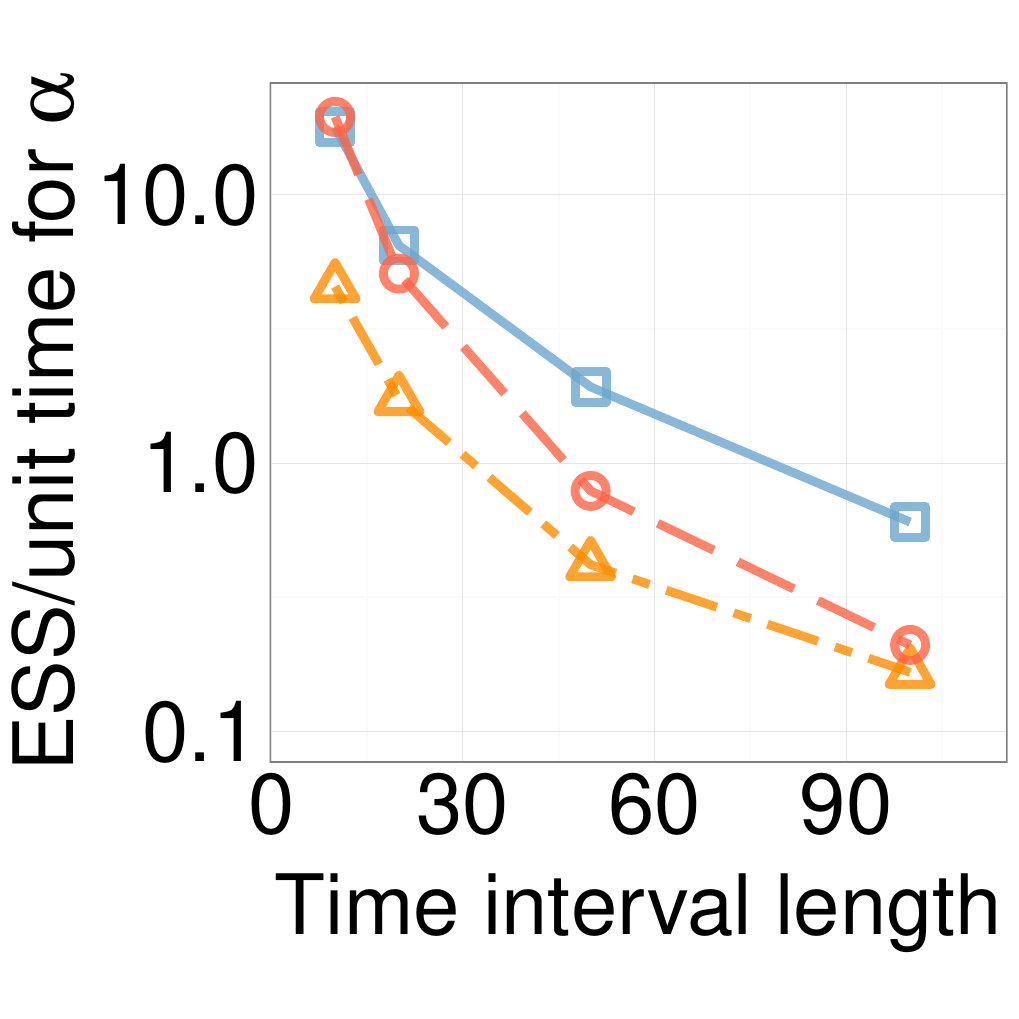}
      \end{minipage}
  \begin{minipage}[hp]{0.24\linewidth}
  \centering
    \includegraphics [width=0.99\textwidth, angle=0]{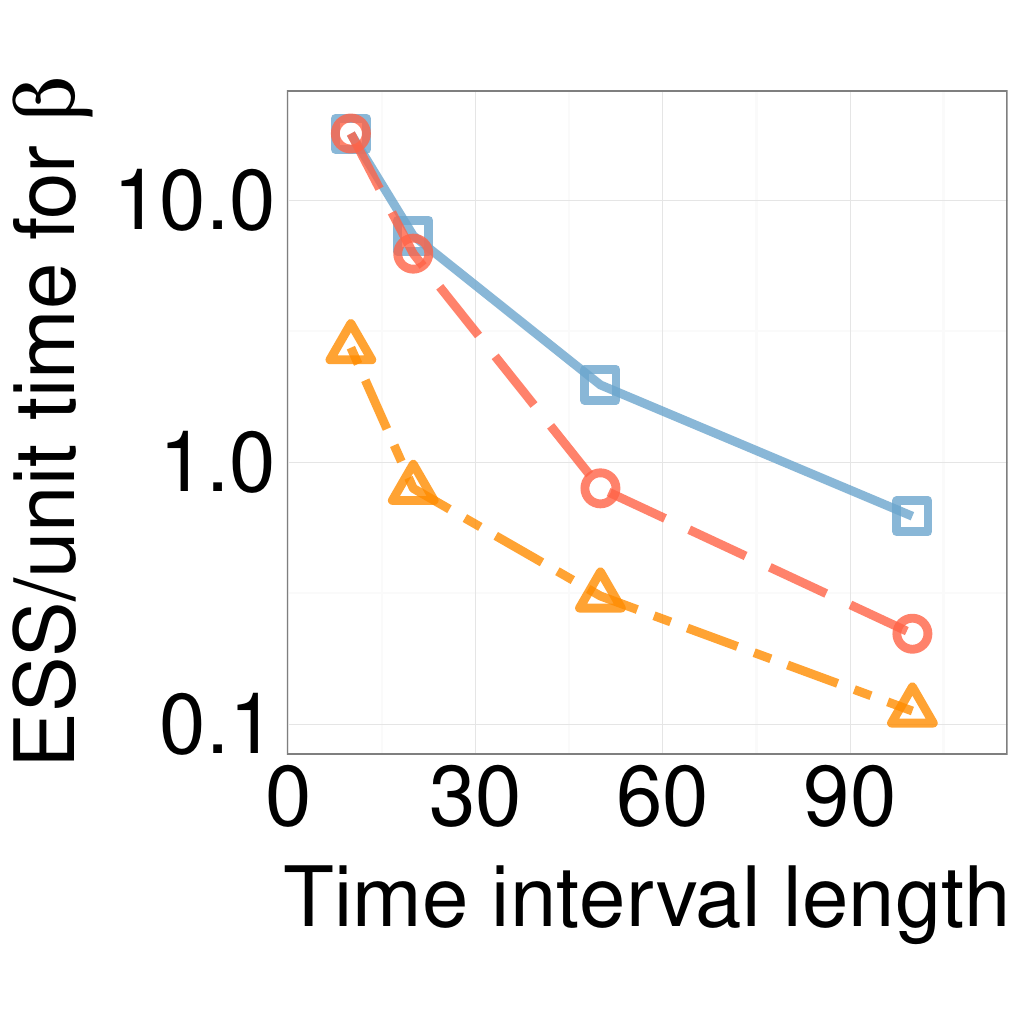}
  \end{minipage}
  \caption{Time interval vs ESS/sec for the synthetic MJP. The left two plots are for $\alpha$ and $\beta$, with the number of observations fixed; in the right two, this grows linearly with the interval length. {Blue squares, yellow triangles and red circles curves} are the symmetrized MH, \naive\ MH and Gibbs algorithm.
  }
     \label{fig:TSS}
  \end{figure}
In figure~\ref{fig:TSS}, we plot ESS per unit time as the observation interval $t_{end}$ increases. 
We consider the 3-state MJP, and as before there are $19$ observations uniformly located over a time interval $(0,t_{end})$.
We consider four settings, with $t_{end}$ equal to $10, 20, 50, 100$. 
For each, we compare our symmetrized MH sampler (with $\kappa$ set to $1$) with the \naive\ MH and Gibbs samplers (with $\kappa$ set to $2$). 
While the performance of the Gibbs sampler is comparable with our symmetrized algorithm for the smallest value of $t_{end}$, its performance is considerably worse for longer time-intervals.  
This is the limitation of Gibbs sampling that motivated this work: when updating $\theta$ conditioned on the MJP trajectory, longer time intervals result in stronger coupling between MJP path and parameters (figure~\ref{fig:hist}), and thus poorer mixing. 
The performance of the \naive\ sampler demonstrates that it is not sufficient just to integrate out the state values of the trajectory, we also have to get around the coupling between the Poisson grid and the parameters. 
Our symmetrized MH-algorithm allows this. 

To the right of figure~\ref{fig:TSS}, we plot results from a similar experiment. 
Now, instead of keeping the number of measurements fixed as we increase the observation interval, we keep the observation {\em rate} fixed at one observation every unit interval of time, so that longer observation intervals have larger number of observations. 
The results are similar to the previous case: Gibbs sampling performs well for small observation intervals, with performance degrading sharply for larger intervals. 

  \subsection{The Jukes and Cantor (JC69) model}~
  The Jukes and Cantor (JC69) model~\citep{jukescantor69} is a popular model of DNA nucleotide substitution.  
  We write its state space as $\{0, 1, 2, 3\}$, representing the four nucleotides $\{A, T, C, G\}$.  
  The model has a single parameter $\alpha$, representing the rate at which the system transitions between any pair of states. 
  Thus, the rate matrix $A$ is given by $A_i = -A_{i,i} = 3\alpha, A_{i, j} = \alpha,i \neq j.$
We place a Gamma$(3,2)$ prior on the parameter $\alpha$.
Figures~\ref{fig:ESS_JC}(b) and (c) compare different samplers: we again see that the symmetrized MH samplers comprehensively outperforms all others.
Part of the reason why the difference is so dramatic here is because now a {\em single} paramter $\alpha \defeq \theta$ defines the transition matrix, implying a stronger coupling between MJP path and parameter. 
We point out that for Gibbs sampling, the conditional distribution over $\theta$ is conjugate to the Gamma prior. We can thus simulate directly from this distribution without any MH proposal (hence its performance remains fixed along the x-axis). 
Despite this, its performance is worse than our symmetrized algorithm.
Particle MCMC performs worse than all the algorithms, and we do not include it in our plots.
Figure~\ref{fig:ESS_JC}(d) compares different settings of $\Omega(\theta,\vartheta)$ for our sampler: again, the simple additive setting $\Omega(\theta,\vartheta) = \max_s A_s(\theta) + \max_s A_s(\vartheta)$ does best.

   \begin{figure}
  \begin{minipage}[!hp]{0.24\linewidth}
	\centering
    \includegraphics[width=0.99\textwidth, angle=0]{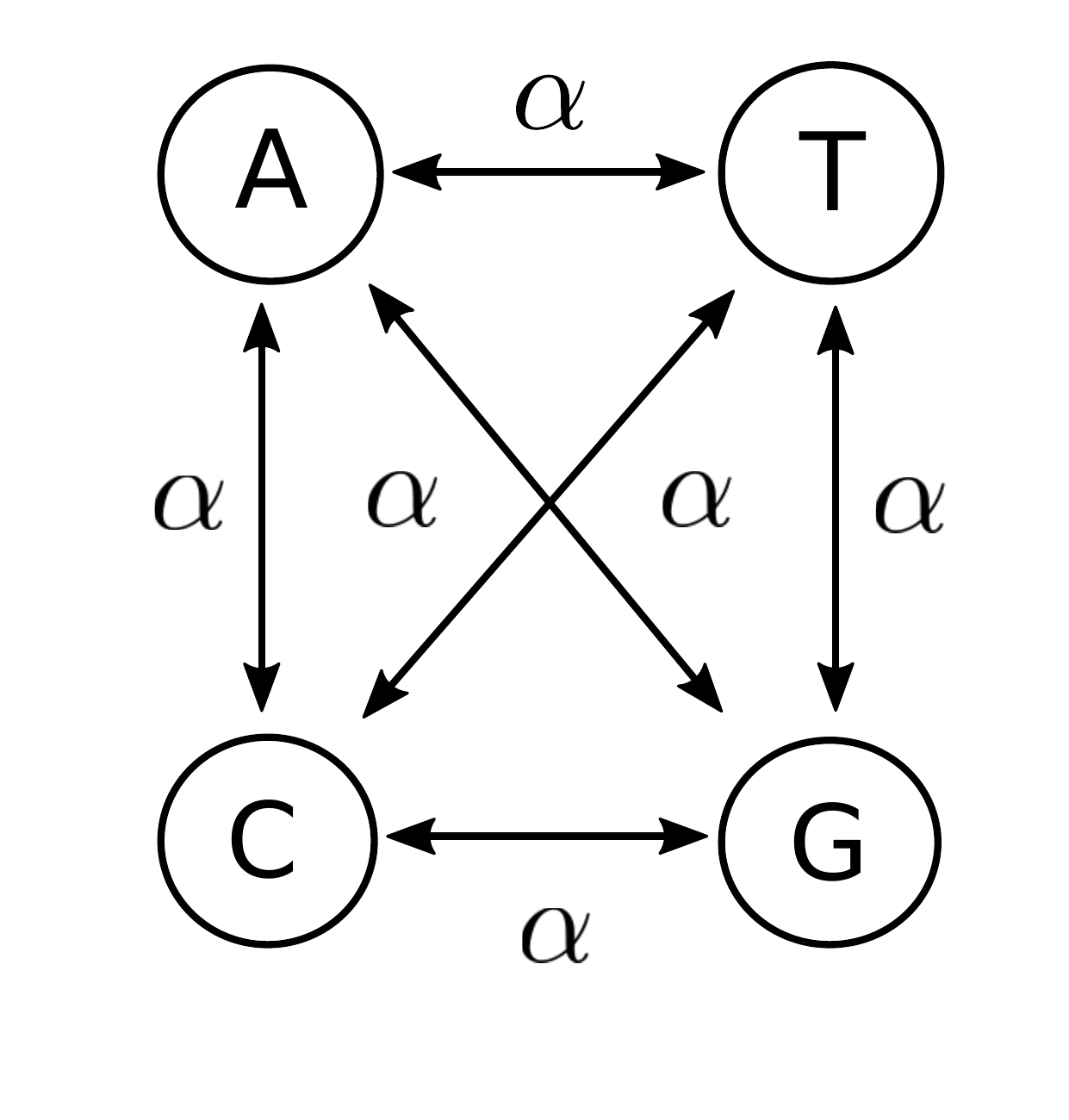} 
	\end{minipage}
	  \begin{minipage}[!hp]{0.24\linewidth}
	\centering
    \includegraphics[width=0.99\textwidth, angle=0]{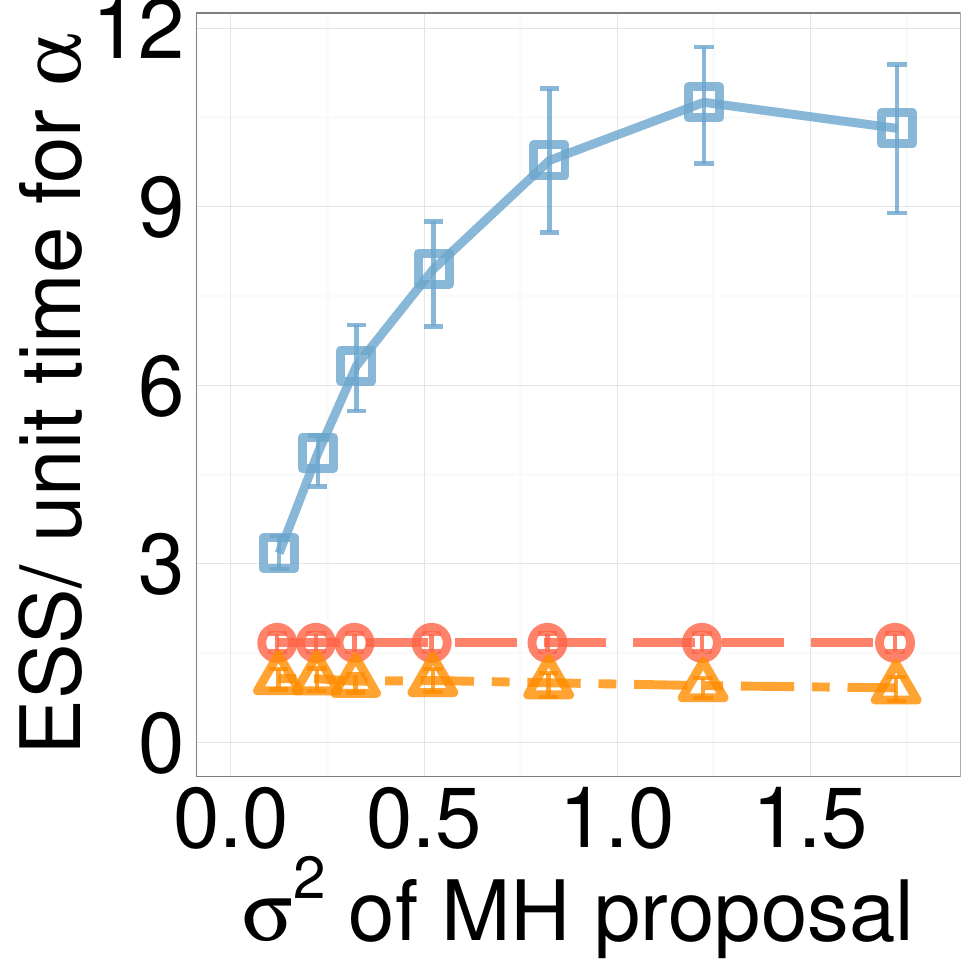}
  \end{minipage}
  \begin{minipage}[!hp]{0.24\linewidth}
	\centering
    \includegraphics[width=0.99\textwidth, angle=0]{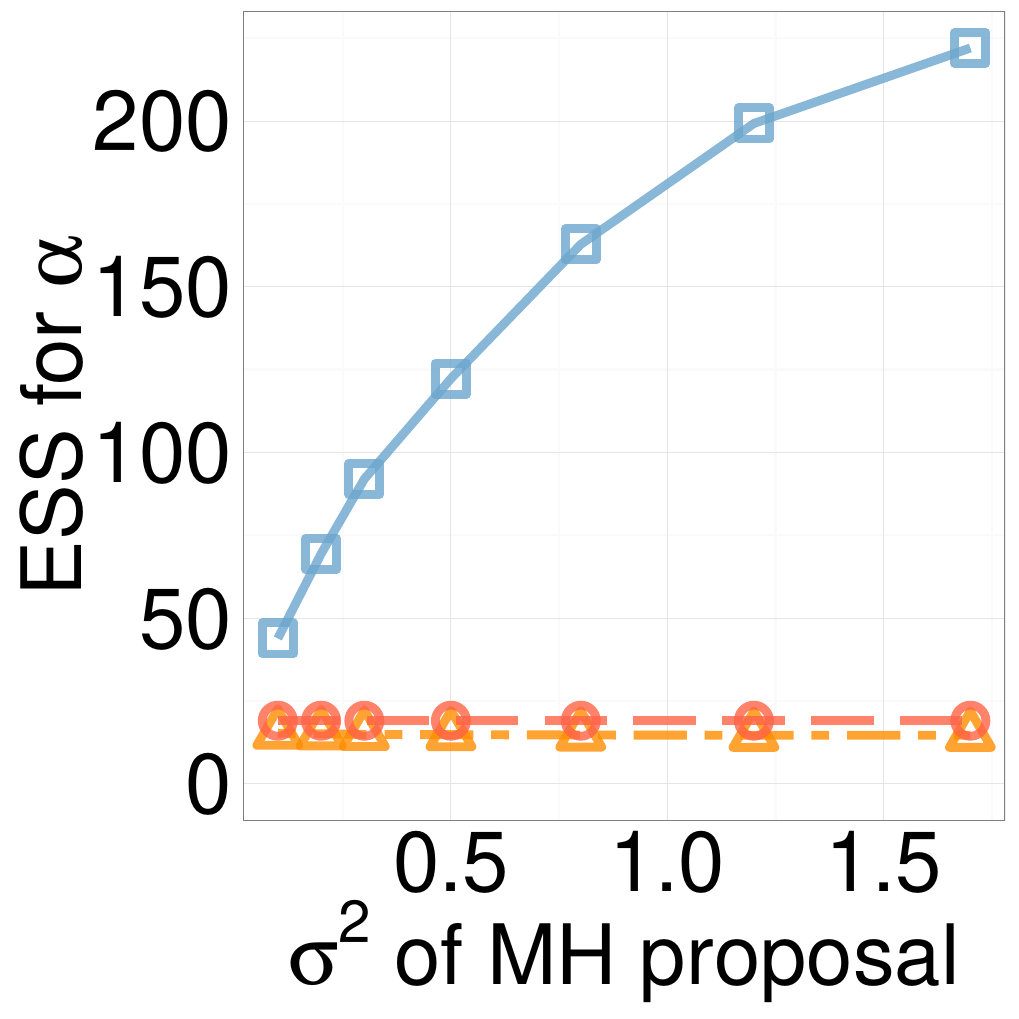}
\end{minipage}
  \begin{minipage}[!hp]{0.24\linewidth}
	\centering
    \includegraphics[width=0.99\textwidth, angle=0]{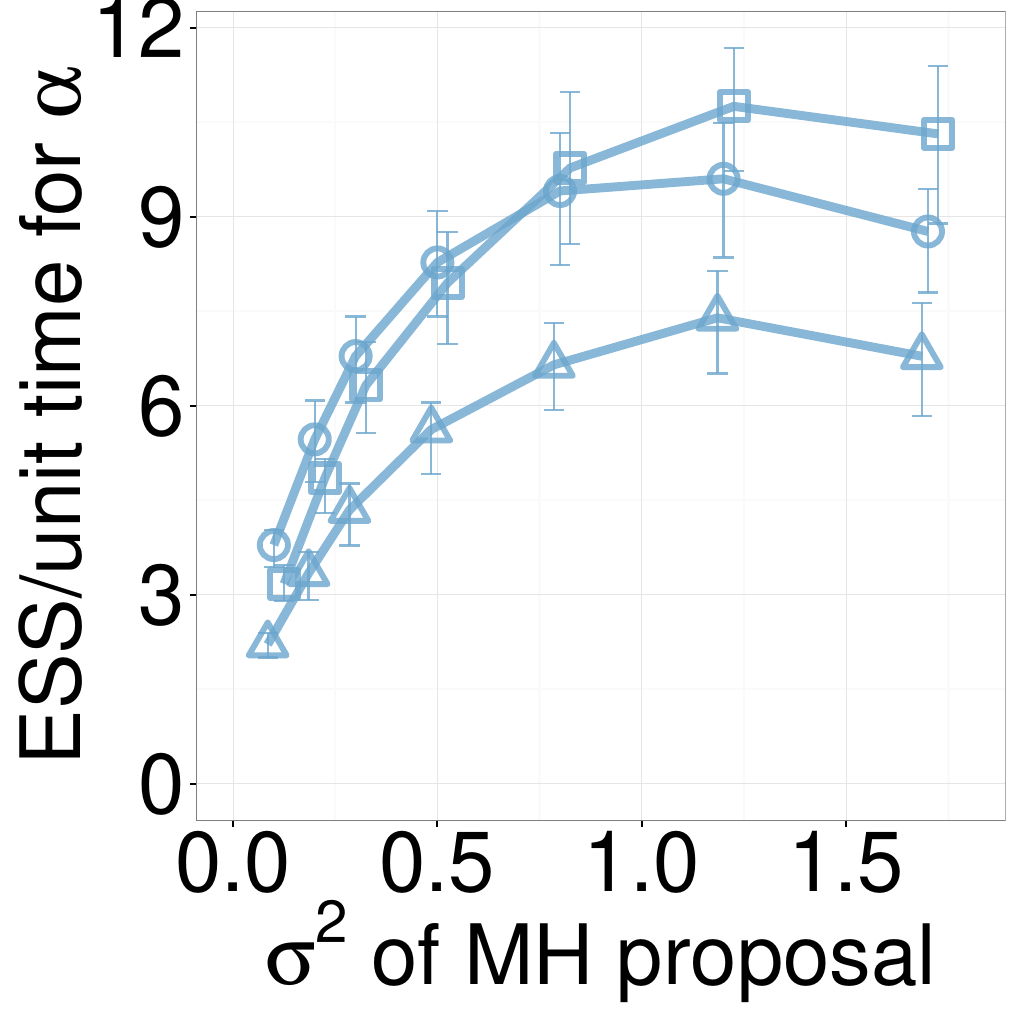}
\end{minipage}
  \caption{The leftmost panel is the Jukes-Cantor (JC69) model. The next two panels from left to right are ESS/sec and ESS per 1000 samples for this. 
    Blue squares, yellow triangles and red circles are the symmetrized MH, \naive\ MH and Gibbs algorithm.
    The rightmost panel looks at different settings of the symmetrized MH algorithm, with squares, circles and triangles corresponding to 
$\Omega(\theta,\vartheta)$ set to $(\max_s A_s(\theta) + \max_s A_s(\vartheta))$, $\max(\max_s A_s(\theta), \max_s A_s(\vartheta))$ and  $1.5(\max_s A_s(\theta) + \max_s A_s(\vartheta))$.
     \label{fig:ESS_JC}
   }
 \end{figure}

  \begin{figure}[H]
  \centering

  \begin{minipage}[!hp]{0.99\linewidth}
  	\centering
    \includegraphics [width=0.24\textwidth, angle=0]{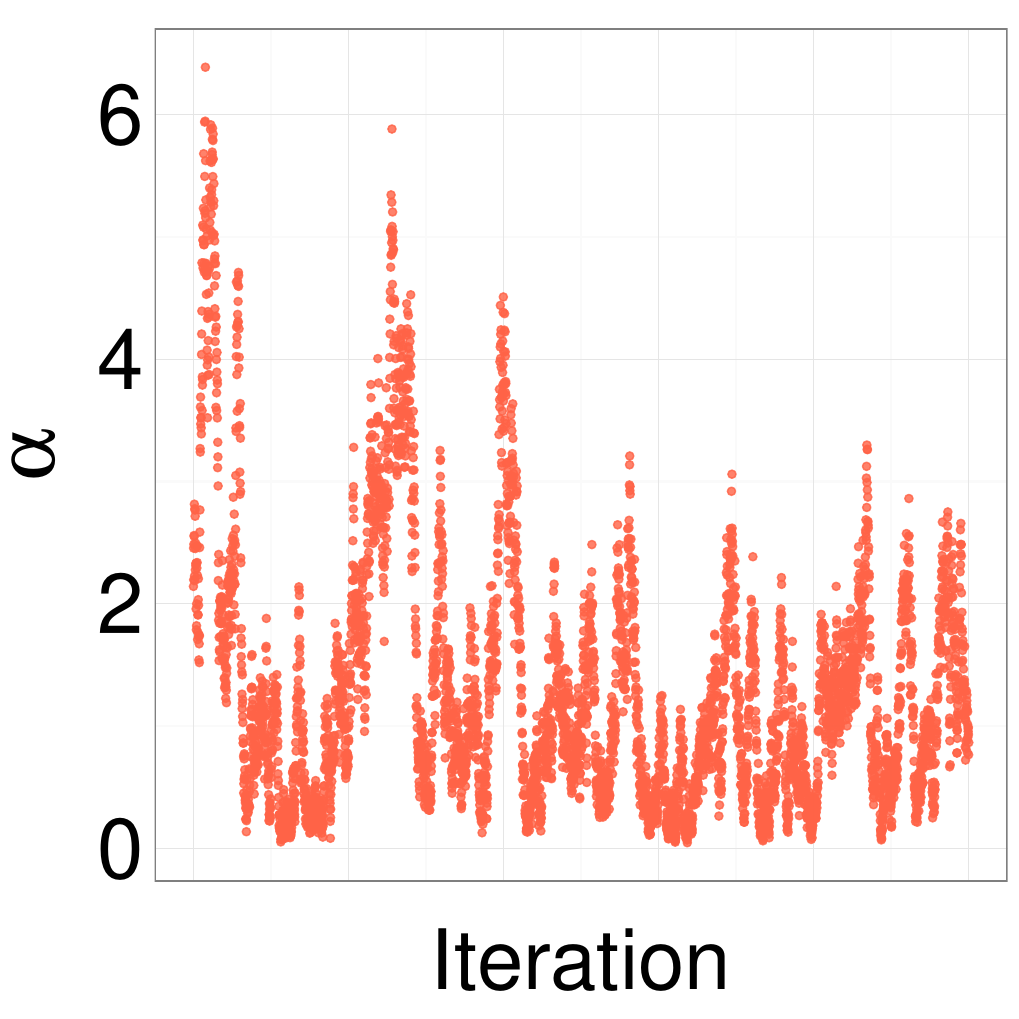}
    \includegraphics [width=0.24\textwidth, angle=0]{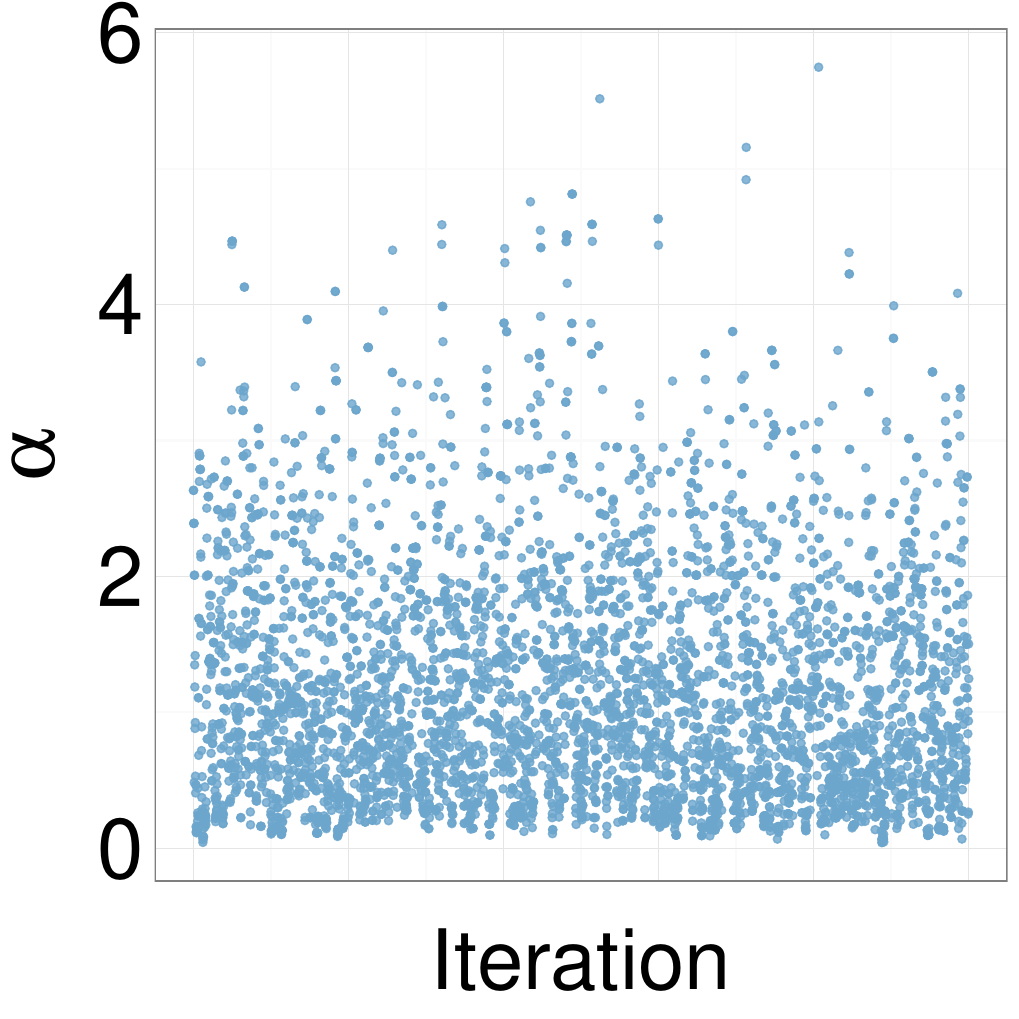}
    \includegraphics [width=0.24\textwidth, angle=0]{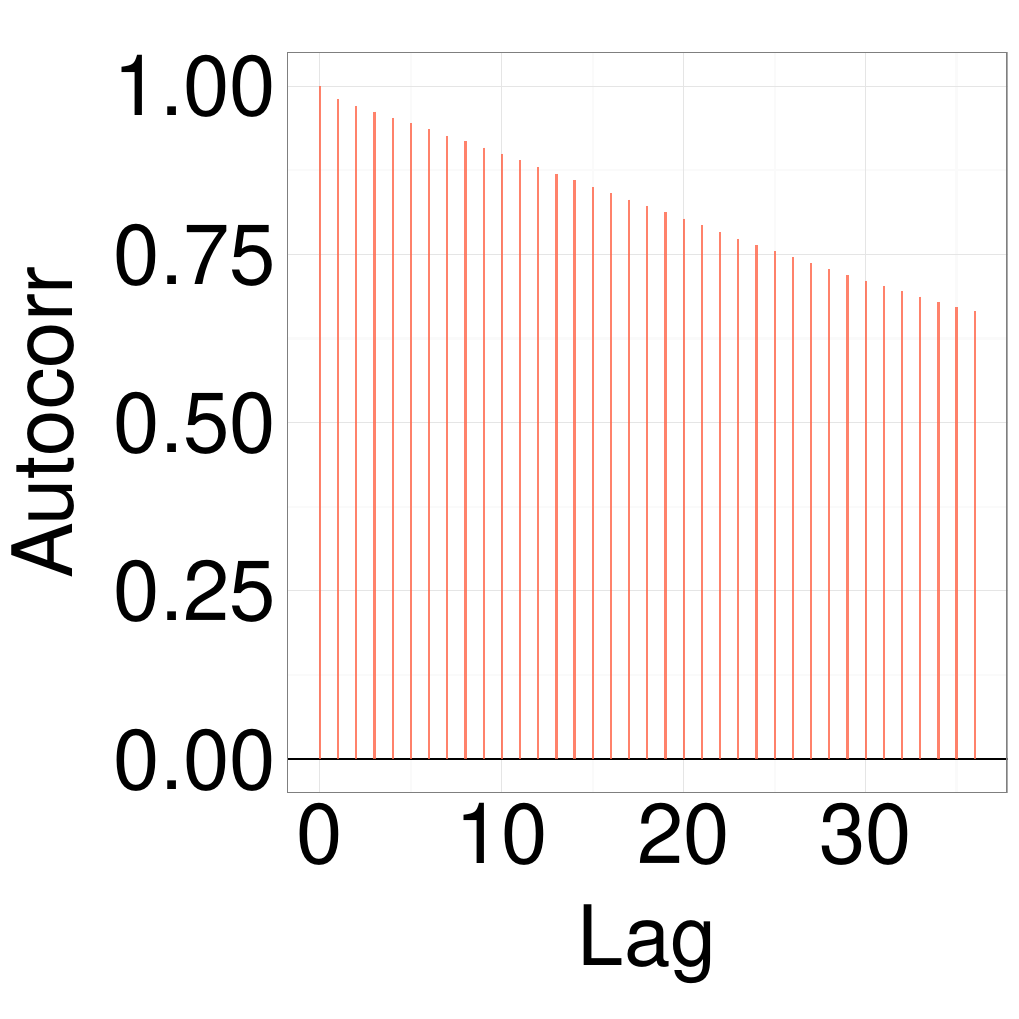}
    \includegraphics [width=0.24\textwidth, angle=0]{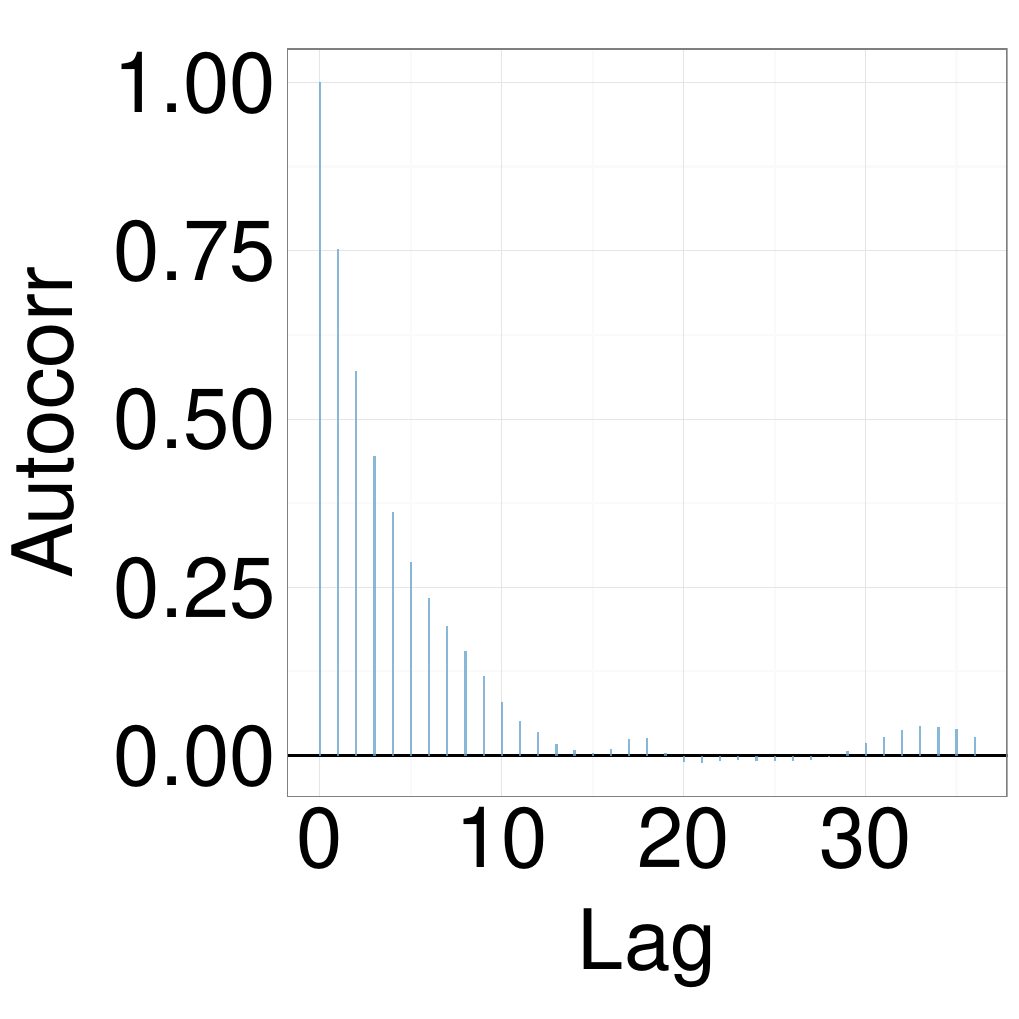}
  \end{minipage}
    \caption{Trace and autocorrelation plots of $\alpha$ for the JC69 model. Left two panels are for Gibbs and the right two for the symmetrized MH algorithm.}
    \label{fig:ACC_JC_m}
    \vspace{-.15in}
  \end{figure}
  Figure~\ref{fig:ACC_JC_m} plots MCMC diagnostics for the Gibbs and symmetrized MH sampler, confirming the previous findings. 
  Both agree on the posterior $P(\alpha|X)$ (figure~\ref{fig:jc_model_vs_t}(a)), with a two-sample Kolmogorov-Smirnov test giving a p-value of $0.97$. 
  Figure~\ref{fig:jc_model_vs_t}(b) plots the average MH acceptance probabilities for the \naive\ and symmetrized MH samplers for different settings of the proposal distribution, again we see that the former has lower acceptance rates because of the $P(W|\theta)$ terms.
  \begin{figure}
  \centering
  \begin{minipage}[!hp]{0.99\linewidth}
  \centering
    \includegraphics [width=0.24\textwidth, angle=0]{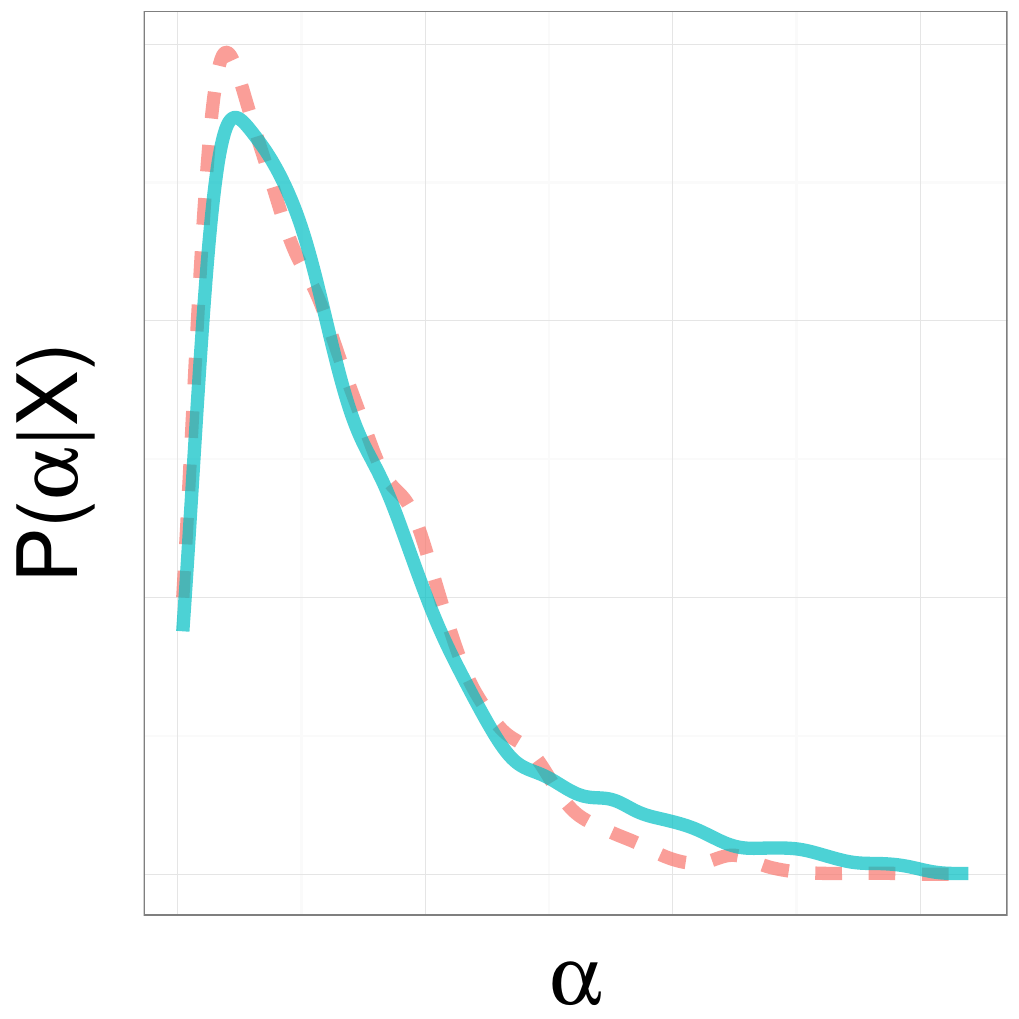}
    \includegraphics [width=0.24\textwidth, angle=0]{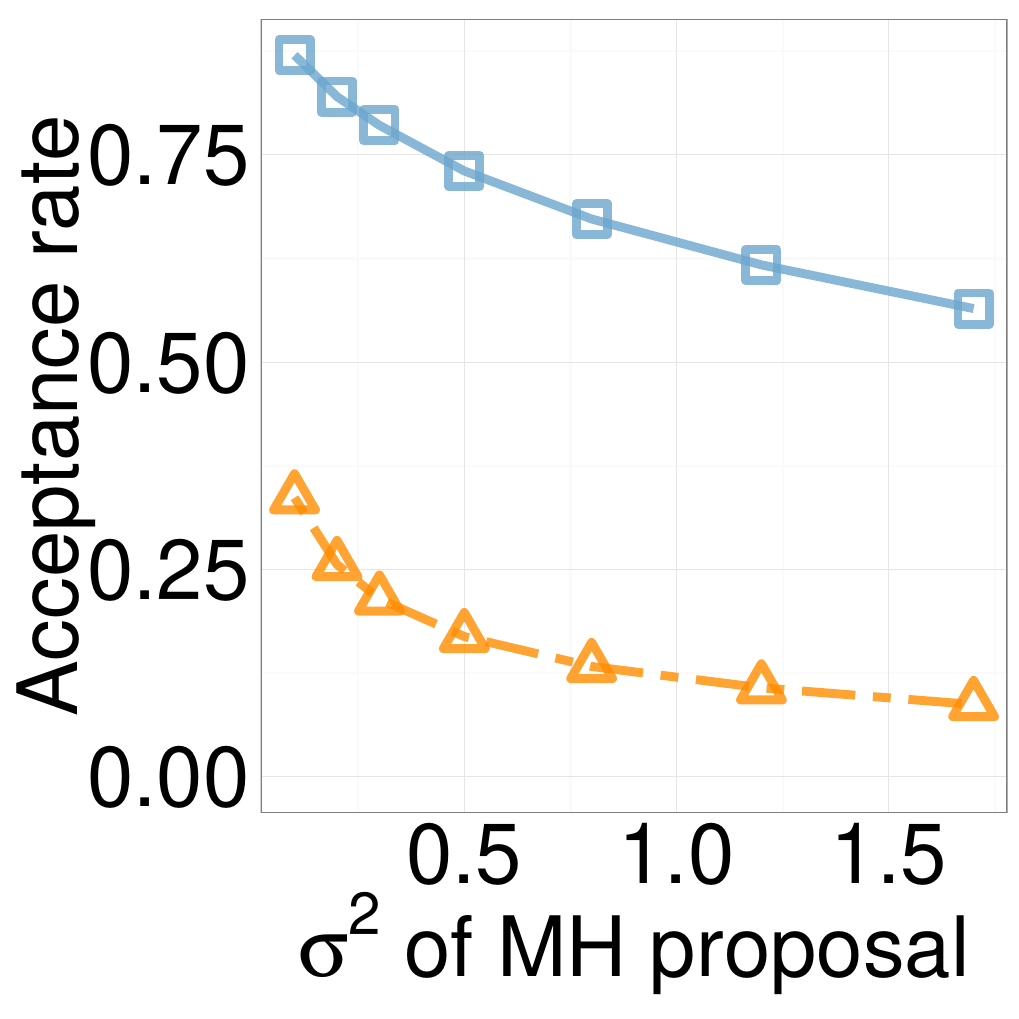}
    \includegraphics [width=0.24\textwidth, angle=0]{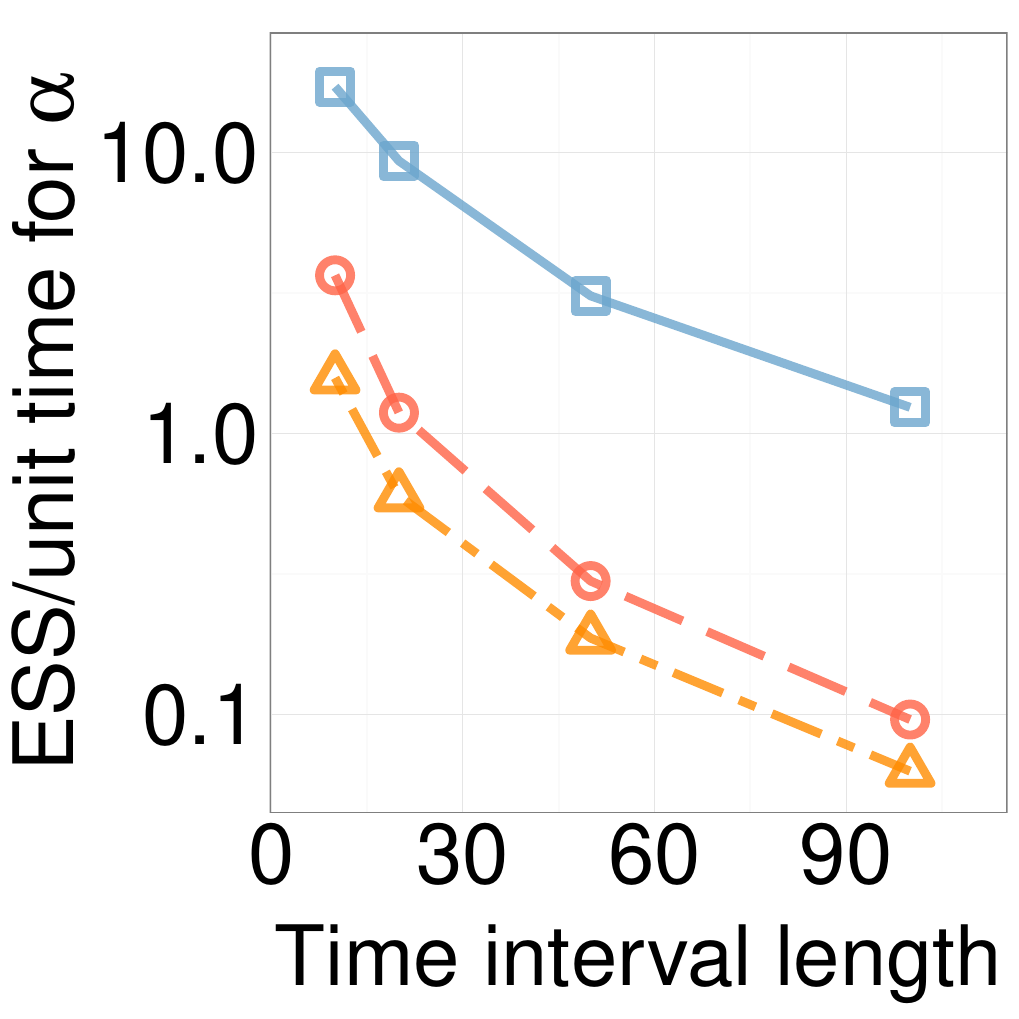}
    \includegraphics [width=0.24\textwidth, angle=0]{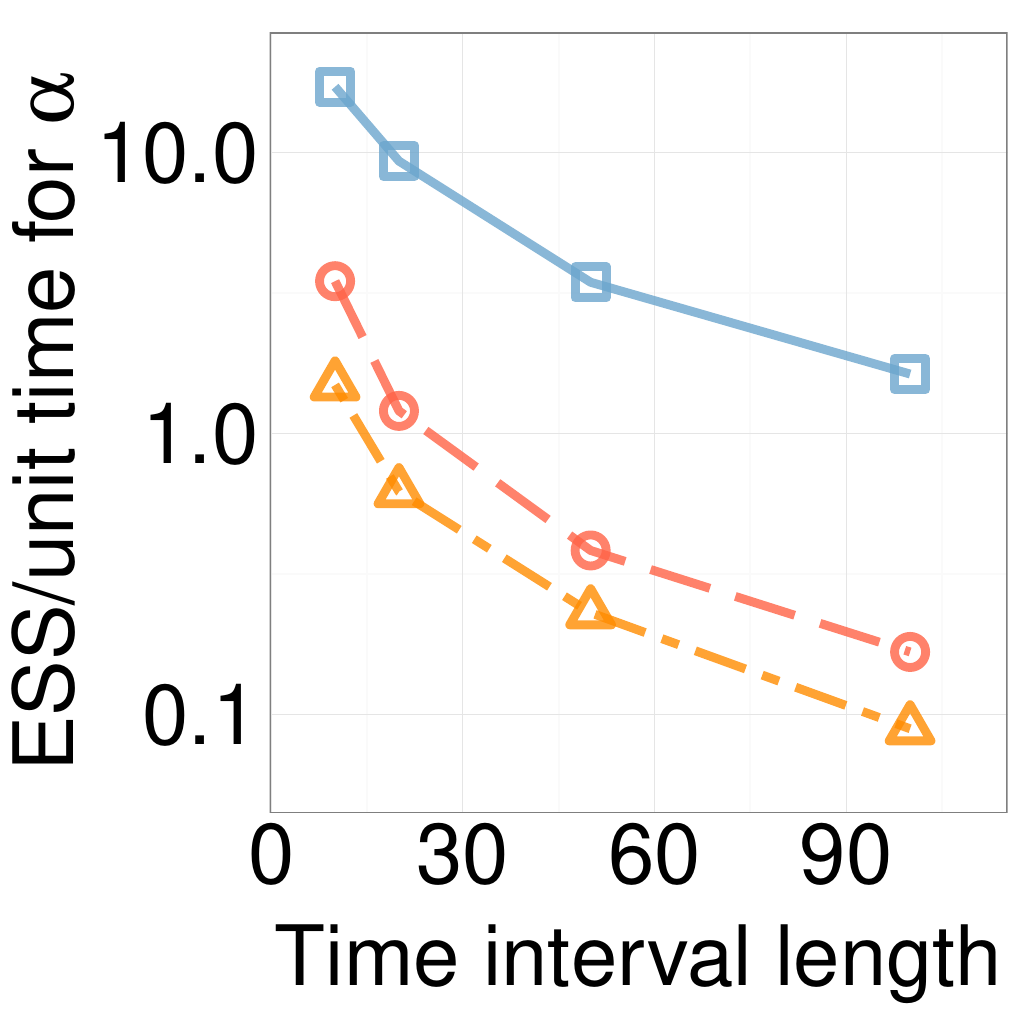}
  \end{minipage}

  \caption{(a) Posterior $P(\alpha|X)$ in the JC69 model for Gibbs (dashed) and symmetrized MH (continuous). (b) MH acceptance rates for \naive\  and symmetrized MH. (c) and (d):
  ESS/sec against $t_{end}$ for $\kappa=2$ with: (c) number of observations fixed, and (d) observation rate fixed. {Squares, triangles and circles} are symmetrized MH, \naive\ MH and Gibbs. }
	\label{fig:jc_model_vs_t}
  \end{figure}
  Figures~\ref{fig:jc_model_vs_t}~(c) and (d) plot the ESS per unit time for the
  different samplers as $t_{end}$ increases. 
  The left plot keeps the number of observations fixed, while the right keeps the observation rate fixed. 
Once again we see that our proposed algorithm
1) performs best over all interval lengths, and 2) suffers a performance
degradation with interval length that is much milder than the other algorithms.

\subsection{An immigration model with finite capacity}\label{sec:immig}~
  \begin{figure}[H]
  \centering
  \begin{minipage}[!hp]{0.24\linewidth}
    \includegraphics [width=0.99\textwidth, angle=0]{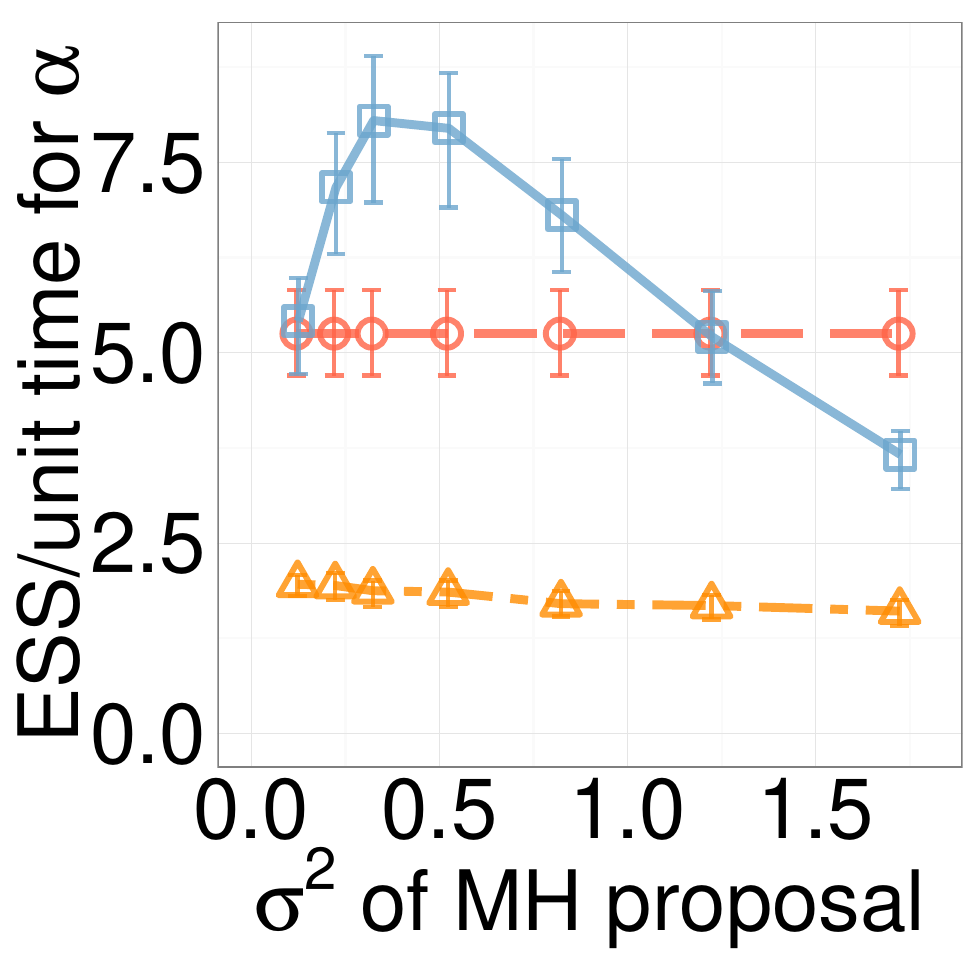}
\end{minipage}
  \begin{minipage}[hp]{0.24\linewidth}
  \centering
    \includegraphics [width=0.99\textwidth, angle=0]{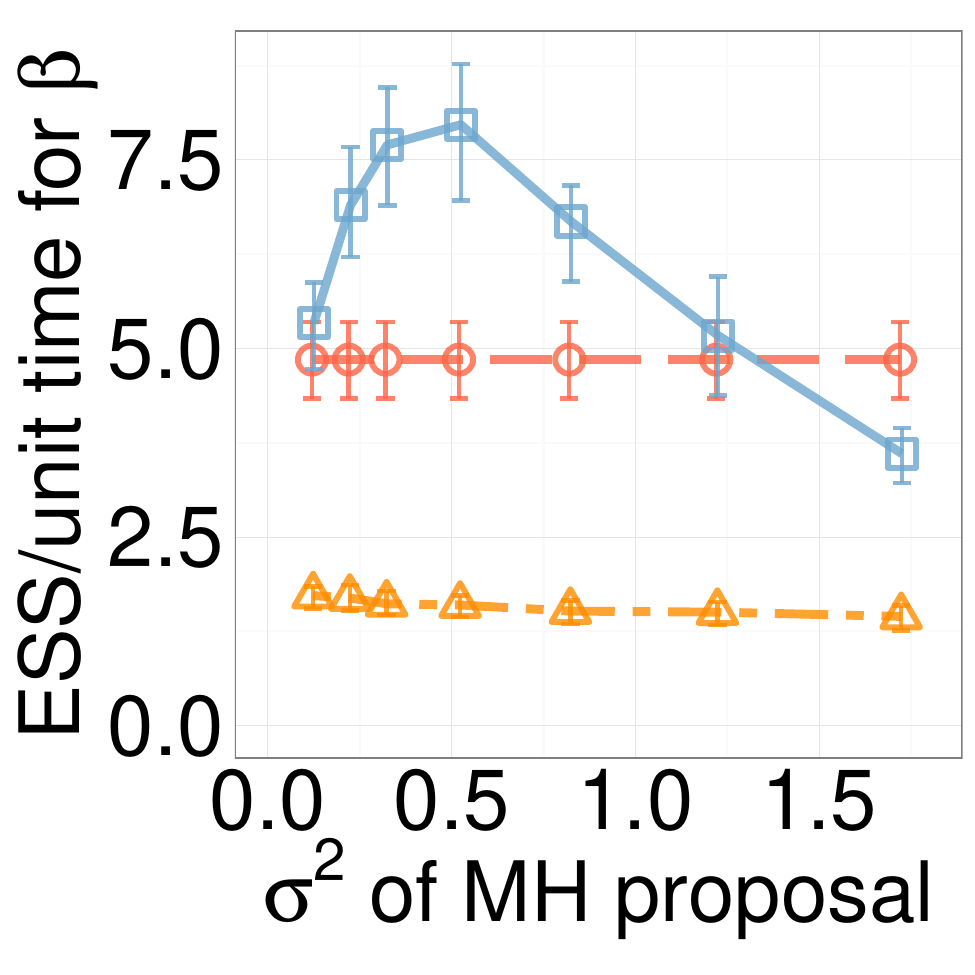}
	\end{minipage}
  \begin{minipage}[hp]{0.24\linewidth}
  \centering
    \includegraphics [width=0.99\textwidth, angle=0]{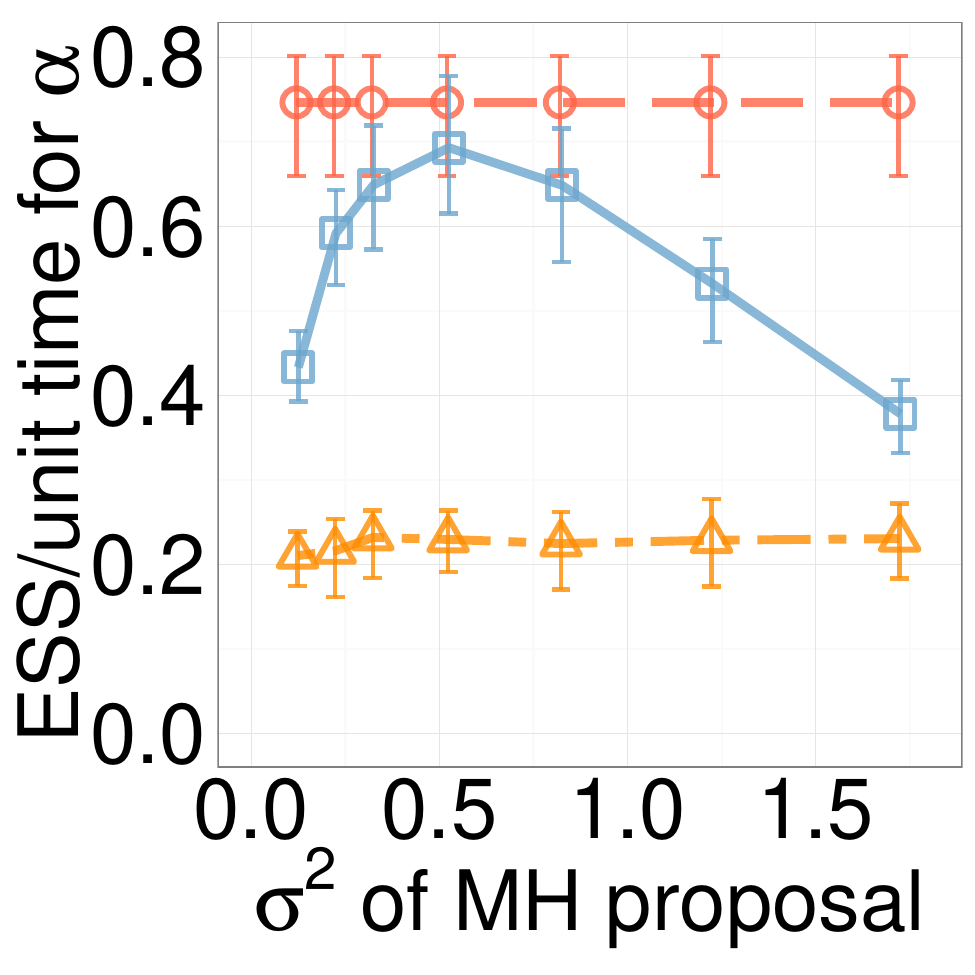}
	\end{minipage}
  \begin{minipage}[hp]{0.24\linewidth}
  \centering
    \includegraphics [width=0.99\textwidth, angle=0]{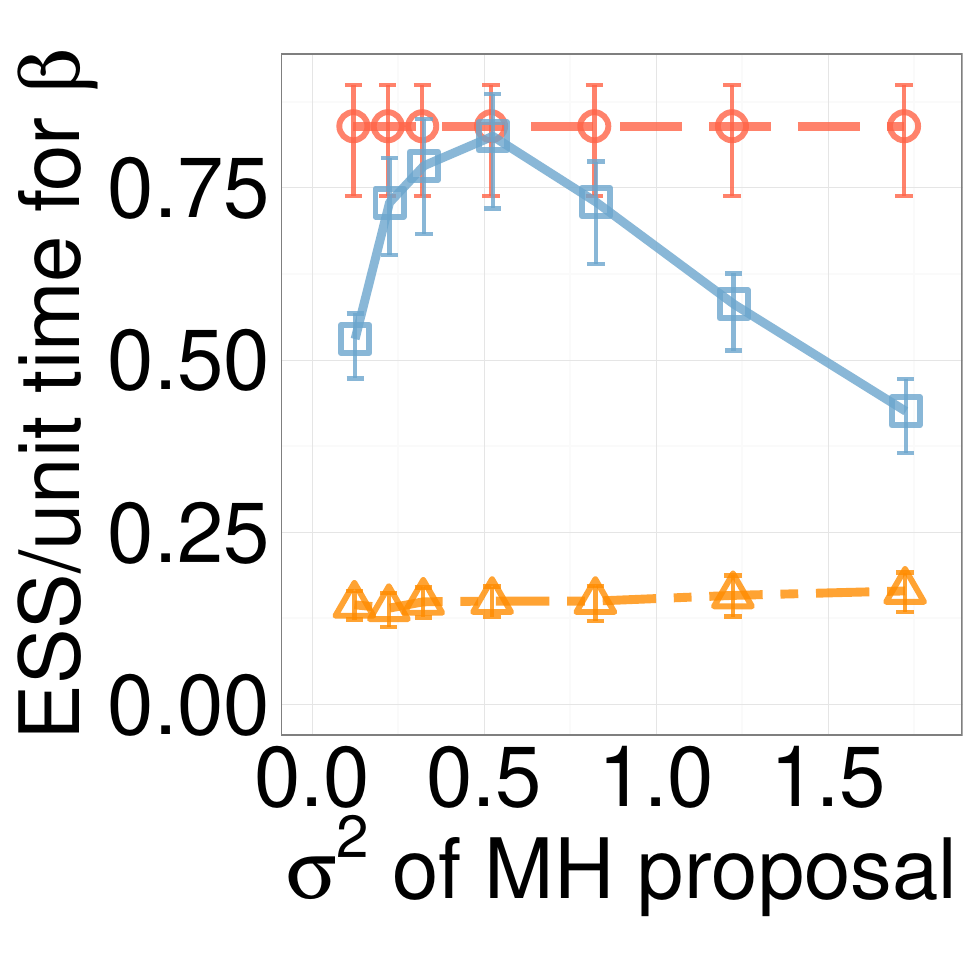}
	\end{minipage}
  \centering
  \begin{minipage}[!hp]{0.24\linewidth}
  \centering
    \includegraphics [width=0.99\textwidth, angle=0]{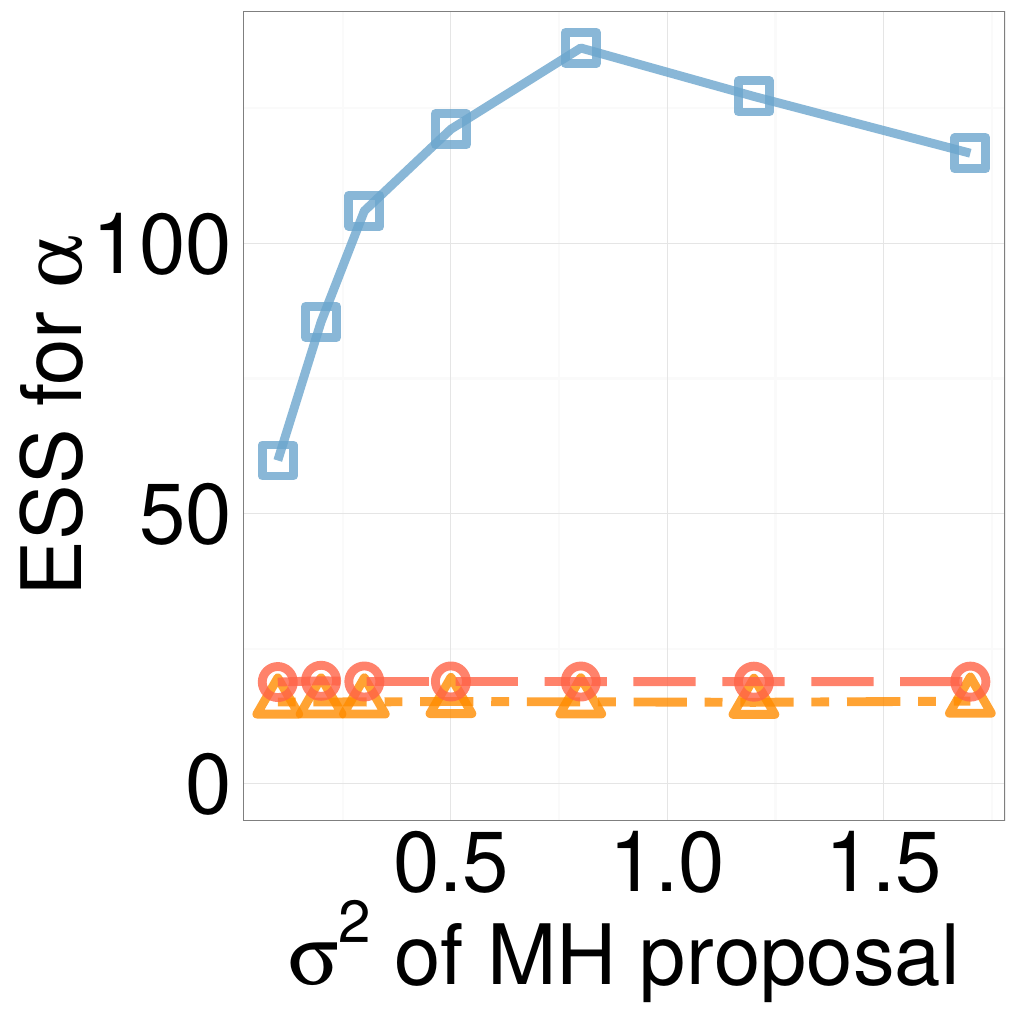}
\end{minipage}
  \begin{minipage}[hp]{0.24\linewidth}
  \centering
    \includegraphics [width=0.99\textwidth, angle=0]{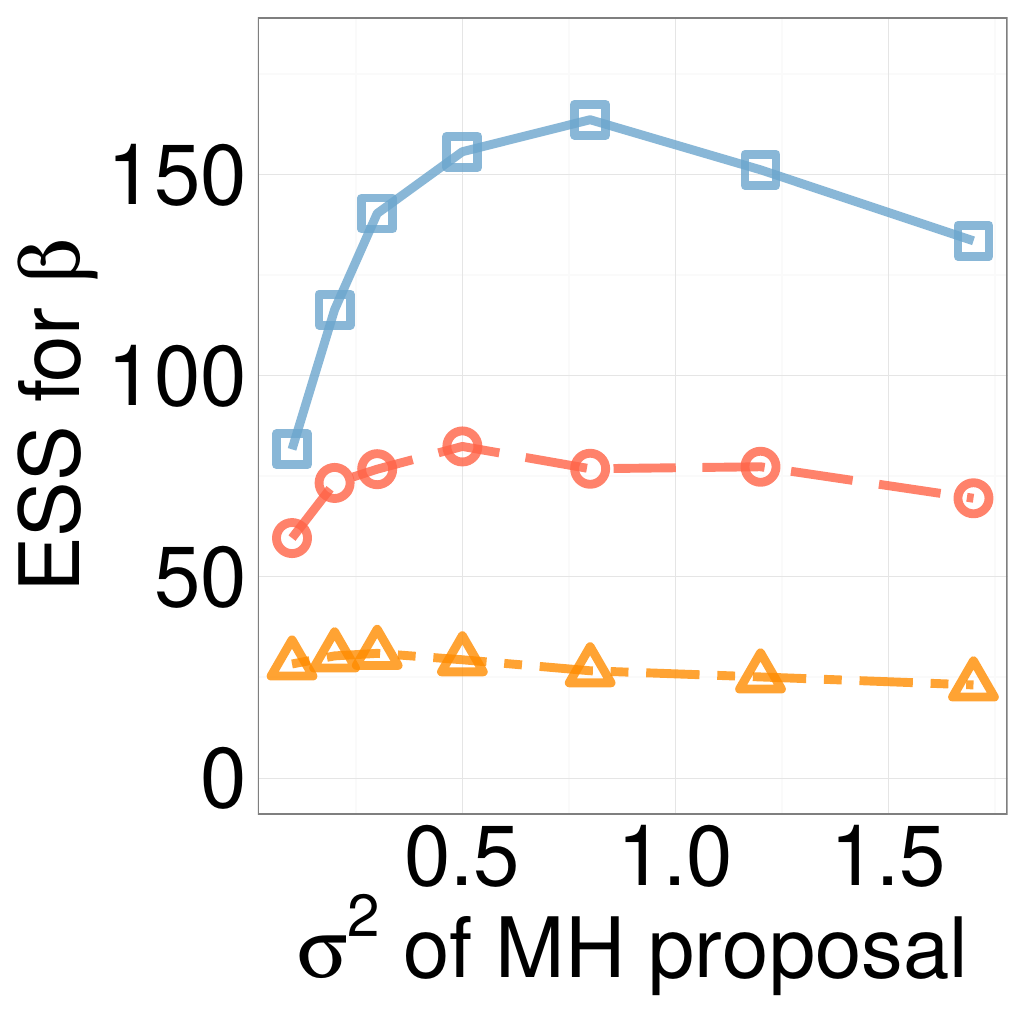}
	\end{minipage}
  \begin{minipage}[hp]{0.24\linewidth}
  \centering
    \includegraphics [width=0.99\textwidth, angle=0]{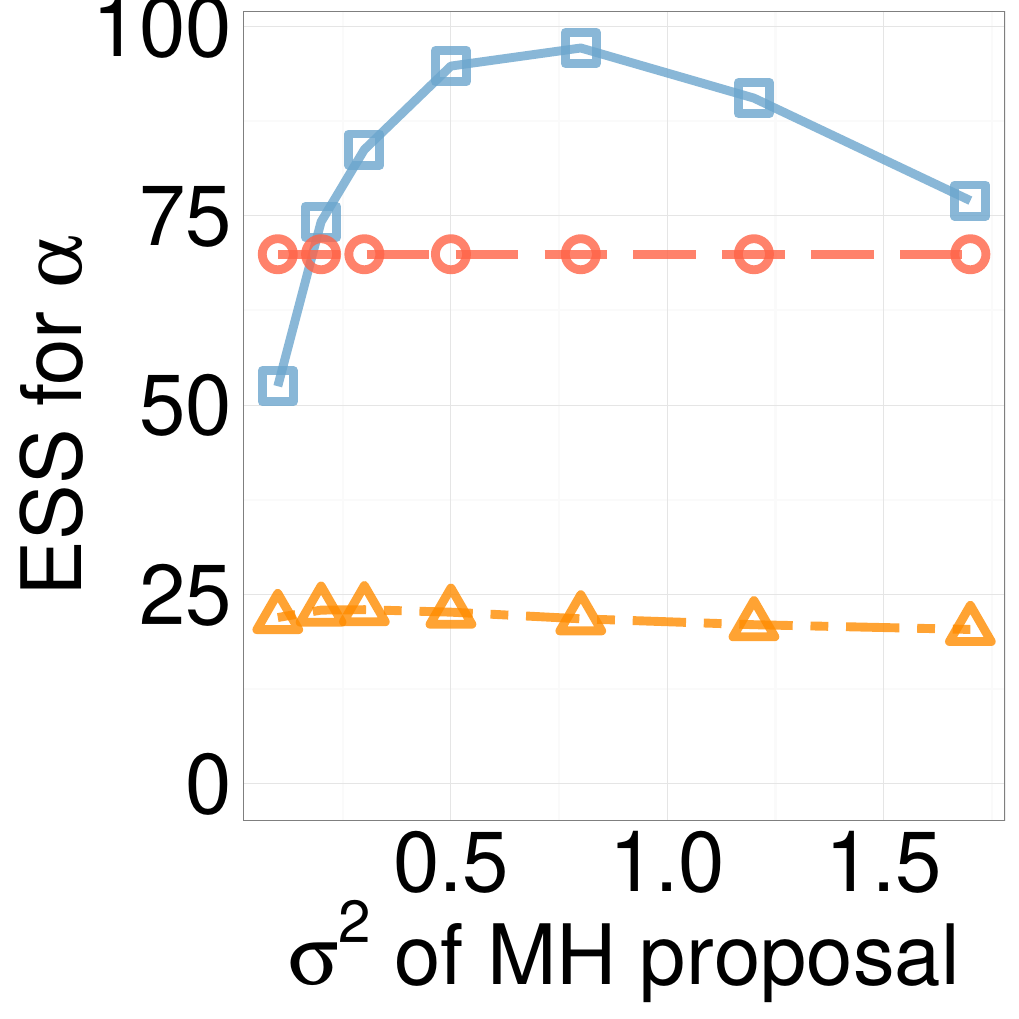}
	\end{minipage}
  \begin{minipage}[hp]{0.24\linewidth}
  \centering
    \includegraphics [width=0.99\textwidth, angle=0]{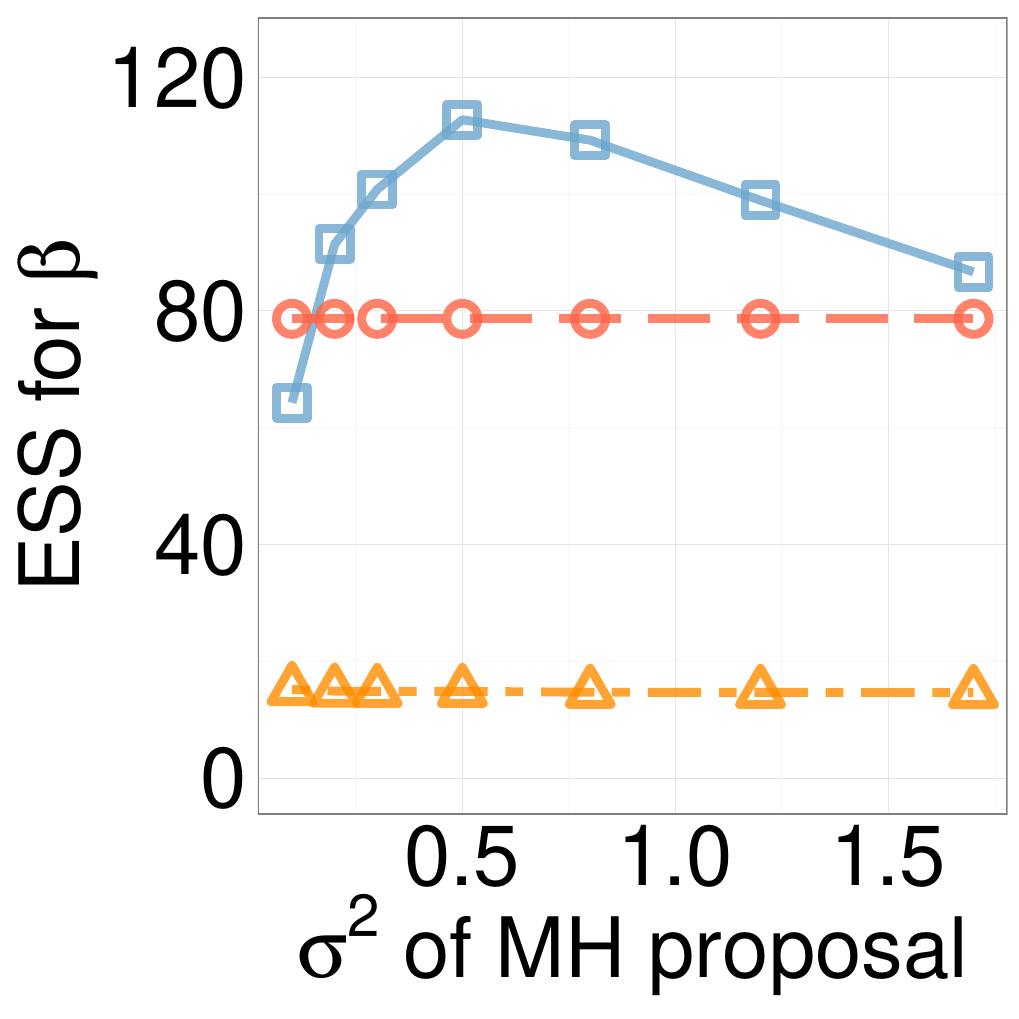}
	\end{minipage}
    \caption{ESS/sec (top row) and ESS per 1000 samples (bottom row) for the immigration model. The left two columns are $\alpha$ and $\beta$ for 3 states, and the right two, for 10 states.
    Squares, triangles and circles are symmetrized MH, \naive\ MH, and Gibbs algorithm. }
     \label{fig:ESS_Q_D10}
  \end{figure}

  \begin{figure}[H]
  \centering
  \begin{minipage}[!hp]{0.24\linewidth}
  \centering
    \includegraphics [width=0.99\textwidth, angle=0]{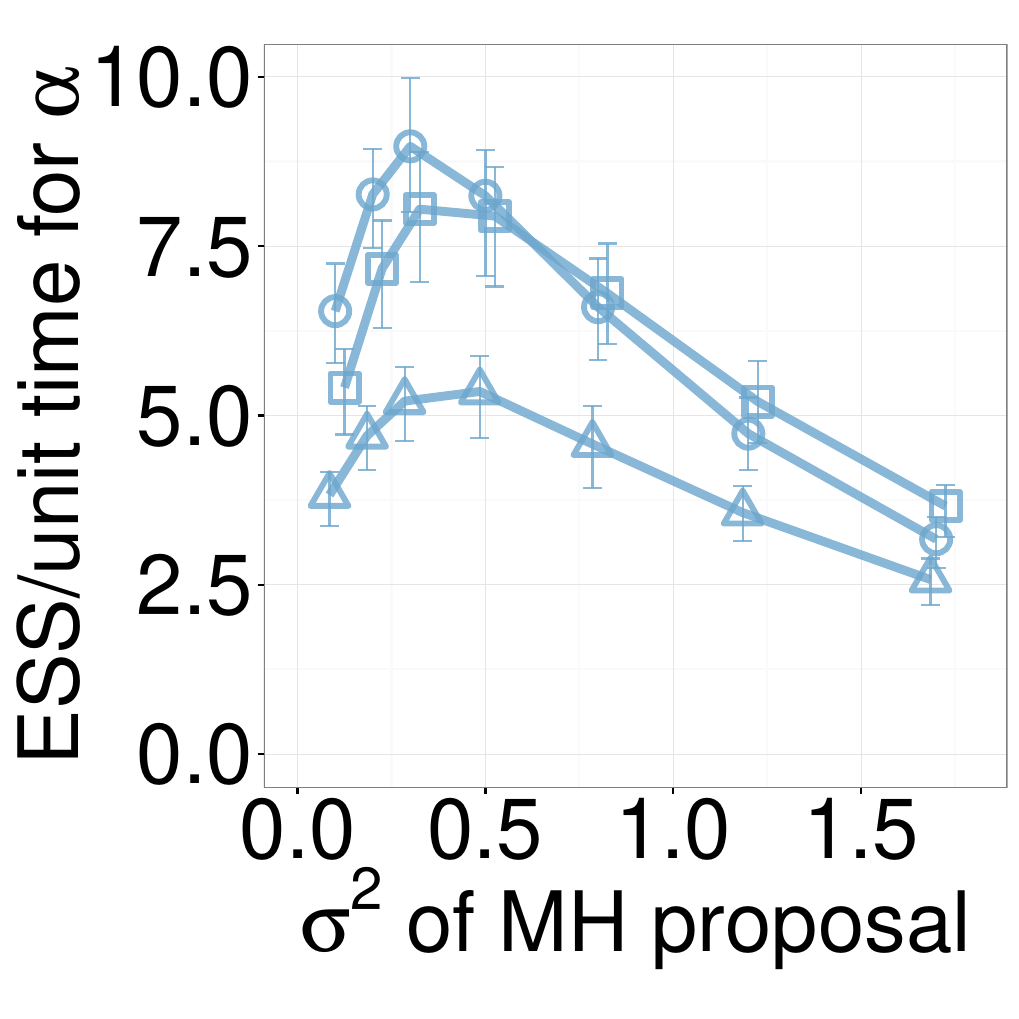}
\end{minipage}
  \begin{minipage}[hp]{0.24\linewidth}
  \centering
    \includegraphics [width=0.99\textwidth, angle=0]{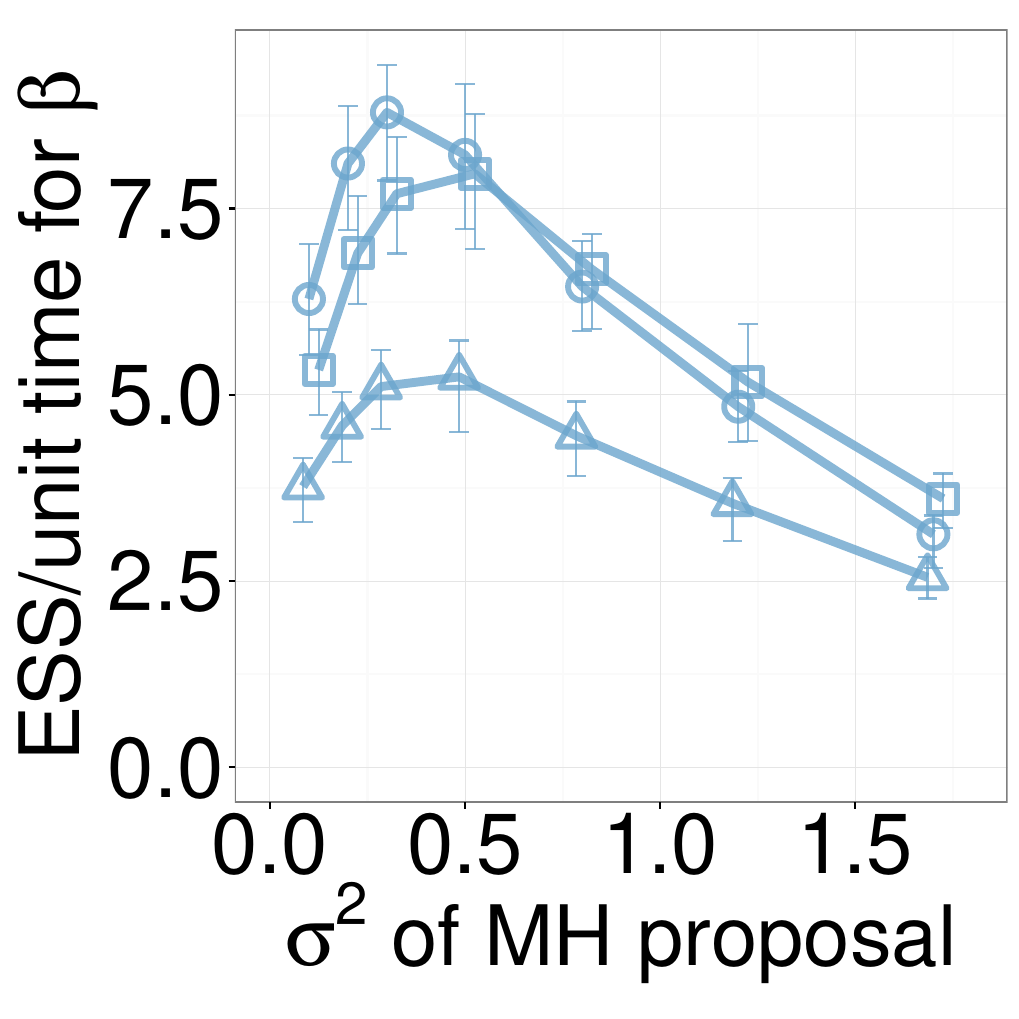}
	\end{minipage}
  \begin{minipage}[hp]{0.24\linewidth}
  \centering
    \includegraphics [width=0.99\textwidth, angle=0]{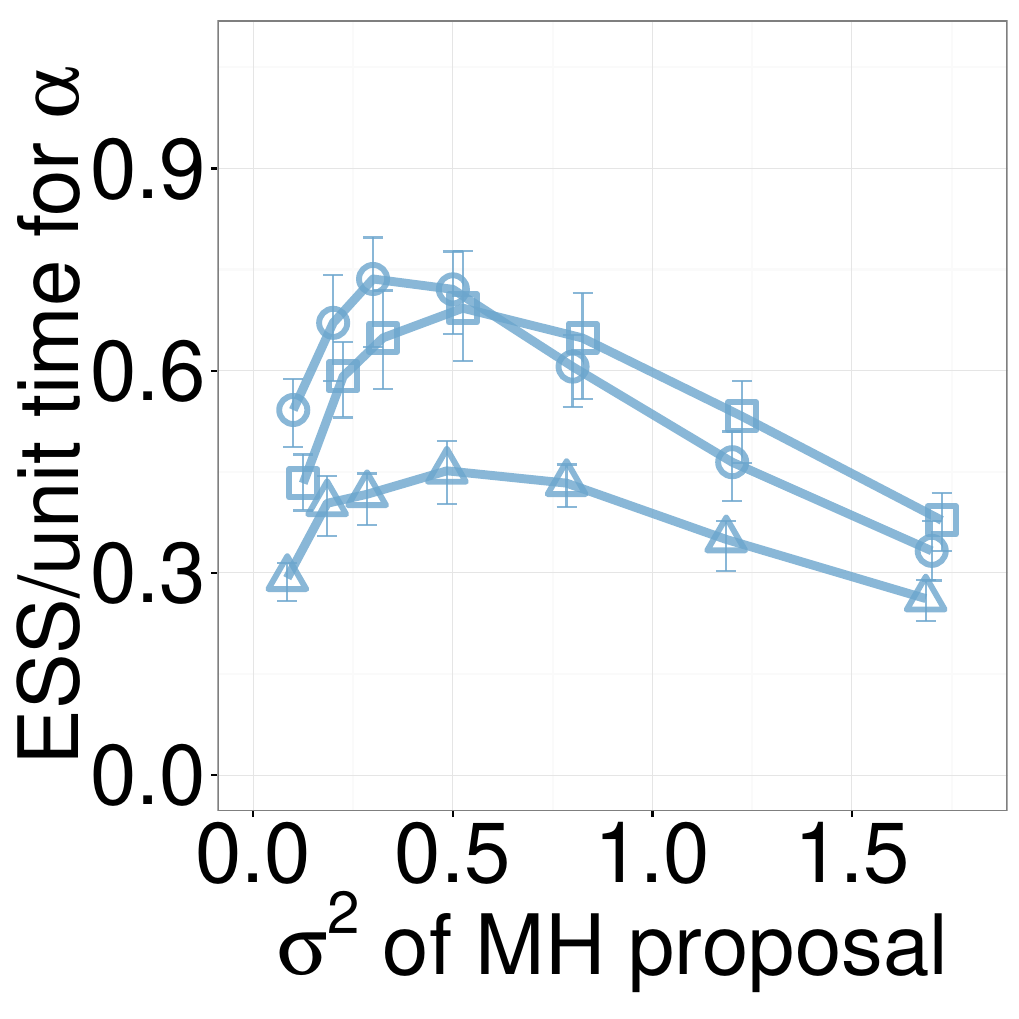}
	\end{minipage}
  \begin{minipage}[hp]{0.24\linewidth}
  \centering
    \includegraphics [width=0.99\textwidth, angle=0]{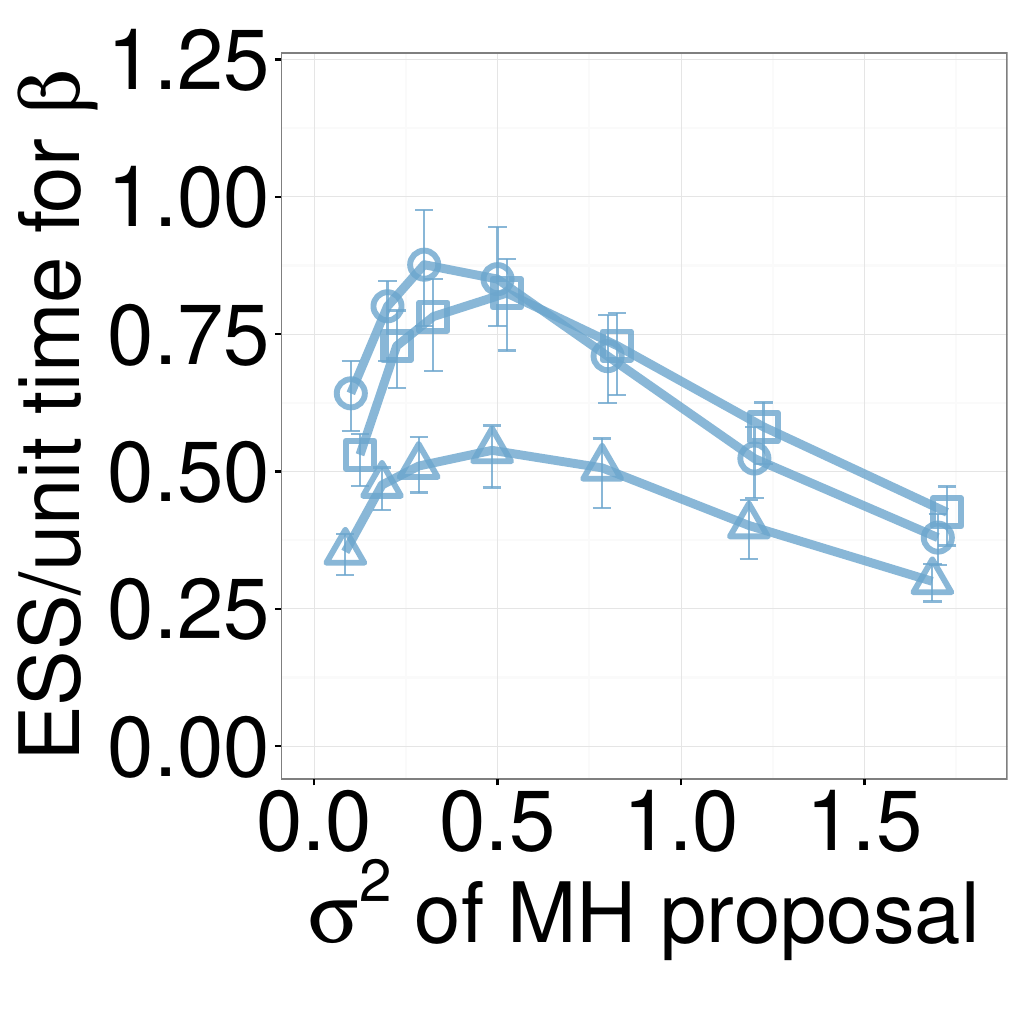}
	\end{minipage}
    \caption{ESS/sec for symmetrized MH for the immigration model for different settings of $\Omega(\theta,\vartheta)$. The left two columns are for $\alpha$ and $\beta$ with 3 states, and the right two, with 10. 
    Squares, circles and triangles correspond to $\Omega(\theta,\vartheta)$ set to $(\max_s A_s(\theta) + \max_s A_s(\vartheta))$, $\max(\max_s A_s(\theta), \max_s A_s(\vartheta))$ and  $1.5(\max_s A_s(\theta) + \max_s A_s(\vartheta))$.
  }
    \label{fig:mhESS_Q}
  \end{figure}

Next, we consider an M/M/N/N queue~\citep{gross2011fundamentals}. 
The state space of this stochastic process is $\{0, 1, 2, 3, \cdots, N - 1\}$ giving the number of customers/jobs/individuals in a system/population. 
Arrivals follow a rate-$\alpha$ Poisson process, moving the process from state 
$i$ to $i+1$ for $i<N$. The system has a capacity of $N$, so any arrivals when 
the current state is $N$ are discarded.  Service times or deaths are 
exponentially distributed, with a rate that is now state-dependent:
the system moves from $i$ to $i - 1$ with rate $i\beta$. 


We follow the same setup as the first experiment:
for $(\alpha_0,\alpha_1,\beta_0,\beta_1)$ equal to $(3,2,5,2)$,
we place Gamma$(\alpha_0,\alpha_1)$, and Gamma$(\beta_0, \beta_1)$ priors on $\alpha$, $\beta$. 
These prior distributions are used to sample transition matrices $A$, which, along with a uniform distribution over initial states, are used to generate MJP trajectories. 
We observe these at integer-valued times according to a Gaussian likelihood.
We again consider three settings: $3, 5$ and $10$ states, with results from $5$ steps included in the supplementary material. 

Figure~\ref{fig:ESS_Q_D10} plots the ESS per unit time (top row) as well as ESS per 1000 samples (bottom row) for the parameters $\alpha$ and $\beta$, again as we change the variance of the proposal kernel. 
The left two columns show these for $\alpha$ and $\beta$ for the MJP state-space having size $3$, while the right two columns show these for size $10$.
Our symmetrized  MH algorithm does best for dimensions $3$ and $5$ (shown in the supplement), although now Gibbs sampling performs best for dimensionality $10$ (although there is no significant different between the best proposal variance for our sampler and the Gibbs sampler).
The Gibbs sampler performs so well partly because the conditionals over $\alpha$ and $\beta$ are conjugate, following simple Gamma distributions. Also, unlike the earlier problem, the rate matrix is tri-diagonal, and governed by two parameters, so that path-parameter coupling is now milder.

  \begin{figure}[H]
  \centering
  \begin{minipage}[!hp]{0.24\linewidth}
    \includegraphics [width=0.99\textwidth, angle=0]{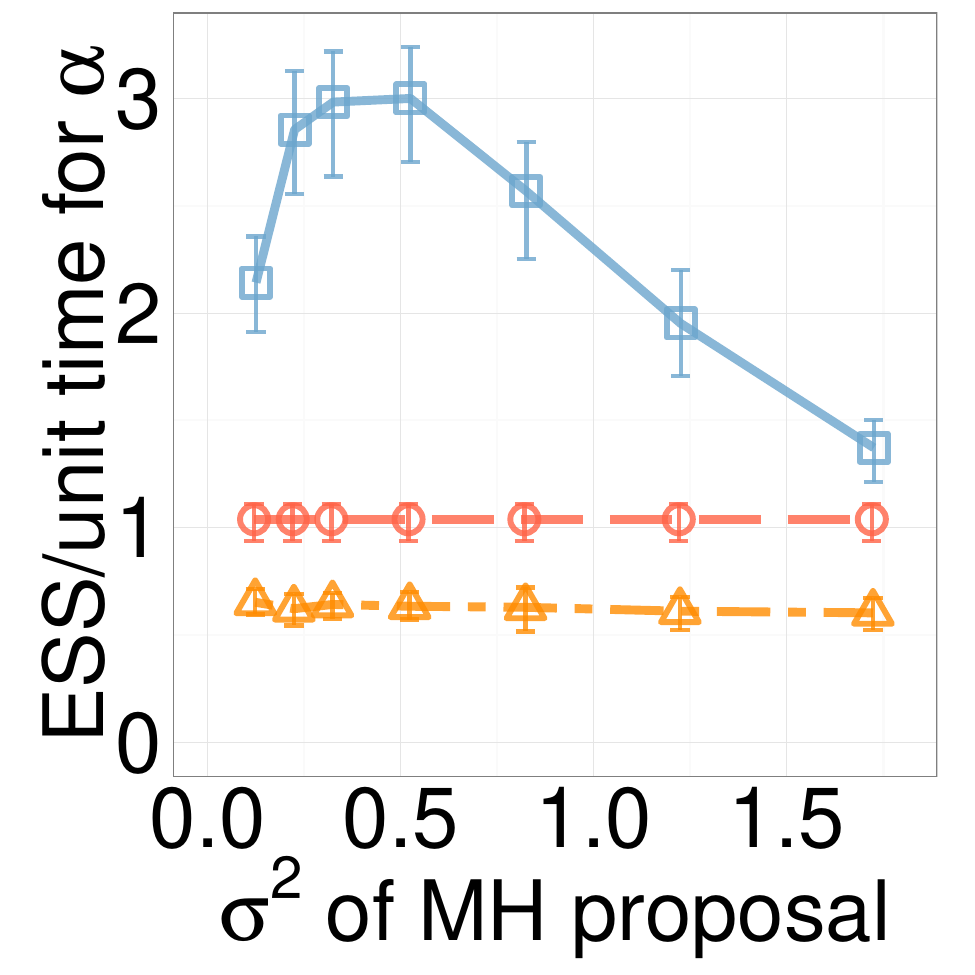}
\end{minipage}
  \begin{minipage}[hp]{0.24\linewidth}
  \centering
    \includegraphics [width=0.99\textwidth, angle=0]{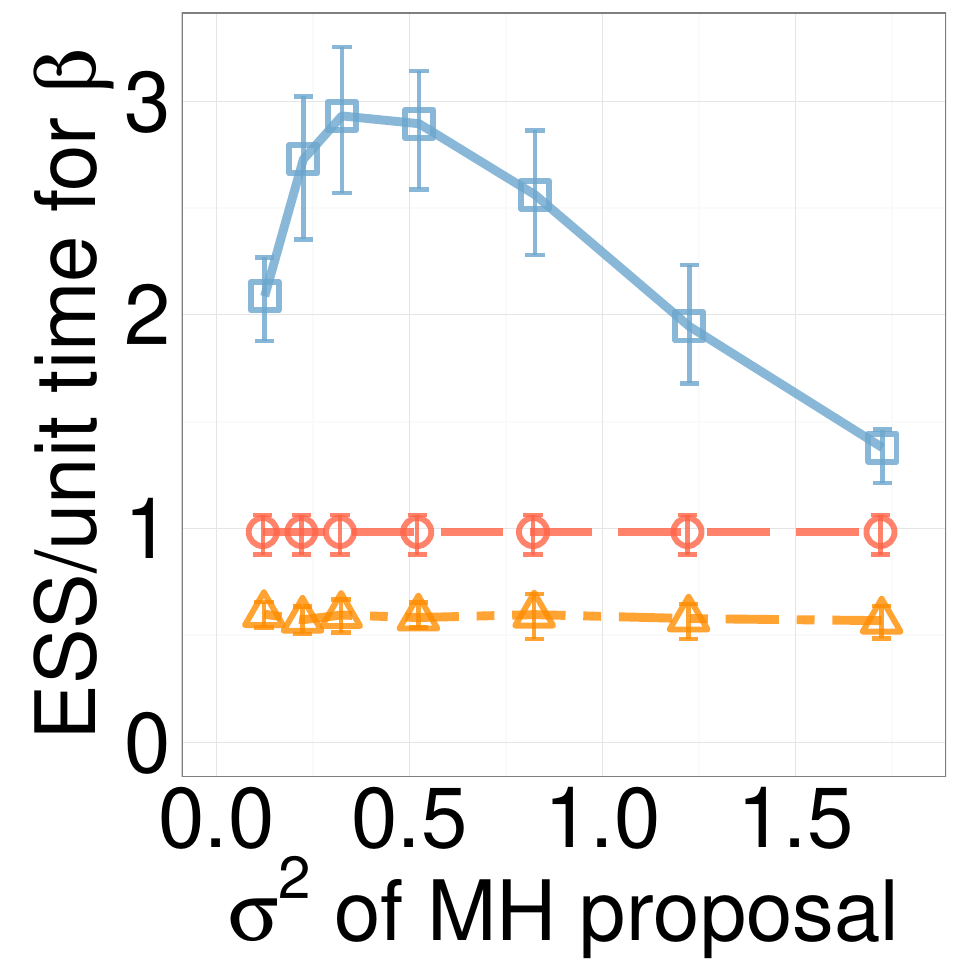}
	\end{minipage}
  \begin{minipage}[hp]{0.24\linewidth}
  \centering
    \includegraphics [width=0.99\textwidth, angle=0]{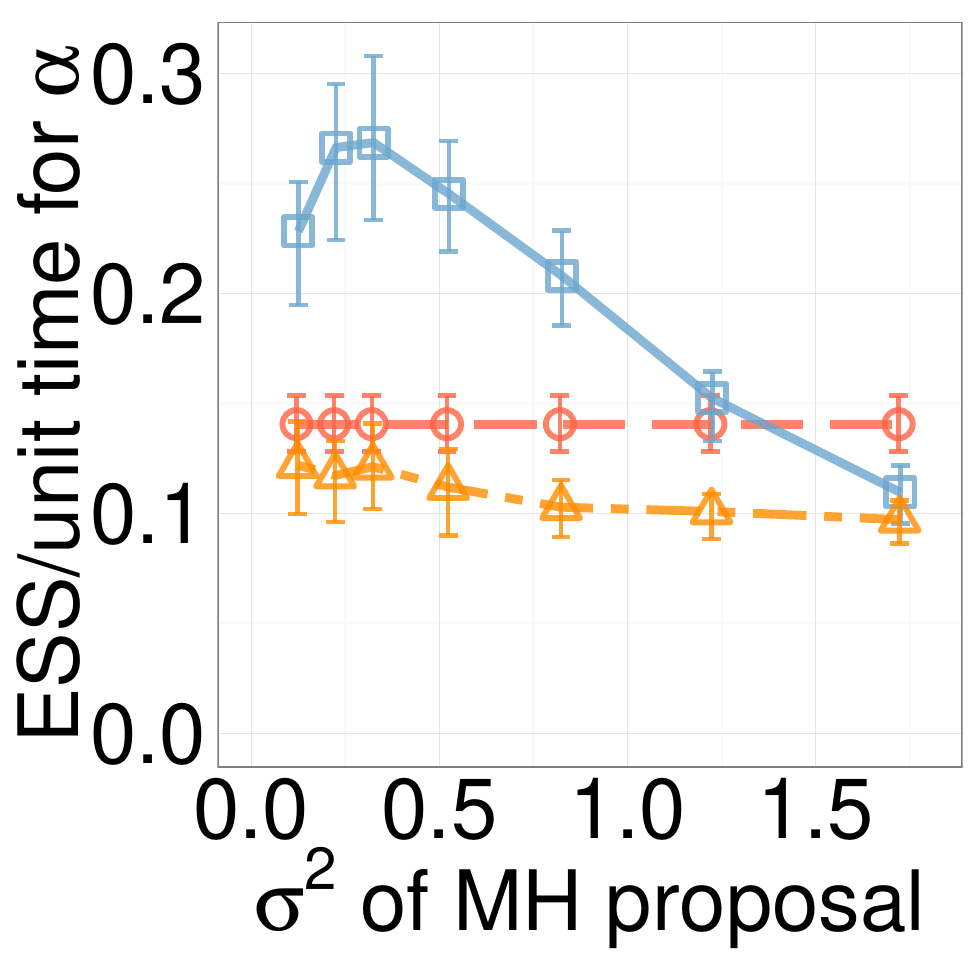}
	\end{minipage}
  \begin{minipage}[hp]{0.24\linewidth}
  \centering
    \includegraphics [width=0.99\textwidth, angle=0]{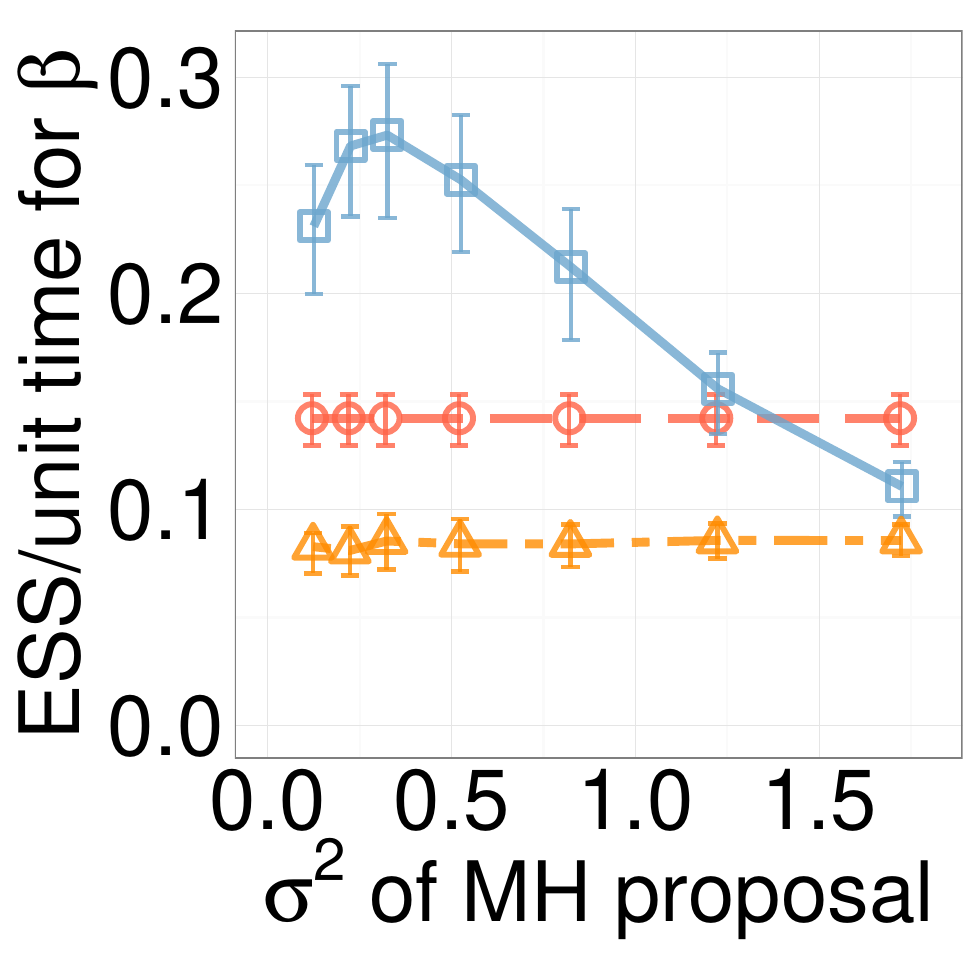}
	\end{minipage}
  \centering
  \begin{minipage}[!hp]{0.24\linewidth}
  \centering
    \includegraphics [width=0.99\textwidth, angle=0]{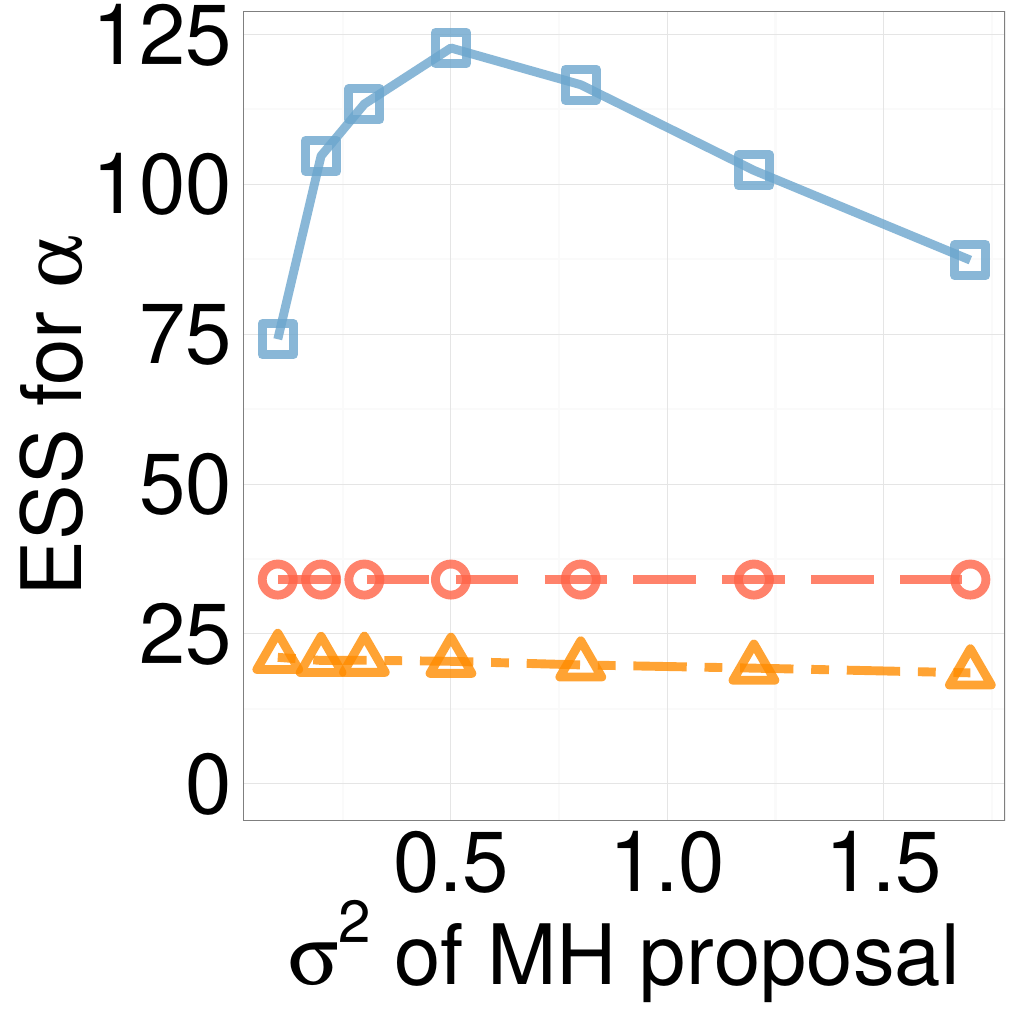}
\end{minipage}
  \begin{minipage}[hp]{0.24\linewidth}
  \centering
    \includegraphics [width=0.99\textwidth, angle=0]{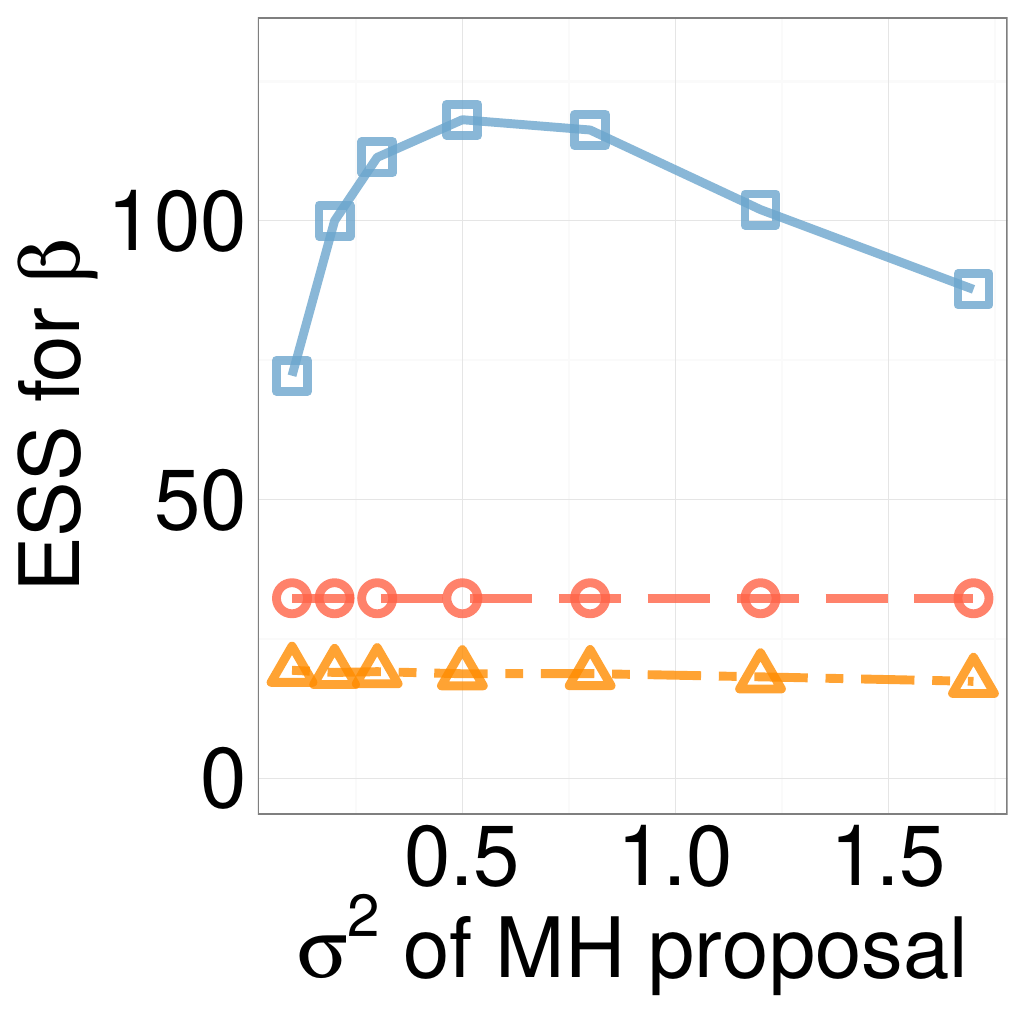}
	\end{minipage}
  \begin{minipage}[hp]{0.24\linewidth}
  \centering
    \includegraphics [width=0.99\textwidth, angle=0]{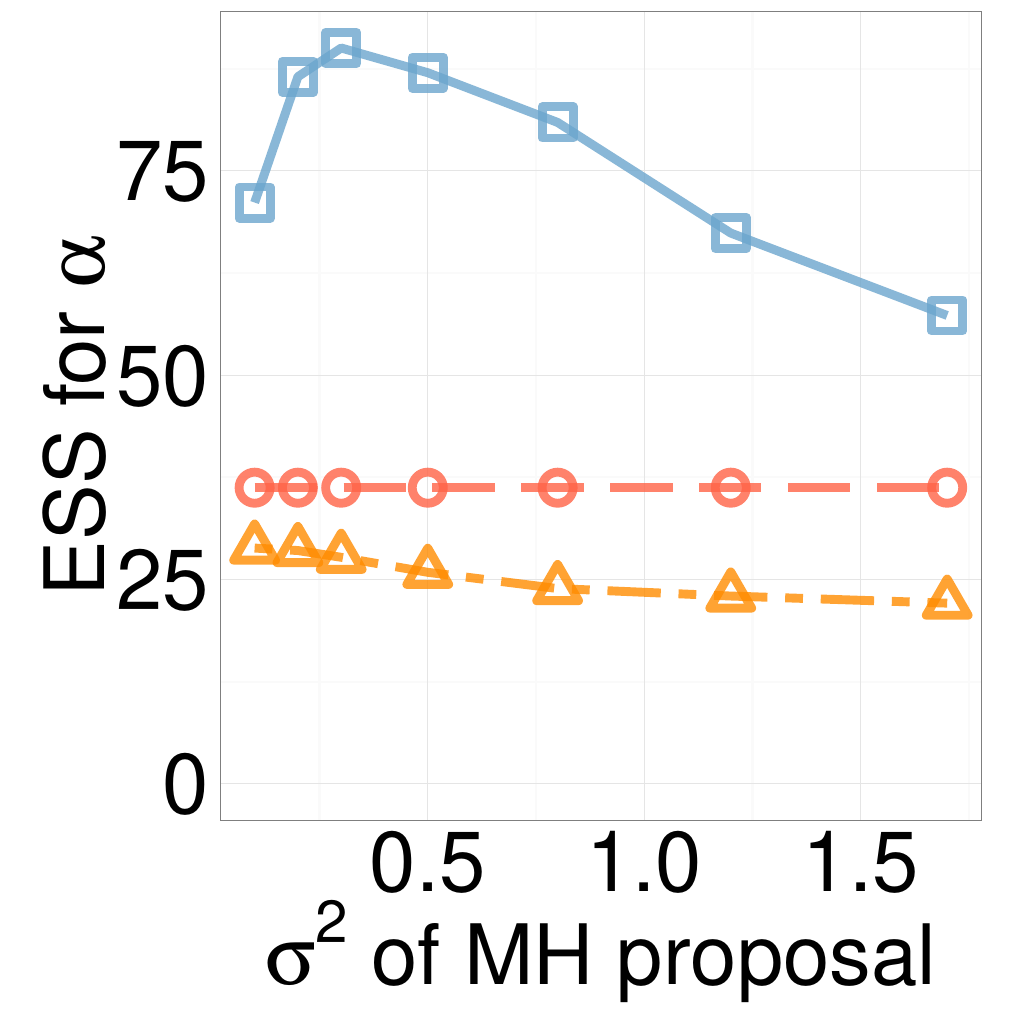}
	\end{minipage}
  \begin{minipage}[hp]{0.24\linewidth}
  \centering
    \includegraphics [width=0.99\textwidth, angle=0]{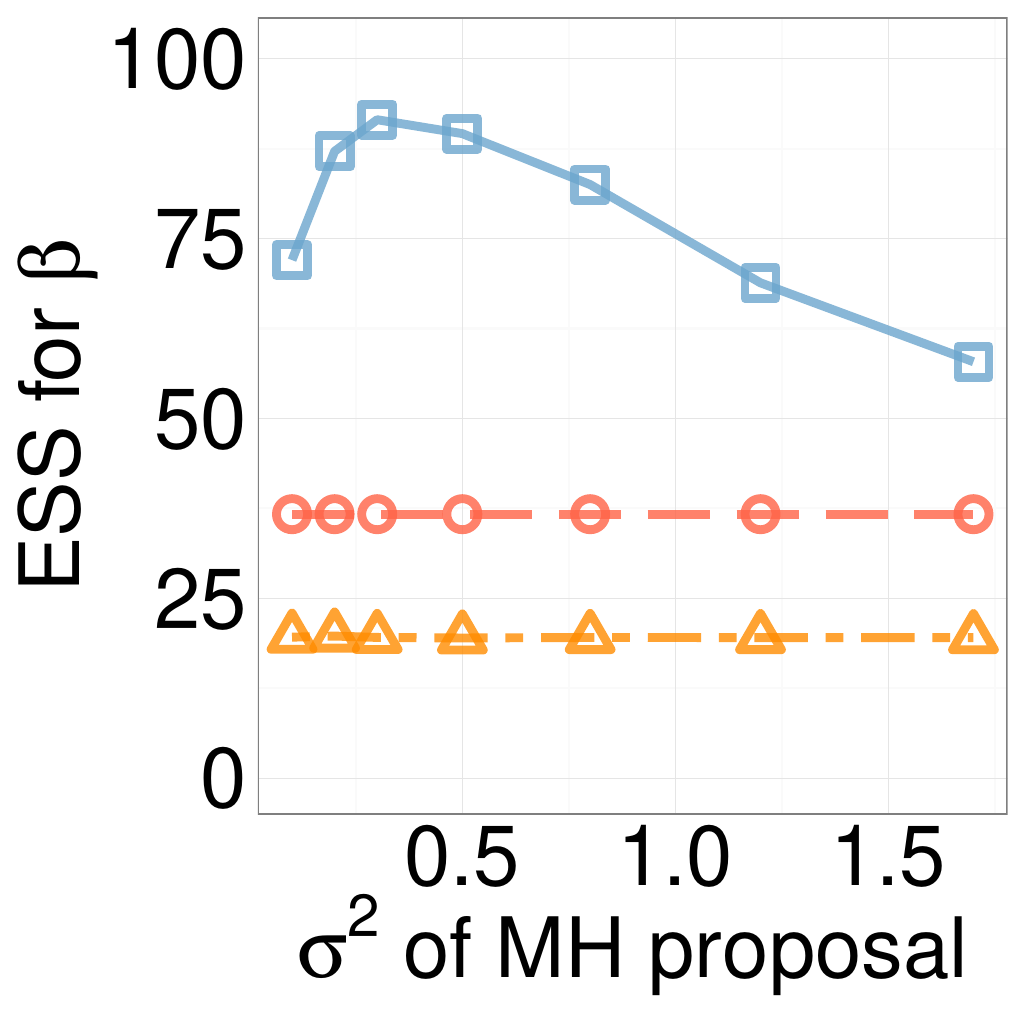}
	\end{minipage}
    \caption{ESS/sec (top row) and ESS per 1000 samples (bottom row) for the time-inhomogeneous immigration model. The left columns are $\alpha$ and $\beta$ for 3 states, and the right two for 10. {Blue squares, yellow triangles and red circles} are the symmetrized MH, \naive\ MH, and Gibbs algorithm. }
     \label{fig:ESS_pc_10}
  \end{figure}
\vspace{-.1in}
  \begin{figure}[H]
  \centering
  \begin{minipage}[!hp]{0.24\linewidth}
  \centering
    \includegraphics [width=0.99\textwidth, angle=0]{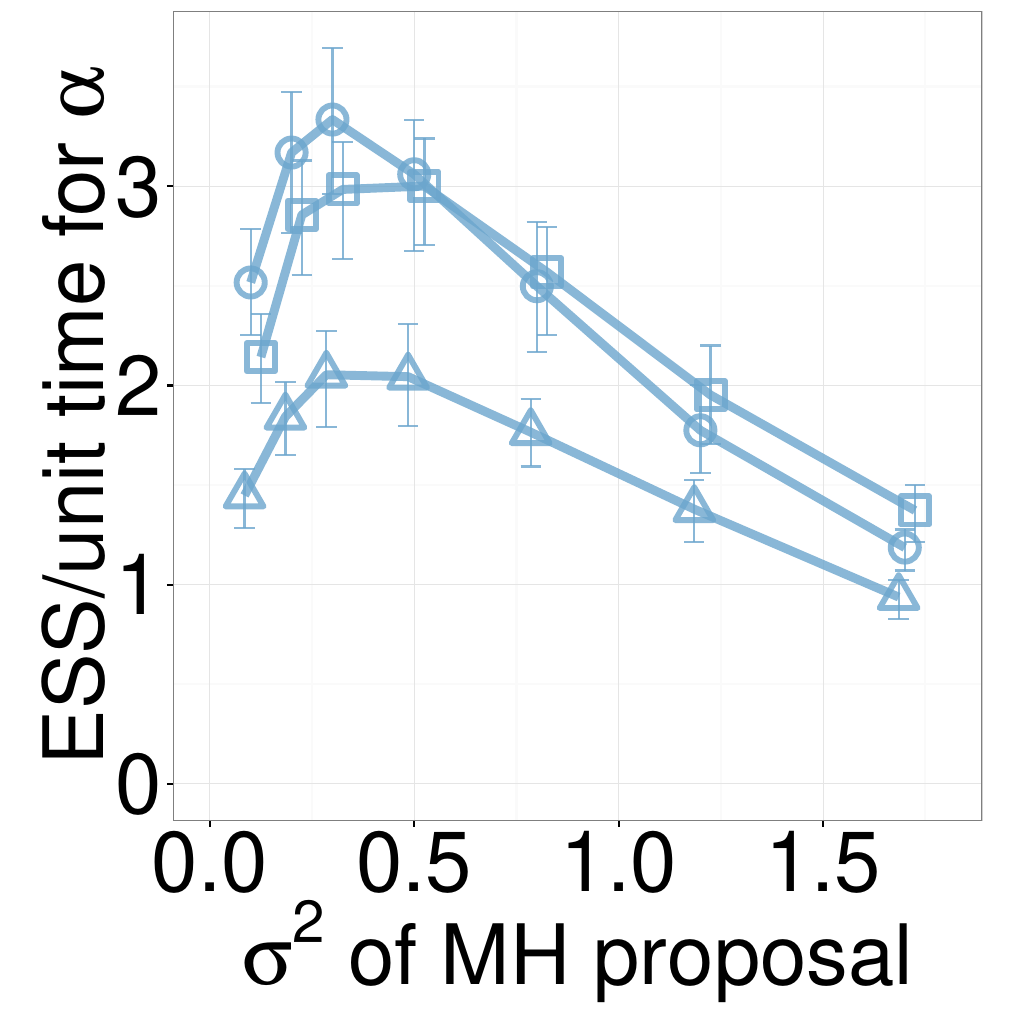}
\end{minipage}
  \begin{minipage}[hp]{0.24\linewidth}
  \centering
    \includegraphics [width=0.99\textwidth, angle=0]{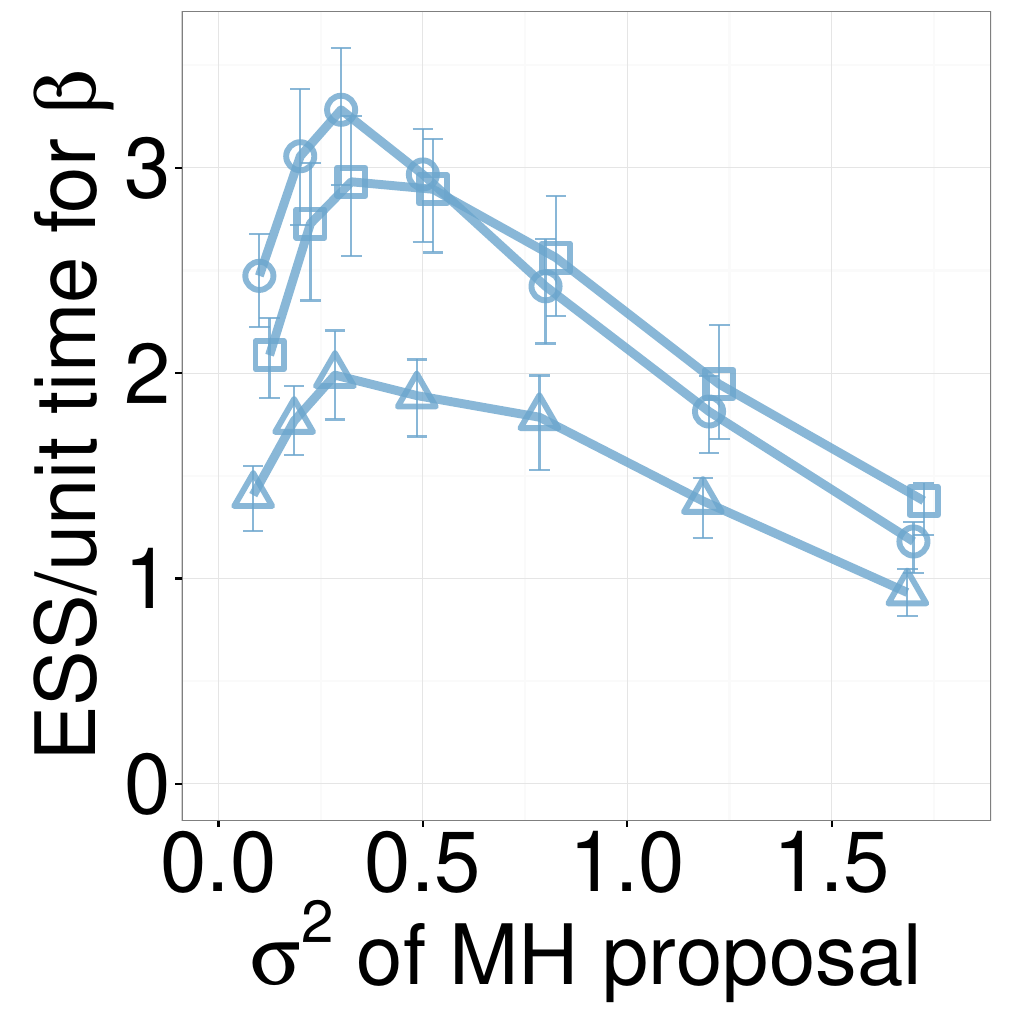}
	\end{minipage}
  \begin{minipage}[hp]{0.24\linewidth}
  \centering
    \includegraphics [width=0.99\textwidth, angle=0]{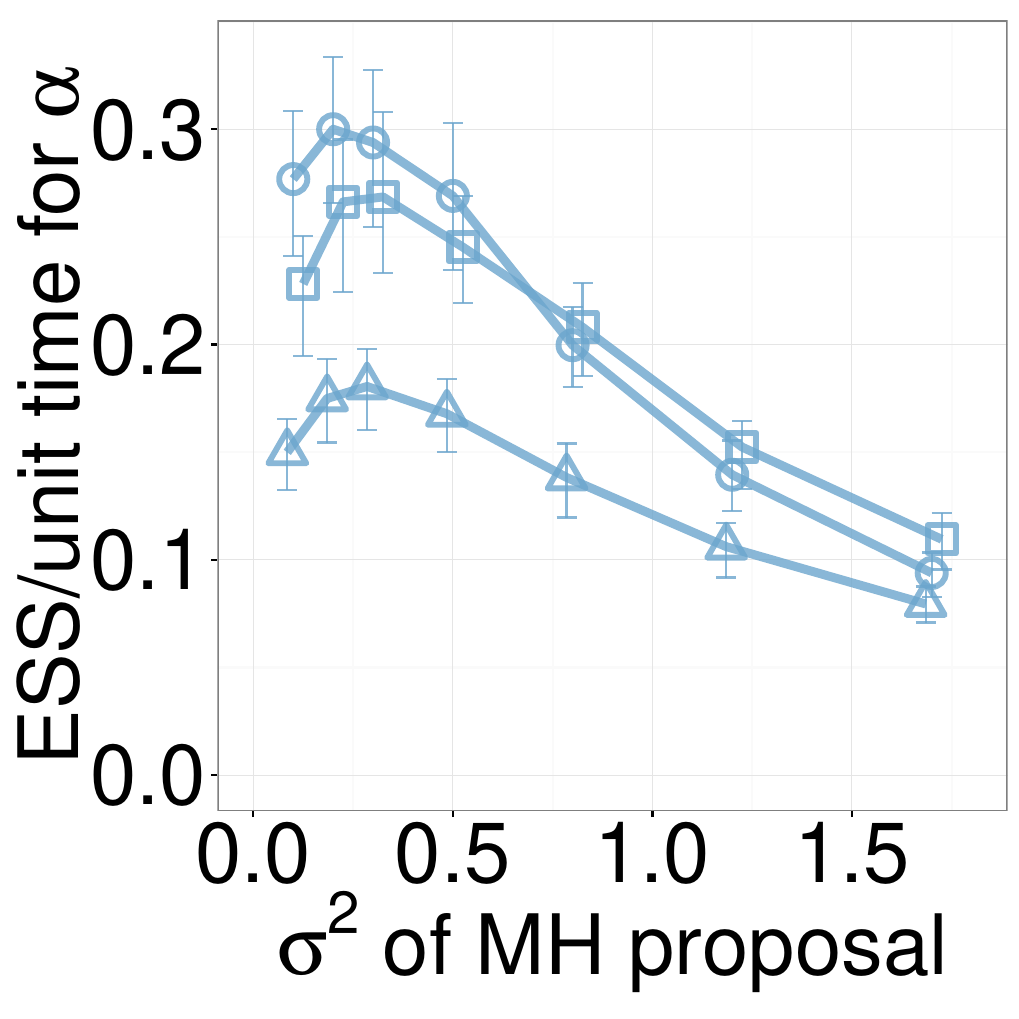}
	\end{minipage}
  \begin{minipage}[hp]{0.24\linewidth}
  \centering
    \includegraphics [width=0.99\textwidth, angle=0]{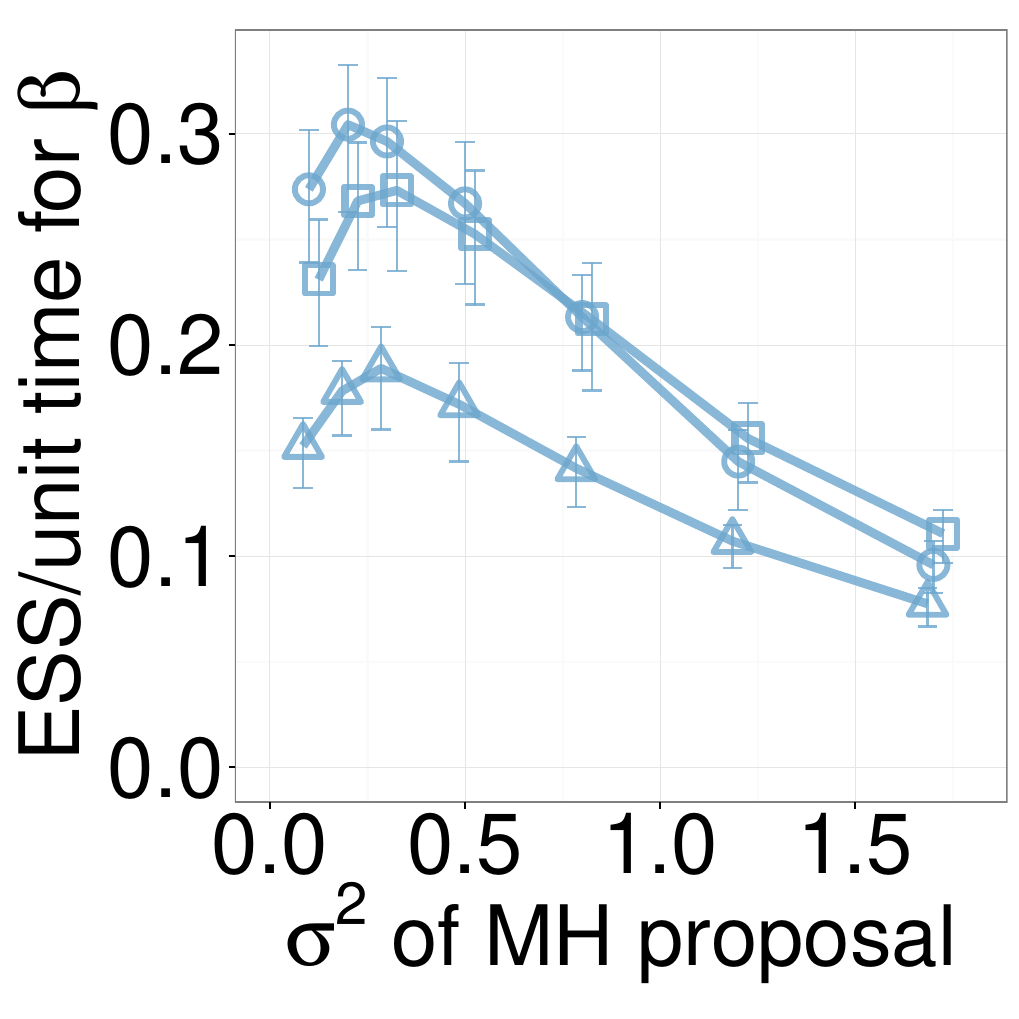}
	\end{minipage}
    \caption{ESS/sec for symmetrized MH for the time-inhomogeneous immigration model for different settings of $\Omega(\theta,\vartheta)$. The left two columns are $\alpha$ and $\beta$ for 3 states, and the right two for 10. 
    Squares, circles and trianges correspond to $\Omega(\theta,\vartheta)$ set to $(\max_s A_s(\theta) + \max_s A_s(\vartheta))$, $\max(\max_s A_s(\theta), \max_s A_s(\vartheta))$ and  $1.5(\max_s A_s(\theta) + \max_s A_s(\vartheta))$.}
     \label{fig:mhESS_CQ}
  \end{figure}

\vspace{-.1in}
\noindent \textbf{A time-inhomogeneous immigration model:}
We extend the previous model to incorporate a known time-inhomogeneity. 
The arrival and death rates are now no longer constant, and are instead given by
$A_{i, i+1}(t) = \alpha w(t) \ (i =0,1,\cdots,N-1)$ respectively.
While it is not difficult to work with sophisticated choices of $w(t)$, we limit ourselves to a simple piecewise-constant $w(t) = \left\lfloor \frac{t}{5} \right\rfloor$. 
Even such a simple change in the original model can dramatically affect the performance of the Gibbs sampler. 

Figure~\ref{fig:ESS_pc_10} plots the ESS per unit time (top row) and ESS per 1000 samples (bottom row) for the parameters $\alpha$ and $\beta$. 
The left two columns show these for this model with capacity $3$, and the right two show these for capacity $10$.  
 Now, the symmetrized MH algorithm is significantly more efficient, comfortably outperforming all samplers (including the Gibbs sampler) over a wide range of settings. 
 We note that increasing the dimensionality of the state space results in a more concentrated posterior, shifting the optimal setting of the proposal variance to smaller values. Figure~\ref{fig:hist} shows prior and conditional distributions over $\alpha$ for $t_{end}$ set to $10$ and $100$, with 3 states.

 \vspace{-.1in}
  \subsection{Chi site data for Escherichia coli}~
 \vspace{-.05in}
  Finally, we consider a dataset recording positions of a particular DNA motif on the {\em E.\ coli} genome. 
  These motifs consist of eight base pairs GCTGGTGG, and are called Chi sites~\citep{FearnSher2006}.
  The rates of occurence of Chi sites provide information about genome 
  segmentation, allowing the identification of regions with high 
  mutation or recombination rates.
  Following~\cite{FearnSher2006}, we use this dataset to infer a two-state piecewise-constant segmentation of the DNA strand. 
  We focus on Chi sites along the inner (lagging) strand of the {\em E.\ coli} genome.  
  We place an MJP prior over this segmentation, and indexing position along the 
  strand with $t$, we write this as $\{S(t), t \in [0,2319.838]\}$. 
  To each state $s \in \{1,2\}$, we assign a rate $\lambda_s$, which together with $S(\cdot)$, defines a piecewise-constant rate function $\lambda_{S(\cdot)}$. 
  We model the Chi-site positions as drawn from a Poisson process with rate 
  $\{\lambda_{S(t)}, t \in [0,2319.838]\}$, resulting in a {Markov-modulated Poisson process}~\citep{scottmmpp03} (see also section~\ref{sec:bayes_model}). 
  MJP transitions from state $1$ to state $2$ have rate $\alpha$ while transitions from state $2$ to state $1$ have rate $\beta$. 
  We place 
  Gamma$(2,2)$, Gamma$(2,3)$, Gamma$(3,2)$, Gamma$(1,2)$ priors on $\alpha$, $\beta$, $\lambda_1$, $\lambda_2$.

  We use this setup to evaluate our symmetrized MH sampler along with Gibbs sampling (other algorithms perform much worse, and we do not include them). 
  For our MH proposal distribution, we follow~\citet{gelman2013bayesian}, and first run {2000} iterations of Gibbs sampling to estimate the posterior covariance of the vector $\theta = (\alpha,\beta,\lambda_1,\lambda_2)$. Call this $\Sigma_\theta$. 
  Our MH proposal distribution is then $q(\nu|\theta) = N(\nu|\theta, \sigma^2\Sigma_\theta)$ for different settings of $\sigma^2$ (we recommend $\sigma^2 = 1$),
where we set $\Omega(\theta, \vartheta) = \max_s A_s(\theta) + \max_s A_s(\vartheta)$. 
  \begin{figure}[]
  \centering
  \begin{minipage}[!hp]{0.99\linewidth}
    \includegraphics [width=0.24\textwidth, angle=0]{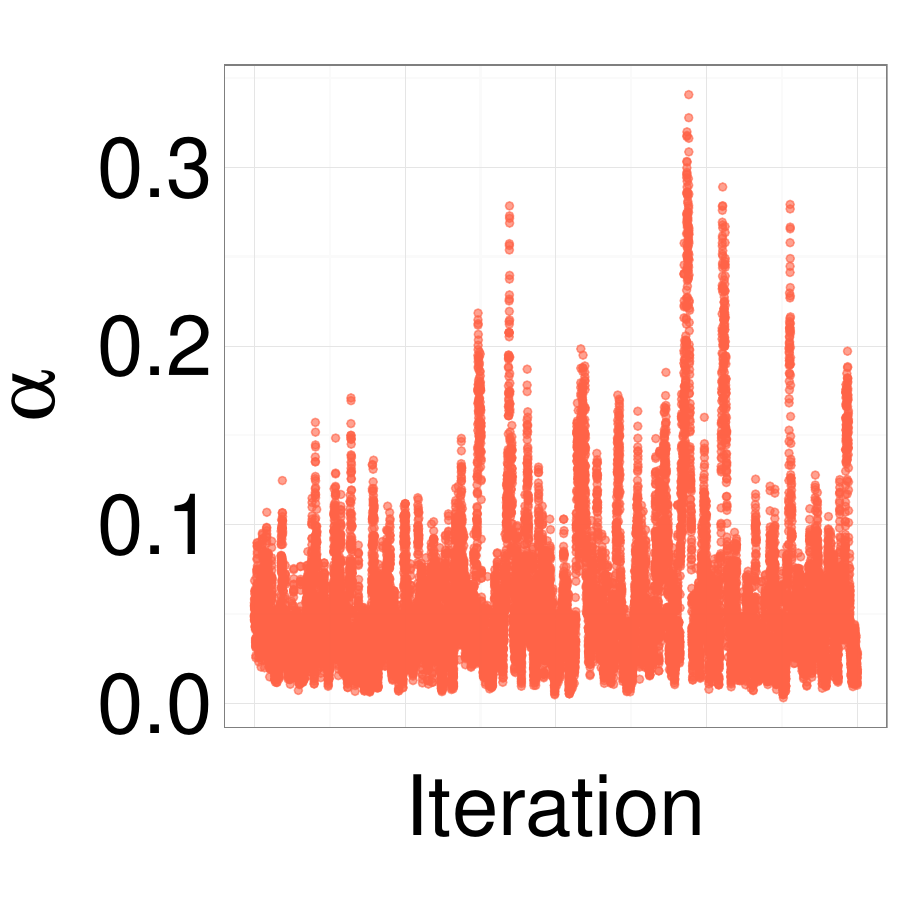}
    \includegraphics [width=0.24\textwidth, angle=0]{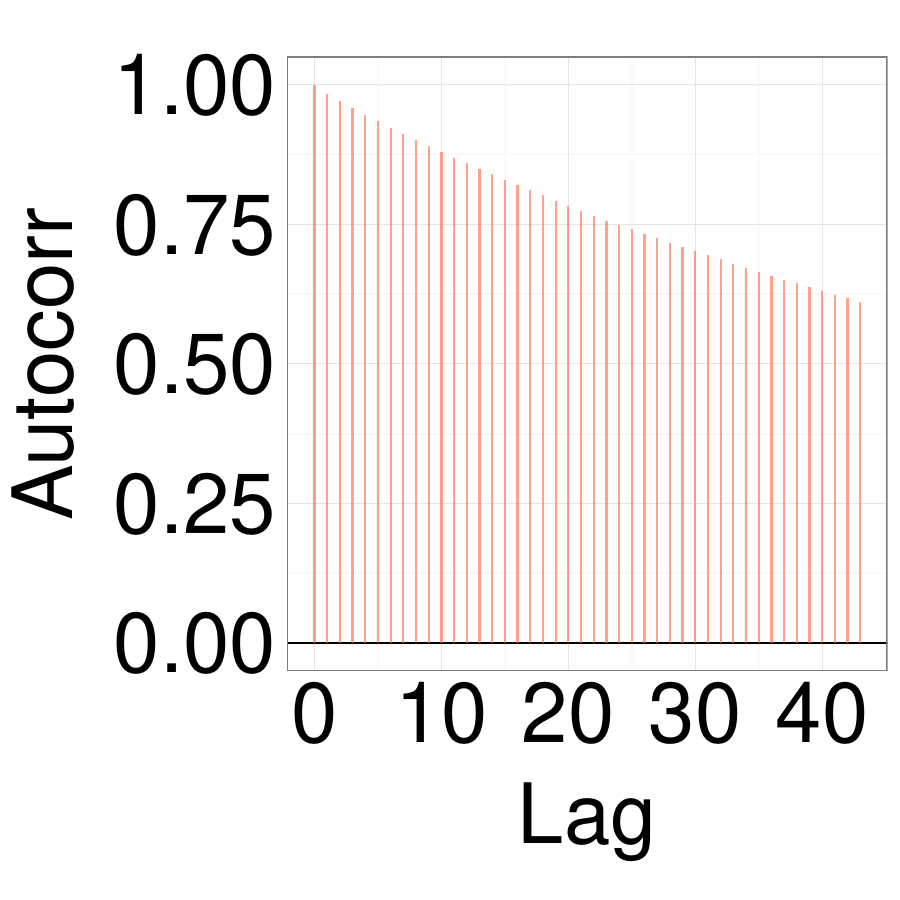}
    \includegraphics [width=0.24\textwidth, angle=0]{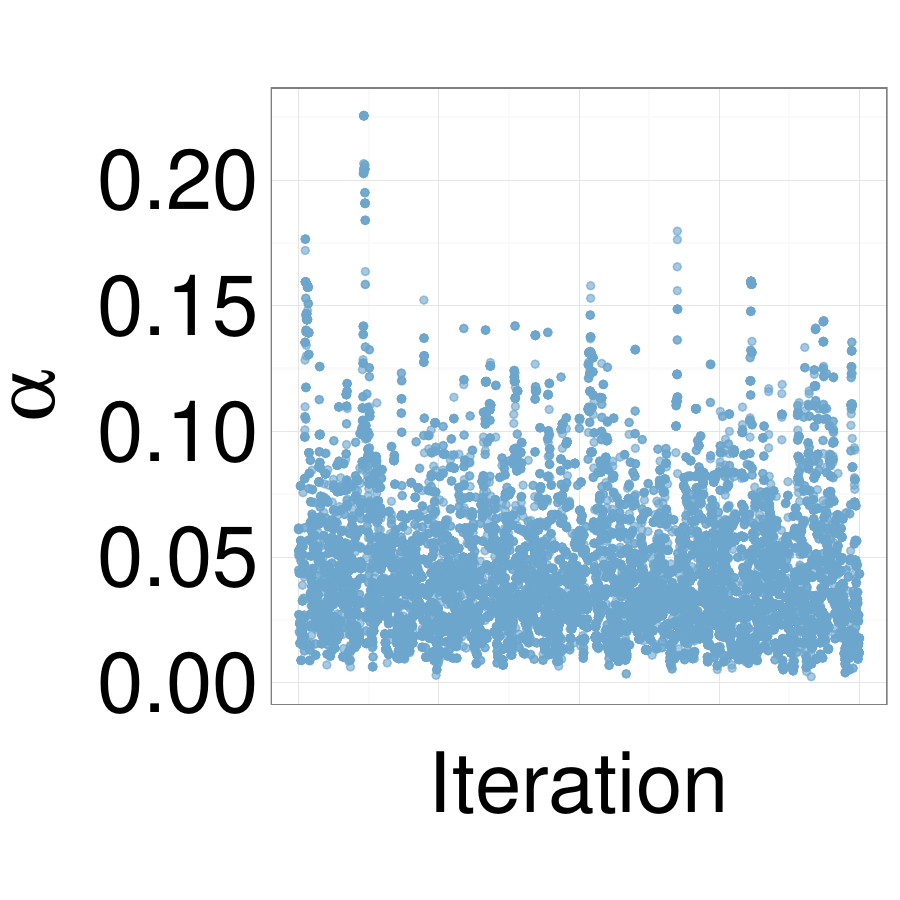}
    \includegraphics [width=0.24\textwidth, angle=0]{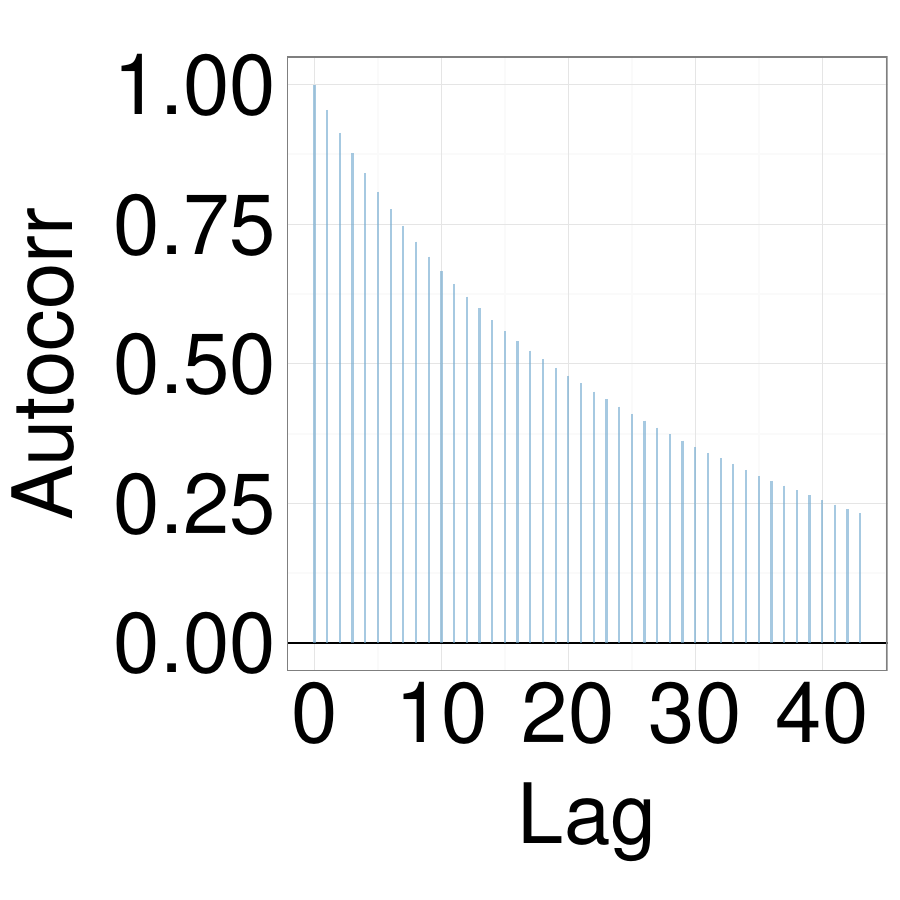}
  \end{minipage}
    \caption{Trace and autocorrelation plots of posterior samples for $\alpha$ for the E.\ Coli data. The left two plots are the Gibbs sampler and the right two are the symmetrized MH. }
     \label{fig:TRACE_ECOLI}
     \vspace{-.1in}
  \end{figure}
  \begin{figure}[]
  \centering
  \begin{minipage}[!hp]{.55\linewidth}
	\includegraphics [width=0.48\textwidth, angle=0]{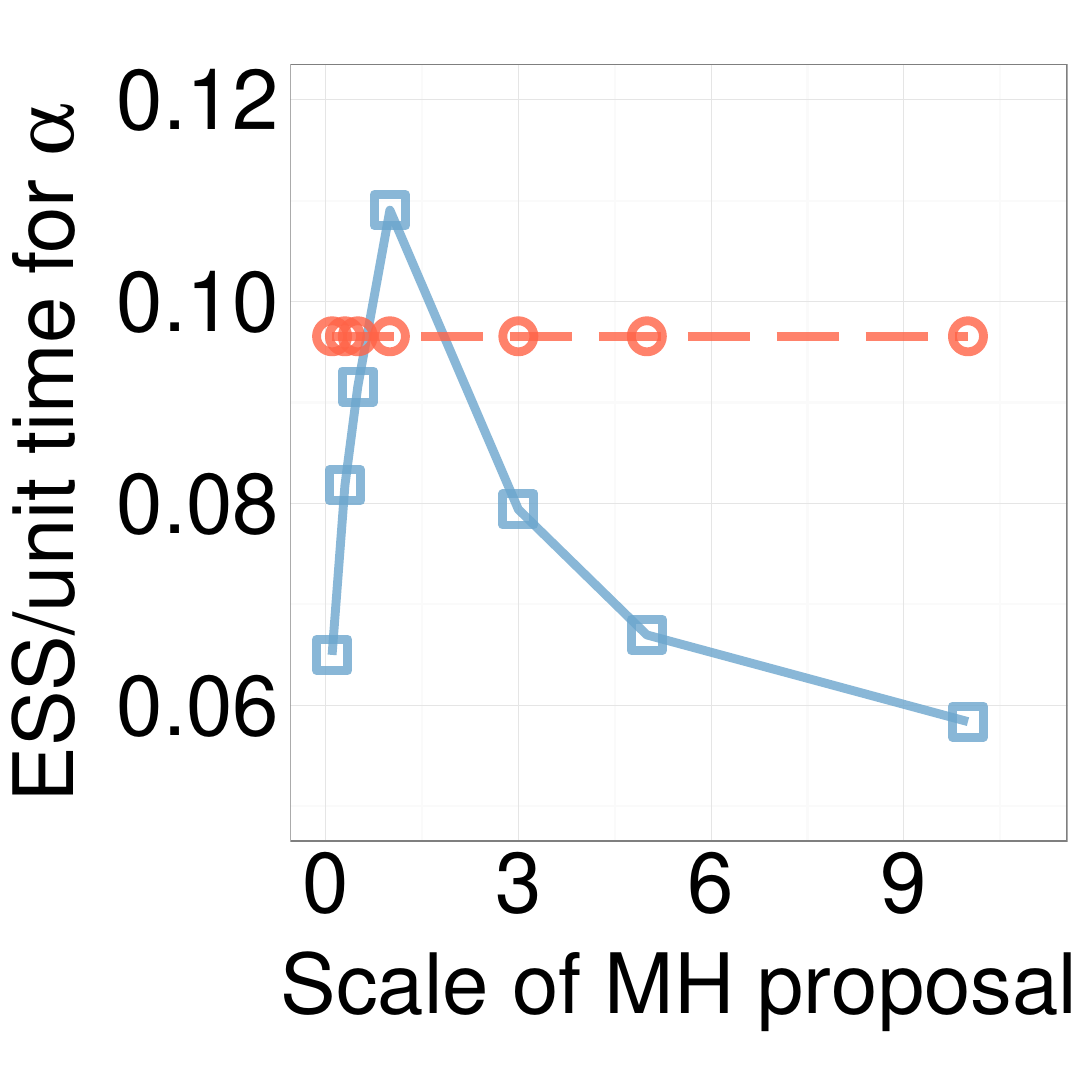}
	\includegraphics [width=0.48\textwidth, angle=0]{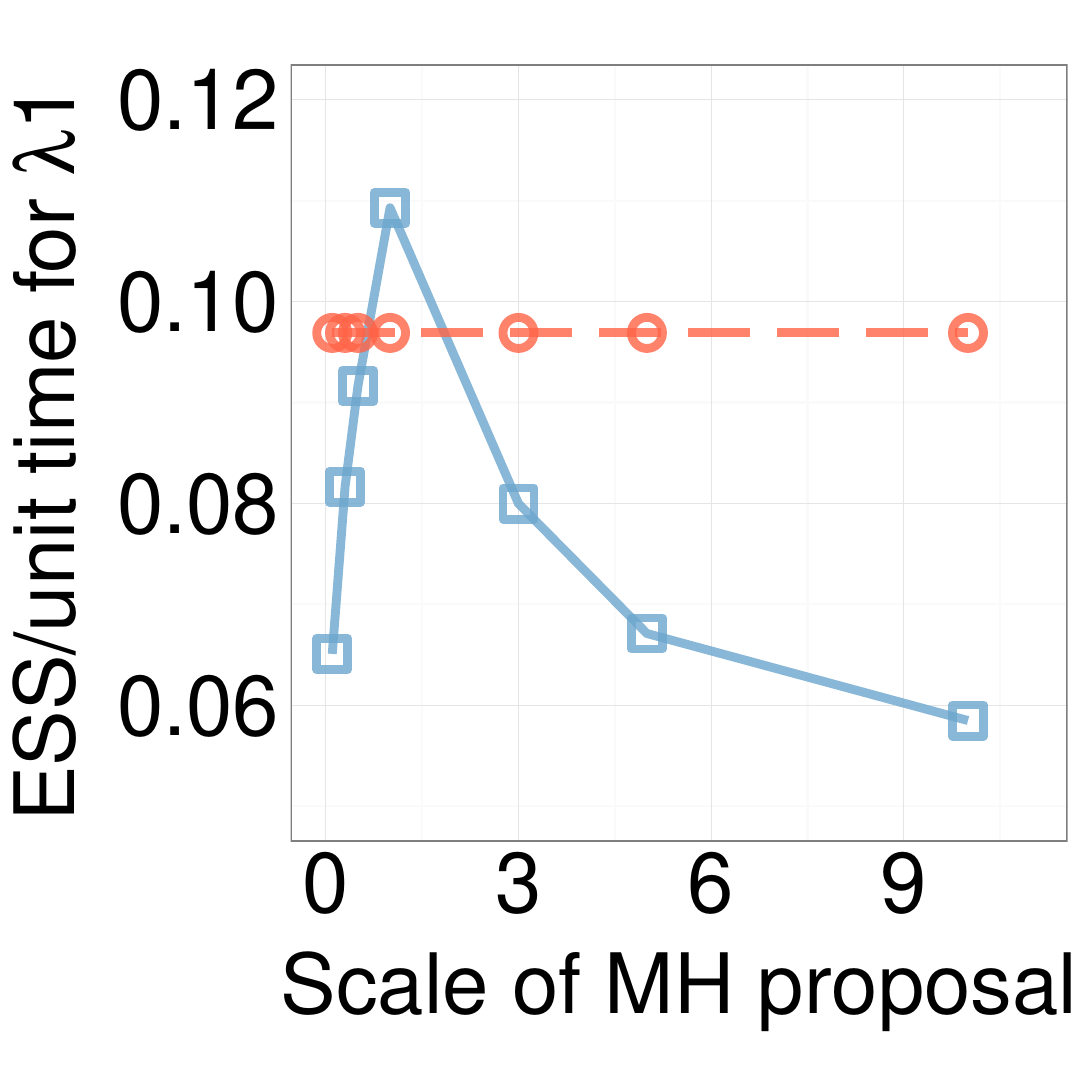}
  \end{minipage}  
  \begin{minipage}[!hp]{.01\linewidth}
  \end{minipage}
  \begin{minipage}[!hp]{.38\linewidth}
    \caption{ESS/sec for $(\alpha,\lambda_1)$ for the E.\ Coli data. 
      Blue squares are symmetrized MH as we vary the variance of the proposal distribution. 
      The red circles are Gibbs. }
     \label{fig:ECOLI}
  \end{minipage}  
\vspace{-.2in}
  \end{figure}

  Figure~\ref{fig:TRACE_ECOLI} shows trace and autocorrelation plots for the parameter $\alpha$ produced by the Gibbs sampler (left) and our proposed sampler with $\kappa$ set to $1$.
  We see that this is a fairly hard MCMC sampling problem, however our sampler clearly outperforms Gibbs, which mixes very poorly.
  Both posterior distributions agreed with each other though, with a two sample-Kolmogorov Smirnov test returning a p-value of $ 0.1641$. 
  
  Figure~\ref{fig:ECOLI} shows the ESS/s for different $\kappa$, for parameters $(\alpha, \lambda_1)$. 
  Both have similar results, and as suggested by the earlier figure, we see that for the typical setting of $\kappa=1$, our sampler ourperforms Gibbs. 
  In this problem though, Gibbs sampling does outperform our method for large or small $\kappa$. 
  This is because a) large or small $\kappa$ mean the proposal variance is too large or too small, and b) the Gibbs conditionals over the parameters are conjugate for this model. 
  We expect the improvements our method offers to be more robust to the proposal distribution for more complex models without such conditional conjugacy.



\section{Geometric ergodicity}
We derive conditions under which our symmetrized MH algorithm inherits mixing properties of an `ideal' sampler that can compute the marginal likelihood $P(X|\theta)$, with the MJP path integrated out. 
This algorithm proposes a new parameter $\vartheta$ from $q(\vartheta|\theta)$, and accepts with probability $\alpha_I(\theta,\vartheta; X) = 1 \wedge \frac{P(X , \vartheta)q(\theta| \vartheta)} {P(X , \theta)q(\vartheta|\theta)}$.  
The resulting Markov chain has transition probability
$P_I(\theta'|\theta) = q(\theta'|\theta)\alpha_I(\theta,\theta';X) + \left[1-\int \dif\vartheta
q(\vartheta|\theta)\alpha_I(\theta,\vartheta;X)\right]\delta_\theta(\theta')$, the first term corresponding to acceptance, and the second, rejection~\citep{meyn2009}.

Our main result is Theorem~\ref{thm:geom_erg}, which shows that if the ideal MCMC sampler is geometrically ergodic, then so is our sampler in Algorithm~\ref{alg:MH_improved}. 
Informally, an MCMC algorithm is geometrically ergodic when the total variation distance between the distribution over states and the stationary distribution decreases geometrically with the number of iterations.
~\citet{meyn2009} provide more details, as well as sufficient conditions that we exploit in Theorem~\ref{thm:geom_erg}.
Geometric ergodicity is an important property of an MCMC chain, guaranteeing that the central limit theorem (CLT) holds for ergodic averages calculated with MCMC samples.
Before diving into the proofs,  we first state our assumptions, 
\begin{assumption}
The uniformization rate is set as $\Omega(\theta, \vartheta) = \Omega(\theta) + \Omega(\vartheta)$, where 
$\Omega(\theta) = k_1 \max_s A_{s}(\theta) + k_0$, for some 
$k_1 > 1, k_0 > 0$.
\label{asmp:unif_rate}
\end{assumption}
\noindent Although it is possible to specify broader conditions under which our result holds, for clarity we focus on this case. 
In our experiments, one of our settings had $k_1=1$. We believe our result holds for this case too, but do not prove it.
We can drop $k_0$ if $\inf_\theta \max_s A_s(\theta) > 0$.

\begin{assumption}
 There exists a positive constant $\theta_0$ such that for any $\theta_x, \theta_y$ satisfying $\| \theta_x \| \ge \| \theta_y \| > \theta_0$, we have $\Omega(\theta_x) \ge \Omega(\theta_y)$.  
  \label{asmp:mono_tail}
\end{assumption}
\noindent 
We make this assumption to avoid book-keeping, so $\Omega(\theta)$ increases monotonically with $\theta$.

\begin{definition}
Let $\pi_\theta$ be the stationary distribution of the MJP with rate-matrix 
$A(\theta)$, and define $D_\theta = \text{diag}(\pi_\theta)$. Define 
$\tilde{A}(\theta) = D_\theta^{-1}A(\theta)D_\theta$, and the 
{\em reversibilization} of $A(\theta)$ as $R_A(\theta) = 
(A(\theta)+\tilde{A}(\theta))/2$. 
\label{def:mjp_symm}
\end{definition}
\noindent This definition is from~\cite{fill1991}, who shows that $R_A(\theta)$ is reversible with real eigenvalues, the smallest being $0$. 
The larger its second smallest eigenvalue, the faster the MJP converges to its 
stationary distribution $\pi_\theta$.
Note that if $A(\theta)$ is reversible, then $R_A(\theta) = A(\theta)$.

\begin{assumption}
  Write $\lambda^{R_A}_2(\theta)$ for the second smallest eigenvalue of
    $R_A(\theta)$. There exist $\mu > 0, \theta_1 > 0$
    such that for all $\theta$ satisfying $ \| \theta \|> \theta_1$, 
    we have $ \lambda^{R_A}_2(\theta) \geq \mu \max_s A_s(\theta)$
    (or equivalently from Assumption~\ref{asmp:unif_rate}, 
    $ \lambda^{R_A}_2(\theta) \geq \mu \Omega(\theta)$),
    and $\min_s \pi_\theta(s) > 0$. 
  \label{asmp:cond_num}
\end{assumption} 
\noindent 
This assumption is the strongest we need, requiring that 
$\lambda^{R_A}_2(\theta)$ (which sets the MJP mixing rate) grows 
at least as fast as $\max_s A_s(\theta)$. 
This is satisfied when $\theta$ is bounded, or when, as in our experiments, all elements of $A(\theta)$ grow 
with $\theta$ at similar rates, controlling the relative stability of 
the least and most stable states.
While not trivial, this is reasonable: the MCMC chain over MJP paths 
will mix well if we can control the mixing of the MJP itself.
A less restrictive assumption would also account for the tail behavior of the prior over $\theta$, though we do not do this.
To better understand this assumption, recall $B(\theta, \theta') = I+\frac{A(\theta)}{\Omega(\theta, \theta')}$
is the transition matrix of the embedded Markov chain from uniformization, which has the same stationary distribution $\pi_\theta$ as $A(\theta)$.
Define the reversibilization $R_B(\theta,\theta')$ of $B(\theta,\theta')$ 
just as we did $R_A(\theta)$ from $A(\theta)$. 
\begin{lemma}
  Consider $\|\theta\| > \max(\theta_0, \theta_1)$ and $\theta'$ such that 
$\frac{1}{K_0} \le \frac{\Omega(\theta')}{\Omega(\theta)} \le K_0 $, 
where $K_0$ 
satisfies $(1 + \frac{1}{K_0})k_1 \ge 2$. 
For all such $(\theta,\theta')$, the Markov chain with transition matrix 
$B(\theta,\theta')$ converges geometrically to stationarity at a rate 
uniformly bounded away from $0$.
  \label{lem:eig_lemma}
\end{lemma}
\begin{proof}
A little algebra gives $R_B(\theta,\theta') = I + R_A(\theta)/\Omega(\theta,\theta')$. It 
follows that both $R_A$ and $R_B$ share the same eigenvectors, with 
eigenvalues satisfying 
$\lambda_{R_B}(\theta, \theta') = 1 - \frac{\lambda_{R_A}(\theta)}{\Omega(\theta,
\theta')}$. 
The second largest eigenvalue $\lambda_2^{R_B}(\theta,\theta')$ of $R_B$ and  second smallest eigenvalue $\lambda^{R_A}_2(\theta,\theta')$ of $R_A$ then satisfy  
$\lambda^{R_B}_2(\theta,\theta') = 1 - \frac{\lambda^{R_A}_2(\theta)}{\Omega(\theta, \theta')}$.
From assumptions~\ref{asmp:unif_rate} and~\ref{asmp:cond_num}, and 
the lemma's assumptions, 
$1 - \lambda^{R_B}_2(\theta,\theta') = \frac{\lambda^{R_A}_2(\theta)}{\Omega(\theta, \theta')} 
\ge \frac{\lambda^{R_A}_2(\theta)}{(K_0+1)\Omega(\theta)} 
\ge \frac{\mu}{K_0+1} $. 
Also, since $(1 + \frac{1}{K_0})k_1 \ge 2$,
\vspace{-.1in}
\begin{align*}
\Omega(\theta, \theta') &= \Omega(\theta) + \Omega(\theta') \ge (1 + \frac{1}{K_0})\Omega(\theta)
 > (1 + \frac{1}{K_0})k_1\max_s A_{s}(\theta) \ge 2\max_s A_{s}(\theta). 
\end{align*}
So for any state $s$, the diagonal element $B_s(\theta, \theta') = 1 - 
\frac{A_s(\theta)}{\Omega(\theta, \theta')}> \frac{1}{2}$.
From~\cite{fill1991}, this diagonal property and the bound 
on $1-\lambda_2^{R_B}(\theta,\theta')$ give the result.
\end{proof}

Our overall proof strategy is to show that on a set with {$\| \theta \|$} and $|W|$ large enough, the conditions of Lemma~\ref{lem:eig_lemma} hold with high probability. 
Lemma~\ref{lem:eig_lemma} then will imply that the distribution over states for the continuous-time MJP and its discrete-time counterpart embedded in $W$ will be brought arbitrarily close to $\pi_\theta$ (and thus to each other), allowing our sampler to inherit mixing properties of the ideal sampler. 
We will exploit the boundedness of the complement of this set to establish a `small-set condition' where the MCMC algorithm forgets its state with some probability. 
These two conditions will be sufficient for geometric ergodicity. 
The next assumption states these small-set conditions for the ideal sampler.

\begin{assumption}
For the ideal sampler with transition probability $p_I(\theta'|\theta)$: \\
i) for each $M$, for the set $\SM_M=\{\theta:\Omega(\theta)\le M \}$, there exists a probability measure $\phi$ and a constant $\kappa_1 > 0$ s.t.\ 
$\alpha_I(\theta, \theta'; X) q(\theta' | \theta) \ge \kappa_1 \phi(\theta')$ for $\theta \in \SM_M$. 
Thus $B_M$ is a $1$-small set. \\
ii) for $M$ large enough, $\exists \rho < 1$ s.\ t.\
$\int \Omega(\ptheta) p_I(\ptheta|\theta) d\ptheta
\leq (1-\rho) \Omega(\theta)+L_I$, $\forall \theta \not\in \SM_M$.
  \label{asmp:ideal_geom}
\end{assumption}
\noindent 
These two conditions are standard small-set and drift conditions necessary for the ideal sampler to satisfy geometric ergodicity. 
The first implies that for $\theta$ in $\SM_M$, the ideal sampler `forgets' its current location with probability $\kappa_1$. 
The second condition ensures that for $\theta$ outside this set, the ideal sampler drifts towards $\SM_M$. 
These two conditions together imply geometric mixing with rate equal or faster than $\kappa_1$~\citep{meyn2009}.
Observe that we have used $\Omega(\theta)$ as the so-called Lyapunov-Foster function to define the drift condition for the ideal sampler. 
This is the most natural choice, though our proof can be tailored to different choices. 
Similarly, we could easily allow $\SM_M$ to be an $n$-small set for any $n\ge 1$ (so the ideal sampler needs $n$ steps before it can forget its current value in $\SM_M$); we restrict ourselves to the $1$-small case for clarity.


\begin{assumption}
$\exists$ $ \ub > \lb > 0$ s.t.
$\prod P(X | s_o, \theta) \in [\lb, \ub]$ for any state $s_o$ and $\theta$.
  \label{asmp:obs_bnd}
\end{assumption}
\noindent This assumption follows~\cite{miasojedow2017}, and holds if
$\theta$ does not include parameters of the observation process (or if so,
the likelihood is finite and nonzero for all settings of $\theta$). We can relax this assumption,
though this will introduce technicalities unrelated to our focus, which 
is on complications in parameter inference arising from the continuous-time
dynamics, rather than the observation process. 

\begin{assumption}
Given the proposal density $q(\ptheta | \theta)$, $\exists \eta_0 > 0, \theta_2 > 0$ 
such that for $\theta$ satisfying $\| \theta \|  > \theta_2$, 
$ \int_\Theta \Omega(\ptheta)^2 q(\ptheta | \theta)d\ptheta \leq \eta_0 \Omega(\theta)^2.$
\label{asmp:integ_bound}
\end{assumption}
\noindent This mild requirement can be satisfied by choosing a proposal 
distribution $q$ that does not attempt to explore large $\theta$'s too 
aggressively.
The next corollary follows from a simple application of the Cauchy-Schwarz inequality, see the supplement for the proof.
\begin{corollary}
Given the proposal density $q(\ptheta | \theta)$, $\exists \eta_1 > 0, \theta_2 > 0$ such that for $\theta$ 
satisfying $\| \theta \|  > \theta_2$, 
$ \int_\Theta \Omega(\ptheta) q(\ptheta | \theta)d\ptheta \leq \eta_1 \Omega(\theta).$
\label{corol:integ_bound}
\end{corollary}



\noindent We need two further assumptions on $q(\theta'|\theta)$. There are satisfied in our experiments.
\begin{assumption}
For any $\epsilon>0$,  there exist finite $M_\epsilon$, $\theta_{3,\epsilon}$
such that for $\theta$ satisfying $\| \theta \|  > \theta_{3,\epsilon}$,
the condition $q(\{\theta': \frac{p(\theta')q(\theta|\theta')}{p(\theta)q(\theta'|\theta)} \le M_\epsilon\}|\theta) 
> 1 - \epsilon$ holds.
  \label{asmp:prior}
\end{assumption}
\noindent This holds, when e.g.\ $p(\theta)$ is a gamma distribution,
and $q(\theta'|\theta)$ is Gaussian.



\begin{assumption}
For any $\epsilon > 0$ and  $K > 1$, there exists $\theta_{4,\epsilon}^K$ such that 
for $\theta$ satisfying $\| \theta \|  > \theta_{4,\epsilon}^K$, the 
condition
$q(\{\theta':\frac{\Omega(\theta')}{\Omega(\theta)} \in 
  \left[\frac{1}{K}, K\right]\} | \theta) > 1 - \epsilon$ holds.
  \label{asmp:omega}
\end{assumption}
\noindent This holds when e.g.\ $q(\theta'|\theta)$ is a centered on $\theta$ and has 
finite variance.





\begin{theorem}
  Under the above assumptions, our symmetrized auxiliary variable MCMC sampler in algorithm~\ref{alg:MH_improved} is geometrically ergodic.  \label{thm:geom_erg}
\end{theorem}
\begin{proof}
\noindent This theorem follows from two lemmas we will prove.
Lemma~\ref{lem:small_set} shows there exist small sets 
$\{(W,\theta,\vartheta): \lambda_1|W| + \Omega(\theta) < M \}$ for 
$\lambda_1, M > 0$, within which our sampler forgets its current state with 
some positive probability. Lemma~\ref{lem:drift}
shows that for appropriate $(\lambda_1,M)$, our sampler drifts towards this set whenever
outside. Together, these two results imply geometric
ergodicity~\citep[Theorems 15.0.1 and Lemma 15.2.8]{meyn2009}.
If $\sup_\theta \Omega(\theta) < \infty$, we just need 
the small set $\{(W,\theta,\vartheta: |W| < M \}$ for some $M$.
\end{proof}
For easier comparison with the ideal sampler, we begin an MCMC iteration from step 5 in Algorithm~\ref{alg:MH_improved}. 
Thus, our sampler operates on $(\theta,\vartheta,W)$,  with $\theta$ the current parameter, $\vartheta$ the auxiliary variable, and $W$ the Poisson grid. 
An MCMC iteration updates this by 
(a) sampling states $V$  with a backward pass, 
(b) discarding $\vartheta$ and self-transition times, (c) sampling $\ptheta$ from $q(\ptheta|\theta)$, 
(d) sampling $U'$ given $(\theta,\ptheta, S, T)$, setting $W' = T \cup U'$, and discarding $S$, (e) proposing to swap 
$(\theta,\ptheta)$ and then (f) accepting or rejecting with a forward pass. 
On acceptance, $\theta' = \ptheta$ and $\vartheta' = \theta$, and on rejection, $\theta'=\theta$ and $\vartheta'=\ptheta$, so that the MCMC state at the end of the iteration is $(\theta',\vartheta', W')$. 
We write $(\theta'',\vartheta'',W'')$ for the MCMC state after two iterations.
Recall that step (a) actually assigns states $V$ to $W$. 
$T$ are the elements of $W$ where $V$ changes value, and $S$ are the corresponding elements of $V$. 
The remaining elements $U$ are the elements of $W$ corresponding to self-transitions. 
For reference, we repeat some of our notation in the supplementary material.

We first bound self-transition probabilities of the embedded Markov chain from $0$: 
\begin{proposition}
The posterior probability that the embedded Markov chain makes a
self-transition,
$P(V_i = V_{i + 1} | W, X, \theta, \vartheta) \ge \delta_1 > 0$,
for 
any $\theta,\vartheta, W$.
\label{prop:self_tr}
\end{proposition}
\noindent The proof (in the supplement) exploits the bounded likelihood from 
assumption~\ref{asmp:obs_bnd}. A simple by-product of the proof is
the following corollary:
\begin{corollary}
$P(V_{i + 1} = s|V_s=s, W, X, \theta, \vartheta) \ge \delta_1 > 0$,
for 
any $\theta,\vartheta, W,s$.
\label{corol:self_tr}
\end{corollary}

\begin{lemma}
  For all $M,h > 0$, the set $B_{h,M} =
\left\lbrace (W, \theta, \vartheta) : |W| \leq h, \theta \in \SM_M
\right\rbrace$ is a 2-small set under our proposed sampler. 
Thus, for all $(W,\theta,\vartheta)$ {in $B_{h, M}$},
the two-step transition probability satisfies 
$P(W'', \theta'',\vartheta'' | W, \theta, \vartheta) \ge \rho_1 
\phi_1(W{''}, \theta'',\vartheta'') $ for a constant $\rho_1$ and a 
probability measure $\phi_1$ independent of the initial state.
\label{lem:small_set}
\end{lemma}
\begin{proof} Recall the definition of $B_M$, and of an $n$-small set from 
  Assumption~\ref{asmp:ideal_geom}. The $1$-step transition probability of our MCMC algorithm
  consists of two terms, corresponding to the proposed parameter being
  accepted and rejected. Discarding the latter, we have 
\begin{align*}
  P(W',\theta',\vartheta'|W,\theta,\vartheta,X)&\geq q(\theta'|\theta) \delta_\theta(\vartheta') \alpha(\theta, \theta', W';X) 
  \sum_{S,T} \left[ P(S,T | W, \theta, \vartheta, X) P(W'| S, T, \theta, \theta') \right]. 
\end{align*}
This follows from steps (c) to (e) in the reordered algorithm.
The summation is over all $(S,T)$ values produced by the backward pass (which are then discarded after sampling $W'$).
We have used the fact that given $(S,T)$, $P(W'|S,T,\theta,\theta',X)$ is independent of  $X$.

We bound the summation over $(S,T)$ by considering only terms with $S$ constant. 
When this constant is  state $s^*$, we write this as $(S=[s^*], T= \emptyset)$. This corresponds to $|W|$ self-transitions 
after starting state $S_0=s^*$. Then the first term in the square brackets becomes
\begin{align*}
  P(S=[s^*], T = \emptyset | W, \theta, &\vartheta, X) =
P(S_0=s^*|X,W, \theta, \vartheta)\prod_{i = 0}^{|W| - 1} 
P(V_{i + 1} = s^* | V_i = s^*,X,W,\theta,\vartheta) \\ 
& \geq P(S_0=s^*|X,W, \theta, \vartheta)\delta_1^{|W|}  
\qquad (\text{from Corollary~\ref{corol:self_tr}}).
\end{align*}
With $S(t)$ fixed at $s^*$, $W'$ is distributed as a Poisson process with rate $\Omega(\theta') + \Omega(\theta) - A_{s^*}(\theta)$.
Write $\PP(W'|R(t))$ for the probability of $W'$ under a rate-$R(t)$ Poisson process on $[0,t_{end}]$, so that $P(W' |  S=[s^*], T = \emptyset, \theta', \theta) = \PP(W'|\Omega(\theta') + \Omega(\theta) - A_{s^*}(\theta))$. Then, from the Poisson superposition theorem, writing $2^{W'}$ for the power set of $W'$, we have 
\begin{align*}
  P(W' |  S=&[s^*], T = \emptyset, \theta', \theta) 
  = \sum_{Z \in 2^{W'}} \PP\left(Z|\Omega(\theta')\right)  
  \PP\left(W'\backslash Z|\Omega(\theta) - A_{s^*}(\theta)\right) \\
            &\geq \PP(W'|\Omega(\theta'))
\PP(\emptyset | \Omega(\theta)-A_{s^*}(\theta) )\\
  & \geq \PP(W' | \Omega(\theta')) \PP(\emptyset | \Omega(\theta) )\\
  & \geq \PP(W' | \Omega(\theta')) \exp(-M t_{end}),
\quad \text{since for $\theta \in B_M$, $\Omega(\theta) \le M$}.
\end{align*}
  Thus we have
  \begin{align*}
    \sum_{S,T} P(S,T,W' | W, \theta, \vartheta, X) & 
    \geq \sum_{{s^*}} P(S{=[s^*]}, T = \emptyset | W, \theta, \vartheta, X)
    P(W' | S{=[s^*]}, T=\emptyset,\theta', \theta)\\
                 &\geq \delta_1^{|W|} \exp(-Mt_{end})
  \PP(W' | \Omega(\theta')). \numberthis
  \label{eq:marg}
  \end{align*}
Next, 
using assumption~\ref{asmp:obs_bnd},
\begin{align*}
\alpha(\theta, \theta', W'; X) &
= 1 \wedge \frac{P(X|W', \theta', \theta) / P(X|\theta')}{P(X|W', \theta,
\theta') / P(X|\theta)} \cdot \frac{P(X | \theta')
q(\theta|\theta')p(\theta')}{P(X | \theta)q(\theta'|\theta)p(\theta)}\\
& \geq 1 \wedge \frac{\lb^2}{\ub^2} \cdot 	\frac{P(X | \theta')
q(\theta|\theta')p(\theta')}{P(X | \theta)q(\theta'|\theta)p(\theta)}
 \geq \alpha_I(\theta, \theta';X)\frac{\lb^2}{\ub^2}.
\numberthis
\label{eq:acc}
\end{align*}
Inside $B_{h,M}$, $|W| \le h$, 
and by assumption~\ref{asmp:ideal_geom}, $q(\theta'|\theta)\alpha_I(\theta,\theta';X) \ge \kappa_1 \phi(\theta')$. Then the three inequalities above let us simplify the equation at the start of the proof:
\begin{align*}
P(W', \theta',\vartheta' | W, \theta, \vartheta)  
  & \geq \frac{\lb^2 }{\ub^2}\delta_1^{h}
\exp(-M t_{end})\delta_\theta(\vartheta')\kappa_1\PP(W'|\Omega(\theta') \phi(\theta') \\
  & \defeq \rho_1 \delta_\theta(\vartheta')\PP(W'| \Omega(\theta') \phi(\theta').
\end{align*}
{Write $F_{Poiss(a)}$ for the CDF of a rate-$a$ Poisson.
The two-step transition satisfies}
\begin{align*}
  P(W'', \theta''&,\vartheta'' | W, \theta, \vartheta)  
       \ge \int_{\SM_{h,M}} P(W'', \theta'',\vartheta'' | W', \theta', \vartheta')
       P(W', \theta',\vartheta' | W, \theta, \vartheta)
       \dif W' \dif\theta' \dif\vartheta' \\
       &\ge \int_{\SM_{h,M}}  \rho_1 \delta_{\theta'}(\vartheta'')\PP(W'' | \Omega(\theta''))\phi(\theta'') \\
       &\qquad \qquad \rho_1 \delta_\theta(\vartheta')\PP(W'|\Omega(\theta')) \phi(\theta')
       \dif W' \dif \theta' \dif \vartheta' \\
       &\ge \rho_1^2 \phi(\theta'')\PP(W''|\Omega(\theta''))
       \int_{\SM_{h,M}} \!\!\!\! \delta_{\theta'}(\vartheta'')
       F_{Poiss(\Omega(\theta'))}(h)\phi(\theta')
       \dif\theta'  \\
       & \ge \rho_1^2 \PP(W''|\Omega(\theta''))\phi(\theta'')\phi(\vartheta'') F_{Poiss(\Omega(\vartheta''))}(h)\delta_{\SM_{h,M}}(\vartheta'') \\
       & \ge \rho_1^2 \PP(W'' | \Omega(\theta''))\phi(\theta'')
       \phi(\vartheta'')\delta_{\SM_{h,M}}(\vartheta'')\exp(-\Omega(\vartheta''))  \numberthis
       \label{eq:density_lowbound}
\end{align*}
The last line uses $F_{Poiss(a)}(h) \ge F_{Poiss(a)}(0) = \exp(-a)\ \forall 
a$, 
and gives our result, 
with $\phi_1(W'',\theta'',\vartheta'') \propto \PP(W'' | \Omega(\theta'')) \phi(\theta'') \phi(\vartheta'')
  \delta_{\SM_{h,M}}(\vartheta'')\exp(-\Omega(\vartheta''))  $.
\end{proof}
\noindent We have established the small set condition: for a point inside $B_{h,M}$ our sampler forgets its state with nonzero probability, sampling a new state from $\phi_1(\cdot)$. 
We next establish a drift condition, showing that outside $B_{h,M}$, the algorithm drifts back towards it (Lemma~\ref{lem:drift}).
We first establish a result needed when $\Mx{\theta}$ is unbounded 
as $\theta$ increases.
This states that the acceptance probabilities of our 
sampler and the ideal sampler can be brought arbitrarily close
outside a small set, so long as $\Omega(\theta)$ and
$\Omega(\theta')$ are sufficiently close.
  \begin{lemma}
  Suppose 
  $\frac{1}{K_0} \le \frac{\Omega(\theta)}{\Omega(\theta')} \leq K_0
  $, for $K_0$ satisfying $(1 + \frac{1}{K_0})k_1 \ge 2$  
  ($k_1$ is from Assumption~\ref{asmp:unif_rate}). Write $\mW$ for the
  minimum number of elements of grid $W$ between any successive pairs of observations.
  For any $\epsilon > 0$, there exist  $w^{K_0}_\epsilon,  \theta_{5, \epsilon}^{K_0} > 0$ such that
  $|P(X| W, \theta, \theta') - P(X | \theta)| < \epsilon$
  for any $(W, \theta)$ with $\mW > w^{K_0}_\epsilon$ and $\| \theta \| > \theta_{5, \epsilon}^{K_0}$.
  \label{lem:eigenvalue_lemma}
  \end{lemma}
  \begin{proof}
%
{From lemma \ref{lem:eig_lemma}, for all $\theta, \theta'$ satisfying 
 the lemma's assumptions, 
 the Markov chain with transition matrix $B(\theta, \theta')$ converges 
 geometrically to stationarity distribution $\pi_\theta$  
 at a rate uniformly bounded away from 0.
}
By setting $\mW$ large enough, for all such $(\theta,\theta')$ and for 
any initial state, the Markov chain would have mixed beween each pair of
observations, with distribution over states returning arbitrarily 
close to $\pi_\theta$.


Write $W_X$ for the indices of the grid $W$ containing observations, 
and $V_X$ for the Markov chain state at these times (illustrated 
in Section~\ref{sec:notation} in the supplementary material).
Let $P_B(V_X | W, \theta, \theta')$ be the probability distribution over
$V_X$ under the Markov chain with transition matrix $B$ given 
$W$ and $P_{st}(V_X|\theta)$ be the probability of $V_X$ sampled
independently under the stationary distribution. 
Let $P(X | W, \theta, \theta')$ be the marginal probability of the 
observations $X$ under that Markov chain $B(\theta,\theta')$ given $W$. Dropping $W$ and 
$\theta'$ from notation, $P(X|\theta)$
is the probability of the observations under the rate-$A(\theta)$ MJP.

From the first paragraph, for $\mW > w_{0}$ for 
large enough $w_0$,
$P_B(V_X | W, \theta, \theta')$ and  $P_{st}(V_X | W, \theta)$ can be
brought $\epsilon'$ close.
Then for any $W$ with $\mW > w_0$, we have
\begin{align*}
  |P(X|W , \theta, \theta') - & P_{st}(X | \theta)| = | \sum_{V_X} P(X|V_X, \theta) [P_B(V_X | W, \theta, \theta') -  P_{st}(V_X | \theta) ]|\\
& \leq \sum_{V_X} P(X | V_X, \theta)|P_B(V_X | W, \theta, \theta') -  P_{st}(V_X | \theta)|\le \epsilon'',
\end{align*}
using $P(X|V_X,\theta) \le \ub$ 
(Assumption~\ref{asmp:obs_bnd}), and 
$\sum_{V_X} |P_B(V_X | W, \theta, \theta') -  P_{st}(V_X | \theta)| < \epsilon$.
For large $\theta$, we prove a similar result in the continuous 
case by uniformization. For any $\theta'$,
\begin{align*}
P(X | \theta) 
= \int \dif W P(X | W, \theta, \theta') \PP(W | \Omega(\theta)+\Omega(\theta')). 
\end{align*}
We split this integral into two parts, one over the set $\{\mW > w_0\}$, and 
the second over its complement. On the former, for $w_0$ large enough, 
$|P(X |W, \theta, \theta')-P_{st}(X|\theta)|
\le \epsilon''$. 
For $\theta$ large enough, $\{\mW > w_0\}$ occurs
with arbitrarily high probability for any $\theta'$. Since the likelihood is bounded, the
integral over the second set can be made arbitrarily small (say, $\epsilon''$ again). 
Finally, from the triangle
inequality,
\begin{align*}
|P(X | \theta) - P(X | W, \theta, \theta')| &\leq |P(X | \theta) -P_{st}(X | \theta) | + | P_{st}(X | \theta) -  P(X | W, \theta, \theta')|\\
       & \leq (\epsilon'' + \epsilon'') + \epsilon'' \defeq \epsilon.
\end{align*}
\vspace{-.2in}
\end{proof}
\noindent The previous lemma bounds the difference in probability of observations under the discrete-time and continuous-time processes for $\theta$ and $|W|$ large enough. 
The next result uses this to bound with high probability the different in acceptance probabilities of the ideal sampler, and our proposed sampler with a grid $W$. 
See the supplement for the proof.
\begin{proposition}
  Let $(W, \theta, \vartheta)$ be the current state of the sampler.
Then, for any $\epsilon$, there exists $\theta_\epsilon > 0$ as well as a set $\E_\epsilon \subseteq \{(W', \theta'): |\alpha_I(\theta,\theta';X) - \alpha(\theta,\theta';W',X)| \le \epsilon\}$, such that for $\theta$ satisfying $\| \theta \| > \theta_\epsilon$ and any $\vartheta$, we have
$P(E_\epsilon|W,\theta,\vartheta) > 1-\epsilon$.
\label{prop:mix0}
\end{proposition}

\begin{lemma}(drift condition) There exist $\delta_2 \in (0, 1), \lambda_1 > 0$ and $L > 0$
  such that \\
  $\mathbb{E}\left[\lambda_1|W'| + \Omega(\theta')  | W, \theta, \vartheta, X\right]
  \leq (1 - \delta_2)\left(\lambda_1|W| + \Omega(\theta)   \right) + L$.
\label{lem:drift}
\end{lemma}
\begin{proof}
Since $W'=T\cup U'$, we consider $\mathbb{E}[|T| |W,\theta,\vartheta,X]$
and $\mathbb{E}[|U'| | W, \theta, \vartheta, X]$ separately.
An upper bound of $\mathbb{E}[|T| | W,\theta,\vartheta, X]$ can be derived
directly from proposition~\ref{prop:self_tr}:
\begin{align*}
\mathbb{E}[|T| |W,\theta,\vartheta,X] &= \mathbb{E}[\sum_{i = 0}^{|W|-1}
  \mathbb{I}_{\{ V_{i + 1} \neq V_i \}}| W, \theta, \vartheta, X]
\leq \sum_{i = 0}^{|W| - 1} (1 - \delta_1) = |W|(1 - \delta_1).
\end{align*}
By corollary \ref{corol:integ_bound}, there exist $\eta_1 , \theta_2$ 
such that for 
$ \| \theta \| > \theta_2$, 
 $ \int \Omega(\ptheta) q(\ptheta | \theta)d\ptheta \leq \eta_1 \Omega(\theta) 
 $. Then,
\begin{align*}
\mathbb{E}[|U'| |W, \theta, \vartheta, X] &= 
\mathbb{E}_{S,T, \ptheta}\mathbb{E}[|U'| | S, T, W, \theta, \vartheta, \ptheta, X] = \mathbb{E}_{S,T, \ptheta}\mathbb{E}[|U'| | S, T, W, \theta, \ptheta] \\
& \leq \mathbb{E}_{S,T, \ptheta} \left[t_{end}\Omega(\theta, \ptheta)\right] = t_{end}\int \Omega(\theta, \ptheta) q(\ptheta | \theta) \dif\ptheta\\
& = t_{end} \left[ \left(  \Omega(\theta) +
\int_\Theta \Omega(\ptheta) q(\ptheta | \theta)\dif\ptheta \right) \right] 
 \leq t_{end} (\eta_1 + 1) \Omega(\theta). 
\end{align*}
To bound $\mathbb{E}\left[\Omega(\theta')  | W, \theta, \vartheta, X\right]$,
consider the transition probability over $(W',\theta')$:
\begin{align*}
  P(\dif W', \dif \theta'&| W, \theta, \vartheta)
=\dif\theta' \dif W' \left[q(\theta' | \theta)
  \sum_{S,T} P(S, T | W, \theta, \vartheta, X)P(W' | S, T, \theta, \theta')
\alpha(\theta, \theta' ; W', X)\right. \\
&\left.+ \int q(\ptheta | \theta) \sum_{S,T} P(S, T|W,\theta,\vartheta,
    X)P(W' | S, T, \theta, \ptheta) ( 1 - {\alpha(\theta, \ptheta ; W', X)})\dif\ptheta
    \delta_\theta(\theta')\right].
\end{align*}
With $P(W' | W, \theta, \vartheta, \theta', X) =
\sum_{S,T} P(S, T | W, \theta, \vartheta, X)P(W' | S, T, \theta, \theta')$,
integrate out $W'$:
\begin{align*}
  P(\dif\theta'| W, \theta, \vartheta) &=\dif\theta' \int \dif W'
  \bigg[q(\theta' | \theta)
     P(W' | W, \theta, \vartheta, \theta', X) \alpha(\theta, \theta' ; W', X) + \\
  &\left.  \int q(\ptheta | \theta)  P(W' |  W, \theta, \vartheta, \ptheta,
X) ( 1 - \alpha(\theta, \ptheta ; W', X))d\ptheta
\delta_\theta(\theta')\right] 
\end{align*}
Let  $\int \Omega(\theta') P(\dif\theta'| W, \theta, \vartheta)
  = I_1(W, \theta, \vartheta) + \Omega(\theta) I_2(W, \theta, \vartheta) 
$, with
\begin{align*}
  &I_1(W, \theta, \vartheta) = \int \dif\theta' \Omega(\theta') q(\theta' | \theta)\int \dif W'P(W' | W, \theta, \vartheta, \theta', X)\alpha(\theta, \theta' ; W', X), \\
&I_2(W, \theta, \vartheta) =\int \dif\ptheta  dW'q(\ptheta | \theta)P(W' | W, \theta, \vartheta, \ptheta, X)(1 - \alpha(\theta, \ptheta ; W', X)).
\end{align*}
{Consider the second term $I_2$.
  From Proposition~\ref{prop:mix0}, for any positive $\epsilon$, there
  exists $\theta_\epsilon > 0$ such that the set $E_{\epsilon}$
  (where $|\alpha(\theta, \ptheta ; X,W') - \alpha_I(\theta, \ptheta ; X)| \le
  \epsilon$) has probability greater than $1-\epsilon$.
  Write $I_{2,\E_{\epsilon}}$ for the integral restricted to this set, 
  and $I_{2,\E_{\epsilon}^c}$ for that over the complement, so that
 $I_{2}= I_{2,\E_{\epsilon}}+I_{2,\E_{\epsilon}^c}$.
 Then for $\theta > \theta_\epsilon$,}
\begin{align*}
I_{2,\E_{\epsilon}}(W, \theta, \vartheta) &= \int_{\E_{\epsilon}} d\ptheta  dW'q(\ptheta | \theta)P(W' | W, \theta, \vartheta, \ptheta, X)(1 - \alpha(\theta, \ptheta ; W', X)) \\
&\le \int_{\E_\epsilon}d\ptheta dW' q(\ptheta | \theta)
  P(W' | W, \theta, \vartheta, \ptheta, X)  [ 1 - (\alpha_I(\theta, \ptheta ; X)-\epsilon)] \\
&\le   \int d\ptheta dW'  q(\ptheta | \theta)
  P(W' | W, \theta, \vartheta, \ptheta, X)
  [ 1 - (\alpha_I(\theta, \ptheta ; X)-\epsilon)] \\
  &\le (1+\epsilon)  - \int  q(\ptheta | \theta) \alpha_I(\theta, \ptheta ; X) d\ptheta, \quad \text{and}  \\
  I_{2,\E_{\epsilon}^c}(W, \theta, \vartheta)  &= \int_{\E^c_{\epsilon}} d\ptheta  dW'q(\ptheta | \theta)P(W' | W, \theta, \vartheta, \ptheta, X)(1 - \alpha(\theta, \ptheta ; W', X)) \\
  &\le \int_{\E^c_{\epsilon}} d\ptheta dW'
  q(\ptheta | \theta) P(W' | W, \theta, \vartheta, \ptheta, X) \le \epsilon.
\end{align*}
We similarly divide the integral $I_1$ into two parts, $I_{1,E_\epsilon}$ 
(over $\E_\epsilon$) and $I_{1,E_\epsilon^c}$ (over its complement 
$\E^c_\epsilon$).
For $\|\theta\|$ large enough, we can bound the acceptance probability 
by $\alpha_I(\theta,\theta'; X) + \epsilon$ on the set $E_\epsilon$, and by corollary 
\ref{corol:integ_bound}, we get 
\begin{align*}
  I_{1,E_\epsilon} & \leq \int_{\E_\epsilon} \Omega(\theta')q(\theta' | \theta) (\alpha_I(\theta, \theta'; X) + \epsilon) d\theta' 
 \leq \int \Omega(\theta')q(\theta' | \theta) \alpha_I(\theta, \theta'; X) d\theta' + \eta_1 \epsilon \Omega(\theta).
\end{align*}
For $I_{1,E_\epsilon^c}$, 
from assumption \ref{asmp:integ_bound}, 
we have $\int_\Theta \Omega(\ptheta)^2 q(\ptheta | \theta)d\ptheta \leq \eta_0 \Omega(\theta)^2$
for $\|\theta\| > \theta_2$.
So, by Cauchy-Schwarz inequality and bounding the acceptance probability by one, we have
\begin{align*}
  \left( I_{1,E_\epsilon^c}\right)^2 
& \le \int_{\E_\epsilon^c} q(\theta' | \theta) P(W' |W, \theta, \vartheta, \theta', X)d\theta'dW'  \int_{\E_\epsilon^c}  \Omega(\theta')^2 q(\theta' | \theta) P(W' |W, \theta, \vartheta, \theta', X)d\theta'dW'   \\
& \le \epsilon \int  \Omega(\theta')^2 q(\theta' | \theta)d\theta' 
 \le \epsilon \eta_0 \Omega(\theta)^2 ,
\end{align*}
giving
$
I_{1,E_\epsilon^c}  \le \sqrt{\epsilon \eta_0 }\Omega(\theta) .
$
Putting these four results together, for $\theta$ satisfying $ \| \theta \| >\max(\theta_2, 
\theta_\epsilon, M) $ 
(where $M$ is from Assumption~\ref{asmp:ideal_geom} on 
the ideal sampler), we have 
\begin{align*}
  \int \Omega(\theta') P(d\theta'| W, \theta, \vartheta)
  & \leq \int \Omega(\theta') q(\theta' | \theta)\alpha_I(\theta, \theta'| X) d\theta'  + \Omega(\theta)\int  q(\ptheta | \theta) (1 - \alpha_I(\theta, \ptheta | X)) d\ptheta+ \\
  &\sqrt{\eta_0\epsilon}\Omega(\theta)  +  \eta_1 \epsilon \Omega(\theta) + 2\epsilon \Omega(\theta)\\
  & \leq (1 - \rho) \Omega(\theta) + (\sqrt{\eta_0\epsilon} +  \eta_1 \epsilon + 2\epsilon) \Omega(\theta) + L_I, \quad \text{giving}
\end{align*}
\begin{align*}
\mathbb{E}[\lambda_1 | W'| &+ \Omega(\theta')| W, \theta, \vartheta, X] \le \lambda_1(1 - \delta_1)|W| + \lambda_1 t_{end} (1 + \eta_1)\Omega(\theta) +\\
&  (1 - \rho) \Omega(\theta) + (\sqrt{\eta_0} \sqrt{\epsilon} +  \eta_1 \epsilon + 2\epsilon) \Omega(\theta) + L_I\\
& = (1 - \delta_1)\lambda_1 |W| + [1 - (\rho - \lambda_1 t_{end} (1 + \eta_1) - (2 + \eta_1)\epsilon - \sqrt{\eta_0 \epsilon})]\Omega(\theta) + L_I\\
& \defeq (1 - \delta_1)\lambda_1 |W| + (1 - \delta_2)\Omega(\theta) + L_I
\end{align*}
For $(\lambda_1,\epsilon)$ small enough, $\delta_2 \in (0,1)$, 
      and  $\delta = \min(\delta_1,\delta_2)$ gives the drift condition.
\end{proof}

 \section{Conclusion}
We have proposed a novel Metropolis-Hastings algorithm for parameter 
inference in Markov jump processes. We use 
uniformization to update the MJP parameters with state-values marginalized 
out, though still conditioning on a random Poisson grid. The 
distribution of this grid depends on the MJP parameters, significantly 
slowing down MCMC mixing. We propose a simple symmetrization scheme to get 
around this dependency. In our experiments, we demonstrate the usefulness 
of this scheme, which outperforms a number of competing baselines.
We also derive conditions under which our sampler inherits geometric 
ergodicity properties of an ideal MCMC sampler.

There are a number of interesting directions for future research.
Our focus was on Metropolis-Hastings algorithms for typical settings,
where the parameters are low dimensional. It is interesting to 
investigate how our ideas extend to schemes like Hamiltonian Monte 
Carlo~\citep{Neal2010} suited for higher-dimensional settings. Another 
direction is to develop and study similar schemes for more complicated 
hierarchical models like mixtures of MJPs or coupled MJPs. While we 
focused only on Markov jump processes, it is also of interest to study 
similar ideas for algorithms for more general processes~\citep{RaoTeh12}. 
It is also important to investigate how similar ideas apply to 
deterministic algorithms like variational Bayes~\citep{OpperVarinf, panzharao17}. From 
a theoretical viewpoint, our proof required the uniformization rate to 
satisfy $\Omega(\theta) \ge k_1 \max_s A_s(\theta) + k_0$ for $k_1 > 1$. 
We believe our result still holds for $k_1 = 1$, and for completeness, 
it would be interesting to prove this.

\section{Supplementary material}
\begin{description}
  \item[Appendix] This file includes a summary of notation used in the main text, proofs not included in the main text, details of the \naive\ and particle MCMC algorithms, as well as experimental results not included in the main text. [\texttt{Appendix\_ZhangRao.pdf}]. 
\item[Python code] This includes code implementing the symmetrized MH algorithm, as well as the E Coli dataset.
  {\texttt{README.txt}} includes instructions. The github repository {\url{https://github.com/varao/ZhangRao_JCGS_code}} also contains the code.
 [\texttt{Code\_ZhangRao.tar.gz}].
\end{description}

\section{Acknowledgements}
We thank the anonymous reviewers whose suggestions helped to significantly improve this manuscript.
We acknowledge the National Science Foundation for funding under grants RI/1816499 and DMS/1812197.

\bibliographystyle{./agsm}
\setcitestyle{authoryear}
\bibliography{bibfile}

\newpage
\bigskip
\begin{center}
{\large\bf APPENDIX to ``Efficient Parameter Sampling for Markov
Jump Processes", by Boqian Zhang and Vinayak Rao}
\end{center}

\subsection{Notation}\label{sec:notation}

\begin{figure}[H]
  \centering
  \begin{minipage}[H]{0.65\linewidth}
  \centering
    \includegraphics [width=0.64\textwidth, angle=0]{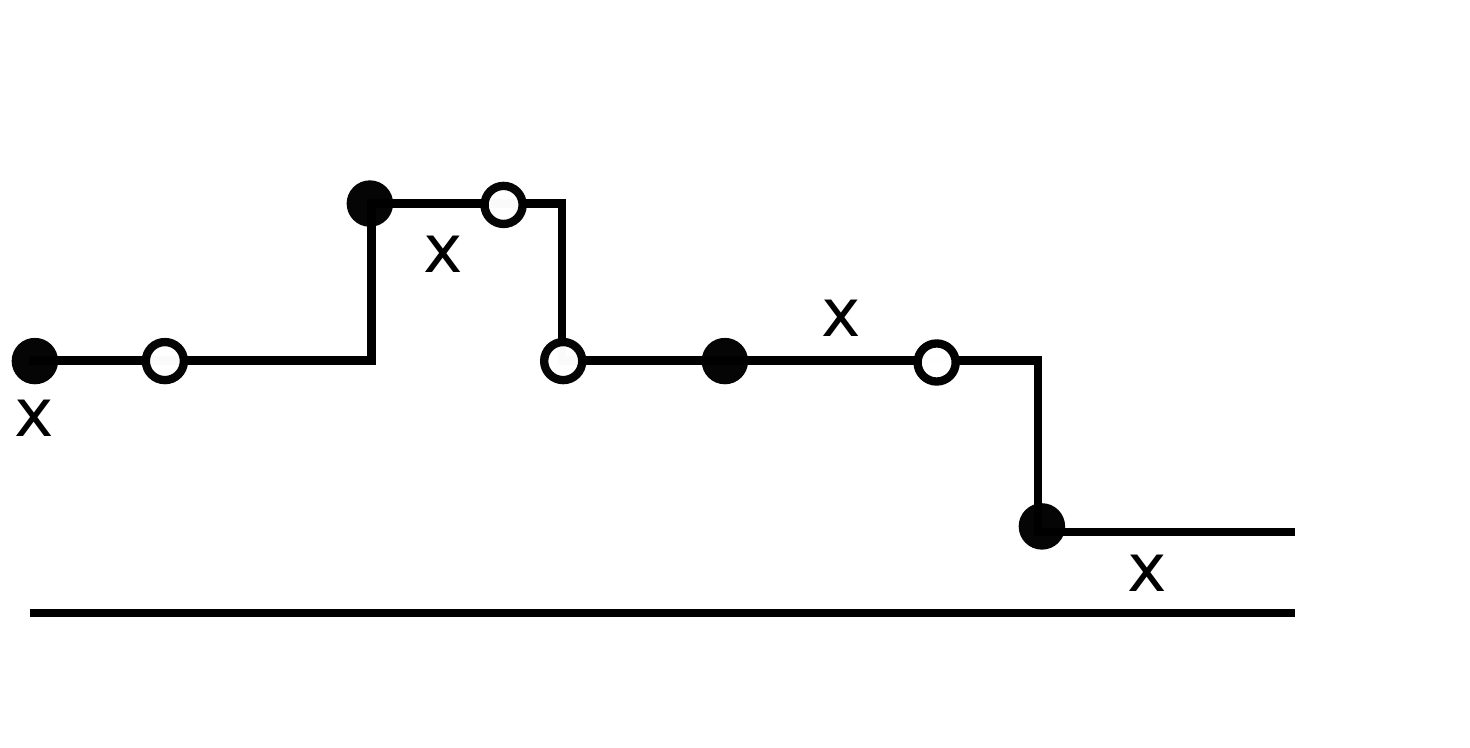}

  \end{minipage}
\caption{an example MJP path}
     \label{fig:vx}
  \end{figure}
{We recall some notation used in our proof. 
  The figure about shows a realization $S(t)$ of an MJP with rate matrix 
  $A(\theta)$ and initial distribution $\pi_0$ over an interval 
  $[0,t_{end}]$. The crosses are observations $X$. $\pi_0$ is the initial 
  distribution over states, and $\pi_\theta$ is the staionary distribution 
  of the MJP.
    $p(\theta)$ is the prior over $\theta$, and $q(\nu|\theta)$ is 
      the proposal distribution.
  \begin{itemize}
    \item The uniformized representation of $S(t)$ is the pair 
      $(V, W)$, with the Poisson grid 
      $W = [w_1, w_2, w_3, w_4, w_5, w_6, w_7]$ and the states 
      $V =[v_0, v_1, v_2, v_3, v_4, v_5, v_6, v_7]$ assigned with through 
      a Markov chain with initial distribution $\pi_0$ and transition 
      matrix $B(\theta,\theta')$. In the figure, the circles (filled and 
      empty) correspond to $W$.
    \item The more standard representation of $S(t)$ is the pair $(S,T)$.
      Here $T$ are the elements of $W$ which are true jump times (when 
      $V$ changes value), and $S$ are the corresponding elements of $V$.
      $U$ are the remaining elements of $W$ corresponding to self-transitions.
      Here, $T = [w_2, w_4, w_7]$ and $U = [w_1, w_3,  w_5, w_6]$. 
    \item The filled circles represent $W_X$, which are the elements of 
      $W$ containing observations. $V_X$ are the states corresponding to 
      $W_X$. In this example, $W_X = [w_2, w_5,w_7] \cup {\{0\}}$ and 
      $V_X = [v_2, v_5, v_7] \cup {\{v_0\}}$. 
    \item We write $\mW$ for the minimum number of elements of $W$ between 
      successive pairs of observations (including start time $0$). 
      In this example, $\mW = \min (3,3,2)= 2$. 
    \item $P(X | W, \theta, \theta')$ is the marginal distribution of 
      $X$ on $W$ under a Markov chain with 
      transition matrix $B(\theta, \theta')$ (after integrating out the 
      state information $V$). Recall that the LHS does not depend 
      on $\theta'$ because of uniformization.
    \item $P(X|\theta)$ is the marginal probability of the 
      observations under the rate-$A(\theta)$ MJP. 
$P(X | \theta) =  \int_W P(X | W, \theta, \theta') P(W|\theta, \theta')dW$.
    \item $P_B(V_X | W, \theta, \theta')$ is the probability 
      distribution over states $V_X$ for the Markov chain with 
      transition matrix $B(\theta, \theta')$ on the grid $W$, with the remaining elements 
      of $V$ integrated out. 
    \item  $P_{st}(V_X|\theta)$ is the probability of $V_X$ when 
      elements of $V_X$ are sampled i.i.d.\ from $\pi_\theta$).
    \item  $P_{st}(X|\theta)$ is the marginal probability of $X$
      when $V_X$ is drawn from $P_{st}(V_X|\theta)$.
  \end{itemize}
} 
\setcounter{theorem}{3}

\subsection{Remaining proofs}
\begin{corollary}
Given the proposal density $q(\ptheta | \theta)$, $\exists \eta_1 > 0, \theta_2 > 0$ such that for $\theta$ 
satisfying $\| \theta \|  > \theta_2$, 
$ \int_\Theta \Omega(\ptheta) q(\ptheta | \theta)d\ptheta \leq \eta_1 \Omega(\theta).$
\end{corollary}
\begin{proof}
From assumption \ref{asmp:integ_bound},  we have $ \int_\Theta \Omega(\ptheta)^2 q(\ptheta | \theta)d\ptheta \leq \eta_0 \Omega(\theta)^2$ for $\theta$ satisfying $\| \theta \|  > \theta_2$.
For such $\theta$, by the Cauchy-Schwarz inequality, we have
\begin{align*}
\left[ \int_\Theta \Omega(\ptheta) q(\ptheta | \theta) d\ptheta \right]^2 &\le \int_\Theta \Omega(\ptheta)^2 q(\ptheta | \theta) d\ptheta \cdot \int_\Theta q(\ptheta | \theta) d\ptheta \le \eta_0 \Omega(\theta)^2.
\end{align*}
So for $\theta$ satisfying $\| \theta \|  > \theta_2$, we have $\int_\Theta \Omega(\ptheta) q(\ptheta | \theta) d\ptheta \le \sqrt{\eta_0} \Omega(\theta).$
\end{proof}

\begin{proposition}
The a posteriori probability that the embedded Markov chain makes a
self-transition,
$P(V_{i + 1} = V_i| W, X, \theta, \vartheta) \ge \delta_1 > 0$,
for 
any $\theta,\vartheta, W$.
\end{proposition}
\begin{proof}
  We use $k_0$ from assumption~\ref{asmp:unif_rate}
  to bound {\em a priori} self-transition probabilities:
  \begin{align*}
    P(V_{i + 1}=s|V_i=s, W, \theta,\vartheta) &= B_{ss}(\theta,\vartheta) =
    1 - \frac{A_{s}(\theta)}{\Omega(\theta, \vartheta)}
    \ge 1 - \frac{A_{s}(\theta)}{\Omega(\theta)} \ge 1-\frac{1}{k_0}.
    \intertext{  We then have}
  P(V_i = V_{i + 1} | W, X, \theta, \vartheta) &= \sum_v P(V_i = V_{i + 1}
  = v | W, X, \theta, \vartheta)
 =\sum_v \frac{P(V_i = V_{i + 1} = v, X | W, \theta, \vartheta)}{P(X | W,
 \theta, \vartheta)} \\
&=\sum_v \frac{P(X | V_i = V_{i + 1} = v, W, \theta, \vartheta)P( V_i =
V_{i + 1} = v|W, \theta, \vartheta)}{P(X | W, \theta, \vartheta)}\\
& \geq \frac{\lb}{\ub}\sum_v P(V_i = V_{i + 1} = v | W, \theta, \vartheta)\\
&=  \frac{\lb}{\ub} \sum_v P(V_{i + 1} = v | V_i = v, W, \theta,\vartheta)P(V_i = v | \theta, \vartheta) \\
& \geq \frac{\lb}{\ub} (1-\frac{1}{k_0}) \assign \delta_1 > 0.
\end{align*}
\qed
\end{proof}
\setcounter{theorem}{6}

\begin{proposition}
  Let $(W, \theta, \vartheta)$ be the current state of the sampler.
Then, for any $\epsilon$, there exists $\theta_\epsilon > 0$ as well as a set $\E_\epsilon \subseteq \{(W', \theta'): |\alpha_I(\theta,\theta';X) - \alpha(\theta,\theta';W',X)| \le \epsilon\}$, such that for $\theta$ satisfying $\| \theta \| > \theta_\epsilon$ and any $\vartheta$, we have
$P(E_\epsilon|W,\theta,\vartheta) > 1-\epsilon$.
\end{proposition}
\begin{proof}
  Fix $\epsilon > 0$ and $K > 1$ satisfying $(1 + \frac{1}{K})k_1 \ge 2$.
  \begin{itemize}
    \item  
From assumption \ref{asmp:prior}, there exist 
$M_\epsilon$ and $\theta_{1,\epsilon}$, such that 
$P(\frac{q(\theta | \theta')p(\theta')}{q(\theta' | \theta)p(\theta)}\leq M_\epsilon) > 1 - \epsilon / 2$ 
{for } $\theta$ satisfying $ \| \theta \| > \theta_{1,\epsilon}$.
    Define  $E_1^\epsilon = \{\theta' s.t.\ \frac{q(\theta | \theta')p(\theta')}{q(\theta' | \theta)p(\theta)}\leq M_\epsilon\}$.
    \item Define $E^K_2 = \{\theta' s.t.\ \frac{\Omega(\theta')}{\Omega(\theta)}\in [1/K, K] \}$.
Following assumption \ref{asmp:omega}, define $\theta_{2, \epsilon}^K$ 
such that $P(E^K_2 | \theta) > 1 - \epsilon / 2$ for all $\theta$ satisfying $\| \theta \| > \theta_{2, \epsilon}^K$.
    \item On the set $E^K_2$, $\Omega(\theta') \le K \Omega(\theta)$ (and also
$\Omega(\theta) \le K \Omega(\theta')$).  Lemma~\ref{lem:eigenvalue_lemma}
ensures that there exist $\theta_{3,\epsilon}^K > 0, w_\epsilon^K > 0$, 
such that for $\mW > w_\epsilon^K$, 
$ \| \theta \| > \theta_{3,\epsilon}^K$ and 
$ \| \theta' \| > \theta_{3,\epsilon}^K$, we have
$|P(X | W, \theta' , \theta) - P(X | \theta' )| < \epsilon$, and
$|P(X | W, \theta , \theta') - P(X | \theta )| < \epsilon$.
Define      $E_{3, \epsilon}^K = \{\theta' s.t. \| \theta'\| > {\theta}_{3,\epsilon}^K \}$.
    \item Define $E_{4, \epsilon}^K =\{W s.t.\ \mW >  {w}_{\epsilon}^K\}$.
Set $\theta_{4,\epsilon}^K$, so that for $\| \theta \| > \theta_{4,\epsilon}^K$, 
$ P(E_{4,\epsilon}^K|E^K_2, E_1^\epsilon ) > 1 - \epsilon.
$. This holds since $W$ comes from a Poisson processes, whose rate can 
be made arbitrarily large by increasing $\Omega(\theta)$.
    \item From assumption \ref{asmp:mono_tail}, there exists $\theta_0$, such that $\Omega(\theta)$ increases as $\| \theta \|$ increases, for $\theta$ satisfying $\| \theta \| > \theta_0$.
Set ${\theta}_\epsilon =\max(\theta_0, \theta_{1, \epsilon},\theta_{2, \epsilon}^K,{\theta}_{3, \epsilon}^K, \theta_{4, \epsilon}^K)$. \\
  \end{itemize}
  Now consider the difference
\begin{align*}
|\alpha(\theta, \theta'; W, X) - \alpha_I(\theta, \theta'; X)| &= \ \mid 1 \wedge \frac{P(X | W, \theta' , \theta)q(\theta | \theta')p(\theta')}{P(X | W, \theta , \theta')q(\theta' | \theta)p(\theta)} - 1 \wedge \frac{P(X | \theta')q(\theta | \theta')p(\theta')}{P(X | \theta)q(\theta' | \theta)p(\theta)} \mid \\
& \leq \ \mid \frac{P(X | W, \theta' , \theta)}{P(X | W, \theta , \theta')} - \frac{P(X | \theta')}{P(X | \theta)}\mid  \frac{q(\theta | \theta')p(\theta')}{q(\theta' | \theta)p(\theta)}.
\end{align*}
%
On $E^\epsilon_1$, $\frac{q(\theta | \theta')p(\theta')}{q(\theta' | \theta)p(\theta)} \le M_\epsilon$.
Since ${P(X | W, \theta , \theta')}$ and ${P(X | \theta)}$
are lower-bounded by $\lb$, for any $\epsilon > 0$ we can 
find a $K$ such that on $E^K_2 \cap E^K_{3,\epsilon}$,
\begin{align*}
|\frac{P(X | W, \theta' , \theta)}{P(X | W, \theta , \theta')} - 
   \frac{P(X | \theta')}{P(X | \theta)}| < \epsilon / M_\epsilon.
\end{align*}
This means that on $E^\epsilon_1 \cap E_2^K \cap E^K_{3,\epsilon}$,
$|\alpha(\theta, \theta', W, X) - \alpha_I(\theta, \theta', X)| < \epsilon$.\\
For $\theta > \max (\theta_{1,\epsilon},\theta^K_{2,\epsilon})$ 
we have 
$
P(E^K_2 E^\epsilon_1) \ge P(E^K_2) + P(E_1^\epsilon) - 1 \ge 1 - \epsilon.
$\\
When $E^K_2$ holds, $\Omega(\theta') \ge \Omega(\theta)/K$.
For $\theta$ large enough, we can ensure $ \| \theta' \| > {\theta}^K_{3,\epsilon}$.
So \begin{align*}
P(E_1^\epsilon E_2^K E_{3,\epsilon}^K E_{4,\epsilon}^K) > (1- \epsilon)^2.
\end{align*} 
Finally, set $E_\epsilon \assign E_1^\epsilon \cap E_2^K \cap E_{3,\epsilon}^K \cap E_{4,\epsilon}^K$ , giving us our result.

\end{proof}

\subsection{Particle MCMC for MJP inference}
\label{sec:pmcmc}
\subsubsection{A sequential Monte Carlo algorithm for MJPs inference}
We describe a sequential Monte Carlo algorithm for MJPs inference that underlies particle MCMC. 
Denote by $S_{[t_1', t_2']}$ the MJP trajectory from time $t_1'$ to time $t_2'$. 
Our target is to sample an MJP trajectory $S_{[0, t_{end}]}$ given $n$ noisy observations $X =(x_1, x_2, ... , x_n)$, at time $t_1^X, t_2^X, ..., t_n^X$. 
The initial value of the Markov jump process trajectory can be simulated from its initial distribution over states: $S(0) \sim \pi_0$. 
$S_{[t_i^X, t_{i + 1}^X]} $, its values over any interval $[t_i^X, t_{i+1}^X]$ can be simulated by Gillespie's algorithm as described in section~\ref{sec:intro} 
For the $i$th observation $x_i$ at time $t^X_i$, denote the likelihood for $S(t^X_i)$ as $P(x_i | S(t^X_i))$.



\begin{algorithm}[H]
  \caption{The SMC sampler for MJP trajectories}
   \label{alg:SMC}
  \begin{tabular}{l l}
   \textbf{Input:  } & \text{Prior $\pi_0$, $n$ observations $X$}, 
                       \text{Number of particles $N$}, rate-matrix $A$.\\
   \textbf{Output:  }& \text{New MJP trajectory $S' (t) = (s'_0, S', T')$}.\\
   \hline
   \end{tabular}
   \begin{algorithmic}[1]
\State Define $t^X_0 = 0$ and $t^X_{n+1} = t_{end}$. 
\State 
Sample initial states for N particles $S^k(0)$ from $\pi_0$, $k = 1,...,N$. 
\For{$i = 1, ..., n+1:$}
	\State (a) For $k = 1,2,...,N$, update particle $k$ from $[0,t^X_{i-1}]$ to $[0,t^X_i]$ by forward simulating $S_{[t^X_{i -1},t^X_i]}^k|S^k(t^X_{i-1})$ via Gillespie's algorithm.
	\State (b) Calculate the weights $w^k_i = P(x_i|S^k(t^X_i))$
and normalize 
$W^k_i = \frac{w^k_i}{\sum_{k = 1}^N w^k_i}.$ 
	\State (c) Sample $J_{i}^k \sim \text{Multi}(\cdot| (W^1_{i},\dotsc,W^N_{i}))$ , $k = 1,2,...,N$.
	\State (d) Set $S_{[0, t^X_i]}^k := S_{[0,t^X_i]}^{J^k_i}$.
	\EndFor
\end{algorithmic}
\end{algorithm}


The SMC algorithm gives us an estimate of the marginal likelihood $P_\theta(X_{1:n})$.
$$ \hat{P}_{\theta} = \hat{P}_{\theta}(X_1) \prod_{i = 2}^n \hat{P}_{\theta}(X_i| X_{1: i- 1}) = \prod_{i = 1}^n \left[ \sum_{k = 1}^N  \frac{1}{N} w_i^k \right].$$

\subsubsection{Particle MCMC algorithm for MJPs inference}

\begin{algorithm}[H]
  \caption{The particle marginal MH sampler for MJP trajectories}
   \label{alg:SMC}
  \begin{tabular}{l l}
   \textbf{Input:  } & \text{The observations $X$, the MJP path $S(t) = (s_0, S, T)$},\\
                       &\text{number of particles $N$}, parameter $\theta$ and $\pi_0$,\\
                     & $P(\theta)$ prior of $\theta$, proposal density $q(\cdot|\cdot)$. \\
   \textbf{Output:  }& \text{New MJP trajectory $S' (t) = (s'_0, S', T')$}.\\
   \hline
   \end{tabular}
   \begin{algorithmic}[1]
\State Sample $\theta^* \sim q(\cdot | \theta)$.
\State Run the SMC algorithm above targeting $P_{\theta^*}(\cdot | X_{1:n})$ to sample $S^*(t)$ from $\hat{P}_{\theta^*}(\cdot | X_{1:n})$ and let $\hat{P}_{\theta^*}$ denote the estimate of the marginal likelihood.\\
Accept $\theta^*, S^*(t)$ with probability $$ \mathtt{acc} = 1 \wedge \frac{\hat{P}_{\theta^*} P(\theta^*)}{\hat{P}_{\theta} P(\theta)} \frac{q(\theta | \theta^*)}{q(\theta^* | \theta)}.$$

\end{algorithmic}
\end{algorithm}

\subsection{Algorithm sketch}
\setlength{\unitlength}{0.8cm}
  \begin{figure}[H]
  \centering
  \begin{minipage}[!hp]{0.45\linewidth}
  \centering
    \includegraphics [width=0.70\textwidth, angle=0]{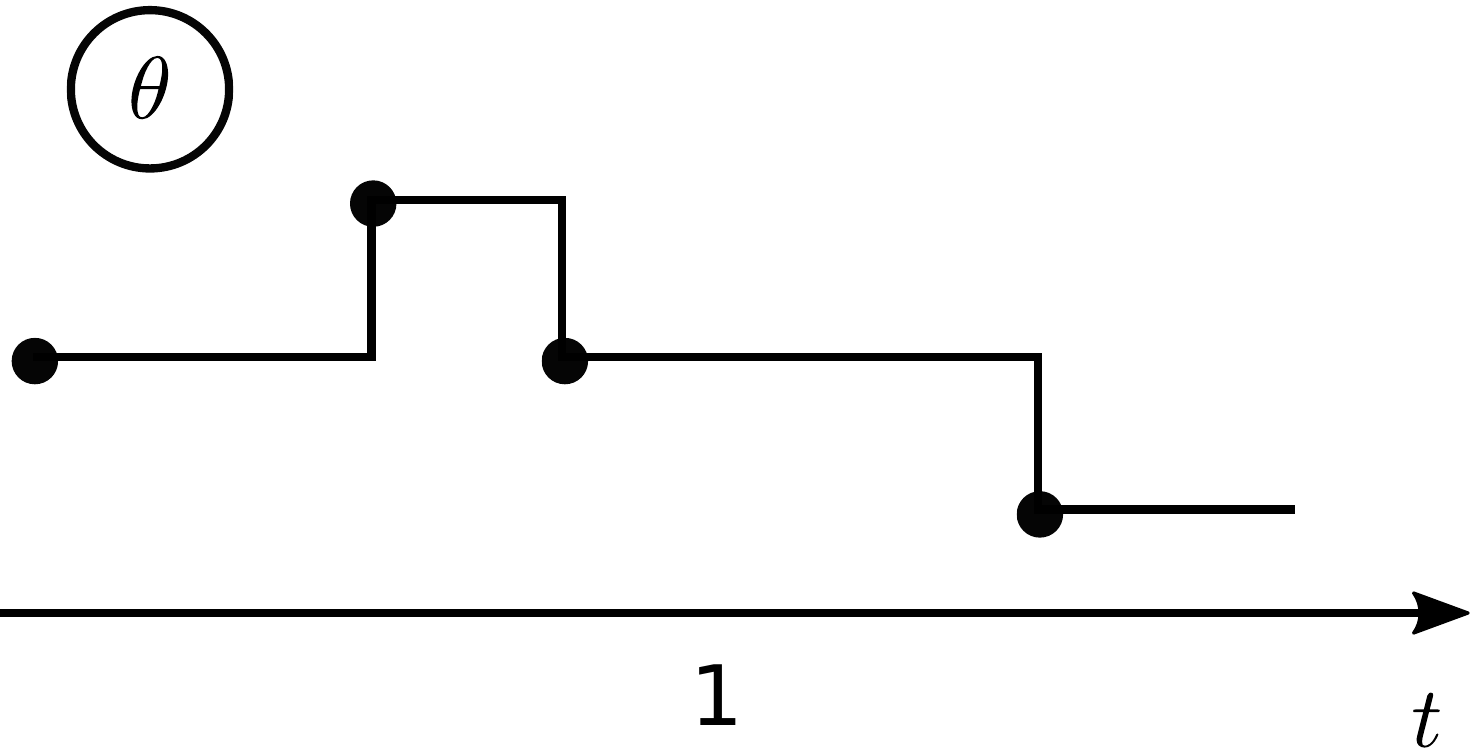}
      \end{minipage}
  \begin{minipage}[!hp]{0.45\linewidth}
  \centering
    \includegraphics [width=0.70\textwidth, angle=0]{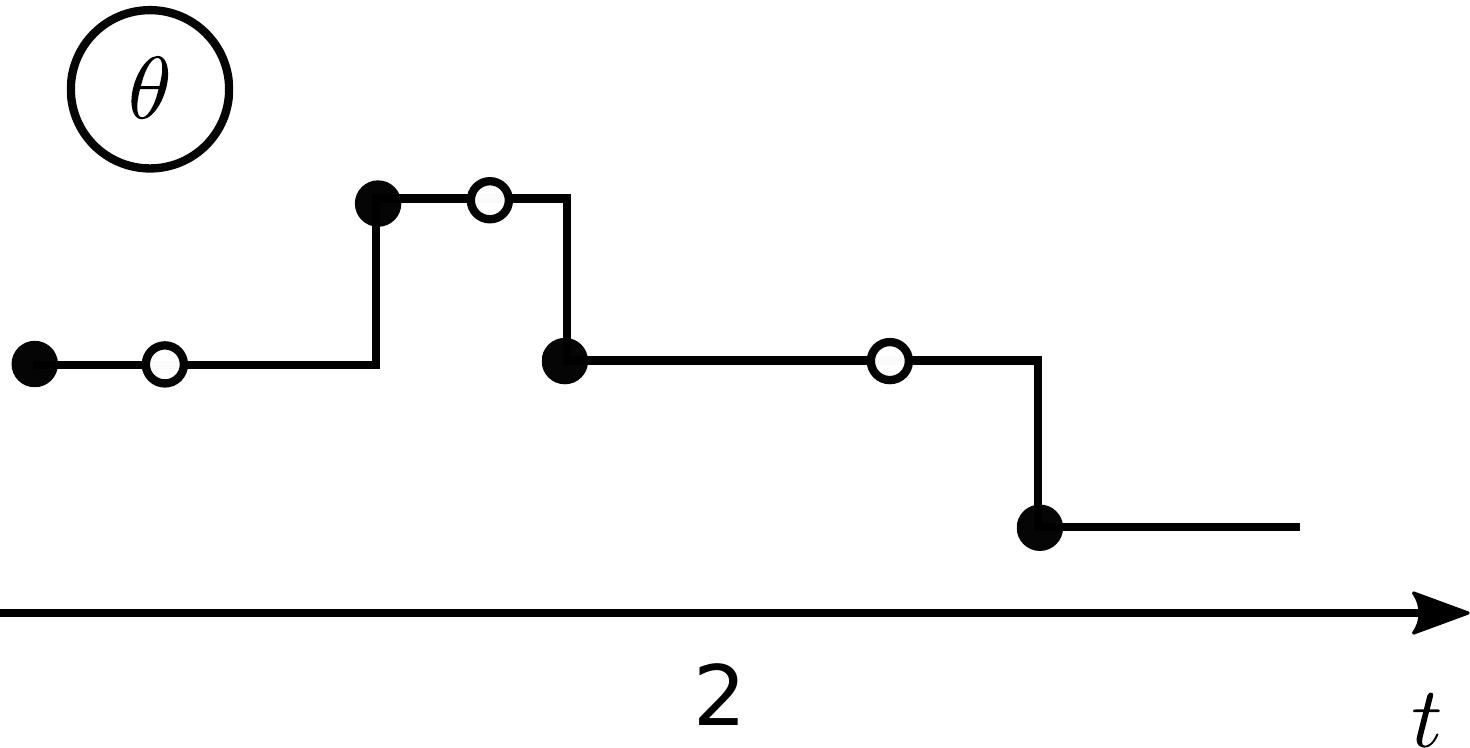}
    \vspace{-0 in}
  \end{minipage}
  \begin{minipage}[!hp]{0.45\linewidth}
  \centering
    \includegraphics [width=0.70\textwidth, angle=0]{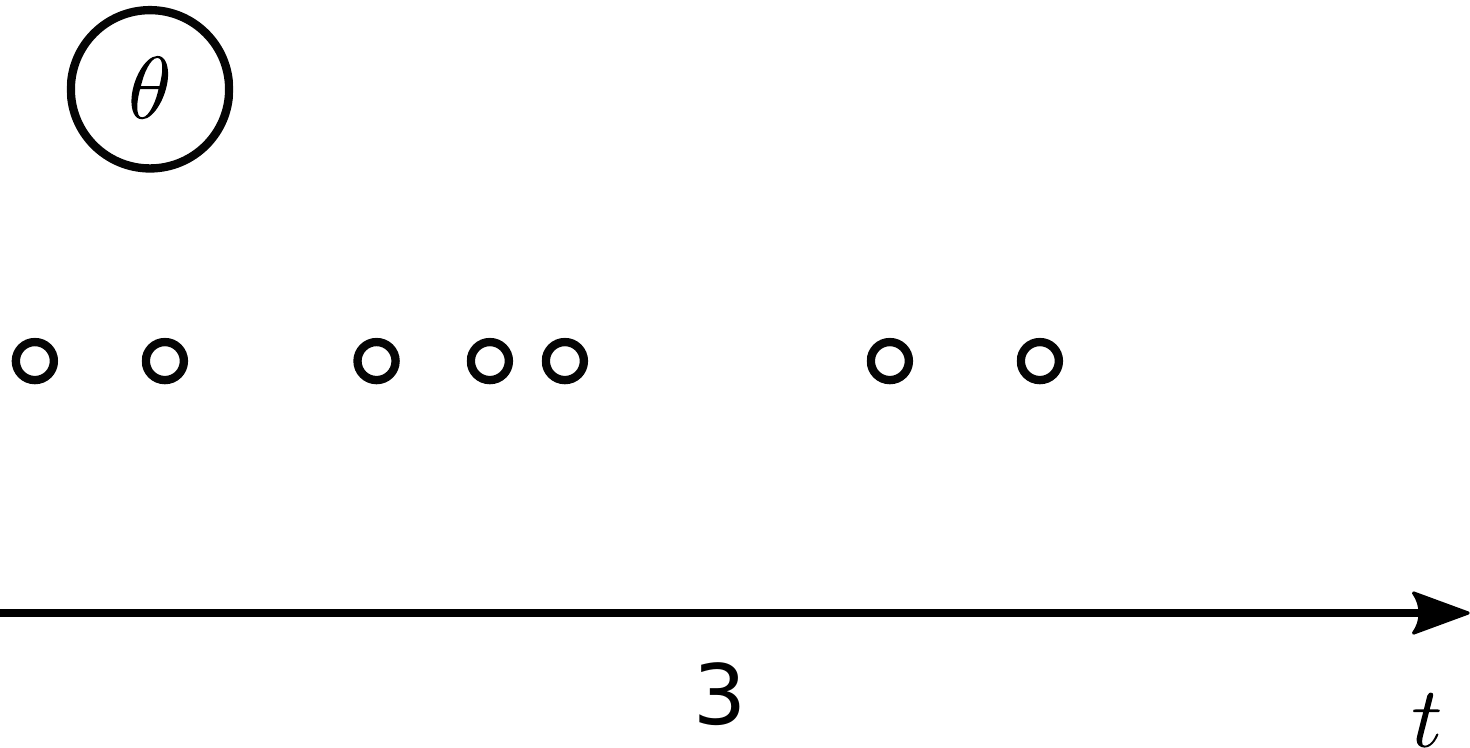}
    \vspace{-0 in}
  \end{minipage}
  \begin{minipage}[!hp]{0.45\linewidth}
  \centering
    \includegraphics [width=0.70\textwidth, angle=0]{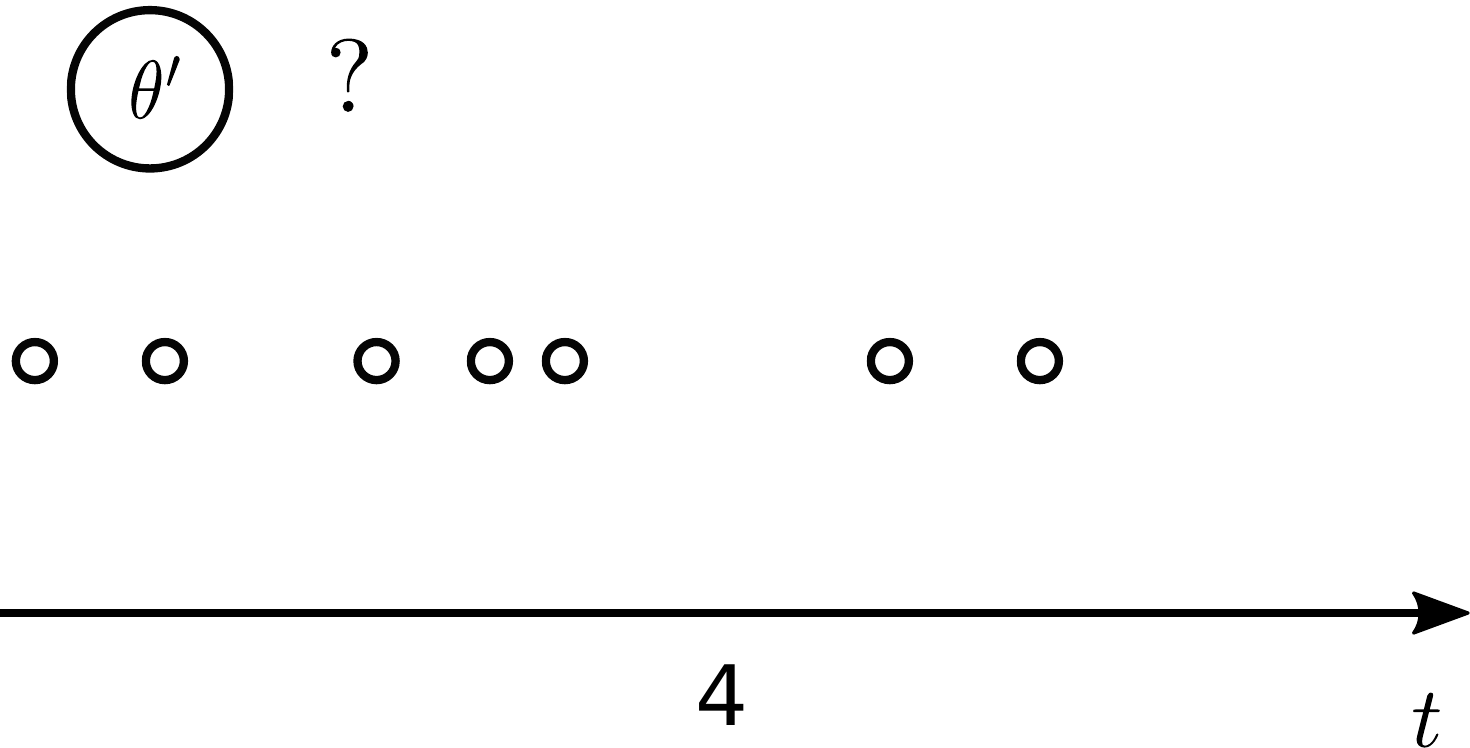}
    \vspace{-0 in}
  \end{minipage}
  \begin{minipage}[!hp]{0.45\linewidth}
  \centering
    \includegraphics [width=0.70\textwidth, angle=0]{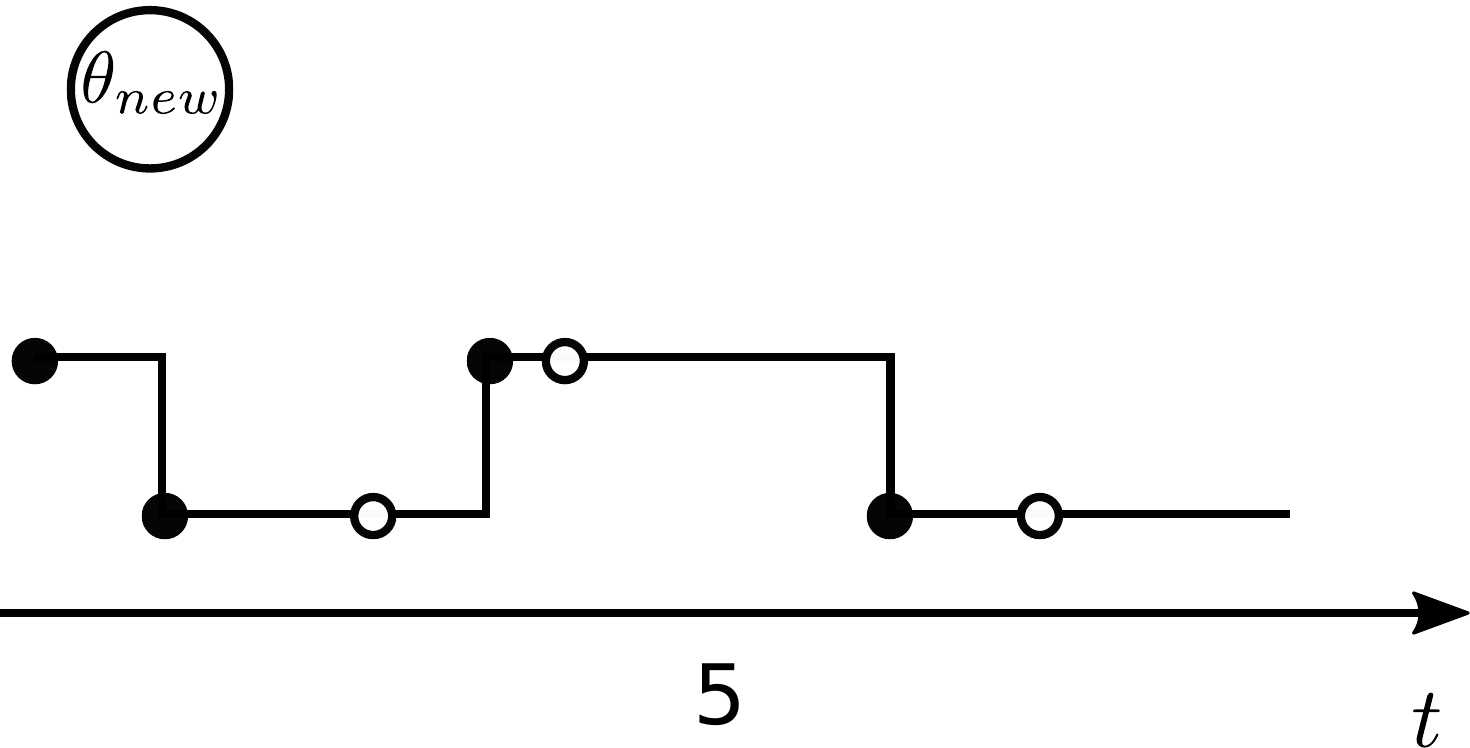}
    \vspace{-0 in}
  \end{minipage}
  \begin{minipage}[!hp]{0.45\linewidth}
  \centering
    \includegraphics [width=0.70\textwidth, angle=0]{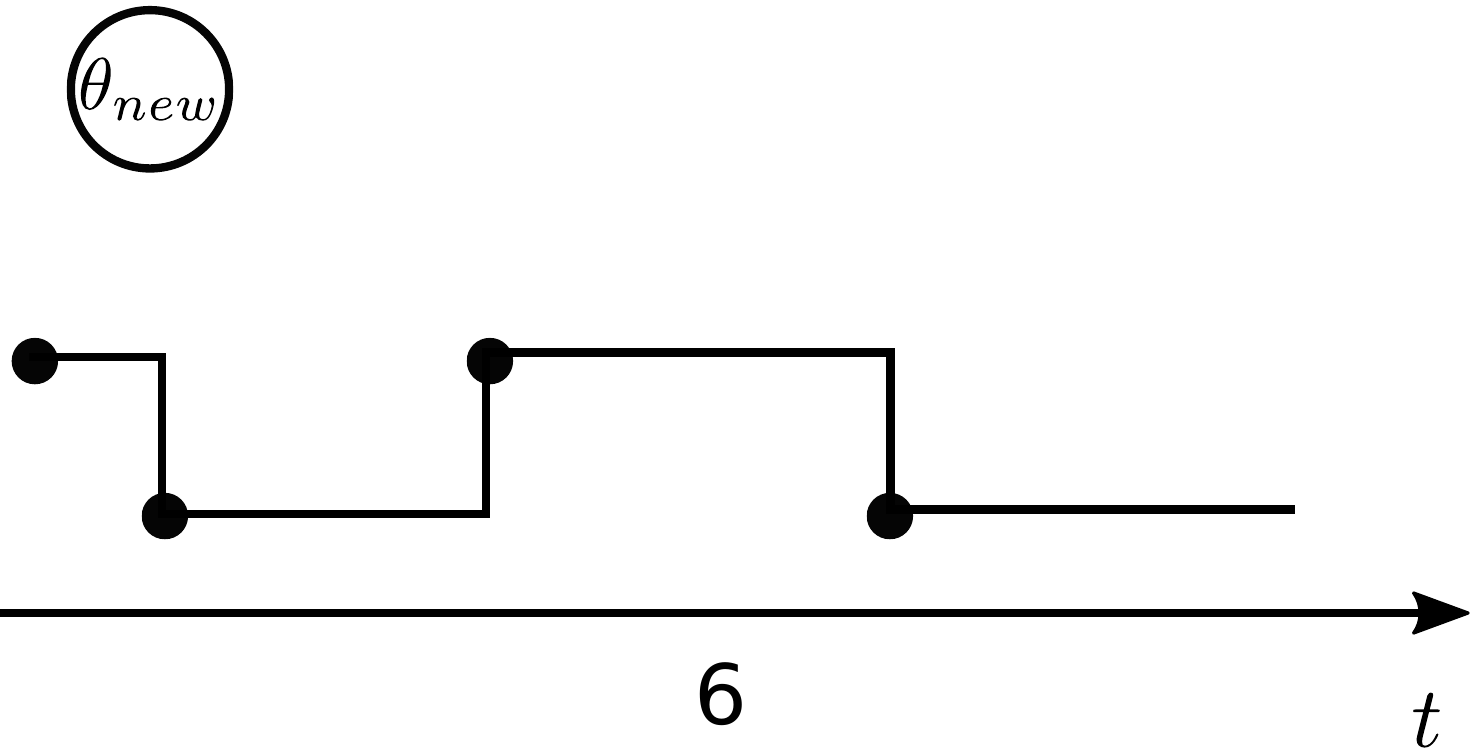}
    \vspace{-0 in}
  \end{minipage}
  \caption{\Naive\ MH-algorithm: Step 1 to 2: sample thinned events
  and discard state information to get a random grid. Step 3: 
propose a new parameter $\theta'$, and accept or reject by making
a forward pass on the grid. Steps 4 to 5: make a backward pass using
the accepted parameter and discard self-transitions to produce a new
trajectory.}
   \label{fig:naive_mh}
  \end{figure}

\subsection{Additional results}
In the following, we evaluate Python implementations of a number of algorithms, focusing our contribution, the symmetrized MH algorithm (algorithm~\ref{alg:MH_improved}), and as well as the \naive\ MH algorithm (algorithm~\ref{alg:MH_naive}).
We evaluate different variants of these algorithms, corresponding to different uniformizing Poisson rates. 
For \naive\ MH, we set $\Omega(\theta) = \kappa \max_s A_s(\theta) $ with $\kappa$  equal to $1.5, 2$ and $3$ (here $\kappa$ must be greater than $1$), 
while for symmetrized MH, where the uniformizing rate depends on both the current and proposed parameters, we consider the settings:
 $\Omega(\theta, \vartheta) = \kappa (\max A(\theta) + \max A(\vartheta))$ 
 ($\kappa = 1$ and $1.5$), and 
$\Omega(\theta, \vartheta) = 1.5 \max(\max A(\theta), \max A(\vartheta))$.
We evaluate two other baselines: Gibbs sampling (Algorithm~\ref{alg:MJP_gibbs}), 
and particle MCMC~\citep[][see also section~\ref{sec:pmcmc} in the appendix]{Andrieu10}. 
Gibbs sampling involves a uniformization step to update the MJP trajectory (step 1 in algorithm~\ref{alg:MJP_gibbs}), for which we use $\Omega(\theta,\vartheta) = \kappa \max_s A_s(\theta)$ for $\kappa=1.5,2,3$. 
Unless specified, our results were obtained from $100$ independent MCMC runs, each of $10000$ iterations.
We found particle MCMC to be more computationally intensive, and limited each run to $3000$ iterations, the number of particles being $5, 10$ and $20$.
For each run of each MCMC algorithm, we calculated the effective sample size (ESS) of the posterior samples of the MJP parameters using the R package \texttt{rcoda}~\citep{Rcoda2006}. 
This estimates the number of independent samples returned by the MCMC algorithm, and dividing this by the runtime of a simulation gives the ESS per unit time (ESS/sec). 
We used this to compare different samplers and different parameter settings. In the following we present the additional results.

  \begin{figure}[H]
  \begin{minipage}[h!]{0.99\linewidth}
    \includegraphics [width=0.24\textwidth, angle=0]{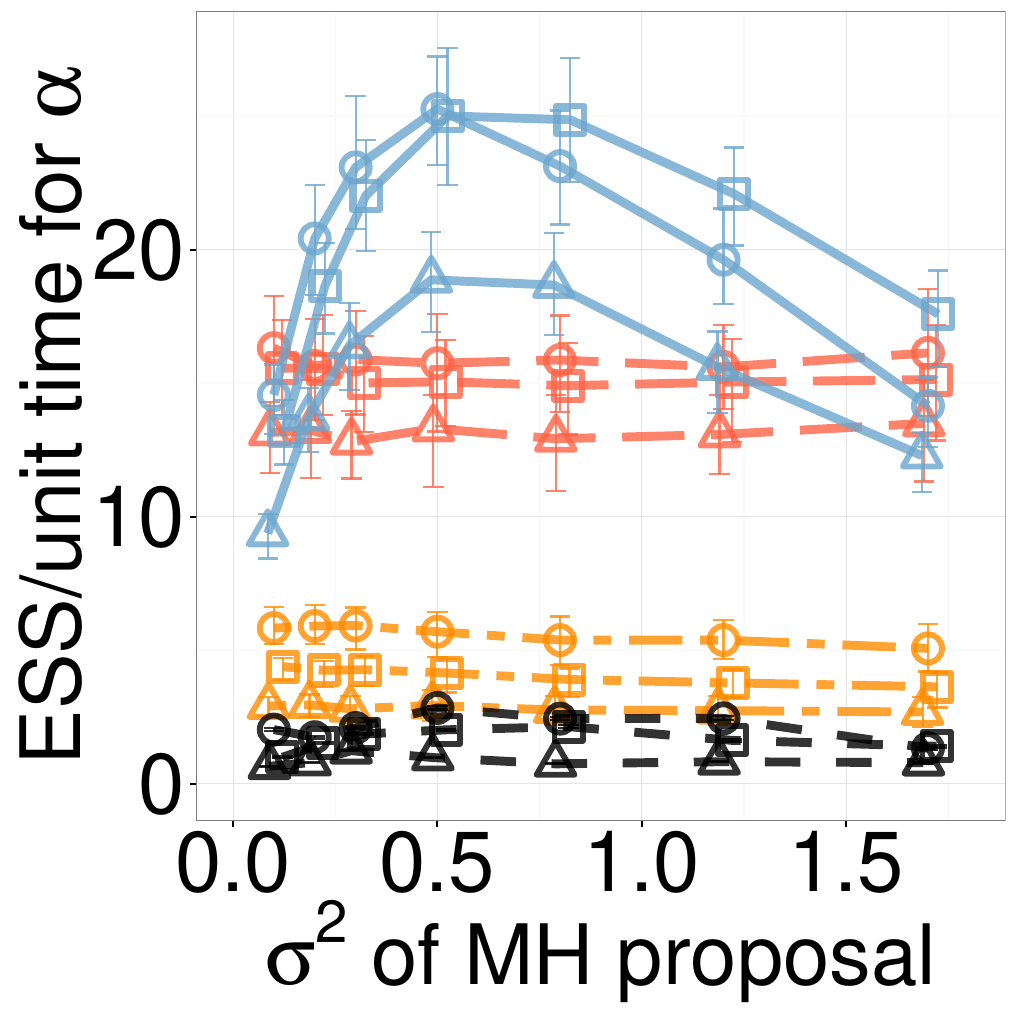}
	\hspace{.6in}
    \includegraphics [width=0.24\textwidth, angle=0]{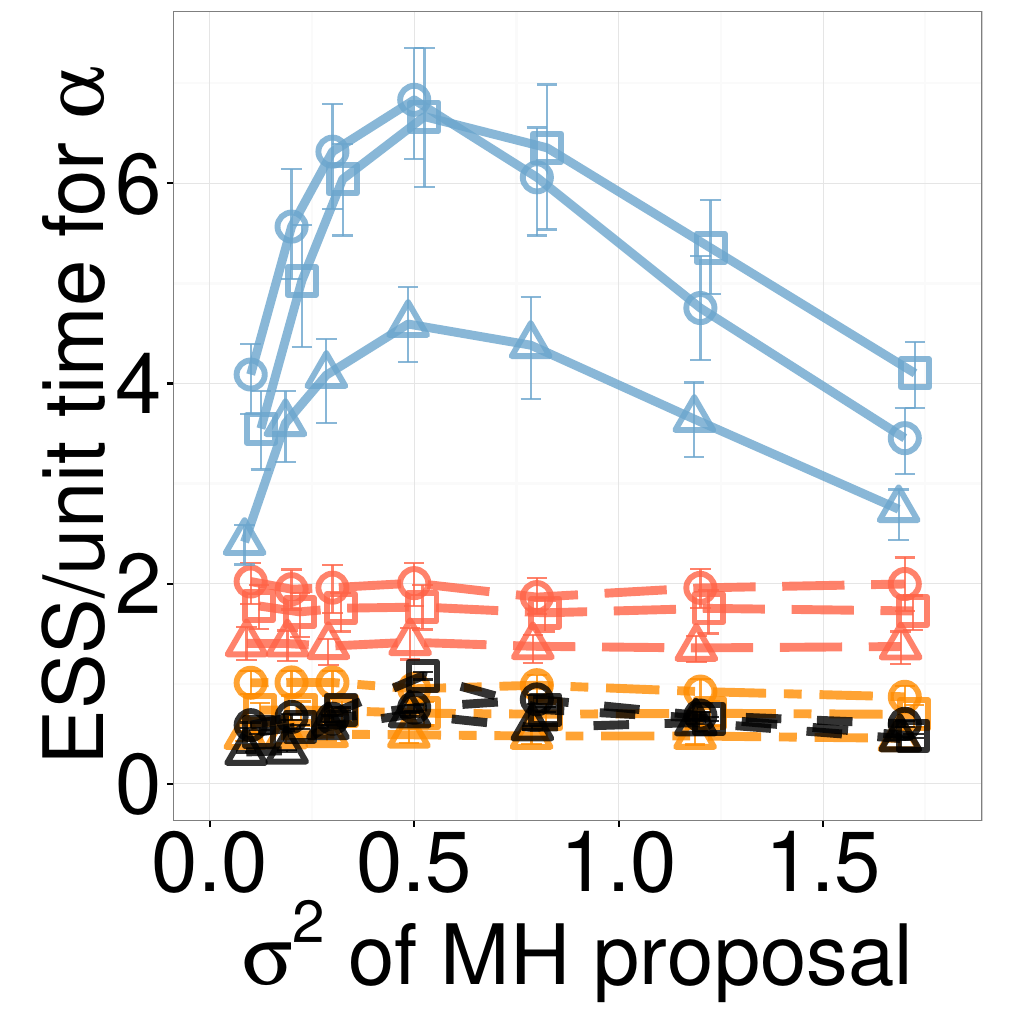}
	\hspace{.6in}
    \includegraphics [width=0.24\textwidth, angle=0]{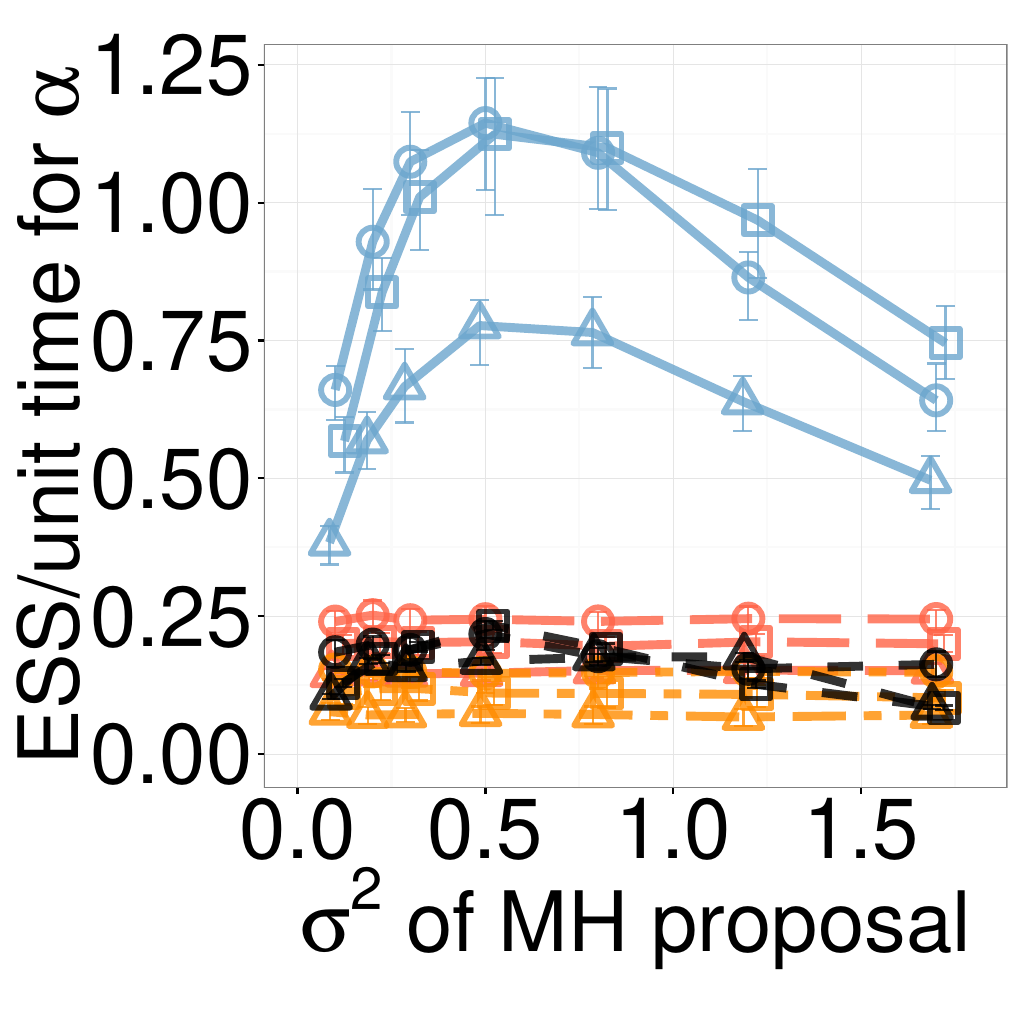}
  \end{minipage}

  \begin{minipage}[h!]{0.99\linewidth}
    \includegraphics [width=0.24\textwidth, angle=0]{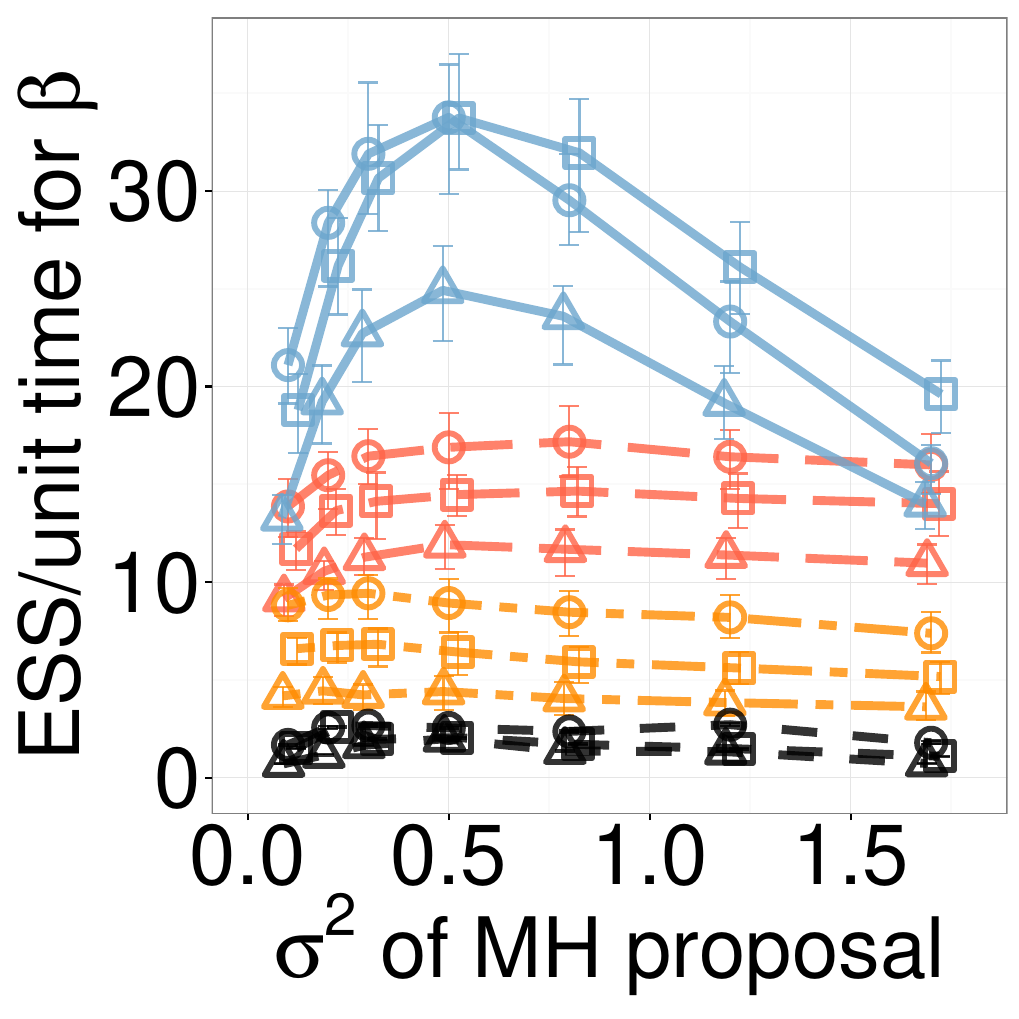}
	\hspace{.6in}
    \includegraphics [width=0.24\textwidth, angle=0]{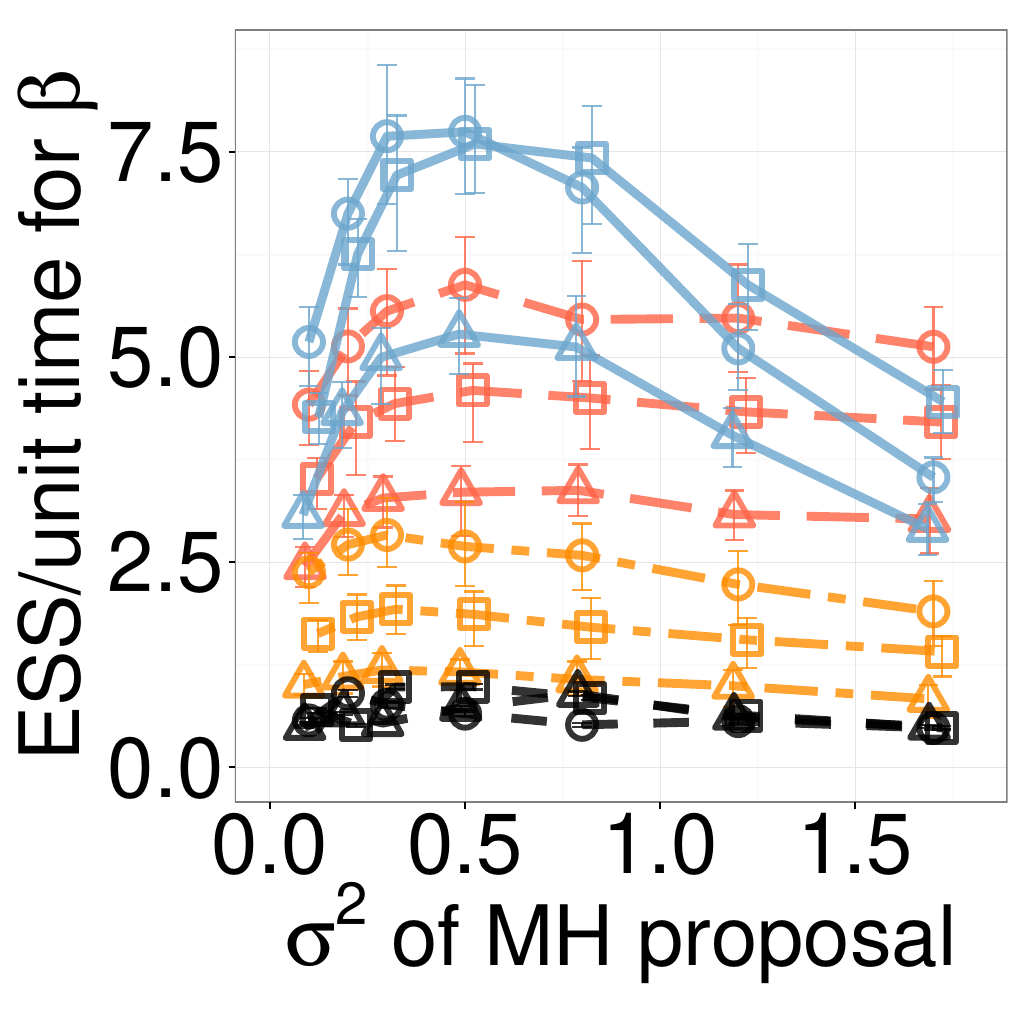}
	\hspace{.6in}
    \includegraphics [width=0.24\textwidth, angle=0]{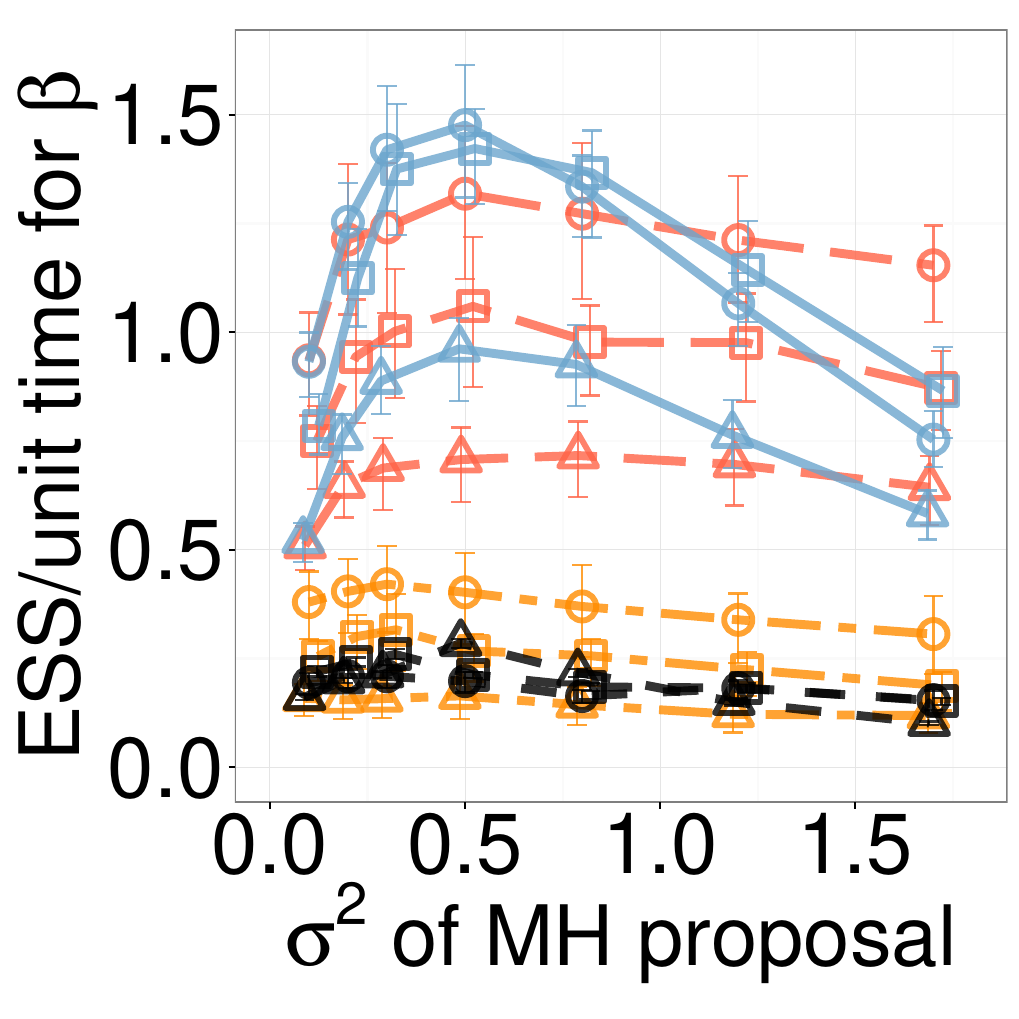}
  \end{minipage}
    \caption{ESS/sec for the synthetic model, the top three are for $\alpha$ for 3 states, 5 states, and 10 states. The bottom three are for $\beta$ for 3 states, 5 states, and 10 states. Blue, yellow, red and black are the symmetrized MH, \naive\ MH, Gibbs and particle MCMC algorithm.  Squares, circles and trianges correspond to $\Omega(\theta,\vartheta)$ set to $(\max_s A_s(\theta) + \max_s A_s(\vartheta))$, $\max(\max_s A_s(\theta), \max_s A_s(\vartheta))$ and  $1.5(\max_s A_s(\theta) + \max_s A_s(\vartheta))$. And for PMCMC, they correspond to 10 particles, 5 particles and 15 particles.}
     \label{fig:ESS_EXP}
  \end{figure}

  \begin{figure}[H]
  \centering
  \begin{minipage}[!hp]{0.99\linewidth}
    \includegraphics [width=0.24\textwidth, angle=0]{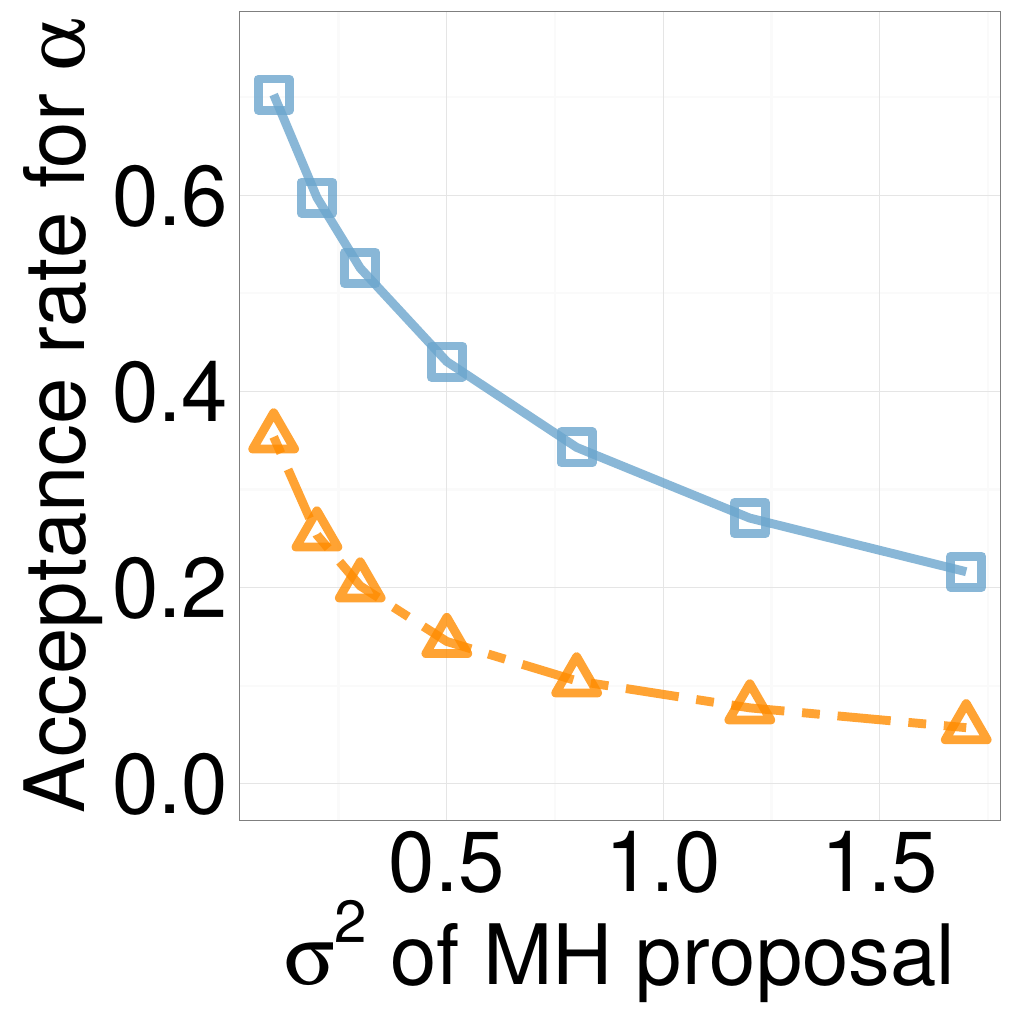}
    \includegraphics [width=0.24\textwidth, angle=0]{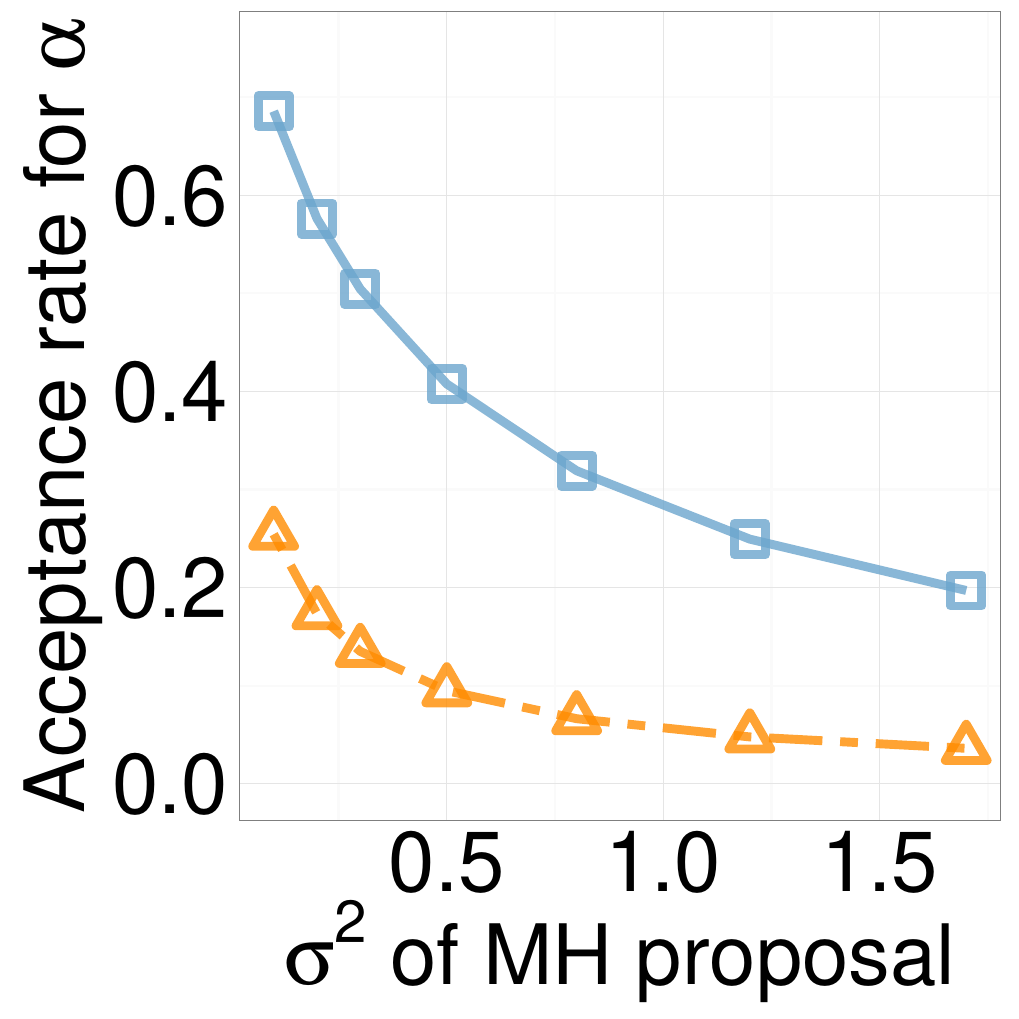}
    \includegraphics [width=0.24\textwidth, angle=0]{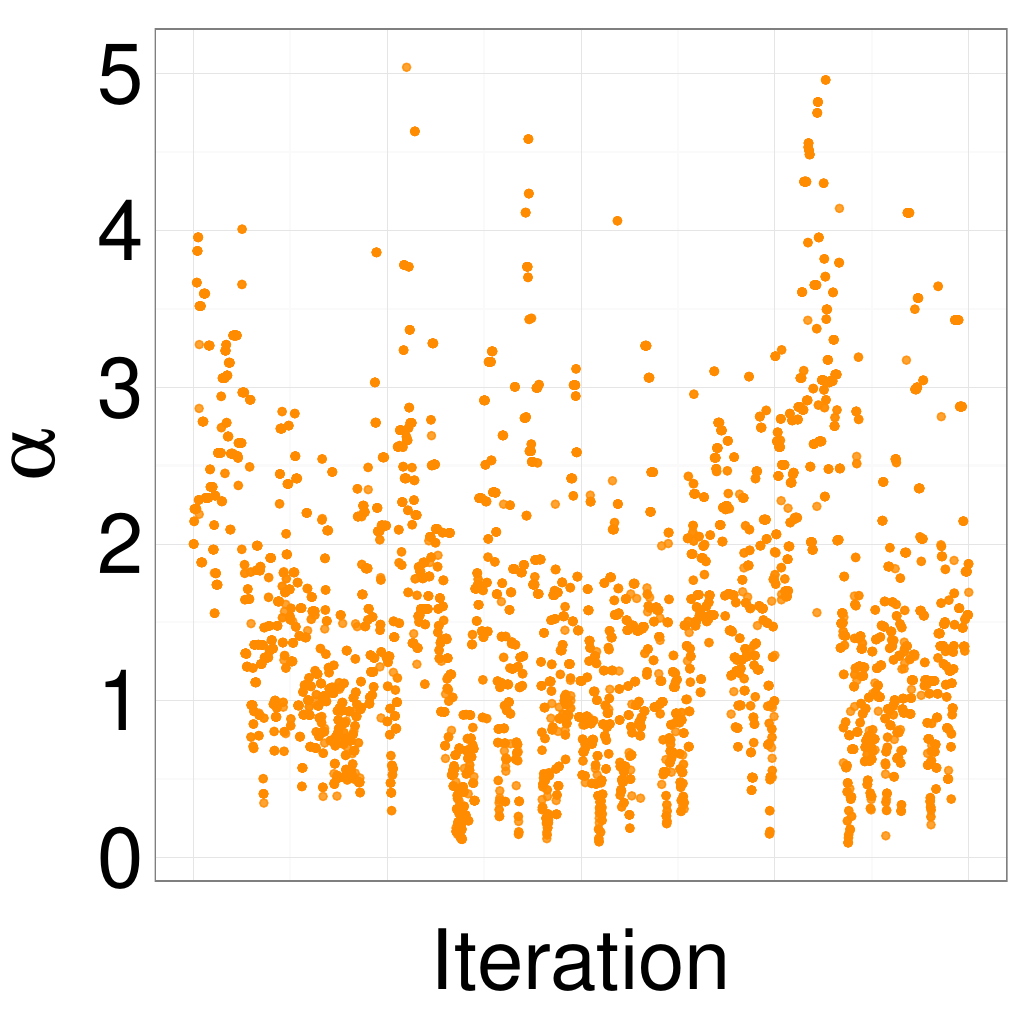}
    \includegraphics [width=0.24\textwidth, angle=0]{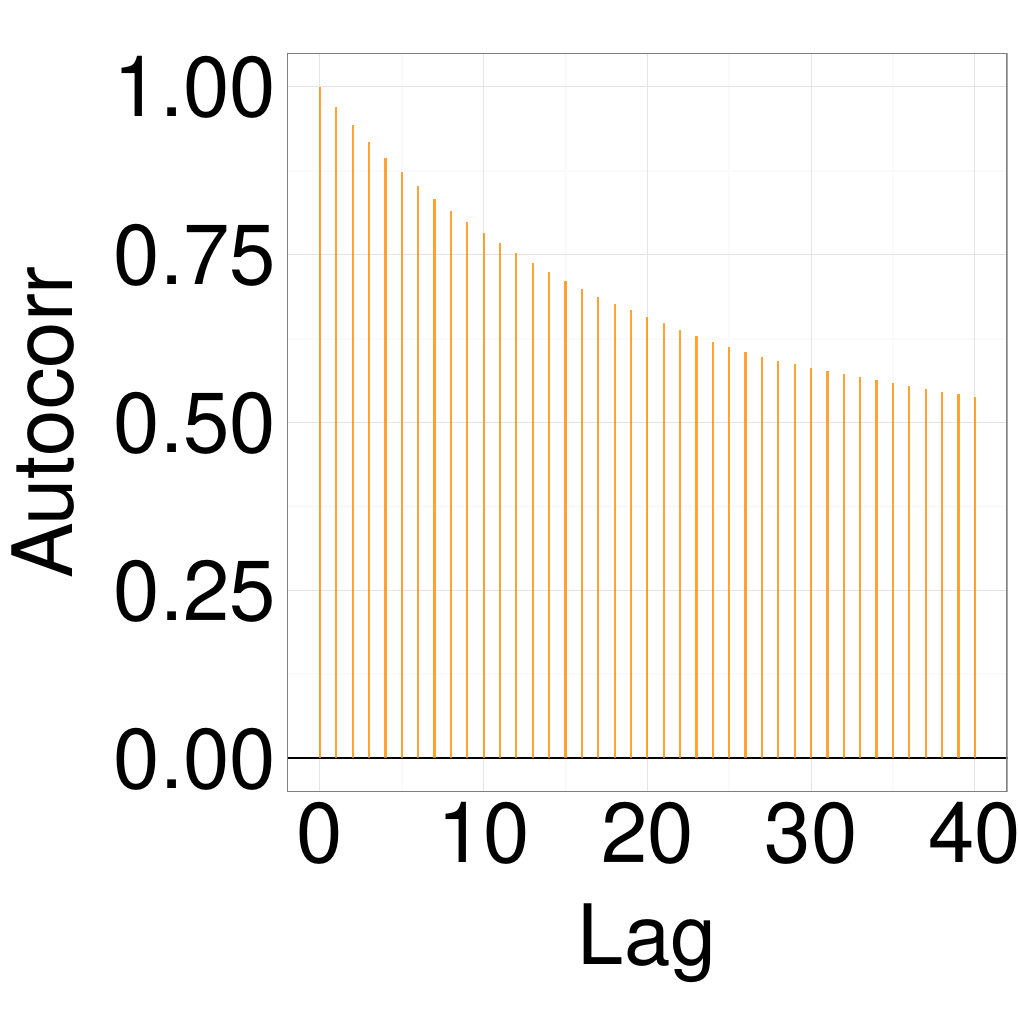}
  \end{minipage}
    \caption{Acceptance Rate for $\alpha$ in the synthetic model (Left two), the first being dimension 3, and the second,dimension 5. Blue square and yellow triangle curves are the symmetrized MH, and \naive\ MH  algorithm. The multiplicative factor is $2$. Trace and autocorrelation plots for \naive\ MH  (right two panels) for the synthetic model with $3$ states.}
     \label{fig:ACC_EXP}
  \end{figure}


  
  \begin{figure}[H]
  \centering
  \begin{minipage}[!hp]{0.99\linewidth}
    \includegraphics [width=0.24\textwidth, angle=0]{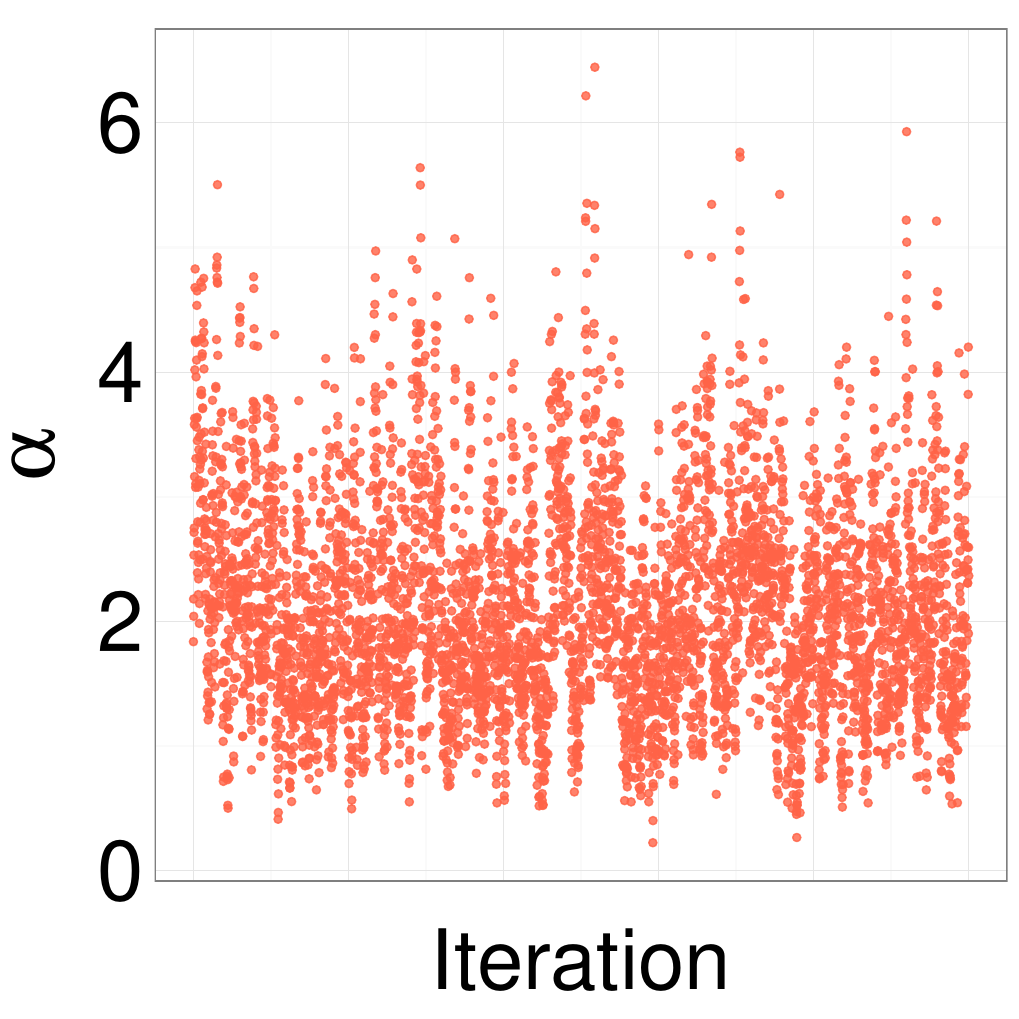}
    \includegraphics [width=0.24\textwidth, angle=0]{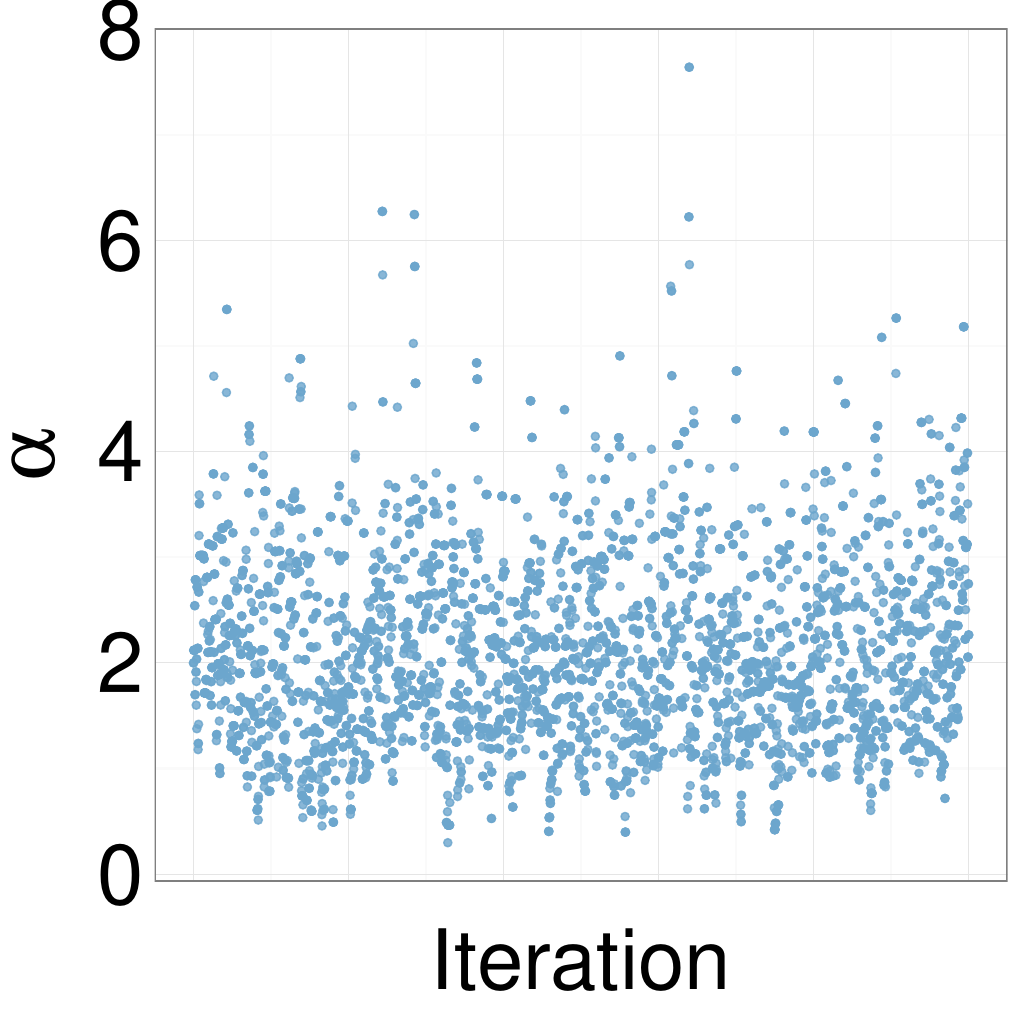}
    \includegraphics [width=0.24\textwidth, angle=0]{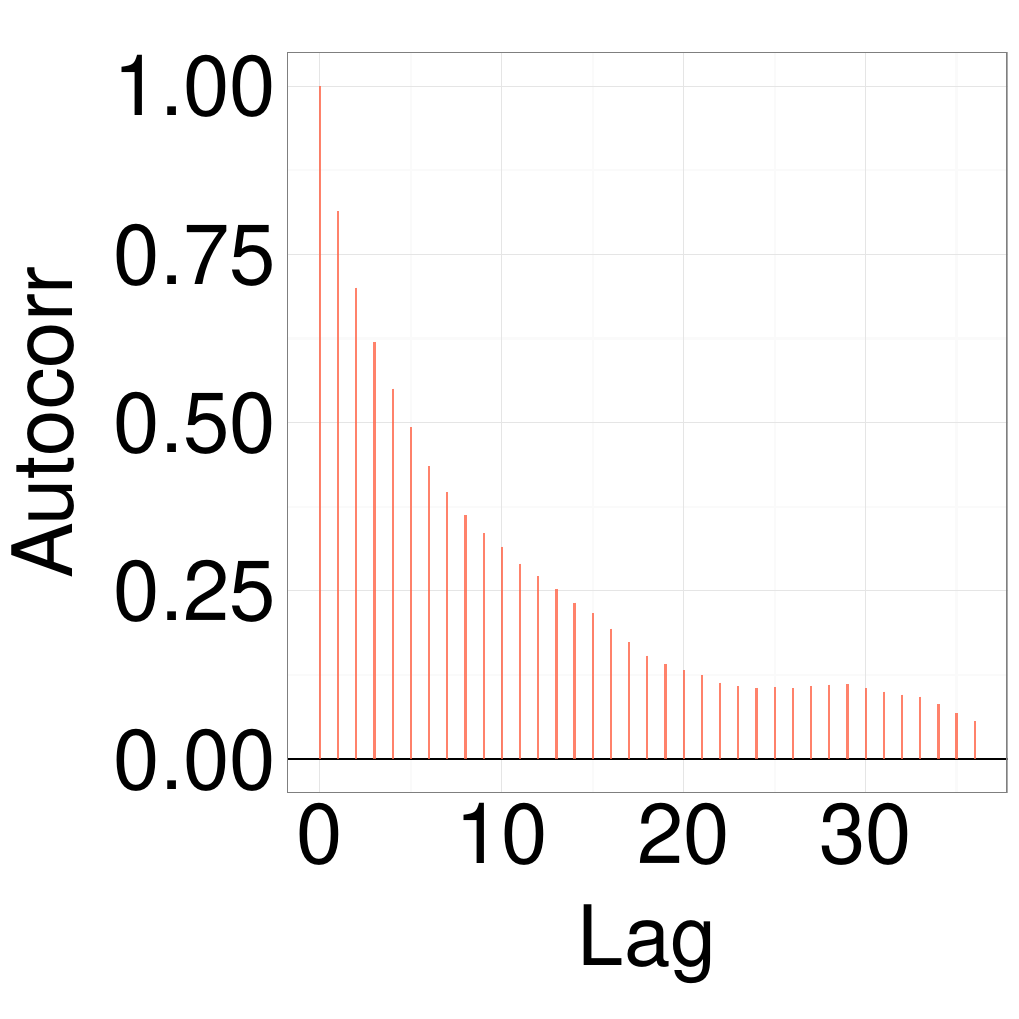}
    \includegraphics [width=0.24\textwidth, angle=0]{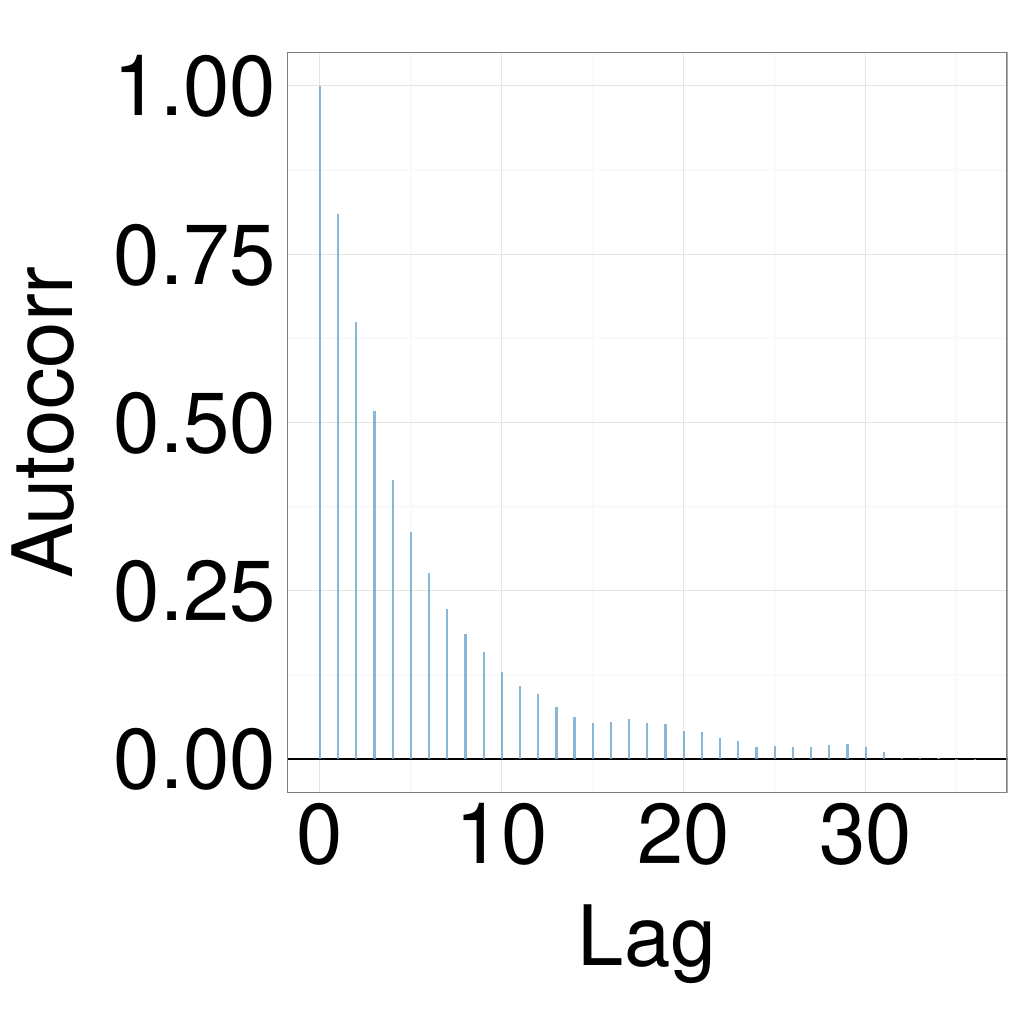}
  \end{minipage}
  
\caption{Trace and autocorrelation plots for Gibbs (left two panels) and symmetrized MH (right two panels). All plots are for the immigration model with $3$ states.}
     \label{fig:TRACE_Q}
  \end{figure}
  
  \begin{figure}[H]
  \centering
  \begin{minipage}[!hp]{0.99\linewidth}
    \includegraphics [width=0.24\textwidth, angle=0]{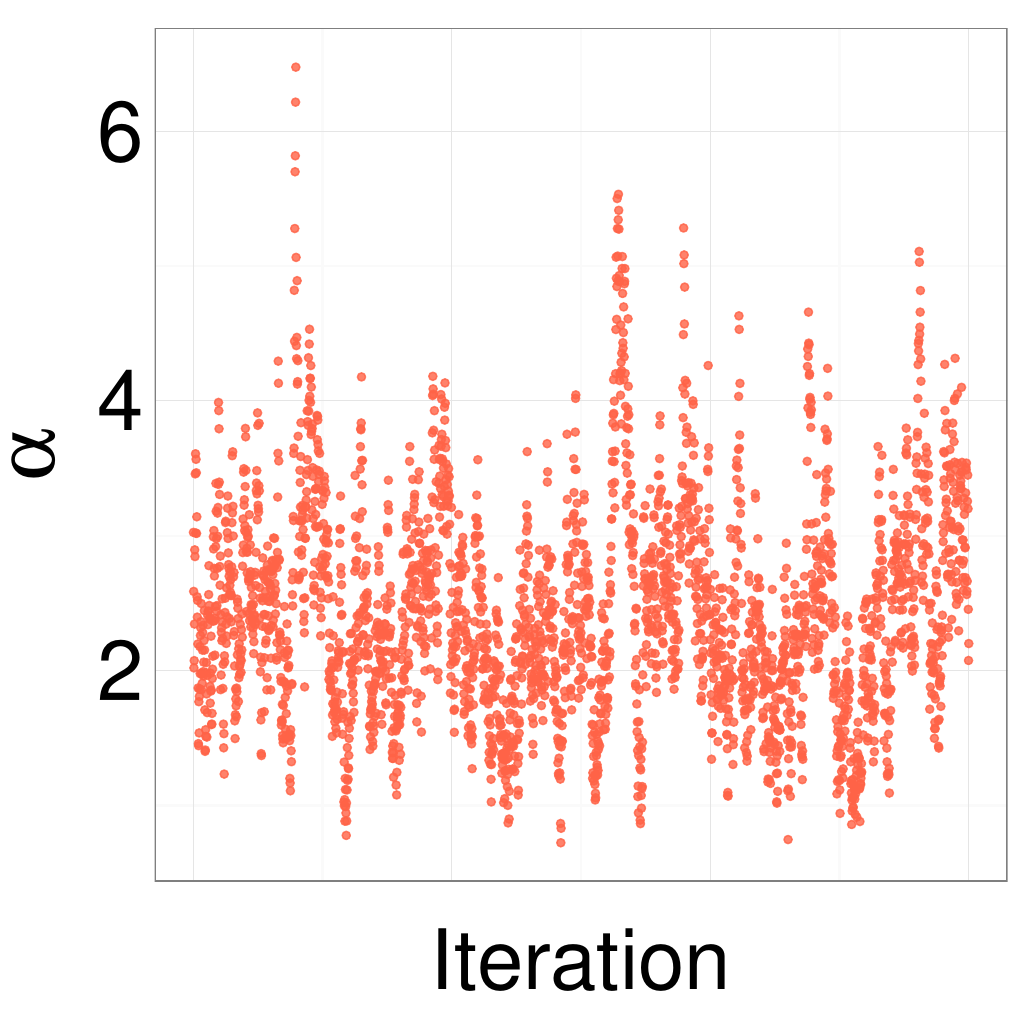}
    \includegraphics [width=0.24\textwidth, angle=0]{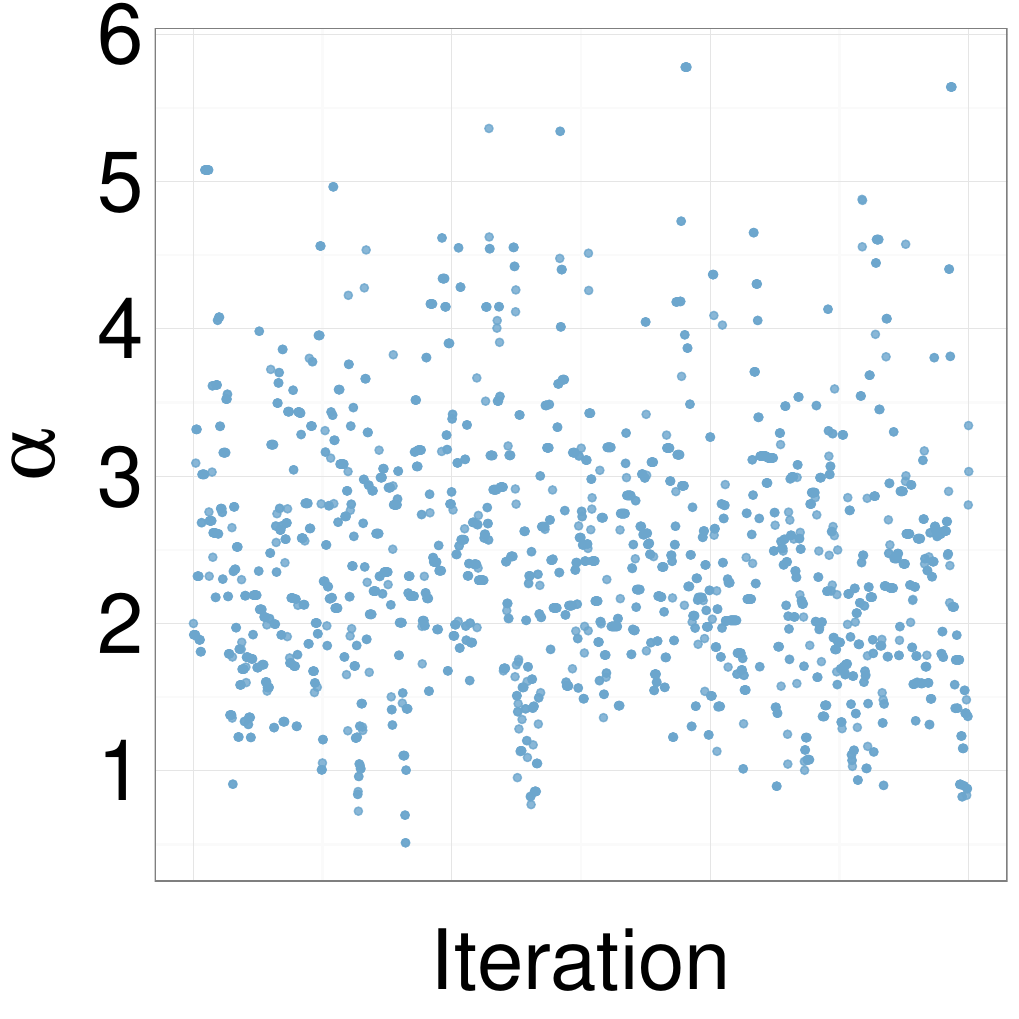}
    \includegraphics [width=0.24\textwidth, angle=0]{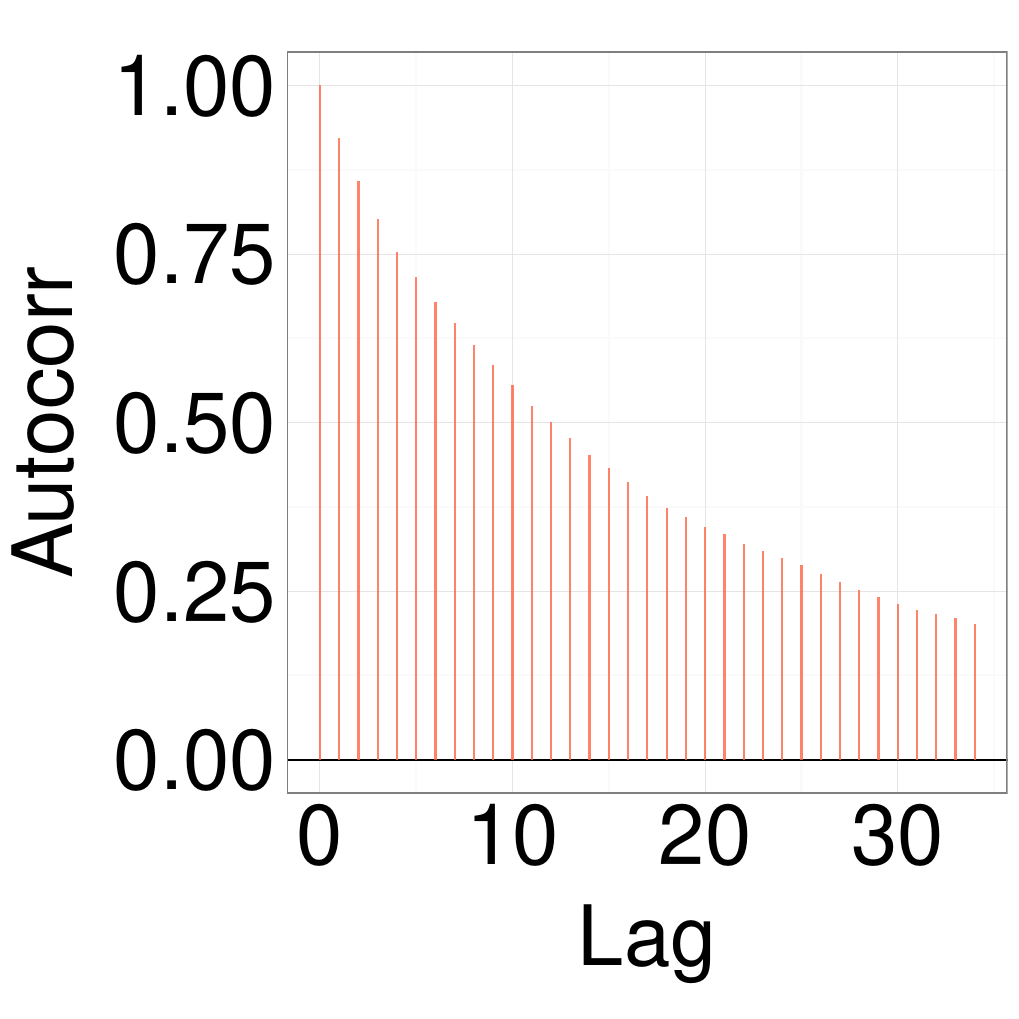}
    \includegraphics [width=0.24\textwidth, angle=0]{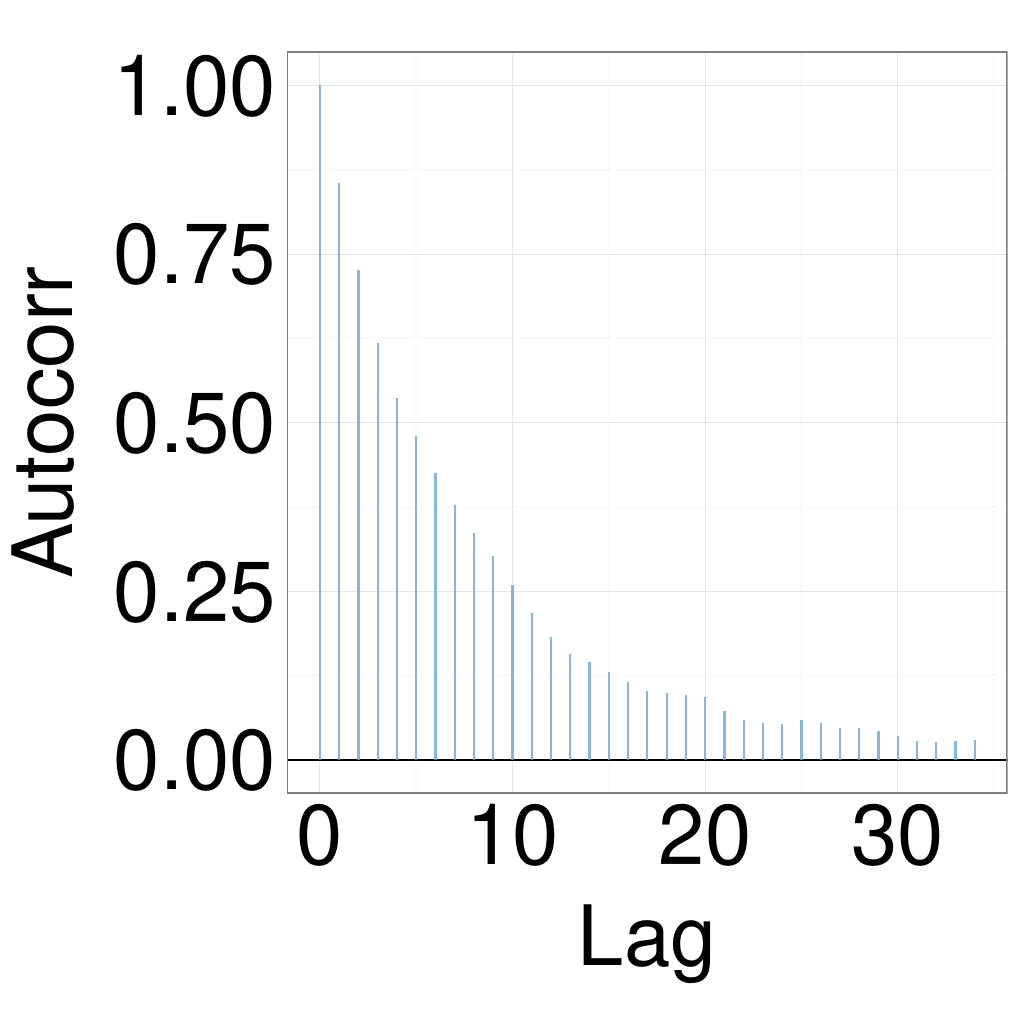}
  \end{minipage}

\caption{Trace and autocorrelation plots for Gibbs (left two panels) and symmetrized MH (right two panels). All plots are for the time-inhomogeneous immigration model with $10$ states.}
 
     \label{fig:TRACE_CQ}
  \end{figure}

  \begin{figure}[H]
  \centering
  \begin{minipage}[h!]{0.99\linewidth}
    \includegraphics [width=0.24\textwidth, angle=0]{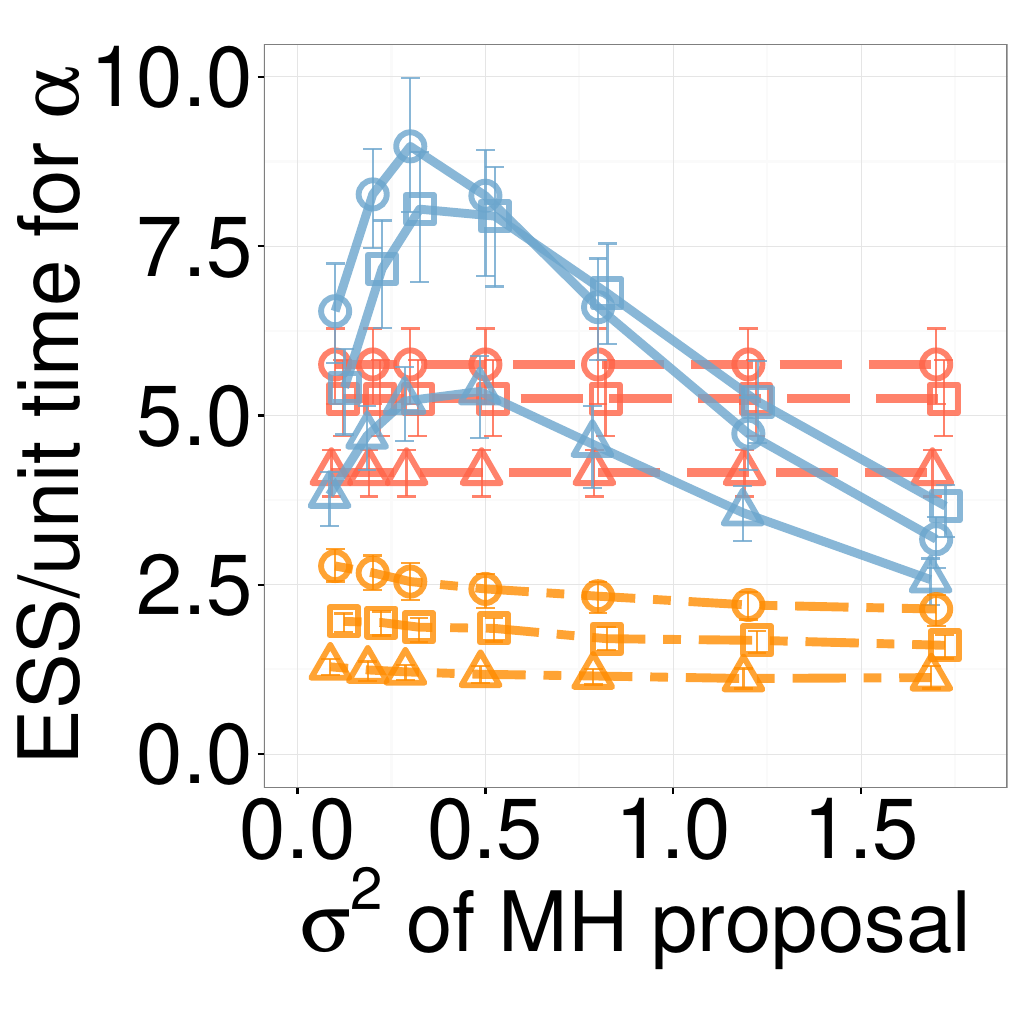}
	\hspace{.6in}
    \includegraphics [width=0.24\textwidth, angle=0]{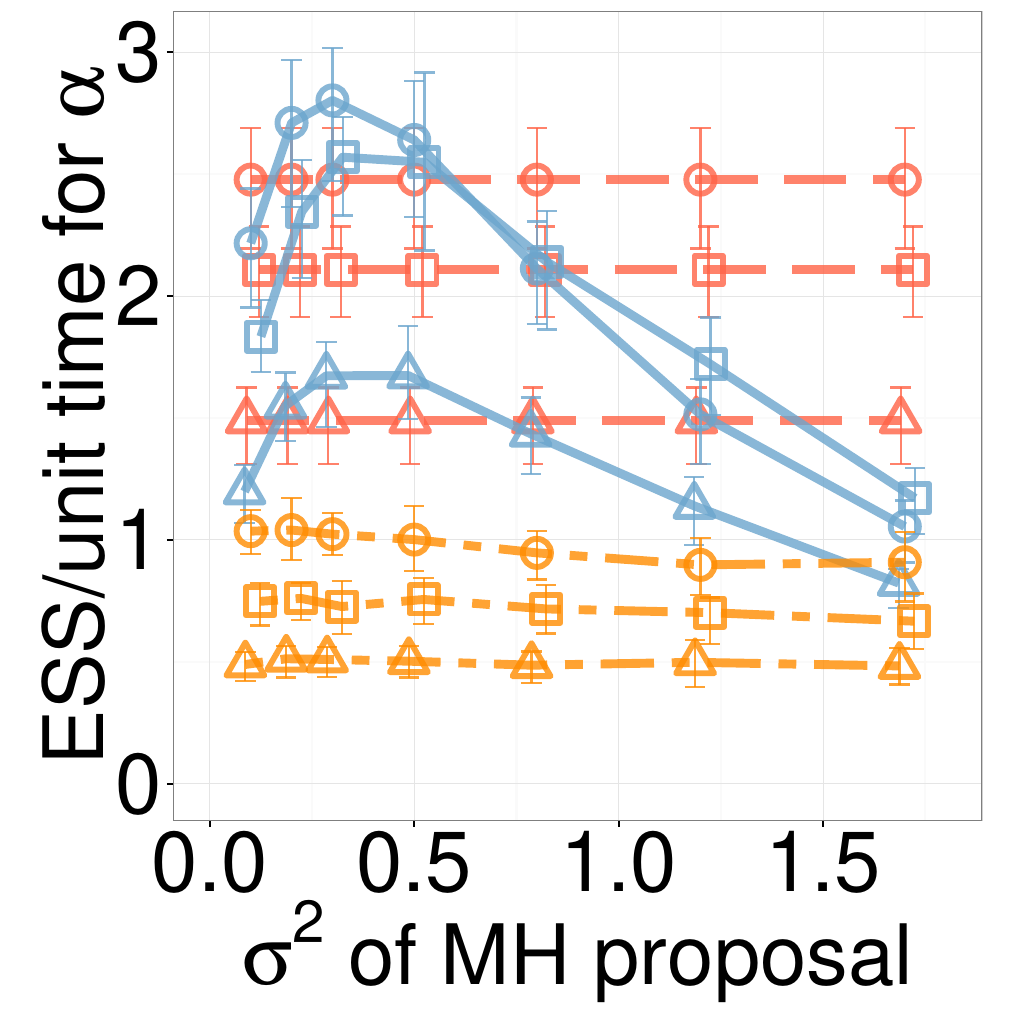}
	\hspace{.6in}
    \includegraphics [width=0.24\textwidth, angle=0]{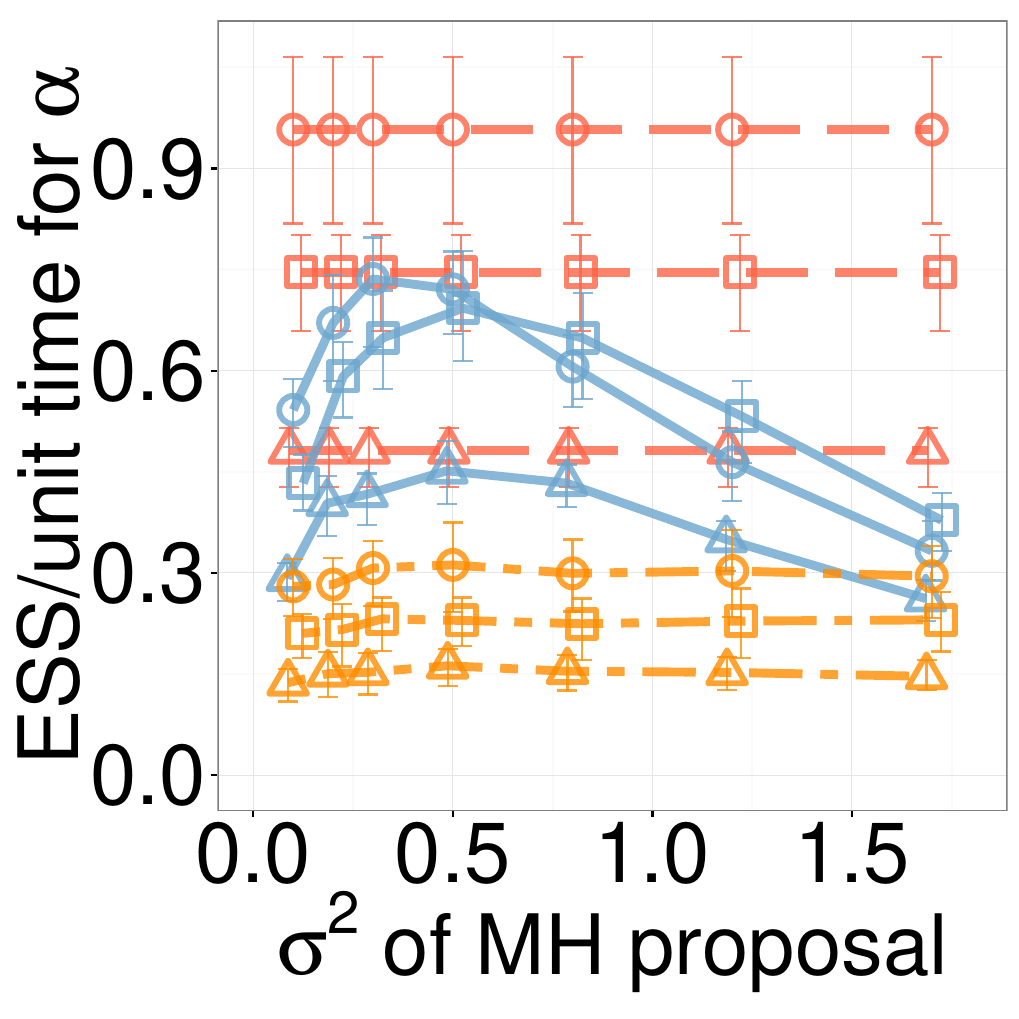}
  \end{minipage}

  \begin{minipage}[h!]{0.99\linewidth}
    \includegraphics [width=0.24\textwidth, angle=0]{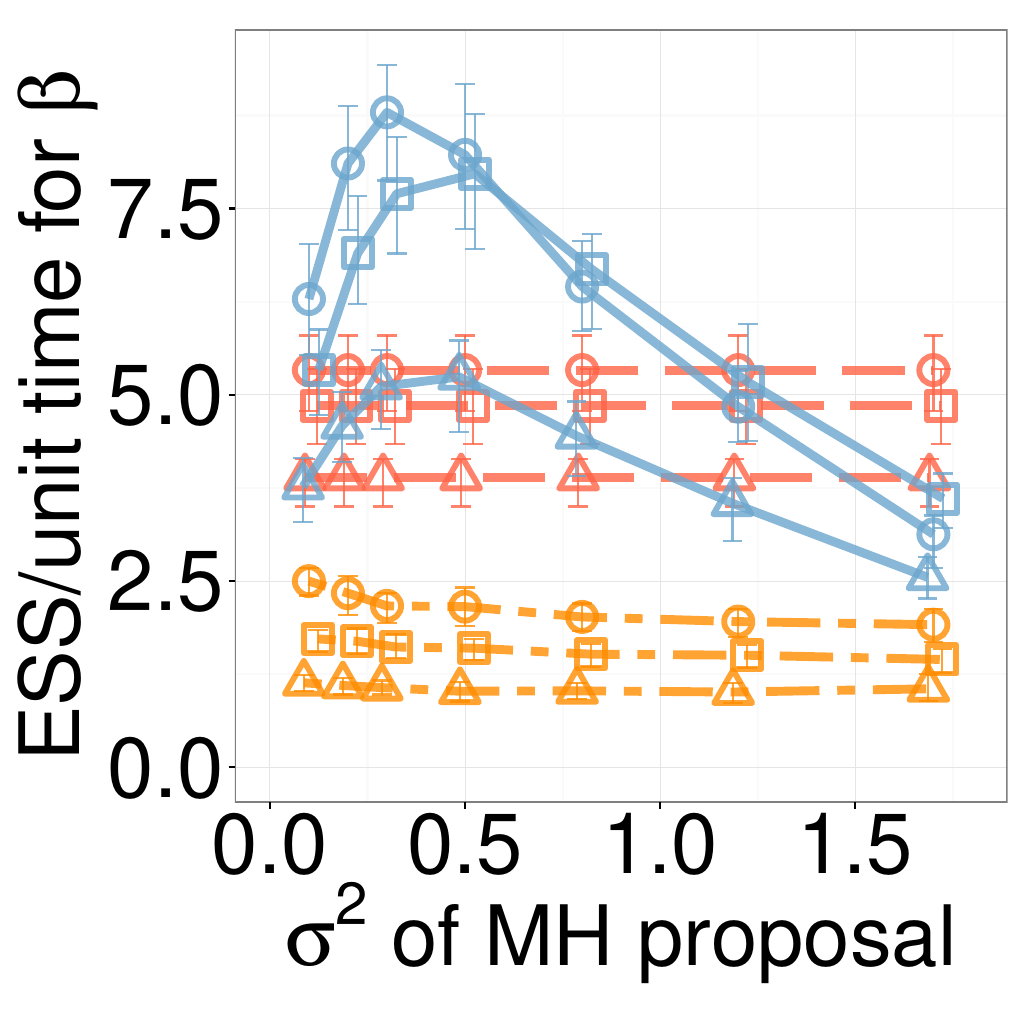}
	\hspace{.6in}
    \includegraphics [width=0.24\textwidth, angle=0]{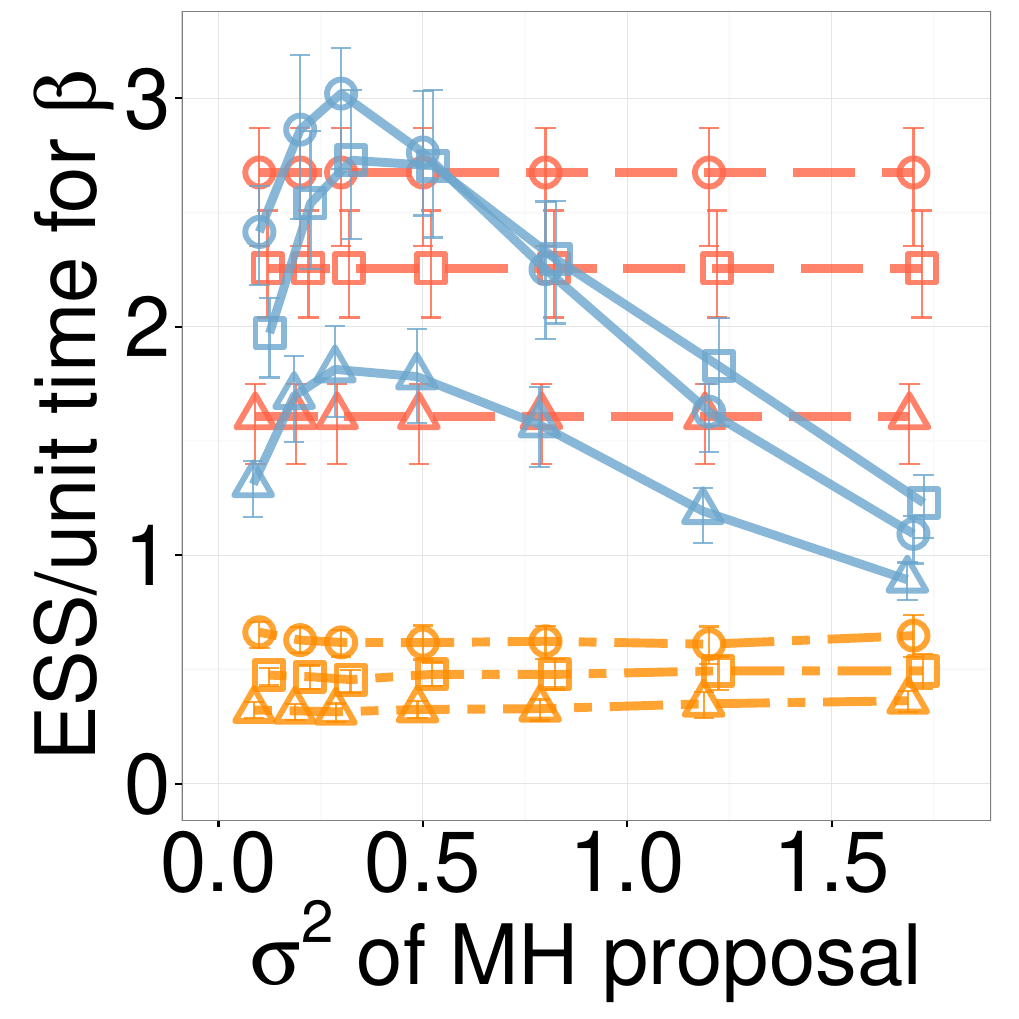}
	\hspace{.6in}
    \includegraphics [width=0.24\textwidth, angle=0]{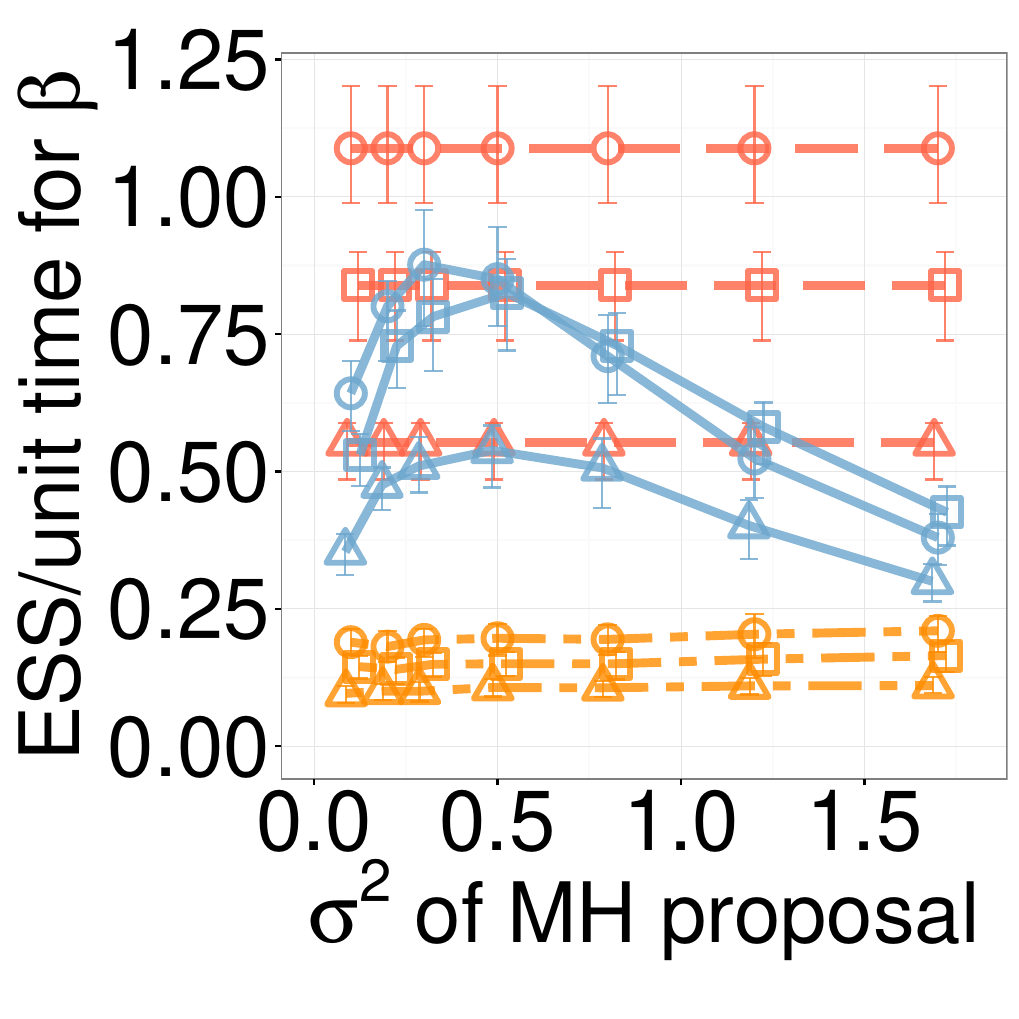}
  \end{minipage}
    \caption{ESS/sec for the immigration model, the top three are for $\alpha$ for 3 states, 5 states, and 10 states. The bottom three are for $\beta$ for 3 states, 5 states, and 10 states. Blue, yellow, and red are the symmetrized MH, \naive\ MH, Gibbs algorithm. Squares, circles and trianges correspond to $\Omega(\theta,\vartheta)$ set to $(\max_s A_s(\theta) + \max_s A_s(\vartheta))$, $\max(\max_s A_s(\theta), \max_s A_s(\vartheta))$ and  $1.5(\max_s A_s(\theta) + \max_s A_s(\vartheta))$.}
     \label{fig:ESS_Q}
  \end{figure}

  \begin{figure}[H]
  \centering
  \begin{minipage}[h!]{0.99\linewidth}
    \includegraphics [width=0.24\textwidth, angle=0]{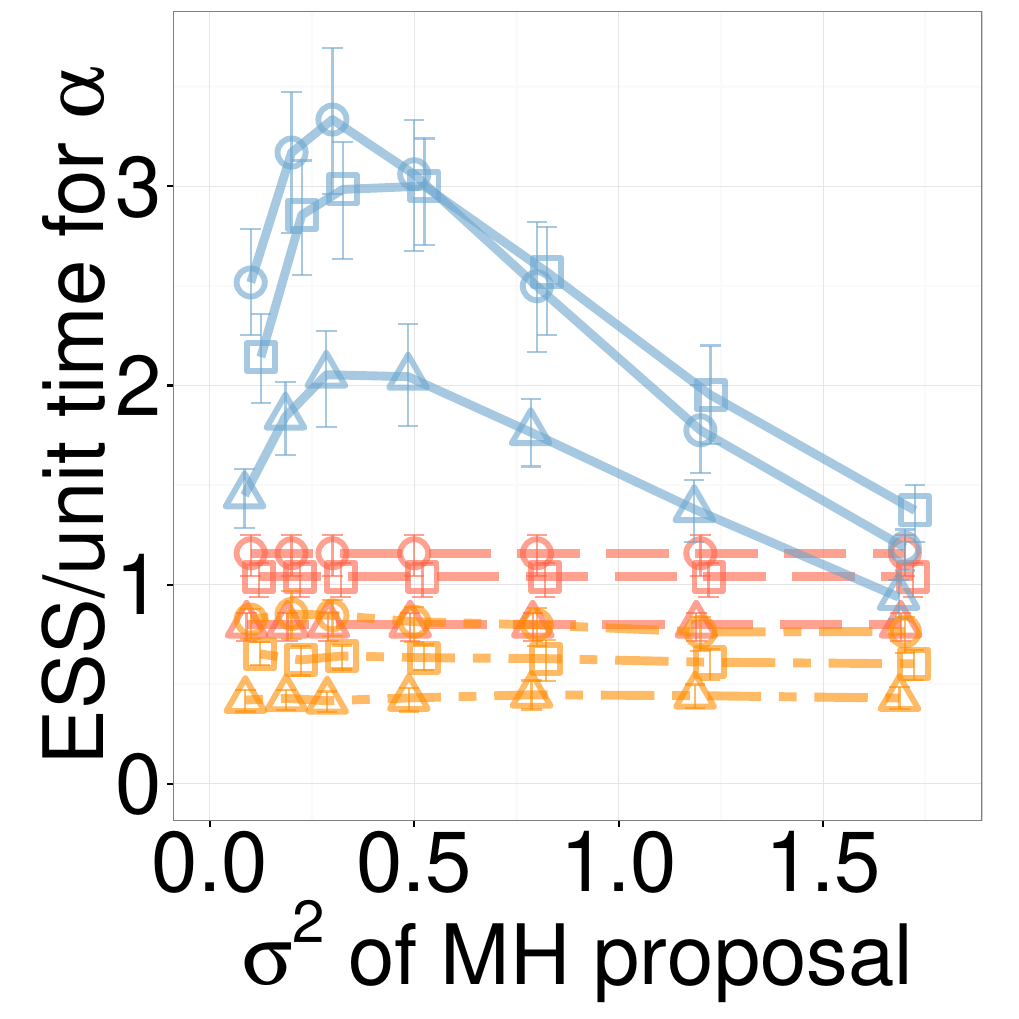}
	\hspace{.6in}
    \includegraphics [width=0.24\textwidth, angle=0]{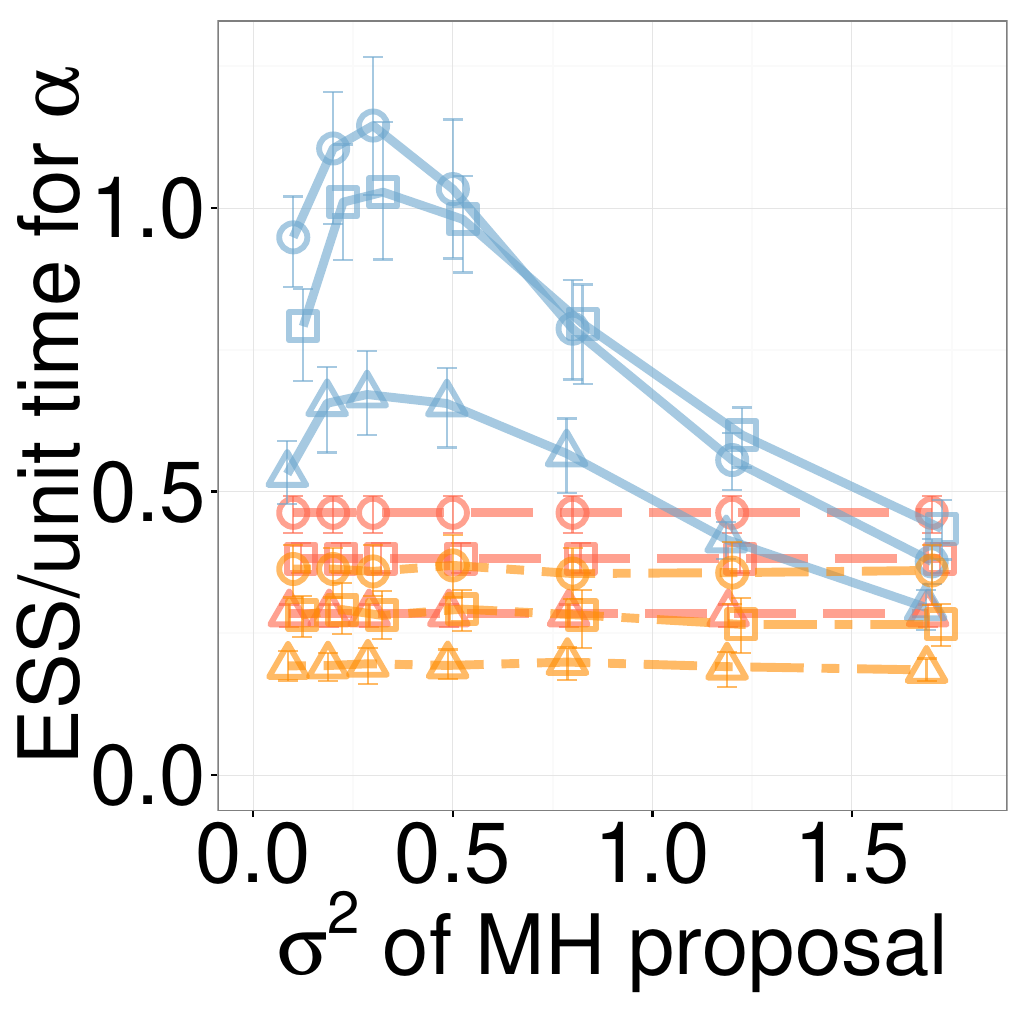}
	\hspace{.6in}
    \includegraphics [width=0.24\textwidth, angle=0]{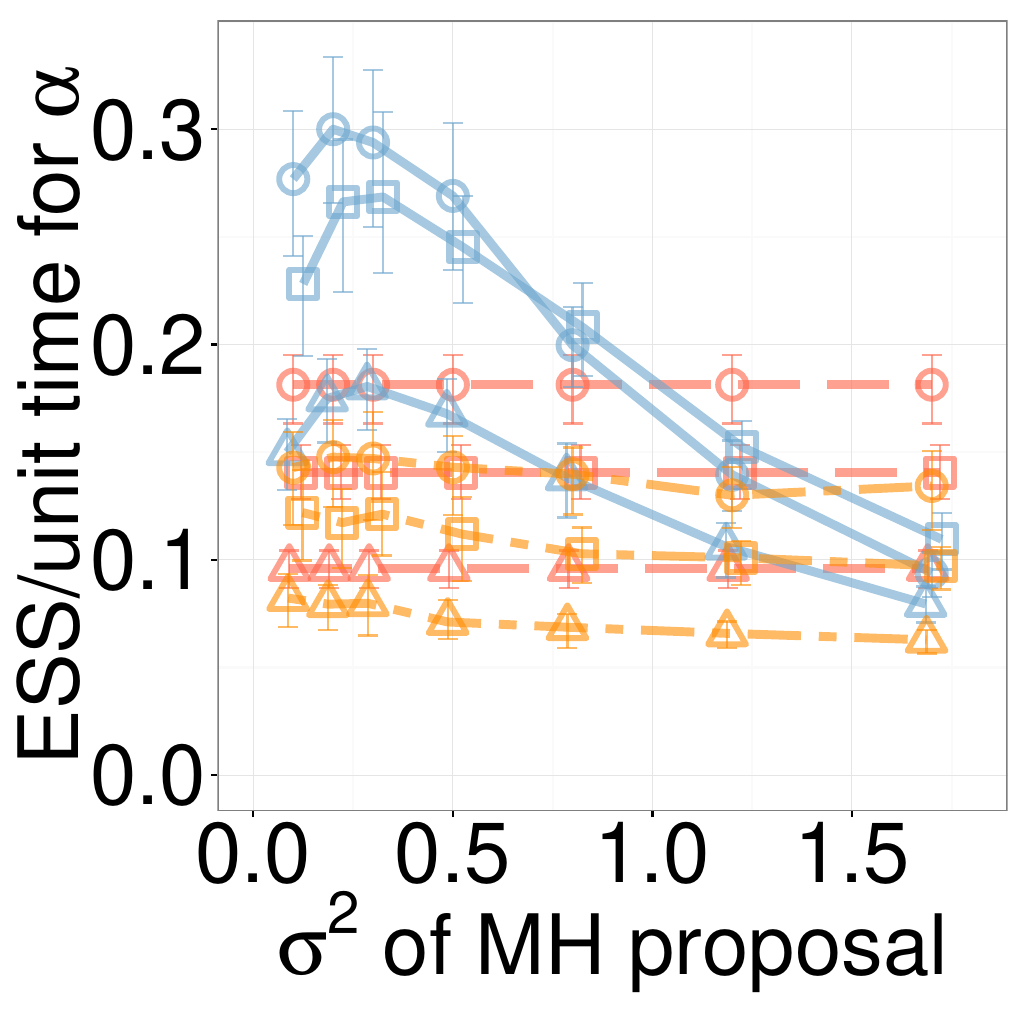}
  \end{minipage}

  \begin{minipage}[h!]{0.99\linewidth}
    \includegraphics [width=0.24\textwidth, angle=0]{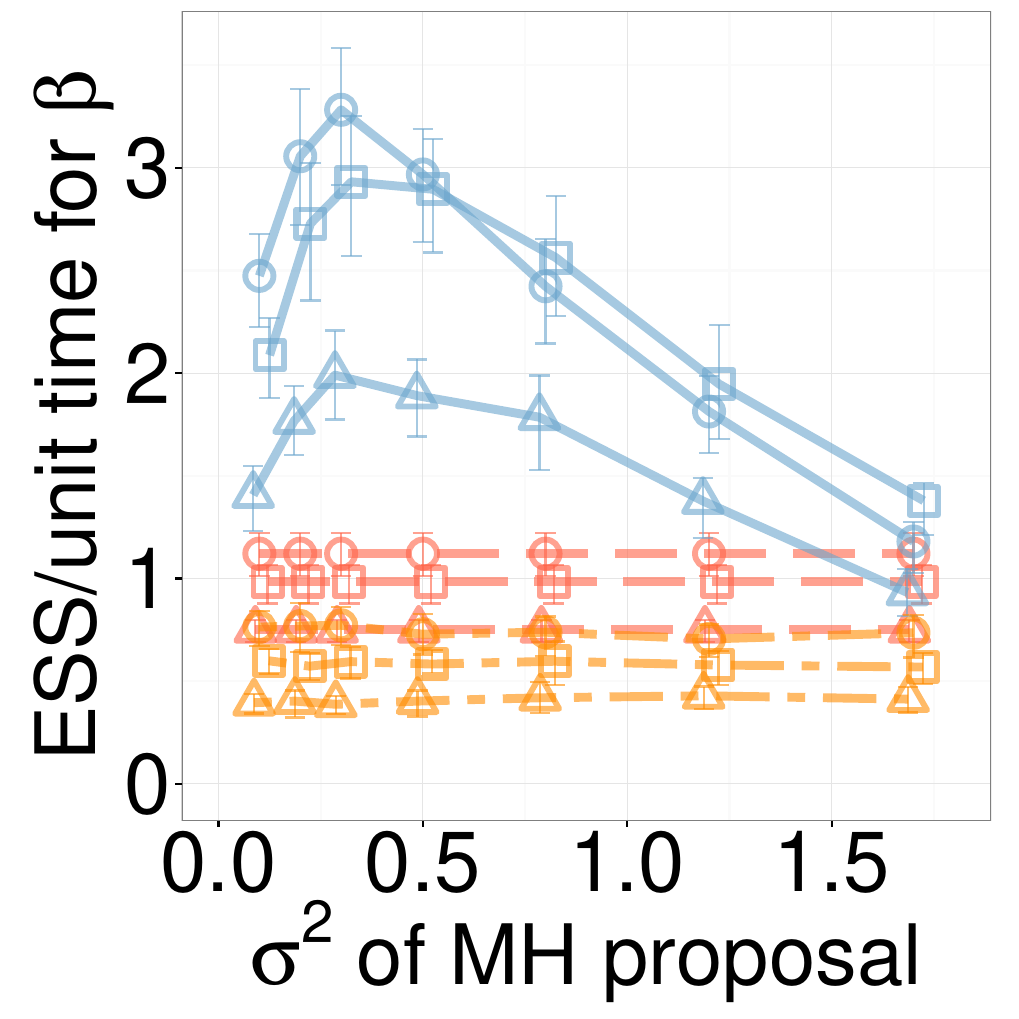}
	\hspace{.6in}
    \includegraphics [width=0.24\textwidth, angle=0]{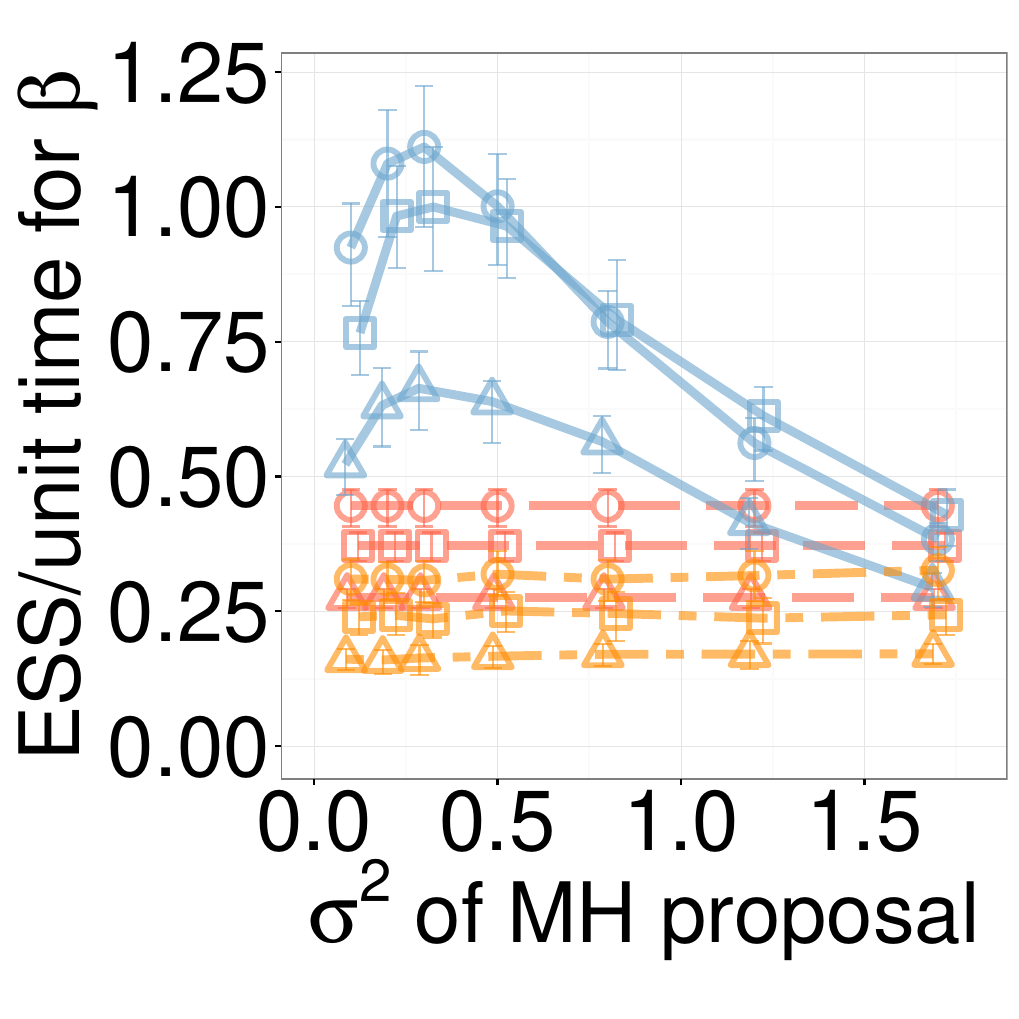}
	\hspace{.6in}
    \includegraphics [width=0.24\textwidth, angle=0]{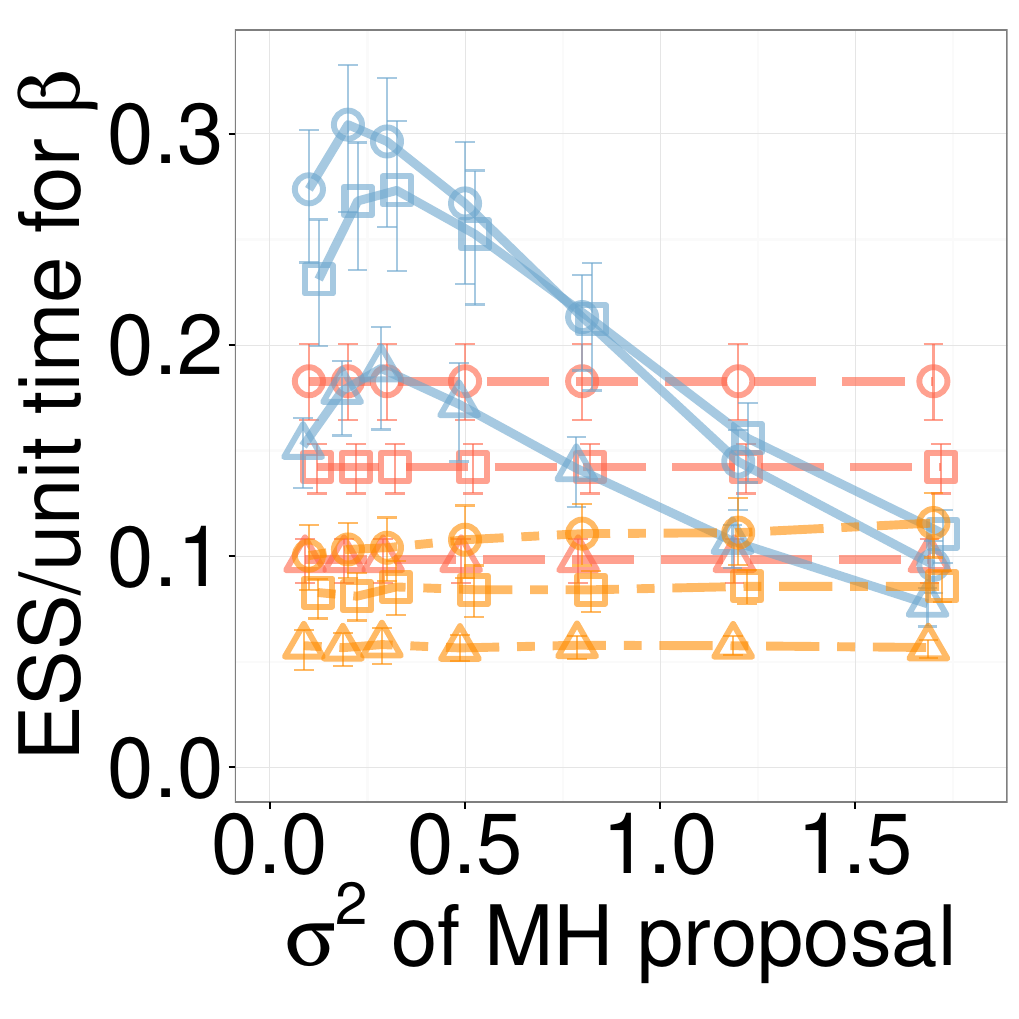}
  \end{minipage}
    \caption{ESS/sec for the time-inhomogeneous immigration model, the top three are for $\alpha$ for 3 states, 5 states, and 10 states. The bottom three are for $\beta$ for 3 states, 5 states, and 10 states. Blue, yellow, and red are the symmetrized MH, \naive\ MH, Gibbs algorithm. Squares, circles and trianges correspond to $\Omega(\theta,\vartheta)$ set to $(\max_s A_s(\theta) + \max_s A_s(\vartheta))$, $\max(\max_s A_s(\theta), \max_s A_s(\vartheta))$ and  $1.5(\max_s A_s(\theta) + \max_s A_s(\vartheta))$.}
     \label{fig:ESS_CQ}
  \end{figure}

	\begin{figure}	
  \centering
  \begin{minipage}[!hp]{0.25\linewidth}
  \centering
    \includegraphics [width=0.99\textwidth, angle=0]{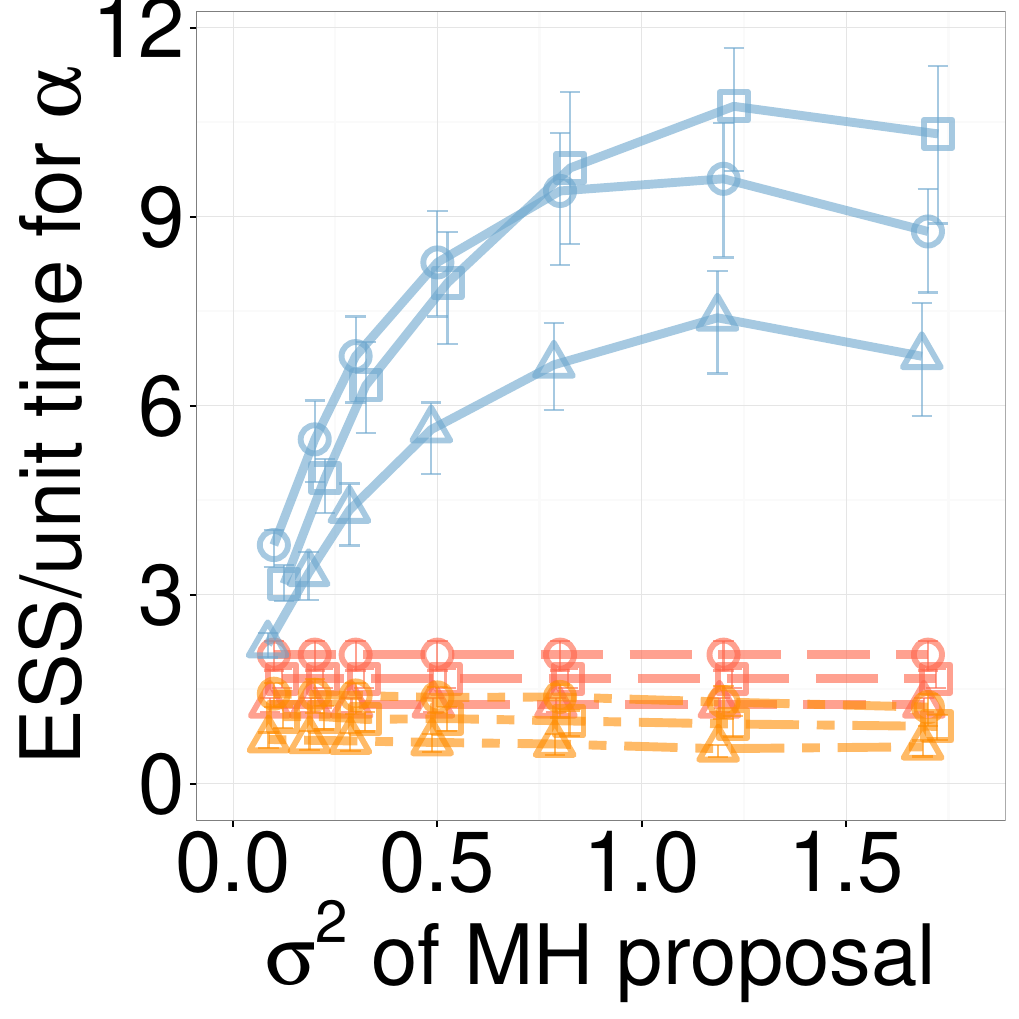}
  \end{minipage}
  \begin{minipage}[!hp]{0.74\linewidth}
    \caption{ESS/sec for the JC immigration model, blue, yellow and red curves are the symmetrized MH,
  \naive\ MH, and Gibbs algorithm.}
     \label{fig:ESS_pc_55}
  \end{minipage}

  \end{figure}

  \begin{figure}[H]
  \centering
  \begin{minipage}[!hp]{0.99\linewidth}
    \includegraphics [width=0.24\textwidth, angle=0]{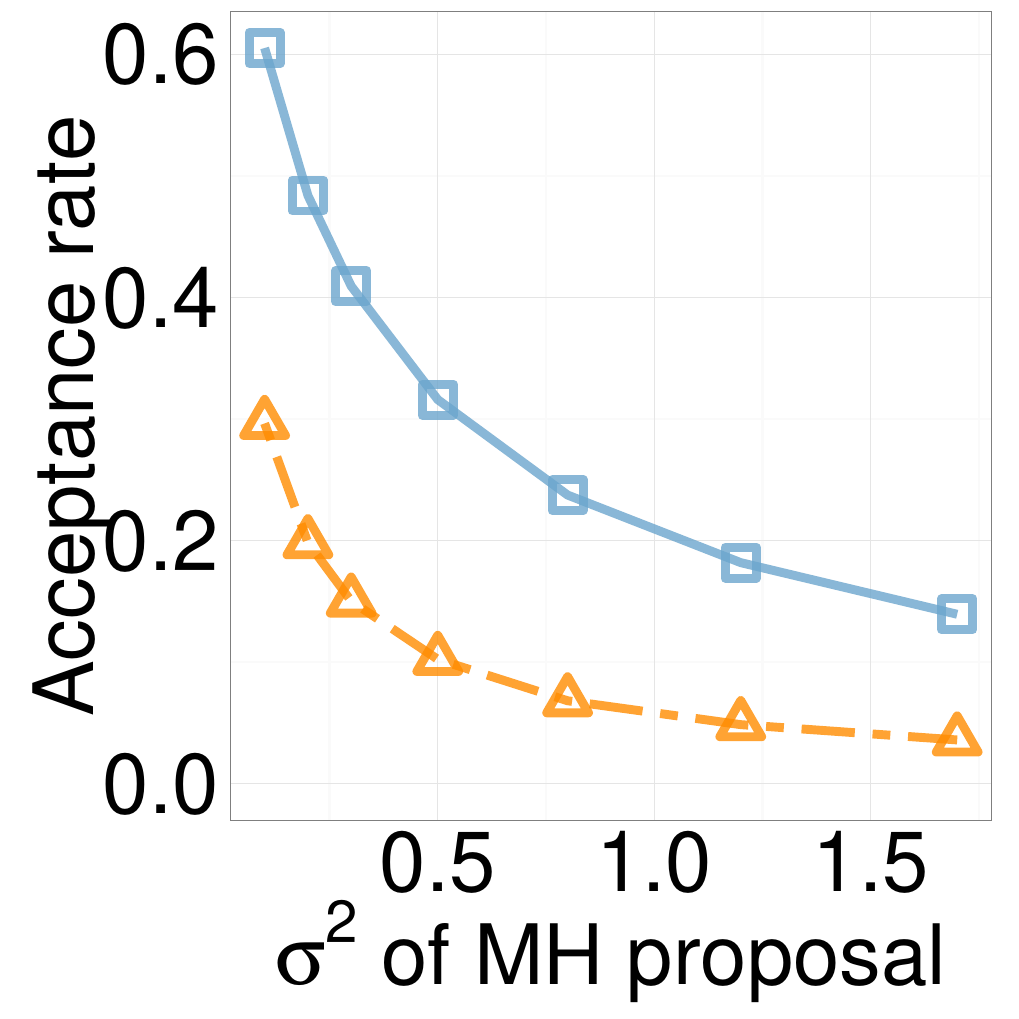}
    \includegraphics [width=0.24\textwidth, angle=0]{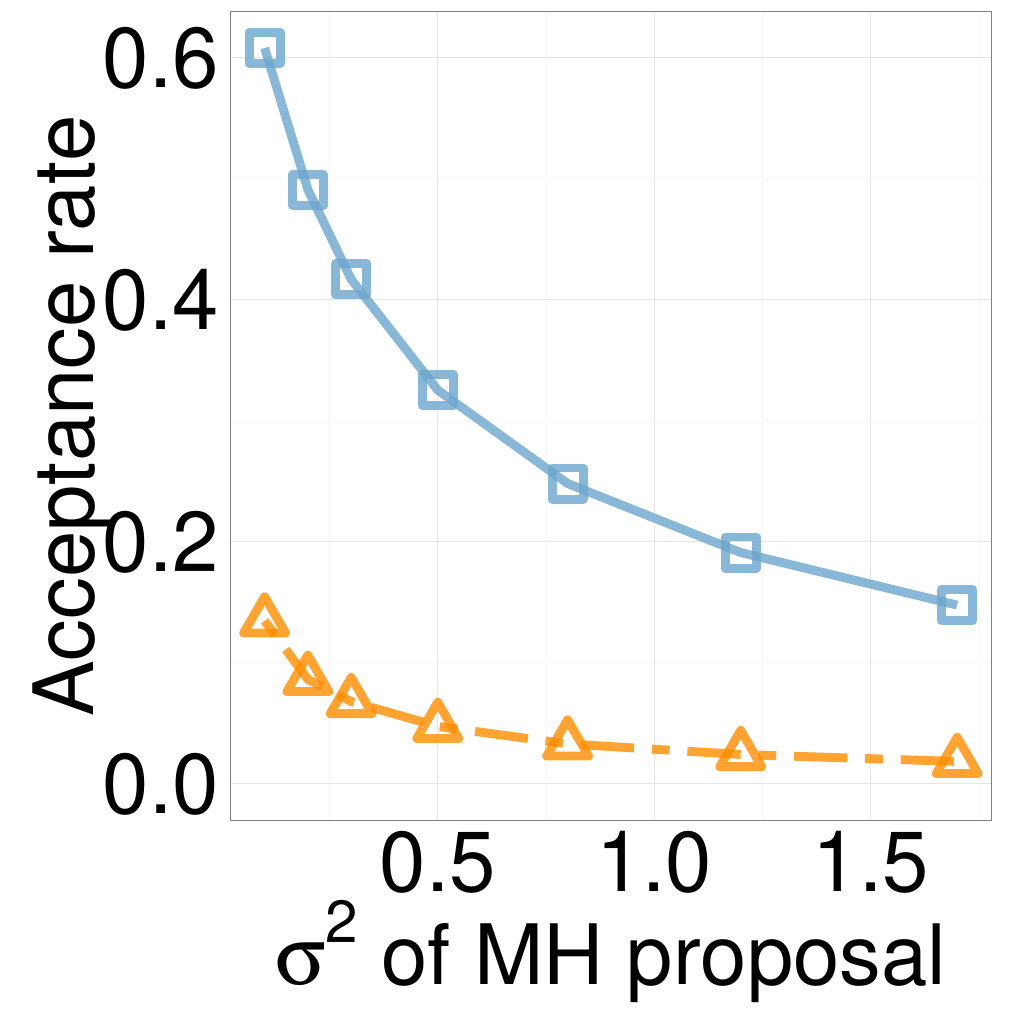}
    \includegraphics [width=0.24\textwidth, angle=0]{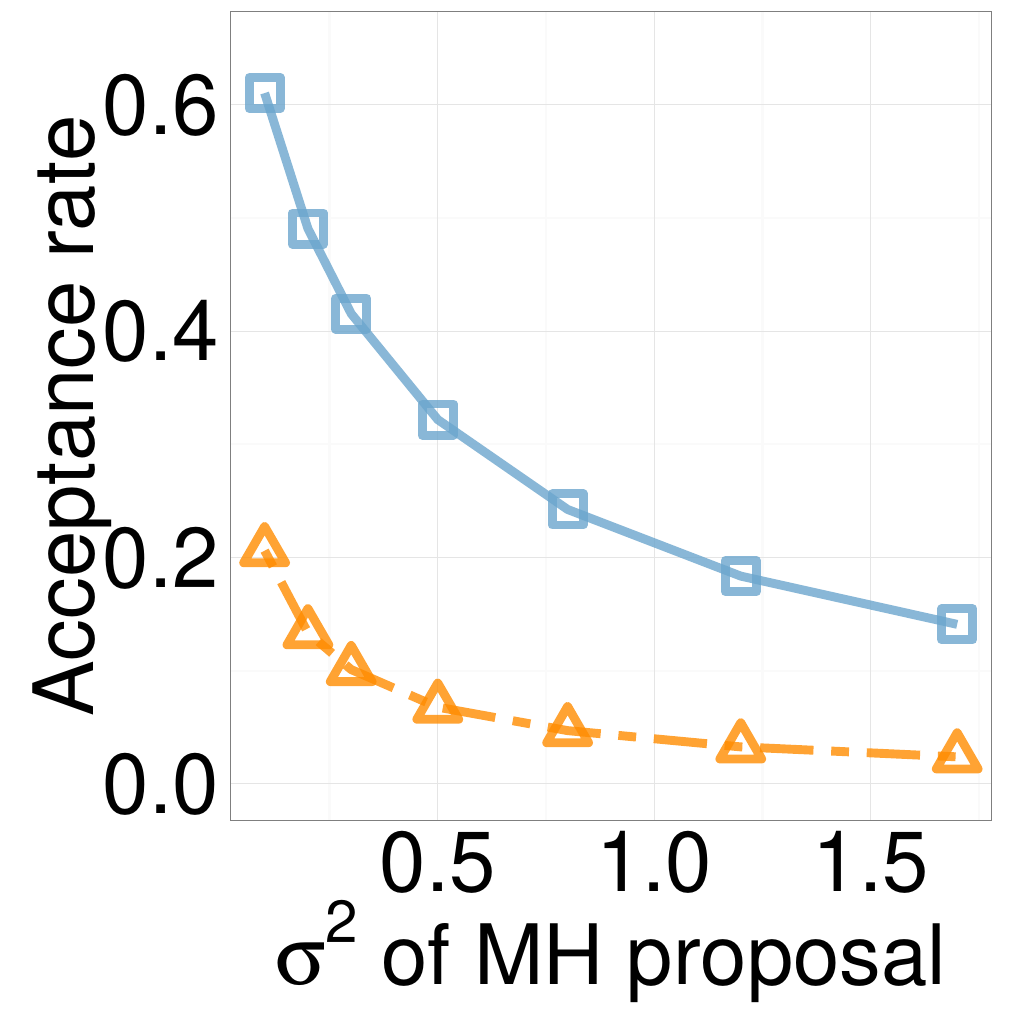}
    \includegraphics [width=0.24\textwidth, angle=0]{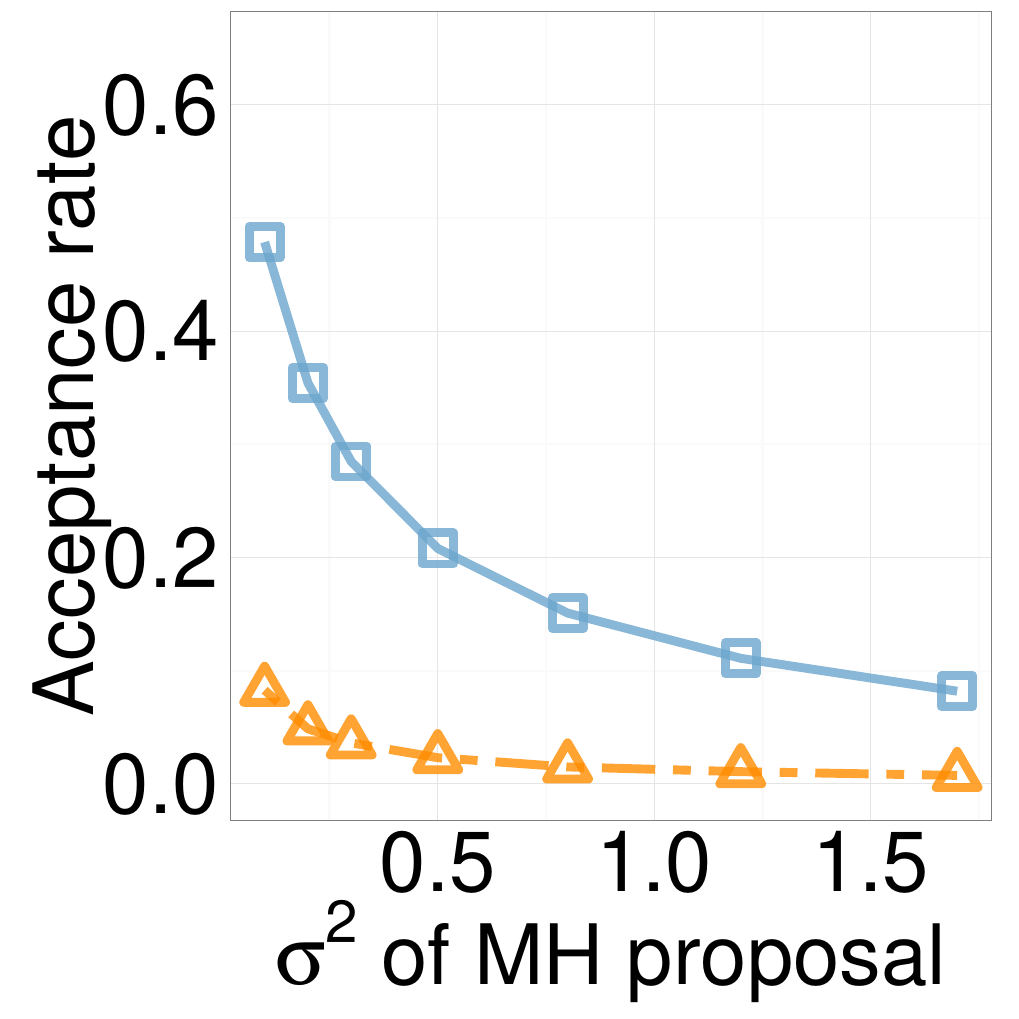}
  \end{minipage}
    \caption{Acceptance Rate for $\alpha$ in the immigration model (left two) and time-inhomogeneous immigration model (right two) , the left two being dimension 3, and the right,dimension 10 and the right two being dimension 3, and the right,dimension 10.  Blue square and yellow triangle curves represent symmetrized MH,
 and \naive\ MH  algorithm. The multiplicative factor is $2$. }
     \label{fig:ACC_Q}
  \end{figure}


  \begin{figure}[H]
  \centering

  \begin{minipage}[!hp]{0.49\linewidth}
    \includegraphics [width=0.49\textwidth, angle=0]{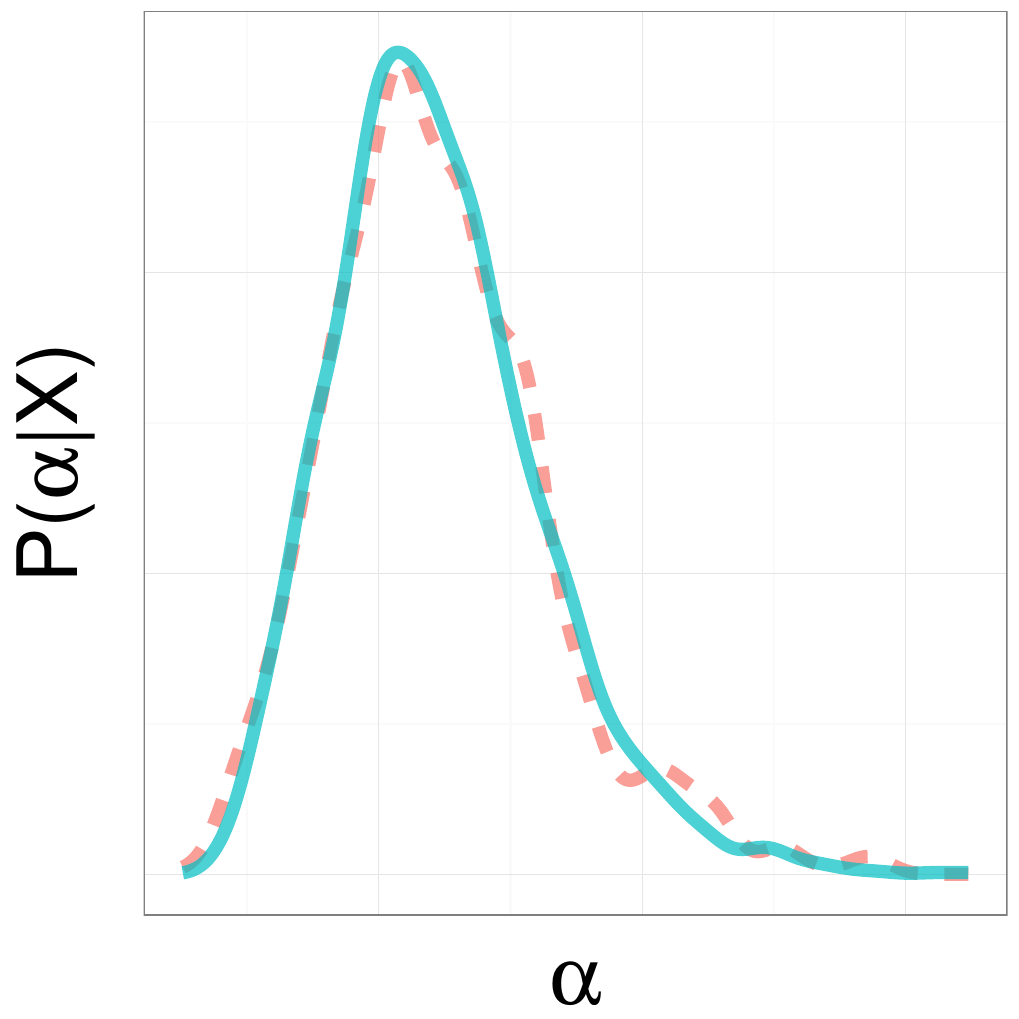}
    \includegraphics [width=0.49\textwidth, angle=0]{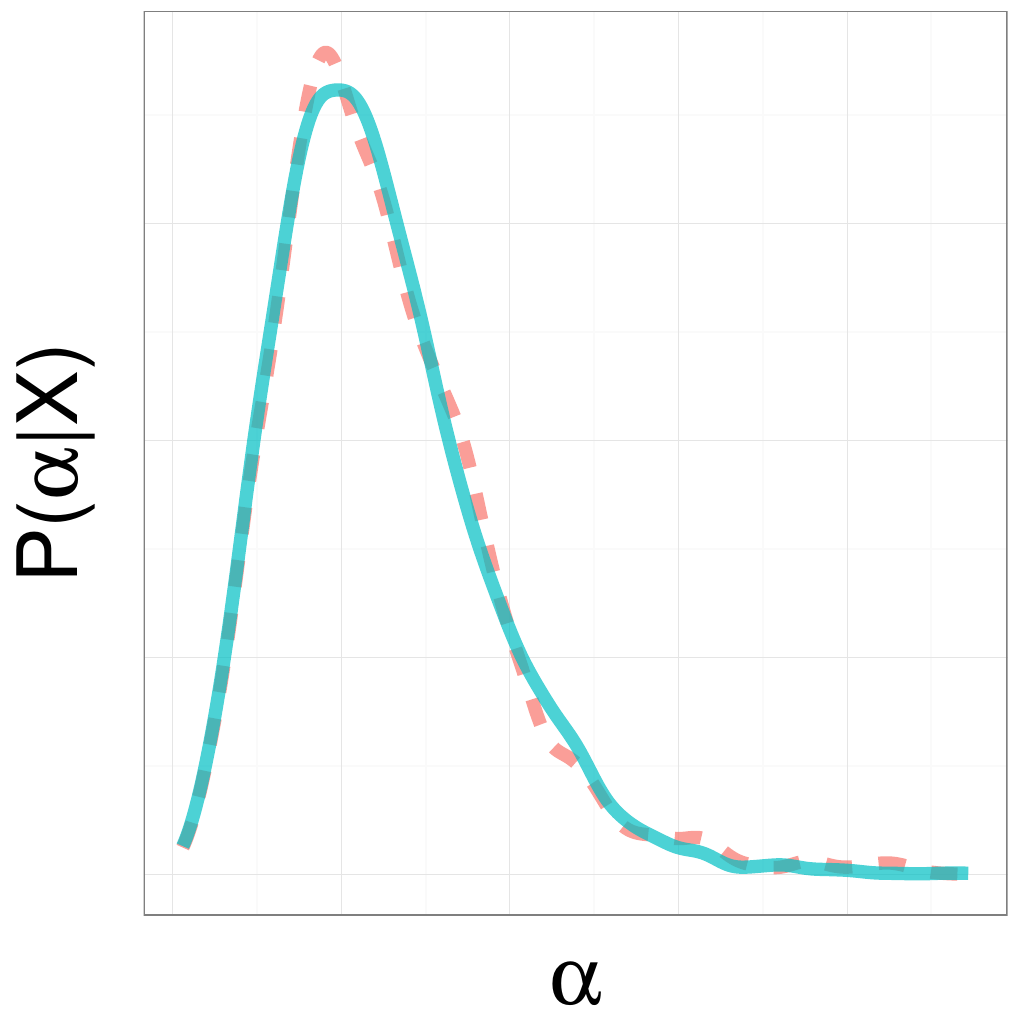}
  \end{minipage}
  \begin{minipage}[!hp]{0.49\linewidth}
    \caption{Posterior $P(\alpha|X)$ from Gibbs (dashed line) and symmetrized MH (solid line) for the immigration model(Left), and time-inhomogeneous immigration model(right)}
     \label{fig:HIST_QCQ}
  \end{minipage}
  \end{figure}

  \begin{figure}[H]
  \centering

  \begin{minipage}[!hp]{0.99\linewidth}
    \includegraphics [width=0.24\textwidth, angle=0]{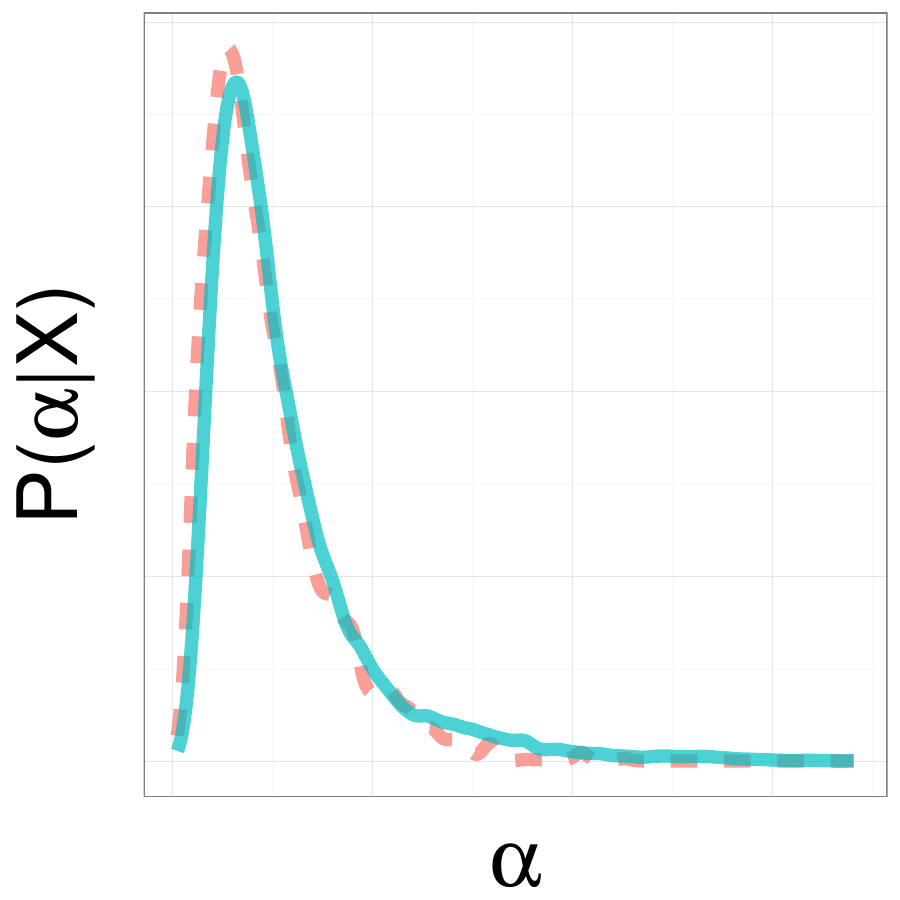}
	\hspace{.6in}
    \includegraphics [width=0.24\textwidth, angle=0]{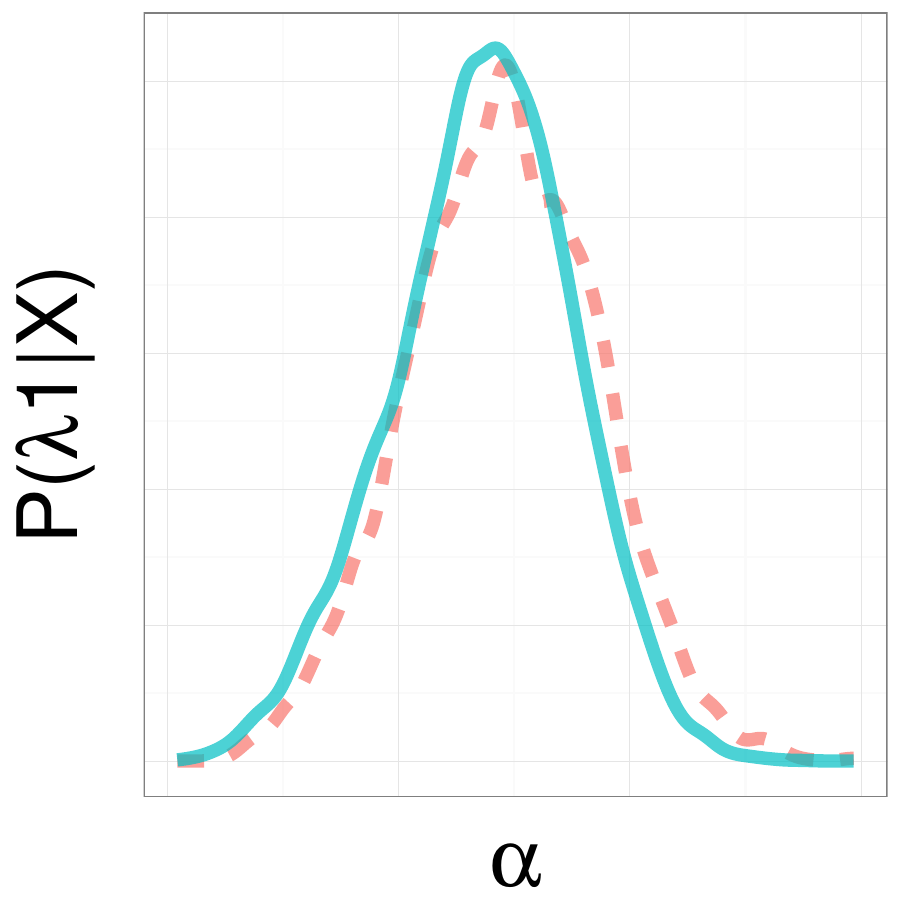}
	\hspace{.6in}
    \includegraphics [width=0.24\textwidth, angle=0]{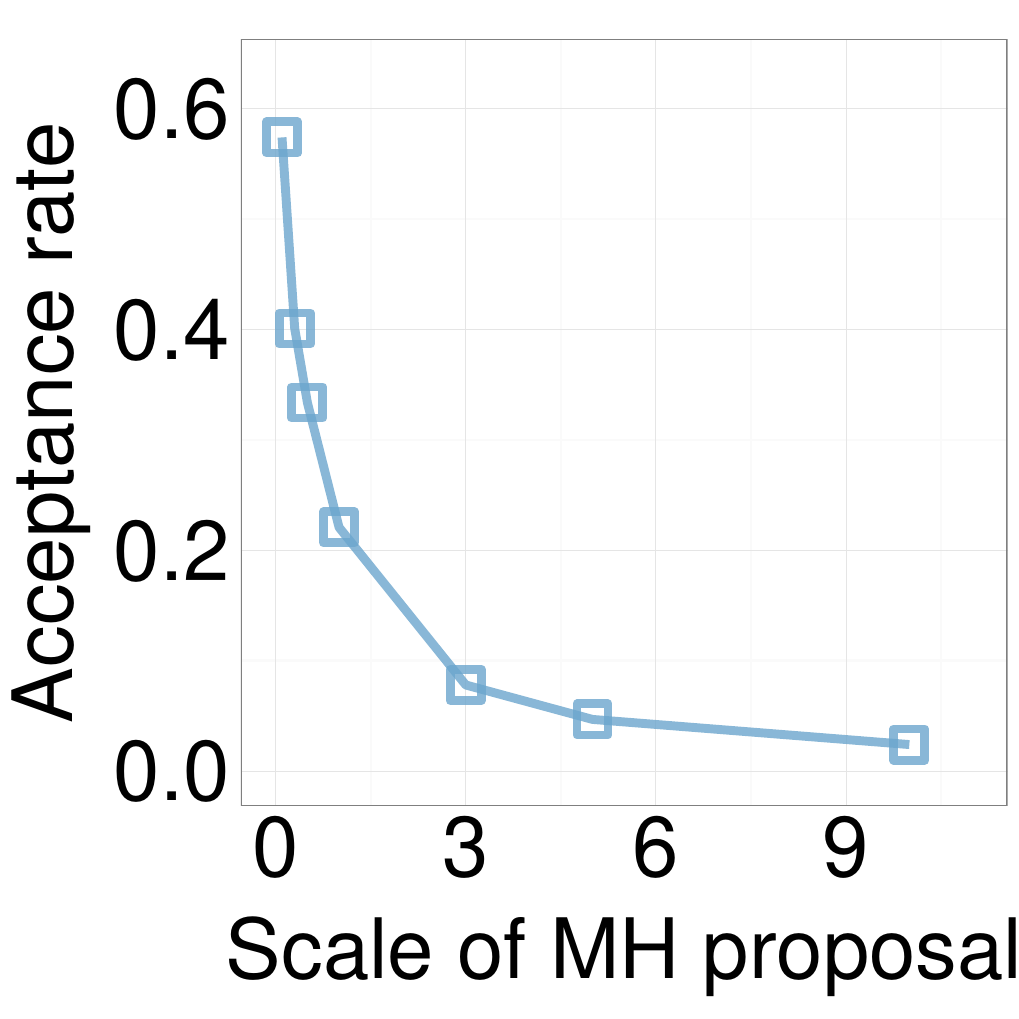}

  \end{minipage}
    \caption{Posterior $P(\alpha|X)$ (left) and $P(\lambda_1|X)$ (middle) from Gibbs (dashed line) and symmetrized MH (solid line) for the E.\ Coli data. Acceptance Rate of $\alpha$ generated by the symmetrized MH algorithm for the E.\ Coli data . The multiplicative factor is $2$. }
     \label{fig:ACC_ECOLI}
  \end{figure}

\end{document}